\pdfoutput=1
\documentclass[acmsmall,timestamp,nonacm=true,screen]{acmart}

\AtBeginDocument{%
  \providecommand\BibTeX{{%
    \normalfont B\kern-0.5em{\scshape i\kern-0.25em b}\kern-0.8em\TeX}}}

\makeatletter
\def\blfootnote{\gdef\@thefnmark{}\@footnotetext}
\makeatother

\newcommand{\vargamma}{\bar\gamma}
\newcommand{\varf}{\bar f}
\newcommand{\vara}{\bar a}
\newcommand{\varlb}{\bar \lambda}
\newcommand{\varxi}{\bar \xi}
\newcommand{\polDelta}{\Delta}
 
\newcommand{\polb}{w}
\newcommand{\polc}{z}
\newcommand{\polbU}{\Gamma}

\newcommand{\polp}{g} 
\newcommand{\polr}{r} 
\newcommand{\polP}{G}
\newcommand{\mult}{\ell} 

\newcommand{\bigO}[1]{\mathchoice{O\left(#1\right)}{O(#1)}{O(#1)}{O(#1)}} 
\newcommand{\softO}[1]{\mathchoice{\tilde{O}\left(#1\right)}{\tilde{O}(#1)}{\tilde{O}(#1)}{\tilde{O}(#1)}} 

\newcommand{\bicost}[1]{c(n,m,#1)}
\newcommand{\bicostrecall}{with~\(c(\cdot)\) from \cref{eq:bicost}}

\newcommand{\ZZ}{\mathbb{Z}} 
\newcommand{\NN}{\mathbb{N}} 
\newcommand{\NNp}{\mathbb{N}_{> 0}} 
\newcommand{\card}[1]{\mathrm{card}(#1)}
\newcommand{\genBy}[1]{\langle #1 \rangle} 
\newcommand{\Span}{\operatorname{Span}} 

\newcommand{\field}{\mathbb{K}} 
\newcommand{\Kbar}{\overline\field}

\newcommand{\ring}{\mathbb{A}}
\newcommand{\bigfield}{\mathbb{L}} 
\newcommand{\quotientL}{\mathbb{L}} 
\newcommand{\matRing}[2]{\field^{#1 \times #2}} 
\newcommand{\vecRing}[1]{\field^{#1}}
\newcommand{\bmatRing}[2]{\Kbar^{#1 \times #2}} 
\newcommand{\bvecRing}[1]{\Kbar^{#1}}
\newcommand{\xRing}{\field[x]} 
\newcommand{\yRing}{\field[y]} 
\newcommand{\ySeries}{\field[\mkern-0.9mu[ y ]\mkern-0.7mu]} 
\newcommand{\xyRing}{\field[x,y]} 
\newcommand{\xmatRing}[2]{\xRing^{#1 \times #2}} 
\newcommand{\xvecRing}[1]{\xRing^{#1}}
\newcommand{\ymatRing}[2]{\yRing^{#1 \times #2}} 
\newcommand{\yvecRing}[1]{\yRing^{#1}}

\makeatletter
\newcommand*{\rem}{%
  \nonscript\mskip-\medmuskip\mkern5mu%
  \mathbin{\operator@font rem}\penalty900\mkern5mu%
  \nonscript\mskip-\medmuskip
}
\makeatother
\newcommand{\rev}{\operatorname{rev}} 
\newcommand{\val}{\operatorname{val}} 
\newcommand{\trsp}[1]{#1^\mathsf{T}} 

\newcommand{\diag}[1]{\mathrm{diag}(#1)}  
\newcommand{\idMat}[1]{\mathrm{I}_{#1}} 
\newcommand{\rank}[1]{\operatorname{rank}(#1)}  
\newcommand{\denom}[1]{\mathcal{D}(#1)} 
\newcommand{\denombig}[1]{\mathcal{D}\!\left(#1\right)} 
\newcommand{\degdet}[1]{\operatorname{deg\,det}(#1)} 

\newcommand{\rmod}[1][m]{\mathcal{M}_{#1}} 
\newcommand{\relmod}[3]{\rmod[#1]^{(#2,#3)}} 
\newcommand{\rmodfa}[1][m]{\relmod{#1}{a}{f}} 
\newcommand{\rmodfg}[1][m]{\relmod{#1}{\gamma}{f}} 
\newcommand{\rmodma}[1][m]{\relmod{#1}{\alpha}{\mu_\gamma}} 
\newcommand{\rmat}[1][m]{R_{#1}} 
\newcommand{\relmat}[3]{\rmat[#1]^{(#2,#3)}} 
\newcommand{\rmatfa}[1][m]{\relmat{#1}{a}{f}} 
\newcommand{\rmatfg}[1][m]{\relmat{#1}{\gamma}{f}} 
\newcommand{\rmatma}[1][m]{\relmat{#1}{\alpha}{\mu_\gamma}}
\newcommand{\dd}[1][m]{\nu_{#1}} 
\newcommand{\detdeg}[3]{\dd[#1]^{(#2,#3)}} 
\newcommand{\ddfa}[1][m]{\detdeg{#1}{a}{f}} 
\newcommand{\ddfg}[1][m]{\detdeg{#1}{\gamma}{f}} 
\newcommand{\ddma}[1][m]{\detdeg{#1}{\alpha}{\mu_\gamma}}
\newcommand{\hk}[1][m,d]{\operatorname{Hk}_{#1}} 
\newcommand{\hankel}[3]{\hk[#1]^{(#2,#3)}} 
\newcommand{\hkfa}[1][m,d]{\hankel{#1}{a}{f}} 
\newcommand{\hankelsimp}[1]{\hk[#1]} 
\newcommand{\hkfasimp}[1][m,d]{\hankelsimp{#1}} 
\newcommand{\hkca}[1][m,d]{\hankel{#1}{\alpha}{\chi_\gamma}} 
\newcommand{\hkma}[1][m,d]{\hankel{#1}{\alpha}{\mu_\gamma}}
\newcommand{\rmodmm}{\rmod[m,m]} 
\newcommand{\rmodlm}{\rmod[\ell,m]} 
\newcommand{\rmodmmfa}{\rmod[m,m]^{(a,f)}} 
\newcommand{\rmodmmfg}{\rmod[m,m]^{(\gamma,f)}} 
\newcommand{\ddmm}{\dd[m,m]} 
\newcommand{\ddlm}{\dd[\ell,m]} 
\newcommand{\ddmmfg}{\dd[m,m]^{(\gamma,f)}} 

\newcommand{\ideal}{\mathcal{I}} 
\newcommand{\quotient}{\mathbb{A}} 
\newcommand{\quotientB}{\mathbb{B}} 
\newcommand{\idealGens}{\genBy{y-a(x),f(x)}} 
\newcommand{\quoF}{\xRing/\genBy{f(x)}} 
\newcommand{\mulmat}[1][a]{M_{#1}} 
\newcommand{\charmat}[1][a]{y\idMat{n} - \mulmat[#1]} 
\newcommand{\charpoly}[1][a]{\chi_{#1}} 
\newcommand{\minpoly}[1][a]{\mu_{#1}} 
\newcommand{\lag}{\mathcal{L}} 

\newcommand{\probainsep}{2n^4}

\newcommand{\flag}{\text{flag}}

\DeclareMathAccent{\myo}{\mathalpha}{operators}{23}

\newif\iftight
\tightfalse
\iftight
\newcommand{\Ramm}{K_{d,2m-1}{\mathstrut}^{\hspace{-12pt}(a,f)}} 
\newcommand{\La}[1][d]{L_{#1,m}{\mathstrut}^{\hspace{-12pt}(a,f)}} 
\newcommand{\Ra}[1][d]{K_{#1,m}{\mathstrut}^{\hspace{-12pt}(a,f)}} 
\newcommand{\Rama}[1][d]{K_{#1,m}{\mathstrut}^{\hspace{-12pt}(\alpha,\mu_{\gamma})}} 
\newcommand{\Lasimp}[1][d]{L_{#1,m}} 
\newcommand{\Rasimp}[1][d]{K_{#1,m}} 
\newcommand{\Laone}[1][n]{L_{1,#1}{\mathstrut}^{\hspace{-9pt}(a,f)}} 
\newcommand{\Raone}[1][n]{K_{1,#1}{\mathstrut}^{\hspace{-9pt}(a,f)}} 
\newcommand{\rhoa}[1][m,d]{\kappa_{#1}{\mathstrut}^{\hspace{-10pt}(a,f)}}
\newcommand{\lba}[1][m,d]{\lambda_{#1}{\mathstrut}^{\hspace{-10pt}(a,f)}}
\newcommand{\rhoasimp}[1][m,d]{\kappa_{#1}}
\newcommand{\lbasimp}[1][m,d]{\lambda_{#1}}
\else
\newcommand{\Ramm}{K_{d,2m-1}^{(a,f)}} 
\newcommand{\Lasimp}[1][d]{L_{{m,#1}}} 
\newcommand{\Rasimp}[1][d]{K_{m,#1}} 
\newcommand{\La}[1][d]{L_{{m,#1}}^{{(a,f)}}} 
\newcommand{\Ra}[1][d]{K_{m,#1}^{(a,f)}} 
\newcommand{\Rama}[1][d]{K_{m,#1}^{(\alpha,\mu_{\gamma})}} 
\newcommand{\Laone}[1][n]{L_{1,#1}^{(a,f)}} 
\newcommand{\Raone}[1][n]{K_{1,#1}^{(a,f)}} 
\newcommand{\rhoa}[1][m,d]{\kappa_{#1}^{(a,f)}} 
\newcommand{\lba}[1][m,d]{\lambda_{#1}^{(a,f)}} 
\newcommand{\rhoasimp}[1][m,d]{\kappa_{#1}} 
\newcommand{\lbasimp}[1][m,d]{\lambda_{#1}} 
\fi

\usepackage{algorithm}
\usepackage{chngcntr}
\counterwithin{algorithm}{section}

\usepackage[noend]{algpseudocode}
\algrenewcomment[1]{\hfill \(\triangleright\) \emph{\small #1}} 

\algnewcommand{\algorithmicand}{\textbf{and}}
\algnewcommand{\algorithmicor}{\textbf{or}}
\algnewcommand{\FOR}{\algorithmicfor}
\algnewcommand{\OR}{\algorithmicor}
\algnewcommand{\AND}{\algorithmicand}
\algnewcommand{\IF}{\algorithmicif}
\algnewcommand{\THEN}{\algorithmicthen}
\algnewcommand{\ELSE}{\algorithmicelse}
\algnewcommand{\Fail}{\textsc{Fail}}
\algnewcommand{\Cert}{\textsc{Cert}}
\algnewcommand{\NoCert}{\textsc{NoCert}}
\algnewcommand{\Flag}{\textsc{Flag}}

\algnewcommand{\CommentLine}[1]{\(\triangleright\) \emph{\small #1}}
\algnewcommand{\InlineFor}[2]{\algorithmicfor\ #1\ \algorithmicdo\ #2} 
\algnewcommand{\InlineIf}[2]{
  \algorithmicif\ #1\ \algorithmicthen\ #2}
\algnewcommand{\InlineIfElse}[3]{
  \algorithmicif\ #1\ \algorithmicthen\ #2\ \algorithmicelse\ #3}

\makeatletter
\newcounter{algorithmicH}
\let\oldalgorithmic\algorithmic
\renewcommand{\algorithmic}{%
  \stepcounter{algorithmicH}
  \oldalgorithmic}
\renewcommand{\theHALG@line}{ALG@line.\thealgorithmicH.\arabic{ALG@line}}
\makeatother

\algrenewcommand\Call[2]{\nameref{#1}\ifthenelse{\equal{#2}{}}{}{\ensuremath{(#2)}}}%
\newcommand{\algoName}[1]{Algorithm \nameref{#1}}
\newcommand{\algoCaptionLabel}[2]{%
     \caption[\textproc{#1}]{\textproc{#1}\ifthenelse{\equal{#2}{}}{}{$(#2)$}}%
     \label{algo:#1}%
     }%

\usepackage{bbm} 
\usepackage{mathdots} 
\usepackage[shortlabels]{enumitem}

\usepackage[capitalize]{cleveref}
\crefname{problem}{Problem}{Problems}
\Crefname{problem}{Problem}{Problems}
\crefname{line}{Step}{Steps}
\Crefname{line}{Step}{Steps}
\creflabelformat{enumi}{(#2#1#3)}

\usepackage{amsthm} 
\newtheorem*{theorem*}{Theorem}
\newtheorem{theorem}{Theorem}[section] 
\newtheorem{remark}[theorem]{Remark}

\usepackage{mathtools}

\begin{document}

\title{Faster Modular Composition}

\author{Vincent Neiger}
\affiliation{%
   \institution{Sorbonne Universit\' e}
   \department{Laboratoire LIP6 UMR 7606 CNRS, Sorbonne Universit\' e}
   \city{Paris}
   \country{France}
   }
   \email{vincent.neiger@lip6.fr}

\author{Bruno Salvy}
\affiliation{%
\institution{Inria}
\department{Laboratoire LIP UMR 5668  Univ. Lyon, CNRS, ENS de Lyon, Inria, UCBL}
\city{Lyon}
\country{France}
}
\email{bruno.salvy@inria.fr}

\author{\'Eric Schost}
 \affiliation{%
   \institution{University of Waterloo}
   \department{Cheriton School of Computer Science}
   \city{Waterloo, ON}
   \postcode{N2L 3G1}
   \country{Canada}
 }
   \email{eric.schost@uwaterloo.ca}

\author{Gilles Villard}
 \affiliation{%
   \institution{CNRS}
   \department{Laboratoire LIP UMR 5668  Univ. Lyon, CNRS, ENS de Lyon, Inria, UCBL}
   \city{Lyon}
   \country{France}
   }
   \email{gilles.villard@cnrs.fr}


\begin{abstract}
  A new Las Vegas algorithm is presented for the composition of two
  polynomials modulo a third one, over an arbitrary field. When the
  degrees of these polynomials are bounded by~$n$, the algorithm
  uses~$\bigO{n^{1.43}}$ field operations, breaking through the $3/2$
  barrier in the exponent for the first time. The previous fastest
  algebraic algorithms, due to Brent and Kung in 1978,
  require~$\bigO{n^{1.63}}$ field operations in general,
  and~${n^{3/2+o(1)}}$ field operations in the special case of
  power series over a field of large enough characteristic. If 
  cubic-time matrix multiplication is used, the new algorithm runs
  in~${n^{5/3+o(1)}}$ operations, while previous ones run
  in~$\bigO{n^2}$ operations.

  Our approach relies on the computation of a matrix of algebraic
  relations that is typically of small size. Randomization is
  used to reduce arbitrary input to this favorable situation.  
\end{abstract}

\ccsdesc[500]{Computing methodologies~Algebraic algorithms}
\ccsdesc[500]{Theory of computation~Algebraic complexity theory}


\maketitle


\section{Introduction}
\label{sec:intro}

\blfootnote{Implementations of
  \Crefrange{algo:ModularComposition-BrentKung}{algo:ModularCompositionBaseCase}
  are available at
  \url{https://github.com/vneiger/faster_modular_composition_SageMath},
  based on the SageMath software (version $\ge$ 9.4 is required and is freely available
  at \url{https://www.sagemath.org/}).
}

\subsection{Problem and Result}
Many fundamental operations over univariate polynomials of degree at
most~$n$ with coefficients in a commutative ring~$\ring$ can be
computed in a number of arithmetic operations in~$\ring$ that is 
quasi-linear in $n$~\cite{GaGe99}. It is the case for multiplication, division with
remainder by a monic polynomial, multipoint evaluation, interpolation
at points whose differences are units in~$\ring$, and greatest common
divisors when~$\ring$ is a field.

In contrast with these operations, improving the cost bound for
\emph{modular composition} is a longstanding open question.  Given
three polynomials $a,f \in \ring[x]$ and $\polp\in \ring[y]$, with
$\deg(a) < n$ and $\deg(\polp) < n$ where \(n=\deg(f)\), and with $f$
monic, this problem is to compute $\polp(a) \rem f$, where the ``$\rem$''
operation takes the remainder of the Euclidean division.

\paragraph{Motivation}
This operation arises in a variety of contexts. For instance, with
$f=x^n$, it amounts to power series composition. For many
applications of power series, composition is the bottleneck. This
is the case for power series
reversion, that can then be reduced to composition with a small
overhead~\cite{BK78}. This is also the case of further
operations such as solving
families of functional equations~\cite{Hoeven02}. 

The application of certain algebra morphisms also translates to
modular composition. Over a field $\field$, for $f$ and $a$ in
$\field[x]$, we denote by $a \bmod f \in \field[x]/\genBy{f}$ the
class of $a$ modulo $f$. Then, for $e$ and $f$ in respectively
$\field[y]$ and $\field[x]$, and for a $\field$-algebra morphism
$\phi:\field[y]/\genBy{e} \to \field[x]/\genBy{f}$, if $\phi(y \bmod
e)=a \bmod f$ then for $g$ in $\field[y]$, the image $\phi(g \bmod e)$
is equal to $g(a) \bmod f$.

Over finite fields, with $e$ and $f$ the same polynomial and $\phi$ the Frobenius endomorphism,
this results in modular composition playing an important role in
algorithms for polynomial
factorization~\cite{GathenShoup1992,KaSh98,KaSh97}. Dedicated
algorithms exist for modular composition over finite fields, with
quasi-linear complexity (they are discussed later). 
Still, there
remains a
variety of questions that can be considered over arbitrary fields, and
which are impacted by modular composition (or closely related
operations such as power projection, discussed later as well):
computing the minimal polynomial of an algebraic
number~\cite{Shoup94,Shoup95,Shoup99}, normal bases
computations~\cite{KaSh98,GiJaSc21}, arithmetic operations with two
algebraic numbers~\cite{BFSS06}, computing with towers of algebraic
extensions~\cite{PS13,PoSc13b,HoeLec19}, Riemann-Roch space
computations~\cite{AbCoLe20,AbCoLe21}, etc.

\paragraph{Previous algorithms}
The most famous algorithm in this area is that of Kedlaya and
Umans, which achieves complexity $n^{1+\epsilon}\log^
{1+o(1)}(q)$ \emph{bit operations} for any given \(\epsilon>0\)~\cite[Cor.~7.2]{KU11} when the
field~$\mathbb K$ is the finite field~$\mathbb F_q$. In contrast, we
deal with an arbitrary field \(\field\) and count \emph{arithmetic operations} in
\(\field\). In this context, the known algorithms have much
higher complexity estimates.

Modular composition can be performed using Horner's algorithm with
modular reduction at each stage, which leads to a complexity
in~$\softO{n^2}$ operations if fast polynomial multiplication is
used. The notation
$c'=\softO{c}$ means that $c'=\bigO{c\log^k(c)}$ for some $k>0$; in
other words, logarithmic factors are dropped. 

In~1978, Brent and
  Kung gave two algorithms that perform
composition modulo $x^n$ (the case of power
series)~\cite{BrentKung1977,BK78}. One relies strongly on Taylor
expansion and runs in~$\softO{n^{3/2}}$ operations; the other one,
using a baby steps/giant steps approach, uses
$\bigO{n^{(\omega+1)/2}}+\softO{n^{3/2}}$ operations, where $\omega\le
3$ is a feasible matrix multiplication exponent (two $n\times n$
matrices can be multiplied in~$\bigO{n^{\omega}}$ operations; the
best known bound is 
$2.371552$~\cite{AlWi21,DWZ22,VWXXZ23}).   This
latter algorithm works verbatim and in the same complexity for
composition modulo an arbitrary polynomial~$f$ of degree~$n$ not
restricted to be~$x^n$~\cite{GathenShoup1992}.
Both these algorithms have remained
essentially the best ones since then. Huang and Pan used fast
rectangular matrix multiplication in the central step of the baby
steps/giant steps algorithm to reduce its complexity
to~$\bigO{n^{\omega_2/2}}+\softO{n^{3/2}}$ \cite{HuPa98}, where $\omega_2\le\omega+1$
is a feasible exponent such that a $n\times n^2$ matrix can be
multiplied by a $n^2\times n$ matrix in $\bigO{n^{\omega_2}}$
operations. The currently best known value gives $\omega_2 \approx
3.250385$~\cite{LGU18,LG23,VWXXZ23}, which makes the previous algebraic complexity
bound~$\bigO{n^{1.626}}$ for modular composition for an
arbitrary~$f$. Even assuming an optimal matrix multiplication, which
means~$\omega=2$, these algorithms do not break the exponent
barrier~$3/2$.

\paragraph{Our result}
The open problem~2.4 in the book of B\"urgisser, Clausen and
Shokrollahi~\cite{BCS97} asks whether Brent and Kung's algorithm can
be improved substantially. The research problem~12.19 in von zur Gathen and
Gerhard's book \cite{GaGe99} asks for a complexity
in~$\softO{n^{1.5}}$ or better. Our main result answers both questions
positively when $\ring$ is a field, with few extra hypotheses.

\begin{theorem}\label{thm:intro}
  Given $a,f \in \xRing$ and $\polp\in \yRing$ with coefficients in
  a field
  $\field$, with \(\deg(f) = n\),
  $\deg(a)$ and $\deg(\polp)$ smaller than $n$, and a tuple $r$ of
  $\bigO{n^{1+1/3}}$ field elements,
  \algoName{algo:ModularComposition}{} returns either $\polp(a)\rem f$
  or {\sc Fail} after~$\softO{n^\kappa}$ arithmetic operations
  in~$\field$, with
  \begin{equation}\label{eq:def-gamma}
    \kappa=1+\frac{1}{\frac{1}{\omega-1}+\frac2{\omega_2-2}}<
    1.43.
  \end{equation}
It returns {\sc Fail} with
  probability at most $(2n^4+18n^2)/\card{S}$
   when the entries of $r$ are chosen uniformly and independently from
   a
  finite subset $S \subseteq \field$.
\end{theorem} 

Here we use an algebraic model of computation: roughly, basic arithmetic
operations \(\{+,-,\times,\div\}\) and zero-tests in the base field~$\field$
are counted at unit cost; for more details, see \cref{sec:preliminaries}.  
As usual with probabilistic algorithms of Las Vegas type, the
algorithm can be repeated until it succeeds, so that only its running
time becomes a random variable.

We assume that the characteristic $p$ of $\field$ is known to the
algorithms. For $\field$ finite and of small cardinality $q$ (namely,
$q
\le 2n^4+18n^2$), the probability statement becomes vacuous. However,
in such cases, one can work in a sufficiently large field, by
constructing an extension of~$\field$ of degree~$\bigO{\log(n)}$
efficiently (see \cite[Sec.\,14.9]{GaGe99} and references therein).
In this extension, each arithmetic operation can be performed in
$\softO{\log(n)}$ arithmetic operations in~$\field$, so that the
asymptotic complexity estimate is unaffected.

We also give a probabilistic algorithm of the Las Vegas type with the
same complexity bound for computing an annihilating polynomial for \(a
\bmod f\), that is, a nonzero polynomial $g \in \field[y]$ such that
$g(a) \rem f=0$.

The improvements brought by fast matrix multiplication on one
hand and by fast rectangular matrix multiplication on the other hand are made
clearer by noting
that the exponent $\kappa$ of composition satisfies 
\[
  \frac43
  \le \kappa = \underbrace{1 + \frac{1}{ \frac{1}{\omega-1} + \frac{2}{\omega_2-2}} }_{< 1.42945}
  \le\underbrace{\frac{\omega+2}{3}}_{< 1.4572}
  \le\underbrace{\frac53}_{<1.666667},
\]
where the first approximation is obtained with the bounds on $\omega$
and $\omega_2$ given above; the second one is
obtained when no fast
rectangular matrix multiplication is used, so that $\omega_2$ simply
becomes $\omega+1$; the last one is obtained when no fast matrix
multiplication is used, thus taking $\omega=3$. In the latter case,
our algorithm is the first subquadratic one for modular
composition. In the other direction, considering the lower bounds
$\omega \ge 2$ and $\omega_2 \ge 3$ shows that $\kappa\ge4/3$, giving
a lower bound on the complexity estimate that can be achieved by
the algorithm designed in this work.

\paragraph{Main steps}
To compute $g(a) \rem f$, our approach relies on first computing a
polynomial $\tilde g$ of ``small degree'' such that $\tilde g(a)\rem
f=g(a)\rem f$. If $a\bmod f$ has a minimal polynomial of small
degree~$\mu(a)$, then one can take $\tilde{g}=g\bmod\mu$. In general,
such a $\mu$ may not exist, and a small degree \emph{univariate} \(\tilde g\)
may not exist either.
However, generically, one can compute a set of
\emph{bivariate}
polynomials~$\mu_i(x,y)$ of ``small degree'' such that $\mu_i
(x,a)\rem f=0$. These are called \emph{relations}.
From these, a small degree bivariate \(\tilde g(x,y)\)
is found via some type of reduction of \(g\) by all \(\mu_i\)'s
simultaneously, ensuring $\tilde g(x,a)\rem f=g(a)\rem f$.
Relations form one of the
main ingredients of the new algorithm, and most of the new technical
results are about them. On the algorithmic side, the
coefficients of these relations are gathered into matrices called
\emph{matrices of relations} and we make heavy use of fast algorithms
on polynomial matrices.

Here and throughout the article,
genericity is understood in the
Zariski sense: a property is generic if it holds outside of a
hypersurface of the corresponding parameter space.
A {\em randomized change of basis} brings $f$ and~$a$ to a
situation where ``small'' matrices of relations
exist. 

This probabilistic
algorithm is proved to
be correct for $f$ separable (i.e., with no repeated roots
in an algebraic closure $\Kbar$ of $\field$) and for $f$
purely inseparable (i.e., with only one root in $\Kbar$).
Modular composition modulo an arbitrary $f$ is reduced to these two
extreme cases by {\em separable decomposition} of $f$~\cite{Lec2008}, Chinese remainder theorem, and
    a slight generalization of a technique called {\em untangling}~\cite{HoeLec17}. The latter allows to transport the composition problem 
    over~$\field$ modulo a factor of the separable decomposition, to a composition problem over a
    quotient algebra with purely inseparable modulus.

\paragraph{Complexity aspects}
Under genericity conditions, $m\times m$ matrices of relations of
``small degree''
$d\leq \lceil n/m \rceil$ are shown to exist (the choice of~$m$
 is optimized below). Their computation
starts from the first~$m$ coefficients of the
  $2md$ polynomials~$x^ia^k\rem f$, for $0\le i<m$ and $0\le k<2d$.
  With
    \begin{equation}
      \label{eq:bicost}
      \bicost{d}={(m+n/d)d^{\omega_2/2}},
    \end{equation}
    these coefficients can be computed in 
  $\softO{m^2d+\bicost{d}}$ operations in \(\field\) (\cref{sec:truncated_pow}).

From these coefficients, a matrix of relations is obtained
by {\em approximant bases}~\cite{GJV03} in~$\softO{m^\omega d}$
operations (\cref{sec:relmat:certify}). Given such a matrix and in the same complexity, linear
system solving over $\xRing$
  \cite{ZLS12} allows us to reduce the univariate $\polp\in\yRing$ to
  a bivariate $\tilde{\polp}\in\xyRing$, of degrees smaller than \(m\)
  and \(d\) in $x$ and $y$, such that
  $\polp(a)\equiv\tilde{\polp}(x,a)\bmod f$ (\cref{sec:bivcomposition}). 

  A generalization due to N\"usken and Ziegler~\cite{NusZie04}
    of Brent and Kung's algorithm to the case of a bivariate
    polynomial~$\tilde \polp(x,y) \in \field[x,y]$
    finally computes the ``uni-bivariate'' composition~$\tilde{\polp}
    (x,a)\rem f$ using
    $\softO{\bicost{d}}$ operations in~$\field$ (\cref{sec:sequence:biv_comp}).
  
Altogether, the costs of the various parts of the algorithm add up to
$\softO{m^\omega d+\bicost{d}}$ operations in $\field$.  Then choosing $m$ and $d= \lceil n/m \rceil$ so as to minimize
$m^\omega d+\bicost{d}$ leads us to $m\sim n^\eta$, where
\begin{equation}\label{eq:def-beta}
  \eta=\frac1{1+\frac{\omega-1}{(\omega_2-2)/2}} \,,
  \end{equation}
which is approximately $0.3131$ with the bounds on~$\omega$
and~$\omega_2$ given above and leads
to to the complexity estimate \(\softO{n^\kappa}\)
of \cref{thm:intro}.  

\subsection{Previous algorithms in special cases}

To compute $\polp(a) \rem f$, previous known improvements upon Brent and Kung's
approach all have requirements on the input, either on some of the polynomials
$f$, $\polp$, and $a$, or on the ring or field of coefficients~---~possibly with
nonalgebraic algorithms. 

\subsubsection{Special modulus}
\label{subsubsec:spemod}

\subsubsection*{Power series}

For the special case $f= x^n$ of power series, Brent and Kung's second
algorithm relying on Taylor expansion performs composition in
only~$\softO{n^{3/2}}$ operations, provided $a'(0)$ and $(\lceil
\sqrt{n\log(n)} \rceil)!$ are invertible in $\ring$; the assumption on
$a'(0)$ can be weakened \cite[Sec.\,3.4.3]{Hoeven02}. In more
variables, even in the specific case $\polp(x, a) \rem f$ handled by the N\"usken-Ziegler
algorithm, we do not know of any algorithm computing composition
faster for power series than modulo arbitrary polynomials.

Faster composition in only~$\softO{n}$ operations for $\polp(a)\rem
x^n$ is possible for many special cases of~$\polp$: when $\polp$ is a
polynomial of degree~$O(1)$, but also when it is a power series
solution of a polynomial equation of degree~$O(1)$ via Newton's
iteration, or when it is a solution of a differential equation (e.g.,
$\exp$), by first forming a differential equation for $\polp(a)$ and
then solving it by Newton's iteration or other divide-and-conquer
algorithms, generally in characteristic~0 or large enough
\citetext{\citealp{Lipson1976,BK78,Hoeven02,BostanChyzakOllivierSalvySchostSedoglavic2007};
  \citealp[\S13.4]{BostanChyzakGiustiLebretonLecerfSalvySchost2017}}.

Similarly, still in the case when~$f=x^n$, if furthermore $a$ has specific
properties, then composition of power series can be performed in~$\softO{n}$
operations. This is the case when $a$ is a polynomial of moderate
degree~\cite{BK78} (it is a
part
of Brent and Kung's fast composition algorithm), an algebraic power
series~\cite{Hoeven02}, but also for a class of truncated power series that can
be obtained via shifts, reversals, scalings, multiplications by polynomials,
exponentials and logarithms~\cite{BostanSalvySchost2008}.

\subsubsection*{Separable polynomials}

Ritzmann observed that for a separable modulus $f(x)=
(x-\epsilon_1)\dotsm(x-\epsilon_n)$ with distinct
$\epsilon_1,\dots,\epsilon_n$ that are known, modular composition
boils down to multipoint evaluation and interpolation
\cite{Ritzmann1986}, which can be computed in~$\softO{n}$ arithmetic
operations. When furthermore the ring of coefficients is~$\ZZ$, he
uses well-chosen $\epsilon_i$'s to give an efficient algorithm for
composing power series, in a nonalgebraic model of computation: if
$\polp$ and $a$ over $\ZZ$ have coefficients bounded in absolute value
by $K$, then $\polp(a) \rem x^n$ can be computed using $\softO{n^2
  \log(K)}$ \emph{bit} operations, which is quasi-optimal since the
output has bit size~$\Omega(n^2\log(K))$ in general.

\subsubsection*{Chinese remainder theorem}

In our work, the cases of power series and of separable polynomials
play an important role as well. We use the observation that if a
factorization $f=f_1\dotsm f_s$ is known with the \(f_i\)'s relatively
prime, then composing modulo $f$ reduces to composing modulo each
$f_i$ and reconstructing the result via the Chinese remainder theorem.
Several consequences of this observation have been discussed by van der Hoeven
and
Lecerf~\cite{HoeLec18}.

\subsubsection{Special rings or fields}\label{ssec:specialrings}

For power series over a ring \(\ring\) of positive characteristic, Bernstein
proposed an algebraic algorithm whose complexity is quasi-linear in~$n$,
with a constant factor that depends on the characteristic of the
ring~\cite{Ber98}. In particular, this algorithm is very efficient over rings
whose characteristic is a product of small primes; if $\ring$ is a ring of
prime characteristic~$p$ then the algorithm  uses $\softO{np}$ operations in
$\ring$.  

A further step forward was achieved by Umans in 2008~\cite{Uma08},
with a new algorithm for modular composition modulo an arbitrary $f$,
over finite fields of small characteristic: if $p$ is $n^{o(1)}$, his
algorithm uses $n^{1+o(1)}$ base field operations.  Later, Kedlaya and
Umans introduced new techniques for composition over finite rings of
the form $(\ZZ/r\ZZ)[z]/\langle h(z)\rangle$, for an integer $r$ and
$h$ monic. For a finite field $\field = \mathbb{F}_q$, their algorithm
runs in $n^{1+\epsilon}\log^{1+o(1)}(q)$ \emph{bit
  operations}~\cite[Cor.\,7.2]{KU11}.

As in Ritzmann's work, a key idea in~\cite{Uma08,KU11} is to exploit
fast multipoint evaluation, but this time in a multivariate
setting. The composition $\polp(a)\rem f$ is reduced to the evaluation
at suitable points of a multivariate polynomial constructed
from~$\polp$ by an inverse Kronecker substitution, decreasing degrees
at the expense of increasing the number of variables. Umans' algorithm
performs the evaluation using the properties of the Frobenius
endomorphism \cite[Thm\,6]{Uma08}, while Kedlaya and Umans' proceeds by lifting to
characteristic zero (which requires working in a bit complexity
model)~\cite[Thm\,6.3]{KU11}. 
These multipoint evaluation algorithms have been extended to arbitrary number of variables and arbitrary finite fields~\cite{BGKM22,BGGKU22}.
For general fields, efficient analogues of these
multivariate multipoint evaluation algorithms are currently unknown.


\subsection{Related questions}

\subsubsection{(Multivariate) multipoint evaluation} \label{subsec:rel_evaluation}

For simplicity, we limit the discussion to the case of a field; most
of it extends to rings, with minor restrictions. The evaluation of a
univariate $\polp\in\xRing_{<n}$ at $n$ points in the field~$\field$,
and conversely the interpolation of a polynomial of degree~$<n$ from
$n$ values, are computable in quasi-linear
complexity~\cite[Chap.\,10]{GaGe99}. For polynomials in at least two
variables, however, the situation becomes tightly related to modular
composition.

The motivation of N\"usken and Ziegler~\cite{NusZie04} was the
evaluation of a polynomial~$\polp\in\xyRing_{<(m,d)}$ at $n$
points~$(x_k,y_k)_{1\le k\le n}$ in general position, with
$md=\bigO{n}$. Their algorithm first computes a univariate
interpolation polynomial such that~$a(x_k)=y_k$ for all~$k$; then the
composition~$b=\polp(x,a(x))\rem f$, where $f=\prod_k{(x-x_k)}$; and
concludes by a univariate multipoint evaluation of~$b$ at
$x_1,\ldots,x_n$. Since the univariate evaluation and interpolation
are performed in essentially linear time, the complexity is dominated
by the ``uni-bivariate'' modular composition~$\polp(x,a)\rem f$.

The case when several points have the same $x$-value can be handled by
an affine change of coordinates~\cite{NusZie04}; another approach,
taken by Kedlaya and Umans, is to pick~$n$ suitable points
$t_1,\ldots,t_n$ in~$\field$, to compute two interpolation
polynomials~$a_x$ and~$a_y$ in $\field[t]$, and thus reduce the
evaluation to the fully bivariate modular
composition~$\polp(a_x,a_y)\rem f$, where
now~$f=\prod_k{(t-t_k)}$. This extends to an arbitrary number of
variables and shows that multipoint evaluation in~$s$ variables
reduces to multivariate modular composition in the same number of
variables~\cite[Thm.\,3.3]{KU11}.

As mentioned in \cref{ssec:specialrings}, Kedlaya and Umans actually
make a heavy use of a converse reduction~\cite[Thm.\,3.1]{KU11}. If
$\polp$ is a polynomial in~$\field[x_1,\dots,x_s]$, the composition
$\polp(a_1 (x),\dots,a_s(x))\rem f$ reduces to a multipoint evaluation
of a polynomial of smaller degree in each of its variables, whose
number is increased.  For the univariate case of composition $(s=1)$
studied here, the smallest possible number of variables for evaluation
would be~2, leading to a bivariate evaluation of a polynomial of
degree~$\sqrt{n}$ at~$\Theta(n^{3/2})$ points, which is too large for
our target complexity. The next possible choice would be a
polynomial of 3~variables in degree~$n^{1/3}$ at~$\Theta(n^{4/3})$
points. Unfortunately, we are not aware of a sufficiently efficient
multipoint evaluation algorithm in 3~or more variables to make this
approach succeed in the algebraic model.


\subsubsection{Bivariate ideals} \label{subsubsec:bivideals}

Viewing the problem of computing $\polp(a)$ modulo $f$ as a problem of
reduction of $\polp$ modulo the ideal~$\mathcal I=\idealGens$, we introduce
bivariate polynomials in a different way from the inverse Kronecker substitution
mentioned above. Gr\"obner bases are commonly used for reductions modulo multivariate
ideals. A division with remainder similar to that
in \cref{eq:div-by-R} below would be
achieved via reduction by an appropriate Gr\"obner basis of $\mathcal I$,
provided we could compute this basis and perform the reduction in good
complexity. However, already the size of the Gr\"obner basis itself may be
$\Theta(n^{3/2})$ (see the example below), hence exceed our target complexity.

For an ideal given by two generic bivariate polynomials of degree~$n$ (hence
the ideal is of degree~$n^2$) and the graded lexicographic order, van der
Hoeven and Larrieu avoid the use of an explicit Gr\"obner basis. They show that
a {\em concise representation} of the basis of size only~$\softO{n^2}$ is
sufficient for reducing a polynomial modulo the ideal in time~$\softO{n^2}$
\cite{HoevenLarrieu2019}; the concise representation consists in particular of
truncations of well chosen polynomials in the ideal.
It is unclear to us whether a similar truncation strategy could be applied
specifically to~$\mathcal I$, whose degree is only~$n$.  Instead, the matrices
of relations we compute give a set of small degree polynomials in~$\mathcal I$
that  may not generate the whole ideal (see \cref{sec:relmat:structure}), but
provide a process of complexity~$\softO{n^\kappa}$ for the reduction
modulo~$\mathcal I$ of~\cref{eq:div-by-R}. 
These polynomials generate the same ideal as the first polynomials in
the Gr\"obner basis of~$\mathcal
I$ for the lexicographic order (see \cref{cor:structure_module}).

The concise representation of Gr\"obner bases has also been exploited by van
der Hoeven and Lecerf for computing the minimal polynomial of the
multiplication by $y$ modulo~$\mathcal J$, when $\mathcal J = \genBy{f_1,f_2}$
is generated by two generic polynomials  $f_1, f_2 \in \field[x,y]$ and
$\field$ is a finite field~\cite[Sec.\,4]{HoeLec21a}.  They apply the
transposition principle to a bivariate modular composition map modulo $\mathcal
J$, then compute the minimal polynomial from the resulting bivariate power
projections~\cite[Sec.\,6]{Kal00-2}. The evaluation of the composition map
modulo $\mathcal J$ is again in~$\softO{n^2}$, thanks to the concise
representation~\cite{HoeLec21a}.  In our case of $\mathcal I =
\idealGens$ and for a generic $a$, matrices of relations allow us to compute
the minimal polynomial of the multiplication by $y$ modulo $\mathcal I$ in
complexity~$\softO{n^\kappa}$ (see \cref{sec:minpoly}); matrices of relations
are obtained via a bivariate power projection process that can be regarded, in
part, as dual to N\"usken and Ziegler's bivariate modular composition
algorithm (\cref{subsubsec:transpNZ}). 

\subsubsection*{Note.} For a sufficiently large field \(\field\), take
$f=(x-1)\dotsm(x-n)$, where $n=k(k+1)/2$, and~$a \in \field[x]$ the polynomial
of degree smaller than~$n$ such that $a(i)=[\sqrt{2i}]$ for $1 \le
i \le n$. Then the
reduced Gr\"obner bases for the graded lexicographic order and for the
lexicographic order, both with $y\prec x$, coincide. They contain one
polynomial with leading term~$x^iy^{k-i}$ for each~$i\in\{0,\dots,k\}$.
Counting the number of monomials of these polynomials shows that this basis has
$k (k+1)(k+2)/3+(k+1)$ monomials; this is of the order of~$n^{3/2}$.  

\subsubsection{Modular composition and multipoint evaluation with precomputation.}

Quasi-linear modular composition \(\polp(a) \rem f\) is feasible after
precomputations on \((f,a)\) only, for $a$ generic and $f$
square free~\cite{NRS20}.

Likewise, after precomputations on the evaluation points and under genericity
assumptions on them, quasi-linear multivariate multipoint evaluation is
feasible \cite{HoeLec21b}, as well as quasi-linear bivariate interpolation
\cite{NRS20}. Furthermore, for bivariate evaluation, genericity can be replaced
by randomization \cite{HoeLec21c}. 

In these works, the precomputation stages are at least as expensive as the
fastest known  corresponding modular composition or multipoint evaluation
algorithms. They have a feature in common with our composition
algorithm: from
\(f,a\) (or from the evaluation points), they compute a set of polynomials
that belong to \(\idealGens\) (or vanish at the points), and allow for
efficient degree reduction of the polynomial to compose with (or to evaluate).
This set is either  akin to several matrices of relations of \(\rmodfa\) for a
small number of values of \(m\) ranging from \(1\) to \(n\) \cite{NRS20}, or is
a collection of well-chosen polynomials in several Gr\"obner bases for subsets
of the points so as to build a multivariate divide and conquer evaluation tree
\cite{HoeLec21b,HoeLec21c}. 

\subsection{Algorithmic Tools}
\label{sec:algorithmic_tools}
 
Our work builds upon a sequence of earlier algorithmic progress that
we now recall. We denote by $\xRing_{<n}$ the set of univariate
polynomials in $x$ with coefficients in~$\field$ and degree less
than~$n$; by $\xyRing_{<(r,s)}$ the bivariate polynomials in \(x,y\)
of bidegree in~$(x,y)$ less than~$(r,s)$.

\subsubsection{Baby steps/giant steps}
\label{subsubsec:bbs}

One of the bottlenecks in algebraic approaches for evaluating $\polp$ at $a$
modulo~$f$ is the computation of successive powers~$1,a,a^2,\dots$ modulo~$f$,
which leads to the question of minimizing the number of powers that are used.
The solution used by Brent and Kung relies on a baby steps/giant steps
scheme~\cite{PS73,BK78}, where only
\[
  1,a,\dots,a^{\lceil\sqrt{n}\rceil}\rem f
  \quad\text{and}\quad
  a^{2\lceil\sqrt{n}\rceil},a^{3\lceil\sqrt{n}\rceil},\ldots \rem f
\]
are computed. The former group forms the baby steps; the latter forms the giant
steps. The problem is then reduced to about~$\sqrt{n}$ modular compositions
``\(\polp_i(a) \rem f\)'' for \(\polp_i\) of degree about $\sqrt{n}$. These
compositions are all obtained simultaneously through the multiplication of two
matrices of sizes roughly~$\sqrt{n}\times \sqrt{n}$ and \(\sqrt{n}\times n\).
This is followed by a less expensive Horner evaluation step using the powers
of~$a^{\lceil \sqrt{n}\rceil}$. See \cref{ssec:BK} for a complete description.

\subsubsection{Projection-Reconstruction}
\label{subsubsec:pr}

Wiedemann's algorithm~\cite{Wie86} finds the minimal polynomial of a
matrix~$A \in \matRing{n}{n}$ by considering the sequence
$(\trsp{v}A^kw)_{k\ge0}$, for two vectors~$v$ and~$w$. This sequence
is linearly recurrent and its generating function $h(y)=\sum_{k \geq
  0} (\trsp{v}A^kw)/y^{k+1}$ is rational; for generic~$v$ and~$w$, the
denominator of $h(y)$ is the minimal polynomial $\mu_A$ of the
matrix~$A$. Writing $d\le n$ for the degree of \(\mu_A\), this
polynomial can be reconstructed efficiently from the first~$2d$ terms
of the sequence by the Berlekamp-Massey algorithm or, equivalently, by
the computation of a Pad\'e approximant. Given the expansion in
\(y^{-1}\) of a rational power series~$h(y)=q(y)/\mu_A (y)$ with
polynomials $q$ and~$\mu_A$ of degree at most~$d-1$ and $d$, this
reconstructs the fraction $(q,\mu_A)$ as a solution of
\begin{equation}\label{eq:pade}
  (h(y)+O(y^{-2d-1}))\mu_A(y)-q(y)=O(y^{-d-1}).
\end{equation}
If the degree of \(\minpoly[A]\) is unknown, one can use this approach with the
upper bound \(d=n\) instead.

Wiedemann's algorithm can be combined with the baby steps/giant steps
paradigm~\citetext{\citealp[Sec.~3]{Kal92}; \citealp{Shoup94};
  \citealp[Algorithm~AP]{KaSh98}}. In particular, when~$A$ is the
matrix~$M_a$ of multiplication by~$a\bmod f$ in the basis~$\mathcal
B=(1,x,\dots,x^{n-1})$ of $\xRing/\genBy{f}$, this was used by Shoup
to compute the {\em minimal polynomial} of the polynomial~$a$
modulo~$f$ \cite{Shoup94,Shoup95,Shoup99}. For irreducible $f$, Shoup
used the vectors $v=w=\mathbf{1}$ (where~$\mathbf{1}$ is the first column of
the
identity matrix), in which case the sequence~$(\trsp{v}A^kw)_{k\ge 0}$
becomes the sequence of \emph{power
  projections}~$(\ell(1),\ell(a),\ell(a^2),\dots)$, where $\ell$ is
the linear form that takes the coefficient of~1 of an element
of~$\xRing/\genBy{f}$ written on the basis~$\mathcal B$. For an
arbitrary $f$, Shoup used a random linear form $\ell$, corresponding
to a random choice of the vector $v$ and $w=\mathbf{1}$.

In either case, the required $2d$ elements of the sequence can be
obtained by left multiplication by $\trsp{v}$ of a matrix whose
columns are the coefficient vectors of~$1$, $a$, $a^2$, \ldots{}
modulo~$f$.  Now, the \emph{right} multiplication of the exact same
matrix by a vector of coefficients corresponds to modular
composition. Using the \emph{transposition principle}, Shoup described
a baby steps/giant steps algorithm that computes the power
projections for an arbitrary linear form~$\ell: \xRing/\genBy{f} \to
\field$ in the same complexity as that of Brent and Kung's
algorithm~\cite{Shoup94,Shoup95,Shoup99}. (See
\cref{subsec:powerproj}.) This principle states that the existence of
an algebraic algorithm for the multiplication of a matrix by a vector
induces the existence of an algorithm for the product of the transpose
of that matrix by a vector, both having essentially the same
complexity \citetext{\citealp[Thm.\,13.20]{BCS97}; \citealp{BLS03}}.

The same idea is used by Shoup for another operation that we also need.
Given~$a,b,f$, the \emph{inverse modular composition} asks for a
polynomial~$\polp$ of least degree such that~$\polp(a)\equiv b\bmod f$ or for a
proof
that no such~$\polp$ exists. This problem reduces to the computation of the
power projections 
\[
  (\ell(1),\ell(a),\dots,\ell(a^{2n-1}))
  \quad\text{and}\quad
  (\ell(b),\ell(ab),\dots,\ell(a^{n-1}b)),
\]
again in the same complexity as that of modular composition, followed by the
resolution of a linear system of Hankel type~\cite[Thm.\,3.5]{Shoup94}. The
latter is known to be equivalent to Pad\'e approximation~\cite{BGY80}, where
\cref{eq:pade} generalizes to
\[
  \left(\sum_{k\ge0}\frac{\ell(a^k)}{y^{k+1}}+O(y^{-2n-1})\right)
  \polp(y)
  -
  q(y)
  =
  \sum_{k\ge0}{\frac{\ell(a^kb)}{y^{k+1}}}+O(y^{-n-1}),
\]
with unknowns a numerator $q(y) \in \yRing_{<n}$ and
the inverse composition
$\polp(y) \in \yRing_{\le n}$.

\subsubsection{Blocks for speed and structure}
\label{subsubsec:blocks}

Coppersmith introduced a block version of Wiedemann's
algorithm~\cite{Coppersmith94}. There, the scalar sequence
$(\trsp{v}A^kw)_{k\ge0}$ is replaced by the matrix sequence
$(\trsp{V}A^kW)_{k\ge0}$ for two
\emph{matrices}~$V\in\matRing{n}{\ell}$ and~$W\in \matRing{n}{m}$: the
generating function $H(y)=\sum_{k \geq 0} (\trsp{V}A^kW)/y^{k+1}$
is a rational~$\ell\times m$ matrix. 

Such a matrix admits an irreducible \emph{matrix fraction
  description}~$N(y)D(y)^{-1}$ with~$N\in\ymatRing{\ell}{m}$
and~$D\in\ymatRing{m}{m}$ two polynomial matrices (see
\cref{sssec:fractions}), and the columns of the denominator matrix~$D$
form a basis of the $\yRing$-module of polynomial vectors
$u\in\yvecRing{m}$ such that $\sum_{i\le \deg(u)}
\trsp{V}A^{k+i}Wu_i=0$ for all \(k\ge 0\), where $u_i$ denotes the
coefficient of~$y^i$ in~$u$~\cite[Lem.\,2.8]{KaVi05}. For $m=1$, this
module is the ideal generated by the minimal polynomial of the
sequence in Wiedemann's algorithm. 

For $1 \leq m \leq \ell \leq n$,
the matrix~$D$ contains more information: for example, for generic~$V$
and~$W$, its invariant factors are the $m$ invariant factors of
largest degree of the characteristic matrix
$y\idMat{n}-A$~\cite[Thm.\,2.12]{KaVi05}, the highest degree one being
the minimal polynomial of \(A\).  Consequently, the determinant of
\(D\) has degree the sum $\dd$ of the degrees of these $m$~invariant
factors, which implies that $\dd \le n$.

The computation of $D$ can be achieved in two steps, which are matrix versions
of the methods used for \(m=1\) in \cref{subsubsec:pr}.  Writing
\(d\) for the degree of \(D\), it is sufficient to compute the
first~$2d$ matrices of the sequence $(\trsp{V}A^kW)_{k\ge0}$,
which can be done by a baby steps/giant steps approach~\cite{KaVi05}. Next,~$D$
is obtained by matrix fraction reconstruction, solving
\begin{equation*}
  \left(\trsp{V}(y\idMat{n}-A)^{-1}W +O(y^{-2d-1})\right)D(y)-N(y)
  =
  O(y^{-d-1})
\end{equation*}
for the unknown $N\in\ymatRing{\ell}{m}$ and $D\in\ymatRing{m}{m}$ of degrees at
most \(d-1\) and \(d\); this can be done efficiently by a generalization of
Pad\'e approximation called minimal approximant bases, whose properties are
recalled in \cref{sec:relmat:approx_basis}. (See \cite{KaVi05,KaYu13} for
bibliographic pointers to algorithms that compute minimal linear generators of
matrix sequences.) The parameter \(d\) plays a major role in the efficiency of
both steps: it is usually unknown \emph{a priori}, and might be as
large as
\(\Theta(n)\). Yet, the interest of this block approach lies in the fact that,
for generic~$V$ and $W$ and \(\ell\ge m\), the matrix \(D\) has degree \(d =
\lceil \dd /m \rceil \le \lceil n/m \rceil\) \cite[Cor.\,6.4]{Vil97:TR}.

\subsubsection{Efficient projections and small bivariate polynomials}
\label{subsubsec:effprojandsmallbivpols}

Special choices of the matrices~$V$ and~$W$ above, with identity blocks,
lead to efficient projections and have been shown to be effective in
the context of black-box matrix inversion~\cite{EGGSV07}. Even simpler
matrices, $X=\trsp{(\idMat{m}\;\;0)}$ and $Y=\trsp{(0\;\;\idMat {m})}$
in $\field^{n\times m}$ with $m\in \{1, \ldots ,n\}$, have been used
by Villard in his fast algorithm for the bivariate resultant of two
bivariate polynomials \(f\) and \(g\) in \(\xyRing\)~\cite{Vil18}.  In
this context, for generic $f$ and $g$, this choice of \(X\) and \(Y\)
is sufficient to ensure that the denominator matrix \(D\) contains
\(m\) ``small'' polynomials in the ideal of~$\xyRing$ generated by~$f$
and \(g\).

\subsection{Overview of the core algorithm}
\label{subsec:overview}

When~$a\bmod f$ has a minimal polynomial~$\minpoly$ of small degree,
$\minpoly$ can be computed efficiently using power projections
$(\ell(1),\ell(a),\ell(a^2),\dots)$ by Shoup's algorithm, since few
terms in the sequence are needed (see
\cref{subsubsec:pr,subsubsec:minpoly}). Then, for composition, one
uses the identity $\polp(a)\equiv \hat{\polp}(a)\bmod f$, where
$\hat{\polp}=\polp\rem\minpoly$. Since \(\hat{\polp}\) has small
degree, this reduces the number of powers of~$a\bmod f$ that need be
considered.

Our
algorithm can be viewed as a block or bivariate version of this approach, \emph{replacing
the univariate polynomial~$\minpoly$ by a collection of \(m\) small bivariate
polynomials in the ideal generated by $y-a(x)$ and~$f(x)$}, for a fixed
parameter \(m\). In a generic situation, while
\(\minpoly\) has degree \(n\), there exists such a collection with degrees
 \(m-1\) and \(\lceil n/m \rceil\) in \(x\) and \(y\). This collection is represented as
a matrix in~$\ymatRing{m}{m}$ and is found efficiently by exploiting the
structure of the matrix of multiplication by~$a\bmod f$. 

\subsubsection*{Matrices of relations}
Let $M_a \in \field^{n\times n}$ be the matrix of multiplication
by~$a\bmod f$ in the basis~$(1,x,\dots,x^{n-1})$.  Following
\cref{subsubsec:blocks}, in the special case where $A=M_a$, if $V$ is
a generic matrix in $\field^{n \times \ell}$, and $W$ is the matrix
$X= \trsp{(\idMat{m}\;\; 0)}$ with $m\le \ell$ and $m\le n$, the block
Wiedemann approach yields a denominator matrix $D \in\ymatRing{m}{m}$
whose columns represent a basis of the $\yRing$-module
\begin{equation}\label{eq:rmodfa}
  \rmodfa = \left\{ \polr(x,y)=\polr_0(y)+\cdots + \polr_{m-1}(y) x^{m-1}
  \mid \polr(x,a(x))\equiv 0\bmod f (x)\right\};
\end{equation}
this follows for instance from~\cite[Lem.~4.2]{Vil97:TR}.  The
elements of this module are algebraic relations of degree less
than~$m$ in~$x$ satisfied by~$a \bmod f$
(\cref{sec:relmat:relations_denominators}). We call \emph{matrix of
  relations} any nonsingular matrix \(\rmatfa\in\ymatRing{m}{m}\) whose
columns are the coefficients of polynomials in~$\rmodfa$
(\cref{sec:relmat:def_and_degdet}), that is, any nonsingular right
multiple of \(D\). 

Given a matrix of relations~$\rmatfa$, the
composition~$\polp(a)\rem f$ is obtained in two steps.
\begin{itemize}
\item First, by \emph{polynomial matrix division} \cite[Thm.\,6.3-15, p.\,389]{Kailath80},
  there exist vectors~$v,w\in\yvecRing{m}$ such that
  \begin{equation}\label{eq:div-by-R}
    \trsp{(\polp(y) \;\; 0 \;\; \cdots \;\; 0)} = \rmatfa w+v,
  \end{equation}
  where $\deg(v) < d$ and $d$ is an upper bound on $\deg(\rmatfa)$;
   finding such vectors takes
  $\softO{m^{\omega}(d+n/m)}$ operations
  (\cref{sec:bivcomposition})~\cite{ZLS12}. Then, by
  design, the bivariate
  polynomial
  \[
  \tilde\polp(x,y)=v_1(y)+\dots+v_{m}(y)x^{m-1}
  \]
  has degree less than~$m$ and~$d$ in~$x$ and~$y$, and is such that
  \(\polp(a)\equiv\tilde \polp(x,a)\bmod f\). 
 
\item The polynomial \(\tilde\polp\) can then be evaluated at~$y=a \bmod f$ by the
  N\"usken-Ziegler algorithm in~$\softO{\bicost{d}}$ operations, \bicostrecall{}
  (\cref{prop:NuskenZiegler}).
\end{itemize}

\subsubsection*{Truncated sequence of projections}
In the block Wiedemann approach, using $X$ as our right projection
matrix, we need the first \(2d\) elements of the matrix sequence
$(VM_a^kX)_{k\ge 0}$, which amounts to a type of bivariate power
projections (see \cref{subsubsec:bivideals}). Unfortunately, we do not know how
to obtain them
efficiently enough for an arbitrary $V$.  Choosing $V=\trsp{X}$, we
design a baby steps/giant steps algorithm in \cref{sec:truncated_pow}
that runs in $\softO{\bicost{d} +m^2d}$ operations.  With this choice,
by fraction reconstruction the sequence~$(\trsp{X}M_a^kX)_{k\ge0}$
yields a denominator $D$ that is a basis of the $\yRing$-module
\[
  \rmodmmfa = \left\{ \polr(x,y) \in \xyRing_{<(m,\cdot)}
  \mid \big[a(x)^k\polr(x,a(x))\rem f(x)\big]_0^{m-1}=0 \text{ for all }
  k\ge 0 \right\},
\]
where $[\,\cdot\,]_0^{m-1}$ is the projection on
$\Span(1,x,\dots,x^{m-1})$. The inclusion $\rmodfa\subseteq\rmodmmfa$
holds but may be strict, leading to a denominator $D$ that is not a
matrix of relations.

\subsubsection*{Matrices of relations of small degree}\label{ssec:matrel_small_degree}

For an arbitrary~$f$ with $f(0)\neq0$ (this is not really a
restriction, see \cref{rmk:fat0}) and a generic~$a$, two important
properties hold (see \cref{subsubsec:genina}): the above inclusion of
modules is an equality~---~making the algorithm correct~---~and a
basis $\rmatfa$ of degree $d=\lceil n/m\rceil$ of~$\rmodfa$ can be
reconstructed from the first $2d$ elements of the
sequence~$(\trsp{X}M_a^kX)_{k \ge 0}$~---~making the algorithm fast.

The reconstruction is done via minimal approximant bases in
\cref{sec:relmat:approx_basis,subsec:candidate}. Directly extending
\cref{subsubsec:blocks}, we would solve the equation at infinity
\begin{equation}
  \label{eq:blockPade}
  \left(\trsp{X}(\charmat)^{-1}X + O(y^{-2d-1})\right) \rmatfa(y) - N(y) = O(y^{-d-1}),
\end{equation}
for unknown matrices $N$ and $\rmatfa$ of degree at most \(d-1\) and
\(d\). For technical reasons coming from the reconstruction algorithm,
we actually use an expansion at $y=0$ rather than at infinity, so that
the sequence we use involves powers of $M_a^{-1}$ instead of $M_a$ (see \cref{rmk:shift}).

Beyond generic cases, a relevant quantity is
\begin{equation} \label{eq:defnu}
  \ddfa = \deg(\sigma_1) + \cdots + \deg(\sigma_{m}),
\end{equation}
where $\sigma_1, \ldots, \sigma_n \in \yRing$ are the invariant factors of
$\charmat$, ordered by decreasing degree. This quantity is at most \(n\), and
it is the degree of the determinant of any basis of $\rmodfa$
(\cref{prop:invariant-factors}). In favorable situations, working with \(d =
\lceil \ddfa / m \rceil\), and \emph{a fortiori} with $\lceil n/m\rceil$, is
sufficient to obtain such a basis~$\rmodfa$.


\subsection{Probabilistic algorithm for \texorpdfstring{$f$}{f} separable or purely inseparable}
\label{subsec:proba_algo}

Our probabilistic algorithm aims at bringing arbitrary inputs to the
favorable situation mentioned above, by means of a random change of
basis.  For a polynomial $\gamma\in\xRing$ such that the minimal
polynomial~$\mu_\gamma$ of $\gamma\bmod f$ has degree~$n$, the powers
\((1,\gamma,\ldots,\gamma^{n-1})\bmod f\) form a basis
of~$\quotient=\xRing/\genBy{f}$. This induces a $\field$-algebra
isomorphism:
$$
  \phi_\gamma:\quotient\to\yRing/\genBy{\mu_{\gamma}}
$$ that maps $\gamma$ to $y$, and more generally $u \in \quotient$ to
  $v$ such that $v(\gamma)\equiv u\bmod f$.

Using $\phi_\gamma$ allows us to transport our problem of modular
composition to the right-hand side. For $a$ in $\xRing_{<n}$ and
$\polp$ in $\yRing$, to find \(\polp(a)\rem f\), this boils down to
the following (see \cref{algo:ModularCompositionBaseCase}):
\begin{itemize}
  \item[---] a forward change of basis: through inverse modular
    composition, compute $\alpha \in \yRing_{<n}$ such that
    $a=\alpha(\gamma) \rem f$; this step also determines the minimal
    polynomial \(\minpoly[\gamma]\);
  \item[---] a modular composition in the new basis: compute \(\beta =\polp
  (\alpha) \rem \minpoly[\gamma]\);
  \item[---] a backward change of basis: the modular composition
    \(\beta(\gamma) \rem f\), which equals $g(a) \rem f$.
\end{itemize}

\subsubsection*{Computational aspects}
The second and third steps are modular compositions. They can
performed efficiently by the approach of \cref{subsec:overview}, by
finding and using matrices of relations \(\rmatfg\) and \(\rmatma\),
as long as certain genericity assumptions hold; this
aspect is discussed below.

The first step, for the forward change of basis, is an instance of
inverse modular composition and the calculation of a minimal
polynomial.  As mentioned in \cref{subsubsec:pr}, Shoup's solutions
recover both $\alpha$ and $\mu_{\gamma}$ from the power projections
$(\ell(1),\ell(\gamma),\dots,\ell(\gamma^{2n-1}))$ and
$(\ell(a),\ell(\gamma a),\dots,\ell(\gamma^{n-1}a))$, in the
complexity of Brent and Kung's modular composition algorithm. Using
matrices of relations we achieve a lower complexity, for a
generic~\(\gamma\), as follows.

\begin{enumerate}
\item \emph{Matrix of relations and minimal polynomial.}  Generalizing
  the power projections of $\gamma$, the algorithm of
  \cref{subsec:overview} computes the first~\(2d\) terms of
  $(\trsp{X}M_\gamma^kX)_{k\ge 0}$, where $d=\lceil n/m\rceil$, and then
  reconstructs a basis $\rmatfg$ of $\rmodfg$ by solving
  \cref{eq:blockPade} (with \(\gamma\) instead of \(a\)). This basis
  gives in particular the minimal polynomial $\minpoly[\gamma]$, which
  appears as an entry of the Hermite normal form of this basis
  (\cref{prop:invariant-factors}).

\item \emph{Bivariate inverse composition.} The
  use of projections $(\ell(a),\ell(\gamma
  a),\dots,\ell(\gamma^{n-1}a))$ is directly generalized by computing the first~$2d$
  terms of $(\trsp{X}M_\gamma^kM_a\mathbf{1})_{k\ge 0}$, where~$ \mathbf{1}$
  is the first column of~$X$, and solving
  \begin{equation}\label{eq:blockShoup}
    \left(\trsp{X}(y\idMat{n}-M_\gamma)^{-1}X+O(y^{-2d-1})\right)v_{\tilde{\alpha}}(y) 
    -v_N(y) =\trsp{X}(y\idMat{n}-M_\gamma)^{-1}M_a\mathbf{1}+O(y^{-d-1})
    \end{equation}
  for polynomial vectors $v_N$ and $v_{\tilde{\alpha}}$
  in~$\yvecRing{m}$ of degree less than \(d\); the entries of the
  vector $v_{\tilde{\alpha}}$ are the coefficients of a bivariate
  polynomial $\tilde \alpha (x,y)$ of small degree such that $\tilde
  \alpha (x,\gamma)\equiv a\bmod f$.

  As for~\cref{eq:blockShoup}, we actually work with an expansion
  at $y=0$ rather than infinity. 

\item \emph{Bivariate \(\tilde\alpha\) to univariate \(\alpha\).}  The
  situation is now symmetric to that of the composition algorithm of
  \cref{subsec:overview}: we consider again \cref{eq:div-by-R}, where
  now $\polp$ is unknown (it is $\alpha$), $v$ is known (it is $\tilde
  \alpha(x,y)$) and both $\rmatfa$ and $v$ have degree at most~$d$,
  so that the polynomial matrix problem can be solved in
  $\softO{m^\omega(d+n/m)}$ operations.
\end{enumerate}
This approach is detailed in \algoName{algo:ChangeOfBasis}{}, with the
steps reordered and combined so as to retrieve both $\rmatfg$ and
$v_{\tilde{\alpha}}$ from a single fraction reconstruction.

\subsubsection*{Probabilistic aspects}
For a generic \(\gamma\), one has \(\deg(\minpoly[\gamma])=n\), so the
isomorphism $\phi_\gamma$ is well defined. Using the Schwartz-Zippel
lemma, it is straightforward to control the probability of having
$\deg(\minpoly[\gamma])<n$. 

For generic $\gamma$, we can then follow the approach described in
\cref{subsec:overview} to perform the last step, modular composition 
by $\gamma$, with the desired complexity. The quantitative aspects can
be worked out as well, and similar considerations  hold for the
first step, inverse modular composition by~$\gamma$.

However, the composition in the second step, $g(\alpha) \rem
\minpoly[\gamma]$, is more delicate to analyze.  We need the equality
of modules 
$\rmodma = \rmodmm^{(\alpha,\minpoly[\gamma])}$, and that a matrix of
relations in this module can be reconstructed from the first $2\lceil
\ddma/m\rceil \le 2 \lceil n/m\rceil$ elements of the corresponding
matrix sequence; the analysis is made difficult by the fact that both
$\alpha$ and $\minpoly[\gamma]$ are nonlinear functions of the random
element $\gamma$.

We prove that this happens for a generic~$\gamma$ in two cases: when
$f$ is {\em separable} in
\cref{sec:composition_randomized:proof_separable}, and when $f$ is
\emph{purely inseparable}, with extra conditions, in
\cref{sec:composition_randomized:proof_inseparable}; the latter case
covers power series composition with~\(f=x^n\). In both situations,
there is a nonzero polynomial~$\Delta$ in~$n$ variables such that the
constraints above hold if $\Delta$ does not vanish at the coefficients
of~$\gamma$. We  choose a random $\gamma$, and the probability of
failure is again bounded by the Schwartz-Zippel lemma.

We do not have a proof that a generic~$\gamma$  satisfies our
requirements for an arbitrary~$f$. Our algorithm for the general case 
proceeds by reduction to the two extreme cases above, separable 
and purely inseparable polynomials.

\subsubsection*{From Monte Carlo to Las Vegas} At this stage, we have a
probabilistic algorithm of Monte Carlo type, that runs in the
announced complexity and returns the correct result with a controlled
probability of error. The next question is to modify the algorithm so
that it detects and reports the unlucky choices of~$\gamma$ for which
its result would be incorrect.

In order to certify the result obtained for a random choice
of~$\gamma\in\quotient$, it would be sufficient to check the following
properties:
\begin{enumerate}
  \item the computed matrix~$\rmatfg$ is a basis of relations of~$\rmodfg$;
  \item the minimal polynomial of \(\gamma\) modulo \(f\) has degree \(n\);
  \item the computed matrix~$\rmatma$ is a basis of relations of $\rmodma$.
\end{enumerate}
However, we do not know how to check that all the columns of a matrix belong to
the ideal $\genBy{f(x),y-\gamma(x)}$ or~$\genBy{\mu_\gamma(x),y-\alpha}$ in
sufficiently low complexity and in a deterministic way. The matrix~$\rmatfg$ is
easier to deal with: as it is expected to behave like in the generic case, its
expected degree structure is known and the matrix can be certified by degree
considerations~(\cref{lem:fraction-reconstruction:item:certif} of
\cref{lem:fraction-reconstruction}, and \cref{prop:algo:ChangeOfBasis}). From
there, the minimal polynomial of \(\gamma\) can be computed efficiently via the
Hermite normal form of \(\rmodfg\), and it remains to check that it has degree
\(n\).

The other matrix, $\rmatma$, carries more information about~$a$ and
cannot be expected to behave as predictably as~$\rmatfg$. Our approach
is to extract from its columns two small degree polynomials~$r$
and~$s$ in~$\xyRing$. Since only two such polynomials are considered,
they can be checked to vanish at~$\alpha \bmod \mu_\gamma$ by the
N\"usken-Ziegler algorithm without affecting the asymptotic
cost. Then, these two polynomials are used to construct a Sylvester
matrix that can be used for composition instead of~$\rmatma$, without
increasing the overall complexity (\cref{algo:MatrixOfRelations}).

\subsubsection*{Note}

Equivalently, the randomization of our probabilistic algorithm can be seen as a
change of projection. Indeed, let $P \in \matRing{n}{n}$ have its $j$th
column formed by the coefficients of $\gamma^{j-1} \rem f$. If $\gamma \bmod f$
generates~$\xRing/\genBy{f}$ and $M_{\alpha}$ is the matrix of multiplication by $\alpha \bmod {\mu_{\gamma}}$ with basis $(1,y,\ldots,y^{n-1})$, then the multiplications by~$\alpha$ and
by~$a$ are related by
\begin{equation}\label{eq:defMalpha}
  M_{\alpha}=P^{-1}M_aP.
\end{equation} 
Hence
\[
  \trsp{X} M_{\alpha}^k X =  (\trsp{X} P^{-1}) M_{a}^k (P X), 
\]
which, for instance on the right side, leads to considering the first $m$
columns of $P$ instead of $X$ for projecting. This amounts to kinds of
structured projections $(\trsp{V}M_a^kW)_{k\ge0}$, i.e. with matrices $V$ and~$W$ 
in a special proper subset of~$\field^{n\times m}$.


\subsection{Algorithm for the general case}

The algorithm of \cref{subsec:proba_algo} is proved to work when~$f$
is either separable, or purely inseparable (for the latter, with extra
conditions that are dealt with in
\cref{sec:composition_randomized:proof_inseparable}).  In
\cref{sec:generalalgo}, we address the general case, by first computing
a separable decomposition of~$f$~\cite{Lec2008}, yielding a
factorization into a product into pairwise coprime terms of the
form~$h_i(x^{p^{e_i}})^{\mult_i}$, with $h_i$ separable and
$e_i,\mult_i$ integers (here, $p$ is the characteristic of $\field$).

Working modulo each factor separately, we are thus left with
the question of composition modulo a polynomial of the form
$h(x^{p^e})^\mult$, with $h$ separable (all
such results are eventually recombined via the Chinese remainder theorem).

For a modulus of the form~$h(x)^\mult$, van~der~Hoeven and Lecerf
showed how composition can be reduced to $\mult$ compositions modulo
$h$, the computation of an annihilating polynomial modulo~$h$, and a
power series composition at precision~$\mult$ with coefficients in
$\bigfield=\xRing/\genBy{h(x)}$ \cite{HoeLec17}.  We extend this
result to the case of moduli of the form~$h(x^{p^e})^\mult$ in
\cref{ssec:modulo_powers}, involving essentially the same steps.  The
first two operations (compositions and annihilating polynomial modulo
$h$) are directly handled by our results so far, but this is not quite
the case for the latter, power series composition with coefficients in
$\bigfield$.

Our algorithms are written assuming they work over a field, as they
perform zero-tests and inversions (compare this with Brent and Kung's
algorithms, for instance, which apply over a ring).  If~$h$ is
irreducible, $\bigfield$ is a field, but if $h$ is only assumed to be 
separable, then $\bigfield$ is only a product of fields. The \emph{dynamic
  evaluation} paradigm~\cite{D5} explains how an algorithm written for
inputs lying in a field can carry over to inputs in a product of
fields, but the original approach induces cost overheads that go
beyond our cost target. Using van der Hoeven and Lecerf's efficient
dynamic evaluation strategy~\cite{HoeLec20}, we show how our algorithm
for power series adapts to this situation~(\cref{sec:overseparable})
without affecting the asymptotic runtime.

\subsection{Outline}

\cref{sec:preliminaries} introduces some notation and our computational model.
\cref{sec:sequence} details baby steps/giant steps techniques used in our
composition algorithm: known ones such as in Brent and Kung's composition, and
new ones such as for computing truncated powers which give access to
$(\trsp{X}M_a^kX)_{k\ge 0}$. \cref{sec:relmat_intro} studies matrices of relations and
how they are used in our composition algorithm, whereas \cref{sec:relmat_comp}
shows how to compute them efficiently by matrix fraction reconstruction under
some assumptions on \((f,a,m)\).  \cref{sec:changeofbasis} presents an
algorithm for the change of basis of \cref{subsec:proba_algo}: it finds the
minimal polynomial \(\minpoly[\gamma]\) and an inverse composition \(\alpha\)
such that \(\alpha(\gamma) \equiv a \bmod f\), under assumptions on
\((f,\gamma,m)\).  \cref{sec:genericity} studies these assumptions, and in
particular gives precise generic situations where they hold.
\cref{sec:composition_randomized} describes our main randomized composition
algorithm and proves its correctness for a separable $f$ and for a purely
inseparable $f$ (generalizing \(f=x^n\)); then \cref{sec:generalalgo} handles
the general case of composition modulo any \(f\). Finally, in
\cref{sec:applications}, we state resulting complexity improvements for several
variants of modular composition and other related problems.


\section{Preliminaries}
\label{sec:preliminaries}


\subsubsection*{Notation}

In this article, $\field$ is an arbitrary field.  For bivariate
polynomials in variables $x$ and $y$, $\deg_x$ and $\deg_y$ give the
degree in \(x\) and in \(y\). For any polynomial or power series
$p=\sum_i{p_ix^i}$, we use the following notation for a ``slice'' of
it: $[p]_j^k=p_j+p_{j+1}x+\dots+p_{j+k}x^k$. The ideal generated by
polynomials $f_1,\dots,f_k$ in an ambient ring (which will be clear
from the context) is denoted by~$\genBy{f_1,\dots,f_k}$.

Vectors, such as elements of \(\vecRing{m}\) or \(\yvecRing{m}\), are seen as
column vectors by default; when row vectors are considered this is explicit in
our notation, e.g.~\(\matRing{1}{m}\) or \(\ymatRing{1}{m}\). We often
identify a polynomial $\polp_0(y)+ \cdots + \polp_{m-1}(y) x^{m-1}$ in
$\xyRing_{<(m,\cdot)}$ with the column vector \(\trsp{(\polp_0 \;\cdots\;
\polp_{m-1})}\) in \(\yvecRing{m}\) of its coefficients on the basis
\((1,x,\ldots,x^{m-1})\) of the \(\yRing\)-module $\xyRing_{<(m,\cdot)}$.

For $a$ and $f$ in $\field[x]$, \(\mulmat\) denotes the
matrix of the linear map of multiplication by~\(a\) in~$\quoF$ 
with basis \((1,x,\ldots,x^{n-1})\), and $\minpoly$,
resp.~\(\charpoly\), denotes the minimal, resp.~characteristic
polynomial of $a$ in $\quoF$ (that is, the minimal and characteristic
polynomials of $\mulmat$). 

Whenever the context is sufficiently clear, particularly in
\cref{sec:relmat_intro,sec:relmat_comp,sec:genericity}, notation such
as\(\rmodfa\), \(\ddfa\) defined in the introduction is shortened
into \(\rmod\), \(\dd\). We  keep the superscripts in important
statements.

\subsubsection*{Computational model}

Our algorithms are written in pseudocode, using standard syntax elements
(for loops, if statements, \ldots). Informally, we count all
arithmetic
operations $\{+,-,\times,\div\}$ and zero-tests in $\field$ at unit cost. The
underlying complexity model is the {\em computation tree}~\cite[Sec.\,4.4]{BCS97}.

A computation tree over $\field$ is a binary tree whose nodes are partitioned
into {\em input nodes} that form an initial segment of the tree
starting at the root, {\em computation nodes} with outdegree
1, {\em
branching nodes} with outdegree 2 and {\em output nodes} at the
leaves. To each node is associated a label.
Computation nodes are labelled
 by constants in $\field$ or
operations in $\{+,-,\times,\div\}$, in which case they also carry references
to two previous input or computation nodes; branching nodes are labelled by
zero-tests, referring to some previously computed quantity. Each leaf $v$ is
labelled 
with a sequence of references $(u_1,\dots,u_{\ell(v)})$ to previous input
or computation nodes. The {\em cost} of a computation tree is  its height
$\tau$, that is, the maximum length of a path from the root to a leaf.

It then makes sense to {\em evaluate} a computation tree at an element of
$\vecRing{s}$~---~called {\em input to the tree}, where $s$ is the number of input nodes, following a path from
the root to a leaf. After the input nodes, the path is constructed as
follows. Each computation node is assigned a value
derived from the label it carries, when it is defined. Otherwise,
e.g., in case of a division by~0, the path stops. At a branching
node
the path branches left or right depending on 
whether the value it refers to vanishes or not. At a leaf~$v$ with
label $(u_1,\dots,u_{\ell(v)})$, the output of the 
computation is the tuple of the values computed at nodes
$u_1,\dots,u_{\ell(v)}$. In that case, the computation tree is called 
\emph{evaluable} at the input.
Overall, the computation requires at most $\tau$ arithmetic operations
in
$\field$.
An algorithm is called \emph{quasi-linear} when the height
of its computation tree is linear (up to logarithmic factors) in the
number of inputs. It is called \emph{quasi-optimal} when this height
is linear (up to logarithmic factors) in the
number of inputs \emph{plus} the maximum number of values returned by
the output nodes. 

A computation tree takes inputs of fixed length. In order to solve a problem
for inputs of arbitrary size and characteristic, we need a {\em family} of trees, parametrized by
the input size and the characteristic. Every algorithm we describe using pseudocode in this article,
and all algorithms that we rely on from the literature, can be described by a
family of computation trees.

The translation from pseudocode to computation tree is usually rather
direct, and as is customary in the
literature, we  do not do it explicitly. In a nutshell, for loops and
recursive calls are ``unrolled''; if statements that test whether a computed
quantity vanishes yield branching nodes, etc. Some operations in our
pseudocode may not be directly available in our model (as we only allow
arithmetic operations in $\field$ and zero-test), but they can be rewritten in
a way that complies with our requirements. This is for instance the case when
we compute the degree of a polynomial (as in Euclid's GCD algorithm): this can
be achieved by scanning its coefficients, in order of decreasing degree, until
a nonzero one is found. 
{We also invoke a result by van der
Hoeven and Lecerf~\cite{HoeLec20} on the transformation of computation trees
for {\em directed evaluation} in \cref{sec:overseparable}; the translation from pseudocode to tree also applies to 
their algorithm.}

The families of trees that we build for modular composition
with arbitrary degree and characteristic are \emph{uniform}, in the sense that an appropriate tree description is generated from the pseudocode and any given degree $n$ and characteristic $p$.

We allow our algorithms to return flags (such as \Fail, or \Cert{}/\NoCert{}).
This can be done in this model, by returning constants in the vector of outputs,
such as $1$ for \Fail{} and $0$ otherwise.

Finally, several of our algorithms rely on randomization; however, we do not
want to introduce another arithmetic operation for the selection of random
field elements. One reason for this is that {the result by van der
Hoeven and Lecerf~\cite{HoeLec20} mentioned above} is
explicitly written in a deterministic model. Instead, ``random'' field elements
are given to our procedures as extra input parameters.


\section{Simultaneous modular operations by matrix multiplication}
\label{sec:sequence}

A key ingredient in fast modular composition algorithms is to turn the
problem into the simultaneous evaluation of polynomials of smaller
degree, and exploit the structure brought by this simultaneity using
matrix multiplication.  In this section, after reviewing Brent and
Kung's original algorithm and giving a direct extension of it, we use
this idea in two further contexts: N\"usken and Ziegler's bivariate
modular composition algorithm, and the computation of truncations of
powers of the form $a^k\rem f$. Both arise in our algorithms, and are
bottlenecks in their complexity.


\subsection{Brent and Kung's algorithm}
\label{ssec:BK}

\subsubsection{Modular composition}

We start with a review of Brent and Kung's algorithm to compute
$\polp(a)\rem f$, pointing out the impact of rectangular matrix
multiplication~\cite{HuPa98} and how the runtime depends on the
degrees of both $f$ and $\polp$~\cite[Fact~3.1]{Shoup94}. This can
be seen as an introduction to the N\"usken-Ziegler algorithm, which
generalizes this approach to a bivariate~$\polp$.

\begin{proposition}
  \label{lemma:BK}
  Given polynomials $f\in\field[x]$ of degree~$n$, $a$ in $\field[x]_{<n}$ and $g$
  in $\field[y]_{<d}$, \algoName{algo:ModularComposition-BrentKung}{}
  computes $\polp(a)\rem f$ using $\softO{(1 + n/d)d^{\omega_2/2}}$
  operations in \(\field\).
\end{proposition}
\begin{proof}
  Correctness follows from noticing that at \cref{algoBK:bi}, $b_i
  \equiv \polp_{ir} + \polp_{ir+1} a + \cdots + \polp_{ir + r-1}
  a^{r-1} \bmod f$ holds for all $i$, where $\polp_j$ is the
  coefficient of degree $j$ in $\polp$ for all $j$. The cost of the
  algorithm comes from $\Theta(d^{1/2})$ multiplications modulo $f$,
  which use $\softO {n d^{1/2}}$ operations in $\field$, and a matrix
  product in sizes $s \times r$ and $r \times n$, with both $s$ and
  $r$ in $\Theta(d^{1/2})$.  This product can be done through $\lceil
  n/d \rceil \le n/d+1$ matrix products in sizes $s \times r$ and $r
  \times d$, each of which takes $\bigO {d^{\omega_2/2}}$ operations
  in $\field$.
\end{proof}

\begin{algorithm}
  \algoCaptionLabel{ModularComposition-BrentKung}{f,a,\polp}
  \begin{algorithmic}[1]
  \Require $f$ of degree $n$ in $\xRing$, $a$ in $\xRing_{<n}$, $\polp$ in $\yRing_{<d}$
  \Ensure $\polp(a)\rem f$
  \State $r \gets \lceil d^{1/2}\rceil$, $s \gets \lceil d / r\rceil$ 
  \State $\hat a_0 \gets 1$
  \State\InlineFor{$i = 1, \ldots, r$}{$\hat{a}_i \gets a \cdot \hat{a}_{i-1} \rem f$}
      \Comment{\(\hat{a}_i = a^i \rem f\)}  
  \State $A \gets$ matrix $({\rm coeff}(\hat{a}_i, j))_{\substack{0\le
  i<r\\0\le j< n}}$ in $\matRing{r}{n}$
      \Comment{coefficient of degree \(j\) of \(\hat{a}_i\)}  
  \State $\polP \gets$ matrix $({\rm coeff}(\polp,{ir+j}))_
  {\substack{0\le i< s\\0\le j<r}}$ in $\matRing{s}{r}$
  \State $B=(b_{i,j})_{\substack{0\le i < s\\ 0 \le j < n}} \gets \polP
  A$ in $\matRing{s}{n}$
  \State\InlineFor{$i = 0, \ldots, s-1$}{$b_i \gets b_{i,0} + \cdots + b_{i,n-1} x^{n-1}$}\label{algoBK:bi}
  \State \Return \(b_0 + b_1 \hat{a}_r + \dots + b_{s-1} \hat{a}_r^{s-1} \rem
  f\) \Comment{Horner evaluation}
  \end{algorithmic}
\end{algorithm}

\subsubsection*{Note}

In the analysis, dividing the matrix product into blocks, as we did,
is suboptimal. Using rectangular matrix multiplication directly, the
runtime can be described by the finer estimate
$\softO{d^{\omega_{2\log(n)/\log(d)}/2}}$.  Here, the notation
$\omega_\theta$ is a feasible exponent for rectangular matrix
multiplication for any real number \(\theta\): there is an algorithm
that multiplies an $n\times\lceil n^\theta\rceil$ matrix by an
$\lceil n^\theta\rceil\times n$ matrix using
$\bigO{n^{\omega_\theta}}$ operations~\cite{LGU18}.  However, this
refinement complicates notation, and would not be of use for our main
results. The same remark holds for several other runtime estimates in
this section, such as \cref{lemma:NZ,prop:sim_trunc_mod_mul}.

\subsubsection{Power projection}
\label{subsec:powerproj}

The transposition principle implies the
existence of an algorithm \Call{algo:PowerProjection}{} with the same
asymptotic runtime as 
\algoName{algo:ModularComposition-BrentKung}{} and
with the following signature~\cite{Shoup94}.

\begin{algorithm}[ht]
  \algoCaptionLabel{PowerProjection}{f,a,d,(r_i)_{0\le i < n}}
  \begin{algorithmic}[1]
  \Require $f$ of degree $n$ in $\xRing$, $a$ in $\xRing_{<n}$, $d$ in $\NN$, $(r_i)_{0\le i < n}$ in $\vecRing{n}$
  \Ensure $(\ell(1),\ell(a),\dots,\ell(a^{d-1} \bmod f))$, with $\ell(b_0 +\cdots + b_{n-1}x^{n-1}) 
    = r_0 b_0 + \cdots + r_{n-1} b_{n-1}$
  \end{algorithmic}
\end{algorithm}

Whereas seeing the details of
\algoName{algo:ModularComposition-BrentKung}{} is useful as a preamble
to the N\"usken-Ziegler algorithm, \algoName{algo:PowerProjection} {}
 only plays the role of a subroutine in one other algorithm given
just below. Moreover, giving its pseudocode would require us to
introduce concepts such as transposed product, that would not used any
further in this text. We refer the reader to~\cite{Shoup99}, which
gives all details but uses classical matrix arithmetic (with
$\omega_2=4$), so the runtime of that version is $\softO{d^2+nd}$ instead
of $\softO{(1 + n/d)d^{\omega_2/2}}$.

\subsubsection{Small minimal polynomial}
\label{subsubsec:minpoly}

Modular composition can be sped up when the minimal polynomial $\mu_a$
of $a$ modulo $f$ has degree at most $d$, for some (small) integer $d \le
n$. To compute $\polp(a)\rem f$, the idea is that once $\mu_a$
is known, $\tilde{\polp} = \polp \rem \mu_a$ can be computed, 
and then
$\polp(a)  \equiv \tilde{\polp}(a) \bmod f$ (see e.g.\
\cite[Sec.\,4.1]{HoeLec18}). The computation of the latter by \cref{lemma:BK}
benefits from $\tilde \polp$ having degree less than $d$.

It remains to discuss how to compute $\mu_a$. Here, we follow an
algorithm of Shoup (the deterministic version, for $f$ irreducible,
is in~\cite[Thm.\,3.4]{Shoup94}; the randomized one is
in~\cite[Sec.\,4]{Shoup95}). We take a random linear form
$\ell:\xRing/\genBy{f} \to \field$ and compute the
sequence $(\ell(a^k \bmod f))_{0 \le k < 2d}$. With high probability, its
minimal polynomial is $\mu_a$; the algorithm verifies whether it is
the case, and returns either a correct result or \Fail.
In \algoName{algo:ModularComposition-SmallMinimalPolynomial}{},
$\mu _a$
is computed using
an Extended Euclidean scheme called
\textproc{MinimalPolynomialForSequence}~\cite[Algo.\,12.9]{GaGe99}.

The following lemma analyses the runtime of this procedure, and the
probability of success. As per our convention at the end of
\cref{sec:preliminaries}, the ``random'' linear form $\ell$ is
actually given as an argument, through the vector
$(r_0,\dots,r_{n-1}) \in \vecRing{n}$ of its coefficients.

\begin{algorithm}
  \algoCaptionLabel{ModularComposition-SmallMinimalPolynomial}{f,a,\polp,d,(r_i)_{0\le i<n}}
  \begin{algorithmic}[1]
  \Require $f$ of degree $n$ in $\xRing$, $a$ in $\xRing_{< n}$,
  $\polp$ in $\yRing_{<n}$, $d$ in $\{1,\dots,n\}$, $(r_i)_{0\le i<n}$ in
  $\vecRing{n}$
  \Ensure $\polp(a) \rem f$ or \Fail
  \State $(v_0,\dots,v_{2d-1}) \gets \textproc{PowerProjection}(f,a,2d,(r_i)_{0\le i<n})$
  \State $\mu \gets \textproc{MinimalPolynomialForSequence}(v_0,\dots,v_
  {2d-1})$\label{step:minpolyseq} 
  \State $t \gets \Call{algo:ModularComposition-BrentKung}{f,a,\mu}$ \label{step:compute_t}
  \If{$t \ne 0$} \Return \Fail
  \Else{} \Return \Call{algo:ModularComposition-BrentKung}{f,a,\polp \rem \mu}
  \EndIf
  \end{algorithmic}
\end{algorithm}

\begin{lemma}\label{lemma:smallminpoly}
  Given $f\in\field[x]$ of degree~$n$, $a$ in~$\field[x]_{<n}$, $g$ in
  $\field[y]_{<n}$, $d$ in $\{1,\dots,n\}$ and $(r_i)_{0\le i<n}$ in
  $\field^n$,
  \algoName{algo:ModularComposition-SmallMinimalPolynomial}{} uses
  $\softO{nd^{(\omega_2/2)-1}}$ operations in $\field$ and returns
  either $\polp(a) \rem f$ or \Fail. If $\mu_a$ has degree at most
  $d$, and the entries of~$(r_i)_{0\le i<n}$ are chosen uniformly and
  independently from a finite subset $S$ of $\field$, then with
  probability at least $1-n/\card{S}$ the algorithm returns $\polp(a)
  \rem f$ and computes $\mu_a$ as a by-product. If $\mu_a$ has degree
  more than $d$, the algorithm returns \Fail.
\end{lemma}
\begin{proof}
  For any given $a$ in $\xRing_{<n}$ and $(r_0,\dots,r_{n-1})$,
  the algorithm computes a polynomial $\mu$ and tests whether $\mu(a)
  \equiv0 \bmod f$; if it is the case, it reduces $\polp$ modulo $\mu$
  before
  doing
  a modular composition. Hence, the output may only be $\polp(a) \rem
  f$ or \Fail, as claimed; it is \Fail~ if and only if the value $t$
  at \cref{step:compute_t} does not vanish.

  By the discussion in \cref{subsec:powerproj} and \cref{lemma:BK},
  the call to $\textproc{PowerProjection}$ takes $\softO{(1+n/d)d^
    {\omega_2/2}}$ operations in~$\field$; because we take $d \le n$,
  this is $\softO{nd^{(\omega_2/2)-1}}$.  \cref{step:minpolyseq} then
  computes a nonzero annihilating polynomial of degree at most~$d$ in
  $\softO{d}$ operations in $\field$~\cite[Algo.\,12.9]{GaGe99}.
  The remaining lines call 
  \algoName{algo:ModularComposition-BrentKung}{} with a last argument
  of
  degree at most $d$, so the cost is $\softO{nd^{(\omega_2/2)-1}}$
  again; this establishes the claim on the runtime.
  
  Suppose first that $\mu_a$ has degree greater than $d$. Then since
  $\deg(\mu) \le d$, $\mu(a) \rem f$ cannot be zero, so the output is
  \Fail, as claimed.

  Finally, suppose that the minimal polynomial $\mu_a$ has degree at
  most $d$ and that the entries of $(r_i)_{0 \le i<n}$ are chosen
  uniformly at random and independently from a set $S$ in
  $\field$. With $M_a$ the multiplication matrix by $a \bmod f$ and
  $\mathbf{1}$~the vector $(1,0,\ldots,0)$, the sequence
  $(M_a^k \mathbf{1})_{k \ge 0}$ is $(a^k \rem f)_{k \ge 0}$ and the sequence
  $(\trsp{(r_i)} M_a^k \mathbf{1})_{k \ge 0}$ is $(\ell(a^k \bmod
  f))_{k \ge 0}$. Following the probabilistic analysis of
  Wiedemann's algorithm~\citetext{\citealp[Lem.~2]{KaPa91};
    \citealp[Lem.~1] {KaSa91}}, the probability that their minimal
  polynomials coincide is at least $1-n/\card{S}$. When this occurs,
  \cref{step:minpolyseq} computes $\mu_a$; the value $t$ at
  \cref{step:compute_t} is then zero, and the output is $\polp(a)\rem
  f$.
\end{proof}

The main idea in this algorithm~---~computing an annihilating
polynomial for $a$ and using it to reduce $\polp$~---~is actually at
the core of our main algorithm as well. Key differences are that we
compute several annihilating polynomials (which we call
relations), and use them to reduce $\polp$ into a bivariate
polynomial. We then apply N\"usken and Ziegler's extension of
Brent and Kung's algorithm, which we present now.


\subsection{Bivariate composition}
\label{sec:sequence:biv_comp}

Here we describe the N\"usken-Ziegler algorithm for
modular composition \cite{NusZie04}, which computes $\polp(x,a) \rem f$ for
a bivariate $\polp$ in $\xyRing$.

First, however, we address the following question: given $f$ of degree
$n$ in $\xRing$, $a$ in $\xRing_{<n}$ and an $s$-tuple
$(\polp_0,\dots,\polp_{s-1})$ in $\xyRing_{<(m,r)}^s$,
compute all compositions
\[
  (\polp_0(x,a) \rem f,\dots,\polp_{s-1}(x,a) \rem f)\in\xvecRing{s}.
\]
The solution designed by N\"usken and Ziegler~\cite{NusZie04} boils
down to a multiplication of polynomial matrices. Writing the
polynomials~$\polp_i$ as the rows of their coefficients in~$y$ gives an
$s\times r$ matrix~$\polP$ whose entries are polynomials in
\(\xRing_{<m}\). Writing the powers of~$1,a,\dots,a^{r-1}\rem f$ in a
column vector~$A$ reduces the simultaneous composition to a
matrix-vector product~$\polP A$. This is turned into a matrix-matrix
product by spreading the coefficients of~$A$ as follows. If
\[
  \polp_i(x,y) = \sum_{0\le j<r} \polp_{i,j}(x) y^j,
\]
then computing the product
\[
  B=
  \begin{pmatrix}
    \polp_{0,0}(x)&\dotsb&\polp_{0,r-1}(x)\\
    \vdots&&\vdots\\
    \polp_{s-1,0}(x)&\dotsb&\polp_{s-1,r-1}(x)
  \end{pmatrix}
  \begin{pmatrix}
    [a^0\rem f]_0^{m-1}&\dotsb&[a^0\rem f]_{(\lceil n/m\rceil-1)m}^{m-1}\\
    \vdots&&\vdots\\
    [a^{r-1}\rem f]_0^{m-1}&\dotsb&[a^{r-1}\rem f]_{(\lceil n/m\rceil-1)m}^{m-1}
  \end{pmatrix}
\]
yields a matrix whose entry~$B_{i,\ell}$ is
\[B_{i,\ell}=\sum_{0\le j<r}{\polp_{i,j}(x)[a^j\rem f]_{\ell m}^{m-1}}.\]
Summing the~$B_{i,\ell}x^{\ell m}$ modulo $f$, for
$\ell=0,\dots,\lceil n/m\rceil-1$, then provides $\polp_i(x,a)\rem f$
for $i=0,\dots,s-1$ at low cost. This is detailed in
\cref{algo:SimultaneousBivariateModularComposition,lemma:NZ}.

\begin{algorithm}
  \algoCaptionLabel{SimultaneousBivariateModularComposition}{f,a,\polp_0,\dots,\polp_{s-1},m,r}
  \begin{algorithmic}[1]
  \Require $f$ of degree $n$ in $\xRing$, $a$ in $\xRing_{<n}$,
  $(\polp_0,\dots,\polp_{s-1})$ in $\xyRing_{<(m,r)}^s$
  \Ensure $(\polp_0(x,a)\rem f,\dots,\polp_{s-1}(x,a)\rem f)$
    \State \(\hat{a}_0 \gets 1\)\label{step:sim_biv_mod_eval:hata}
    \State\InlineFor{\(i = 1, \ldots, r-1\)}{$\hat{a}_i \gets a \cdot \hat{a}_{i-1} \rem f$}
          \Comment{\(\hat{a}_i = a^i \rem f\)}    \label{step:sim_biv_mod_eval:powers}
    \State $A \gets$ matrix $([\hat{a}_i]_{jm}^{m-1})_{
    \begin{subarray}{l}0\le i<r\\ 0\le j< \lceil n/m\rceil\end{subarray}}$
    in \(\xmatRing{r}{\lceil n/m \rceil}_{<m}\) \label{step:sim_biv_mod_eval:A}
    \State $\polP \gets$ matrix $(\polp_{i,j}(x))_{\substack{1\le i\le
    s\\0\le j<r}}$ in \(\xmatRing{s}{r}_{<m}\), where \(\polp_i(x,y) =
    \sum_j \polp_{i,j}(x) y^j\)
 \label{step:sim_biv_mod_eval:P}
    \State $B = (B_{i,j})_{\begin{subarray}{l}{0 \le i < r}\\ 0 \le
    j<
    \lceil
    n/m\rceil\end{subarray}} \gets \polP A$ 
    \label{step:sim_biv_mod_eval:pmm}
    \State\InlineFor{$i=0,\dots,s-1$}{$b_i \gets (\sum_{0\le j<\lceil n/m\rceil}{B_{i,j}x^{jm}}) \rem f$} \label{step:sim_biv_mod_eval:bi}
    \State \Return \((b_0,\ldots,b_{s-1})\)
  \end{algorithmic}
\end{algorithm}

\begin{lemma}[{\cite[Lem.\,10({\normalfont iii})]{NusZie04}}]
  \label{lemma:NZ}
  \algoName{algo:SimultaneousBivariateModularComposition}{}
  computes $(\polp_0(x,a)\rem f,\dots,\polp_{s-1}(x,a)\rem f)$. 
  Assuming
  $s\in \softO{r}$, it uses
  $\softO{\bicost{r^2}}=\softO{(m+n/r^2)r^{\omega_2}}$ operations in
  $\field$, \bicostrecall.
\end{lemma}
\begin{proof}
  \cref{step:sim_biv_mod_eval:hata,step:sim_biv_mod_eval:powers} use $\softO{rn}$
  operations. Similarly, for each~$i=0,\dots,s-1$,
  \cref{step:sim_biv_mod_eval:bi} uses $\lceil n/m\rceil$ additions
  in~$\bigO{m}$
  operations each and one reduction in~$\softO{m+n}$ operations. The
  total cost of \cref{step:sim_biv_mod_eval:bi} is
  thus~$\softO{s(n+m)}$.

  \cref{step:sim_biv_mod_eval:A,step:sim_biv_mod_eval:P} do not use any arithmetic
  operation. The most expensive step is \cref{step:sim_biv_mod_eval:pmm}, the product of an $s\times r$
  matrix by an $r\times \lceil n/m\rceil$ matrix, both with entries in
  $\xRing_{<m}$. Using the same kind of block decomposition as
  in~\cref{lemma:BK}, this is done using $\lceil \lceil n/m \rceil/r^2 \rceil \in
  n/(mr^2) +\bigO 1 $ products in sizes $s \times r$ and $r \times r^2$.
  With the assumption~$s\in \softO{r}$, each of them uses $\softO{m
    r^{\omega_2}}$ operations in $\field$, for a total of
  $\softO{\bicost{r^2}}=\softO{(m+n/r^2)r^{\omega_2}}$ operations
  in $\field$.

  The other steps, in~$\softO{(r+s)(n+m)}=\softO{rm+rn}$, are at
  most of the same
  order, since $\omega_2 \ge 3$.
\end{proof}

\algoName{algo:SimultaneousBivariateModularComposition}{} is the
central step in
bivariate composition as showed in 
\algoName{algo:BivariateModularComposition}{},
leading to the complexity stated in \cref{prop:NuskenZiegler}.

\begin{algorithm}
  \caption[\textproc{BivariateModularComposition}]{\textproc{BivariateModularComposition}{$(f,a,\polp)$}
 \hfill (N\"usken-Ziegler algorithm \cite{NusZie04})}
  \label{algo:BivariateModularComposition}
  
  \begin{algorithmic}[1]
  \Require $f$  of degree $n$ in  $\xRing$, $a$ in $\xRing_{<n}$, $\polp$ in $\xyRing_{<(m,d)}$
  \Ensure $\polp(x,a)\rem f$
    \State $r \gets \lceil d^{1/2}\rceil$, \(s \gets \lceil d / r\rceil\)\label{bivmodcomp:r}
    \State Write $\polp(x,y)=\polp_0 (x,y) + \polp_1(x,y)y^r + \dots + \polp_{s-1}(x,y)y^{r(s-1)}$ with $\deg_y(\polp_i)<r$ for $0 \le i< s$\label{bivmodcomp:p}
    \State $(b_0,\ldots,b_{s-1}) \gets
    \Call{algo:SimultaneousBivariateModularComposition}{f,a,\polp_0,\dots,\polp_{s-1},m,r}$\newline
    \Comment{$b_i = \polp_i(x,a)\rem f$} \label{bivmodcomp:bi} 
    \State $\hat{a} \gets a^r\rem f$ \Comment{is computed in the\label{bivmodcomp:a}
    previous step}
    \State \Return \(b_0+ b_1 \hat{a} + \dots + b_{s-1} \hat{a}^{s-1}\) \Comment{Horner evaluation}\label{bivmodcomp:ret}
  \end{algorithmic}
\end{algorithm}

\begin{proposition}[{\cite[Thm.\,9]{NusZie04}}]
  \label{prop:NuskenZiegler}
  Given $f\in\field[x]$ of degree~$n$, $a\in\field[x]_{<n}$,
  $\polp$ in $\xyRing_{<(m,d)}$,
\algoName{algo:BivariateModularComposition}{} computes $\polp
(x,a) \rem f$
  using $\softO{\bicost{d}}= \softO{(m+n/d)d^{\omega_2/2}}$
  operations in $\field$, \bicostrecall.
\end{proposition}
\begin{proof}
The correctness of the algorithm is straightforward. For the
complexity analysis, we first note that $s \sim r \sim d^{1/2}$.
\Cref{lemma:NZ} then shows that the complexity of
\cref{bivmodcomp:bi} is $\softO{\bicost{r^2}}=
\softO{\bicost{d}}$. The other task involving arithmetic operations
is the final Horner evaluation which costs \(\softO{rn}\).
As in the proof of \cref{lemma:NZ}, this is smaller than
the other part, since $\omega_2 \ge 3$.
\end{proof}

\subsection{Sequence of truncated modular powers} \label{sec:truncated_pow}

Another key ingredient in our composition algorithm also relies on polynomial
matrix multiplication. To our knowledge this is a new algorithm, whose
properties are summarized in the next lemma.
\begin{lemma}
  \label{prop:sim_trunc_mod_mul}
  Given $f$ of degree $n$ in $\xRing$, $(p_0,\dots,p_{r-1})$ in
  $\xRing_{<n}^r$, $ (q_0,\dots,q_{s-1})$ in $\xRing_{<n}^s$ and
  $m\in\mathbb{N}_{>0}$,
\algoName{algo:SimultaneousTruncatedModularMultiplication}{}
computes the
  simultaneous truncated modular multiplications
\[\{[p_iq_j\rem f]_0^{m-1}\mid 0\le i < r,0\le j < s\}.\]
  If \(s \in \softO{r}\), it uses $\softO{\bicost{r^2}}=\softO{
  (m+n/r^2)r^{\omega_2}}$
  operations in~$\field$, \bicostrecall.
\end{lemma}
The basic approach to this problem is to first compute all the
products $p_iq_j$ modulo $f$ and then truncate the computed
polynomials. However, this produces an intermediate result of size
$nrs$, which is $\Theta(nr^2)$ when $s$ is in $\Theta(r)$, and is
larger than our target complexity.

Hereafter, we use the reversal of a polynomial $p\in\xRing$ with
respect to \(m\in\NN\) defined by $\rev(p,m) = x^m p(1/x)$; when
$m=\deg(p)$ this is the classical reciprocal of the polynomial \(p\).

\begin{algorithm}
  \algoCaptionLabel{SimultaneousTruncatedModularMultiplication}{f,(p_i)_{i<r},(q_j)_{j<s},m}
  
  \begin{algorithmic}[1]
      \Require
      $f$ of degree $n$ in $\xRing$, $(p_0,\dots,p_{r-1})$ in $\xRing_{<n}^r$,
      $(q_0,\dots,q_{s-1})$ in $\field [x]_{<n}^s$, $m\in\mathbb{N}_{>0}$

      \Ensure
      $([p_iq_j\rem f]_0^{m-1})_{\substack{0\le i < r\\ 0\le j<s}}$

    \State $(\ell,t) \gets$ (quotient,remainder) in the Euclidean
    division $n-m-1=\ell m+t$ with $\ell=0$ if $m\ge n$
           \label{step:sim_trunc_mod_mul:init}
    \State \InlineFor{$i=0,\dots,r-1$}{$\bar{p}_i \gets \rev(p_i,n-1)$}
    \State \InlineFor{$j=0,\dots,s-1$}{$\bar{q}_j \gets$ power series expansion $\rev(q_j,n-1) /\rev(f,n) \rem x^{n-1}$}    \label{step:sim_trunc_mod_mul:revq_invrev}
    \State Form the matrices
      \[P_1 \gets ([\bar{p}_i]_{jm+t}^{m-1})_{\substack{0\le i< r\\
      0\le j\le\ell}} \in \xmatRing{r}{(\ell+1)}_{<m} \quad
      P_2 \gets ([\bar{p}_i]_{jm+t}^{m-1})_{\substack{0\le i< r\\ 0
      \le j < \ell}} \in \xmatRing{r}{\ell}_{<m}\]
    \Statex {(note that \(P_2\) is the \(r \times \ell\) left submatrix of \(P_1\)),} and
      \[Q_1 \gets ([\bar{q}_j]_{(\ell-i)m}^{m-1})_{\substack{0\le
      i\le\ell\\ 0\le j< s}} \in \xmatRing{(\ell+1)}{s}_{<m} \quad
      Q_2 \gets ([\bar{q}_j]_{(\ell-1-i)m}^{m-1})_{\substack{0\le i <
      \ell\\ 0\le j< s}} \in \xmatRing{\ell}{s}_{<m}\]
      \State $H \gets [P_1 Q_1]_0^{m-1} + [P_2 Q_2]_m^{m-1} + \left(
      \left[[\bar{p}_i]_0^{t-1}\,[\bar{q}_j]_{\ell m+1}^
      {m-1+t-1}\right]_{t-1}^{m-1}\right)_
      {\substack{0\le i< r\\ 0\le j< s}}$
          \Comment{$H = (\bar h_{i,j})_{i,j}$ is in \(\xmatRing{r}{s}_{<m}\)}
         \label{step:sim_trunc_mod_mul:truncated_quotient}
       \State \InlineFor{$i=0,\dots,r-1$ and $j=0,\dots,s-1$}{$r_{i,j} \gets (p_iq_j-\rev(\bar h_{i,j},m-1)f) \rem x^m$}
      \label{step:sim_trunc_mod_mul:truncated_remainder}
    \State \Return $(r_{i,j})_{\substack{0\le i < r\\ 0\le j <
    s}}$
  \end{algorithmic}
\end{algorithm}

\begin{proof}
  For all $i < r$ and $j < s$, let $p_iq_j=h_{i,j}f+r_{i,j}$, with
  $\deg(r_{i,j})<n$, be the Euclidean division of the product~$p_iq_j$
  by the polynomial~$f$. The main task of the algorithm is to compute
  the truncated quotients~$[h_{i,j}]_0^{m-1}$
  (\crefrange{step:sim_trunc_mod_mul:init}{step:sim_trunc_mod_mul:truncated_quotient});
  from there, the truncated remainders~$[r_{i,j}]_0^{m-1}$ are easily
  obtained (\cref{step:sim_trunc_mod_mul:truncated_remainder}) at a
  total cost of $\softO{mrs}$ operations.

  For the efficient computation of the quotients $h_{i,j}$, we rely on
  the classical approach via reciprocals and power series operations.
  More specifically we use the identity
  \[
    \rev(h_{i,j},n-2)=
    \frac{\rev(p_i,n-1)\rev(q_j,n-1)}{\rev(f,n)}
    \rem x^{n-1} = \bar p_i \bar q_j \rem x^{n-1},
  \]
  obtained by evaluating $p_iq_j=h_{i,j}f+r_{i,j}$ at $1/x$ and
  multiplying by~$x^{n-2}/f(1/x)=x^{2n-2}/\rev(f,n)$;
  here we have $\bar{p}_i=\rev(p_i,n-1)$ and $\bar{q}_j=
  \rev(q_j,n-1)/\rev(f,n) \rem x^{n-1}$, as in the
  pseudocode.  The idea of our algorithm is to compute only the last
  $m$ coefficients of this expansion by means of two polynomial matrix
  multiplications.

  For any
  $t\in\{0,\dots,m-1\}$, for any polynomials $a,b$ written as
  \[
    a = [a]_0^{t-1} + x^t \sum_{i\ge 0} a_i x^{im} \text{ with}\, \deg(a_i) < m,
    \quad
    b = \sum_{j\ge 0} b_j x^{jm} \text{ with}\, \deg(b_j) < m,
  \]
  and for any positive integer $\ell$, one has
  \begin{equation}\label{eq:prodFG}
    [ab]_{\ell m+t}^{m-1} = \left[\sum_{i+j=\ell}a_ib_j\right]_0^{m-1}
    + \left[\sum_{i+j=\ell-1}a_ib_j\right]_{m}^{m-1}
    + \left[[a]_0^{t-1}[b]_{\ell m+1}^{m-1+t-1}\right]_0^{m-1}.
  \end{equation}
  (The last summand is a product of small degree polynomials that is 0
  when $t=0$.) We use this formula with $\ell$ and $t$ as defined
  in \cref{step:sim_trunc_mod_mul:init}, so that the left-hand side is
  $[ab]_{n-m-1}^{m-1}$; applying this to $a = \bar p_i$ and $b= \bar
  q_j$ gives $\bar h_{i,j}=[\rev(h_{i,j},n-2)]_{n-m-1}^{m-1}$, and thus 
  $[h_{i,j}]_0^{m-1}$  by reversal.

  Since $\ell \sim n/m$, using this formula for a single pair $i,j$
  requires $\softO{n}$ operations in \(\field\) and thus is as costly
  as computing $\bar p_i \bar q_j \rem x^n$.  In our algorithm the gain comes from
  using this formula simultaneously for several products, in which
  case matrix multiplication helps.

  The first multiplication in 
  \cref{step:sim_trunc_mod_mul:truncated_quotient} is the matrix
  product
  \[\begin{pmatrix}
  [\bar{p}_0]_t^{m-1}&\dotsb&[\bar{p}_0]_{n-m-1}^{m-1}\\
  \vdots&&\vdots\\
  [\bar p_{r-1}]_t^{m-1}&\dotsb&[\bar p_{r-1}]_{n-m-1}^{m-1}
  \end{pmatrix}
  \begin{pmatrix}
[\bar q_0]_{\ell m}^{m-1}&\dotsb&[\bar q_{s-1}]_{\ell m}^{m-1}\\
\vdots&&\vdots\\
[\bar q_0]_{0}^{m-1}&\dotsb&[\bar q_{s-1}]_{0}^{m-1}\\
  \end{pmatrix}
  .\]
 Its entries
  are the first summand in~\cref{eq:prodFG} for $a=\bar p_i$ and
  $b=\bar q_j$, for $0\le i < r$ and $0\le j < s$. Similarly, the second
  summand in~\cref{eq:prodFG} is obtained from the matrix product
  \[\begin{pmatrix}
  [\bar{p}_0]_t^{m-1}&\dotsb&[\bar{p}_0]_{n-2(m-1)}^{m-1}\\
  \vdots&&\vdots\\
  [\bar p_{r-1}]_t^{m-1}&\dotsb&[\bar p_{r-1}]_{n-2(m-1)}^{m-1}
  \end{pmatrix}
  \begin{pmatrix}
[\bar q_0]_{(\ell-1) m}^{m-1}&\dotsb&[\bar q_{s-1}]_{(\ell-1) m}^{m-1}\\
\vdots&&\vdots\\
[\bar q_0]_{0}^{m-1}&\dotsb&[\bar q_{s-1}]_{0}^{m-1}\\
  \end{pmatrix}.\]
 In terms of complexity, the multiplication \(P_1 Q_1\) involves \(r
 \times (\ell+1)\) and \((\ell+1) \times s\) matrices, while \(P_2
 Q_2\) involves \(r \times \ell\) and \(\ell \times s\) matrices; all
 four operands have degree less than \(m\). 

 Since $\ell =\lfloor (n-1)/m \rfloor-1$, we have $\ell+1\le n/m$, so
 each matrix product can be done using at most $\lceil n/(mr^2) \rceil
 \le n/(mr^2) + 1$ products in sizes $r \times r^2$ and $r^2 \times
 s$. Since \(s\in\softO{r}\), each of these take 
 $\softO{m r^{\omega_2}}$, for a total cost 
 of $\bicost{r^2} = \softO{(m+n/r^2)r^{\omega_2}}$.

   The other operations performed by the algorithm are $\bigO{r}$
   power series expansions at precision $n-1$ in $\softO{n}$
   operations each (precisely, one inverse and \(t\) multiplications,
   see \cref{step:sim_trunc_mod_mul:revq_invrev}), and $\bigO{r^2}$
   power series expansions at precision~$m$ in~$\softO{m}$ operations
   each (precisely, at most \(3rs\) multiplications and \(rs\)
   subtractions, see
   \cref{step:sim_trunc_mod_mul:truncated_quotient,step:sim_trunc_mod_mul:truncated_remainder}).
   This amounts to a total of \(\softO{n r + m r^2}\) operations, and
   can thus be neglected, since $\omega_2 \ge 3$.
\end{proof}
Using simultaneous truncated modular multiplication combined with a
baby steps/giant steps strategy leads to
\algoName{algo:TruncatedPowers}, with the following properties.

\begin{proposition}
  \label{prop:TruncatedPowers}
  Given $f$ in $\xRing$ of degree $n$, $a$ and $b$ in $\xRing_{<n}$,
  $m$ and $d$ in \(\NN_{>0}\), \algoName{algo:TruncatedPowers}
  computes the truncations
  \[
    [a^kb\rem f]_0^{m-1},\qquad 0\le k<d 
  \]
  using~$\softO{\bicost{d}}= \softO{(m+n/d)d^{\omega_2/2}}$ operations in~$\field$,
  \bicostrecall.
\end{proposition}
\begin{proof}
  The algorithm computes $1,a,\ldots,a^{r-1} \rem f$ and
  $b,ba^{r},\ldots,ba^{(s-1)r} \rem f$, which costs $\softO{nr}$ operations
  in~$\field$ since $r\sim s$. From these two sets of polynomials,
  \algoName{algo:SimultaneousTruncatedModularMultiplication}
  is then used to
  compute $[ba^k\rem f]_0^{m-1}$ for $0 \le k \le rs-1$ using
  $\softO{(m+n/r^2)r^{\omega_2}}$ operations, by \cref{prop:sim_trunc_mod_mul};
  since $\omega_2\ge3$, this is larger than $\softO{nr}$. The choice of $s$
  makes $(s-1)r<d\le rs$, so the output consists of the terms \(k=i+rj\) for
  \(j<s-1\) and \(i<r\), and for \(j=s-1\) and
  $i< d-(s-1)r \in\{1,\dots,r\}$.
\end{proof}

\begin{algorithm}
  \algoCaptionLabel{TruncatedPowers}{f,a,b,m,d}
  \begin{algorithmic}[1]
    \Require $f$ of degree $n$ in $\xRing$, $a$ and $b$ in $\xRing_{<n}$, $m$ and $d$ in \(\NN_{>0}\)
    \Ensure the truncated powers $[ba^k\rem f]_0^{m-1}$ for $0\le k<d$
    \State $r \gets \lceil d^{1/2}\rceil$; $s \gets \lceil d/r \rceil$ 
    \State \(\hat{a}_0 \gets 1\);
    \InlineFor{\(i=1,\ldots,r\)}{$\hat{a}_i \gets a \cdot \hat{a}_{i-1} \rem f$}
    \Comment{\(\hat{a}_i = a^i \rem f\)}
    \label{step:trunc_pow:baby_powers}
    \State \(\bar{a}_0 \gets b\);
    \InlineFor{\(j=1,\ldots,s-1\)}{$\bar{a}_j \gets \hat{a}_r \cdot \bar{a}_{j-1} \rem f$}
    \Comment{\(\bar{a}_j = b a^{jr} \rem f\)}
    \State $(c_{i,j})_{\substack{0\le i < r\\ 0\le j < s}}
    \gets$
    \Statex $\quad\Call{algo:SimultaneousTruncatedModularMultiplication}{f,\hat{a}_0,\ldots,\hat{a}_{r-1},\bar{a}_0,\ldots,\bar{a}_{s-1},m}$

\State \InlineFor{\(i=0,\ldots,r-1\) and \(j=0,\ldots,s-2\)}{$r_{i+rj} \gets c_{i,j}$}
\Statex \InlineFor{\(i=0,\ldots,d-1-(s-1)r\)}{$r_{i+r(s-1)} \gets c_{i,s-1}$}
\State \Return $(r_{k})_{0\le k < d}$    
  \end{algorithmic}
\end{algorithm}

Finally, \algoName{algo:BlockTruncatedPowers} computes truncations of
products of the form $x^i a^k \rem f$, which are needed in our
composition algorithm; here, we assume that $f(0)$ is nonzero (see
\cref{rmk:fat0}).

\begin{algorithm}
  \algoCaptionLabel{BlockTruncatedPowers}{f,a,m,d}
  \begin{algorithmic}[1]
    \Require $f$ of degree $n$ in $\xRing$, with $f_0=f(0)\neq0$, $a$ in
    $\xRing_{<n}$, $m$ and $d$ in \(\NN_{>0}\)
    \Ensure the truncated powers $[x^ia^k\rem f]_0^{m-1}$, for $0\le i<m$ and $0\le k<d$
    \State $(r_{k})_{0\le k < d}\gets\Call{algo:TruncatedPowers}
    {f,a,x^{m-1}\rem f,2m-1,d}$
    \label{nonblockforblock}
    \Comment{$r_k=[x^{m-1}a^k]_0^{2m-2}$}
    \State{$f_n\gets\operatorname{coeff}(f,n)$}\Comment{leading
    coefficient}
    \For{$k=0,\dots,d-1$} \label{step:trunc_pow:outer_for}
    \State \(a_{m-1,k} \gets r_{k}\)
    \For{$i=m-1,\dots,2,1$} \label{loopm2}
    \State $c \gets -a_{i,k}(0)/f_0$
    \State $a_{i-1,k} \gets (a_{i,k} + c [f]_0^{m+i-1}) / x$
    \Comment{$a_{i-1,k} = [x^{i-1}a^k\rem f]_0^{m+i-2}$}
    \EndFor
    \EndFor
    \State \Return $([a_{i,k}]_0^{m-1})_{\substack{0 \le i < m\\ 0 \le
    k < d}}$ \end{algorithmic}
\end{algorithm}

\begin{proposition}
  \label{prop:block_truncated_powers}
  Given $f$ in $\field[x]$ of degree~$n$ with $f(0)\neq0$, $a$ in
  $\field[x]_{<n}$, $m$ and $d$ in $\mathbb{N}_{>0}$,
  \algoName{algo:BlockTruncatedPowers} computes
  \[
    [x^i a^k \rem f]_0^{m-1}, \quad  0\le i<m,\ 0\le k<d
  \]
  using $\softO{\bicost{d}} + \bigO{m^2d} = \softO{(m+n/d)d^{\omega_2/2}} +
  \bigO{m^2d}$ operations in~$\field$, \bicostrecall.
\end{proposition}
\begin{proof}
  \cref{prop:TruncatedPowers} shows that the first step
  computes the sequence $[x^{m-1}a^k]_0^{2m-2}$ for $k=0,\dots,d-1$ in
  the announced complexity.  The remaining truncations are obtained
  from the identity
  \[
    [xp\rem f]_0^j= x[p\rem f]_0^{j-1}-\frac{p_{n-1}}{f_n}[f]_0^j,
  \]
  for any integer \(j\) and polynomial \(p\),
  where \(p_{n-1}\) is the coefficient of degree \(n-1\) of \(p \rem
  f\) and $f_n$ is the coefficient of degree~$n$ in~$f$. If
  we know $[xp\rem f]_0^j$, we get \(-p_{n-1}f_0/f_n\) as its constant
  coefficient, whence $p_{n-1}$ since $f_0\neq0$ and from there $
  [p\rem
    f]_0^{j-1}$ is easily obtained. At iteration $k$ of the loop at
  \cref{step:trunc_pow:outer_for}, the truncation $[x^{m-1} a^k \rem
    f]_0^{2m-2}$ computed previously is used to deduce all $
       [x^{m-1-i} a^k \rem f]_0^{2m-2-i}$ for $1 \le i < m$ in
       $\bigO{m^2}$ operations. Thus this loop has a total cost of
       $\bigO{m^2d}$ operations.
\end{proof}

\begin{remark}
  \label{rmk:fat0}
  The assumption $f(0)\neq 0$ is harmless in the context of
  modular composition: in the computation of $\polp(a)\rem f$, one can
  rather evaluate $\polp(y)$ at $a(x+c)$ modulo $f(x+c)$ for a
  randomly chosen~$c \in \field$, and unshift the result.  See
  \cref{step:shifta-mainalgo,step:return-mainalgo} in
  \algoName{algo:ModularCompositionBaseCase}.
\end{remark}


\subsection{Notes}
\subsubsection{Linear algebra interpretation}
\label{subsec:linalginterpret}

Representing polynomials by their vector of coefficients leads to
viewing the operations performed by Algorithms
\nameref{algo:BivariateModularComposition} and
\nameref{algo:TruncatedPowers} as computing the product of special
matrices by column vectors. Recall the notation~$M_a$ for the $n\times
n$ matrix of multiplication by $a\bmod f$ in the
basis~$(1,x,\dots,x^{n-1})$, and $X$ for the
matrix~$\trsp{(\idMat{m}\;\;0)} \in \field^{n \times m}$ with $m \in
\{1, \ldots, n\}$. Then Algorithms
\nameref{algo:BivariateModularComposition} and
\nameref{algo:TruncatedPowers} correspond respectively to
multiplication by

\begin{equation}\label{eq:factor_Hk}
  \Ra=\begin{pmatrix}X&\dotsb&M_a^{d-1}X\end{pmatrix}
  \in\matRing{n}{(md)}
  \quad\text{and}\quad
  \La=\begin{pmatrix} \trsp{X}\\ \vdots\\ \trsp{X}M_a^{d-1}\\ \end{pmatrix}
  \in\matRing{(md)}{n}.
\end{equation}

Indeed, $\Ra$ is the matrix of the mapping $\rhoa$ of bivariate modular
composition with bounded degrees, as computed by
\algoName{algo:BivariateModularComposition}:
\begin{align*}
\rhoa:\xyRing_{<(m,d)}&\rightarrow \xRing_{< n}\\
\polp(x,y)&\mapsto \polp(x,a)\rem f.
\end{align*}
On the other hand, $\La$ represents the mapping $\lba$ that extracts the
low-degree part of multiplications by powers of~$a$, as computed by
\algoName{algo:TruncatedPowers}:
\begin{align*}
  \lba:\xRing_{<n}&\rightarrow\xRing_{<m}^d\\
  b&\mapsto([b\rem f]_0^{m-1},\dots,[ba^{d-1}\rem f]_0^{m-1}).
\end{align*}
These maps and matrices play an important role in the study of the
generic behavior of our algorithm starting from 
\cref{sec:genericity}.

\subsubsection{Complexity equivalence}

\Cref{prop:NuskenZiegler} (\algoName{algo:BivariateModularComposition}) and
\Cref{prop:TruncatedPowers} (\algoName{algo:TruncatedPowers}) give
similar complexity bounds for the evaluation of $\rhoa$ and $\lba$,
but the computational equivalence of these problems, possibly up to
some conditions, is is still unclear to us in general.

However, for $m=1$, when $f(0)\neq0$, these two problems are indeed
equivalent.  This is a consequence of the transposition principle in
an indirect way, starting from the equality \[\Laone v_b =
\trsp{(M_b\Raone)}\mathbf{1},\] where $v_b$ is the vector associated
to $b$, $M_b$ is the matrix of multiplication by $b\bmod f$,
and~$\mathbf{1}$ is the first canonical vector.  First, this equality
gives a way to evaluate $\lba[1,n]$ for the cost of one multiplication
by~$\trsp{M_b}$ (i.e., $\softO{n}$ by the transposition principle),
plus one multiplication by the transpose of $\Raone$, which has the
same asymptotic cost as that of $\Raone$ itself, by the same
principle. Conversely, if $v=\trsp{M_b}\mathbf{1}$, the equality reads
$\trsp{(\Raone)}v= \Laone v_b$, so that, again by the transposition
principle, the evaluation of $\rhoa[1,n]$ reduces to that of
$\lba[1,n]$ provided~$v_b$ can be computed from~$v$ in low
complexity. When $f(0)\neq 0$, this can be done in $\softO{n}$ by
solving a linear system of Hankel type~\cite[Sec.\,3]{Shoup99}.

If $f(0)=0$, it is unclear whether such a reduction holds: in the
special case $f=x^n$, the map $\lba[1,n]$ becomes much simpler, as it
simply computes the sequence $a_0^i b_0$ for $0\le i < d$, where $a_0$
and $b_0$ are the constant coefficients of $a$ and $b$. This only
requires a linear number~$\bigO{d}$ of operations. On the other hand,
$\kappa_{1,d}$ is the composition of a univariate polynomial
\(\polp(y)\) of degree less than~$d$ with the power series~$a(x)$ and
no quasi-linear complexity result is known for this operation.


\subsubsection{Transposition of the N\"usken-Ziegler algorithm} \label{subsubsec:transpNZ}
Finally, we discuss a different approach to
\algoName{algo:BlockTruncatedPowers}, that actually bypasses
\algoName{algo:TruncatedPowers} altogether, and uses the transpose of 
\algoName{algo:BivariateModularComposition} instead.

\algoName{algo:BlockTruncatedPowers} computes the \(m\times m\)
projections $H_k=\trsp{X}M_{a^k}X$, for $k=0,\dots,d-1$, using the
fact that for $k < d$, $H_k$ can be deduced in $\bigO{m^2}$ operations
(for loop at \cref{loopm2}) from the column vector $\trsp{\bar
  X}M_{a^k}u$ of size $2m-1$, where $\bar X
=\trsp{(\idMat{2m-1}\;\;0)} \in \field^{(2m-1)\times n}$ and $u$ is
the $m$th column of \(X\), i.e.~the \(m\)th canonical vector
(\cref{nonblockforblock}). (Here we have taken $m\leq (n+1)/2$.)
\algoName{algo:TruncatedPowers} computes the vectors $\trsp{\bar
  X}M_a^ku$ for $0\leq k < d$ using $\softO{\bicost{d}}$ operations.

Alternatively, we can consider a recursion similar to the one in the
proof of \cref{prop:block_truncated_powers}, but now for learning a
new coefficient of a polynomial rather than a coefficient of a new
polynomial. Assuming $f(0)\neq 0$, for a polynomial $p$ one has
\[
  [xp\rem f]_{i+1}^0= [p\rem f]_i^0 + (c/f_0) [f]_{i+1}^0,
\]
where $c$ is the coefficient of degree $0$ of \(xp \rem f\): we see
that from the row vector $\trsp{\mathbf{1}}M_{x^{-m+1}a^k}\bar X$, one
can also deduce $H_k=\trsp{X}M_{a^k}X$ using $\bigO{m^2}$ operations.

Now, if we set $v=\trsp{M_{x^{-m+1}}}\mathbf{1}$, computing
$\trsp{v}\Ramm$ precisely gives all vectors
$\trsp{\mathbf{1}}M_{x^{-m+1}a^k}\bar X$, for $0 \le k < d$. Since $v$ can be computed in quasi-linear time, the
application of the transposition principle to
\algoName{algo:BivariateModularComposition} shows that these vectors
can be computed using $\softO{\bicost{d}}$ operations. Altogether,
this gives an alternative to \algoName{algo:BlockTruncatedPowers} with
the same asymptotic complexity.


\section{Matrices of relations for composition}
\label{sec:relmat_intro}

The heart of our algorithm for finding $\polp(a) \rem f$ is the
computation of a \emph{matrix of relations}, which gives a collection
of polynomials of small degree in the ideal~$\mathcal I$ generated by
$y-a$ and $f$ in~$\field[x,y]$. For a given positive integer~$m$,
these polynomials are in the $\yRing$-module $\rmodfa$ obtained by
degree restriction as $\mathcal{I}\cap\xyRing_{<(m,\cdot)}$. 

In \cref{sec:relmat:structure} we show that the invariant factors
of~$\rmodfa$ are the $m$ invariant factors of highest degree of the
characteristic matrix $y\idMat {n}-M_a$, where $M_a$ is the matrix of
multiplication by $a\bmod f$.  Once a matrix of relations has been
obtained, it can be used to perform composition by reducing univariate
composition to a small bivariate composition problem; this is
described in \cref{sec:bivcomposition}. Finally, in
\cref{sec:relmat:polynomials}, the results of this section are applied
to the efficient computation of annihilating polynomials for $a$
modulo $f$.

In all of \cref{sec:relmat_intro}, notation such as \(\rmodfa\) and
\(\ddfa\) is shortened into \(\rmod\) and \(\dd\), except for the main
definitions and statements, as there is no ambiguity as to the
dependency on $a$ or $f$.

\subsection{Structure of the module of relations}
\label{sec:relmat:structure}
This section introduces the module of relations $\rmodfa$ and relates it to the
characteristic matrix.

\subsubsection{Definitions}
\label{sec:relmat:def_and_degdet}

\paragraph{Relations}

We call \emph{relations} the polynomials of the ideal \(\ideal =
\idealGens\) of $\field[x,y]$; these are the bivariate polynomials
$\polr(x,y)$ such that $\polr(x,a)\equiv 0\bmod f$, i.e., they are
algebraic
relations satisfied by $a\bmod f$. We are interested in those
relations whose \(x\)-degree is bounded from above by a given positive
integer \(m\). They form the $\yRing$-module
\begin{equation*}
  \rmodfa = \left\{ \polr(x,y) \in \xyRing_{<(m,\cdot)}
                      \mid \polr(x,a(x))\equiv 0 \bmod f \right\}
          =  \ideal \cap \xyRing_{<(m,\cdot)},
\end{equation*}
which is denoted~$\rmod$ when $a$ and $f$ are clear from the context.

This is a $\yRing$-submodule of $\xyRing_{<(m,\cdot)}$, itself a free
\(\yRing\)-module with basis \((1,x,\ldots,x^{m-1})\); {we refer to \cite[Part
III]{DumFoo04} for basic notions of module theory and modules over
principal
ideal domains}. As stated in
\cref{sec:preliminaries}, we often identify a polynomial
$\polr_0(y)+ \cdots + \polr_{m-1}(y) x^{m-1}$ in
$\xyRing_{<(m,\cdot)}$ with the column vector \(\trsp{(\polr_0
  \;\cdots\; \polr_{m-1})}\) in \(\yvecRing{m}\) of its coefficients
on that basis. Since~\(\yRing\) is a principal ideal domain, \(\rmod\)
is free as well, and it has rank \(m\) since it contains
$\minpoly\yvecRing{m}$, where $\minpoly$ is the minimal polynomial of
$a \bmod f$.

In terms of ideals, there is a chain of inclusions \(\{0\} = \genBy{\rmod[0]}
\subseteq \cdots \subseteq \genBy{\rmod[n+1]} = \ideal\); the latter identity
follows from the fact that \(y-a\) and \(f\) have \(x\)-degree less than
\(n+1\). Furthermore \(\rmod[1] \neq \{0\}\) since \(\minpoly\) belongs to
\(\ideal\cap\yRing\). For small \(m\), the module \(\rmod\) may not contain all
the information in~\(\ideal\): the inclusion \(\genBy{\rmod} \subset \ideal\)
can be strict.

\paragraph{Matrix and basis of relations, determinantal degree} 

A \emph{matrix of relations of \(\rmod\)} is any nonsingular matrix
in~$\ymatRing{m}{m}$ whose columns are elements of the module~$\rmod$
(represented as column vectors). Such a matrix is further called a \emph{basis}
of relations if its columns generate \(\rmod\); all bases of relations of
\(\rmod\) can be obtained from any single one of them via right multiplication
by a unimodular matrix in \(\ymatRing{m}{m}\), i.e., a matrix whose determinant
is in \(\field\setminus\{0\}\). It follows that any matrix of relations of
\(\rmod\) is a square, nonsingular right multiple of any basis of relations of
\(\rmod\), and therefore bases of relations are exactly the matrices of
relations whose determinant has minimal degree. This degree is called the
\emph{determinantal degree} of the module $\rmod$. 

\subsubsection{Relation to invariant factors}

As a finitely generated module over a principal ideal domain,
\(\rmod\) has an invariant factor decomposition. The next result shows
that these invariant factors can be found in any triangular basis of
\(\rmod\), and that the largest of these factors is precisely
\(\minpoly\), the minimal polynomial of \(a\) modulo \(f\). It also
relates the degrees of these factors to the quantity $\ddfa$ (written
more simply as \(\dd\) when context is clear), already highlighted in
\cref{eq:defnu}, and which plays an important role in the analysis of
our approach. 

\begin{proposition}  \label{prop:invariant-factors}
  Let $B$ be an upper triangular basis of $\rmodfa$ for some $m\ge
  1$. Then its diagonal entries are the invariant factors of
  \(\rmodfa\), up to multiplication by nonzero elements of
  \(\field\). A number \(k\le\min(m,n)\) of these invariant factors
  are
  nontrivial, and these nontrivial ones are the~$k$ invariant
  factors of highest degree of the characteristic matrix~$\charmat$,
  which is a basis of relations of~$\rmodfa[n]$. The determinantal
  degree of $\rmodfa$ is the sum~$\ddfa$ of the degrees of these
  invariant factors, hence it satisfies $\min(m,n)\le \ddfa \le n$.
\end{proposition}

\subsubsection{Proof of \cref{prop:invariant-factors}}

Our proof relies on Lazard's structure theorem~\cite{Lazard85} on lexicographic
Gr\"obner bases in $\xyRing$. Here, the \emph{degree} of a zero-dimensional
ideal $\ideal \subset \xyRing$ is the dimension of the $\field$-vector space
$\xyRing/\ideal$.
\begin{lemma}[Lazard's structure theorem for bivariate ideals]
  Let \(\ideal\) be a zero-dimensional ideal of degree \(n\) in \(\xyRing\).
  Any minimal Gr\"obner basis of~\(\ideal\) for the \((y \prec
  x)\)-lexicographic order has the form
  \(\{\polr_0(y)h_k(x,y),\polr_1(y)h_{k-1}(x,y),\ldots,\polr_k(y)h_0(x,y)\}\) for some \(k
  \ge 1\), where
  \[
    \left\{
      \begin{array}{l}
        \polr_k=h_k=1\\
        n \ge \deg(\polr_0) > \cdots > \deg(\polr_k) = 0\\
        n \ge \deg_x(h_0) > \cdots > \deg_x(h_k) = 0 \\
        \text{for } 0 \le i < k,\ \polr_i \in \yRing \text{ is divisible by } \polr_{i+1} \\
        \text{for } 0 \le i \le k,\ h_i \in \xyRing \text{ has leading monomial a power of } x. \\
      \end{array}
    \right.
  \]
\end{lemma}

\begin{proof}
  The form of a minimal
  Gr\"obner basis of \(\ideal\) is given by Lazard's 
  result~\cite[Thm.\,1]{Lazard85}. The
  additional assumption that
  \(\ideal\) is zero-dimensional ensures that this Gr\"obner basis
  contains a polynomial whose leading term is a power of \(y\), hence
  \(h_k=1\), and one whose leading term is a power of \(x\), hence
  \(\polr_k=1\). Since \(\ideal\) has degree \(n\), there are precisely
  \(n\) monomials that are not multiples of the leading monomials of
  \(\{\polr_i h_{k-i} \mid 0 \le i \le k\}\). These leading monomials are
  \(\{x^{\deg_x(h_{k-i})} y^{\deg(\polr_i)} \mid 0 \le i \le k\}\), whence
  the bounds \(\deg(\polr_0) \le n\) and \(\deg_x(h_0) \le n\).
\end{proof}

\begin{corollary}  \label{cor:structure_module}
With the same notation, when~$\mathcal{I}=\langle f, y-a\rangle$ and
$m\ge 1$, a basis of
$\rmodfa=\mathcal{I}\cap\xyRing_{<(m,\cdot)}$ is  given by the
first $m$ polynomials in the sequence
  \begin{equation}\label{eq:sequence_lazard}
    (x^j\polr_0h_k)_{0 \le j < \delta_k},
    \dots,
    (x^j\polr_{k-1}h_1)_{0 \le j < \delta_1},
    (x^jh_0)_{j \ge 0},
  \end{equation}
  where \(\delta_i =
   \deg_x(h_{i-1})-\deg_x(h_i)\). If
   $s=\deg_x(h_0)=\delta_1+\dots+\delta_k$, the nontrivial invariant
   factors of \(\rmodfa\) are the first
  \(\min(m,s)\)
  polynomials in
  \begin{equation}\label{eq:inv_fact_lazard}
\bigl (  \underbrace{\polr_0, \ldots, \polr_0}_{\delta_k}, \;\; \ldots, \;\; 
\underbrace{\polr_{k-1}, \ldots, \polr_{k-1}}_{\delta_1} \bigr ).
  \end{equation}
\end{corollary}
\begin{proof}
  The polynomials in the sequence in~\cref{eq:sequence_lazard} form a
  (nonfinite) Gr\"obner basis of \(\ideal\), made of polynomials of
  \(x\)-degree \(0,1,2,\ldots\) respectively~\cite[Prop.\,1]{Lazard85}. By
  design, the first $m$ elements in this sequence belong to \(\rmod\), and
  considering their $x$-degrees shows that they are $\yRing$-linearly
  independent.  

  Any polynomial \(p(x,y) \in \rmod\) is a \(\yRing\)-linear combination of the
  first \(m\) of these polynomials. Indeed, it can be divided by the Gr\"obner
  basis with a remainder equal to~\(0\); in view of its degree in~$x$, only
  these~$m$ polynomials are involved in the division.
  This proves the claim on the basis of~\(\rmod[m]\) described in
  \cref{eq:sequence_lazard}. 

  The matrix \(T(y) \in \ymatRing{m}{m}\) representing this basis (with basis
  elements written in columns) is upper triangular, with its first $\min(m,s)$
  diagonal entries being the first $\min(m,s)$ polynomials
  in~\cref{eq:inv_fact_lazard} in this order, and with its remaining diagonal
  entries being nonzero elements of \(\field\). Furthermore each of these
  diagonal entries divides all other entries in the same column, hence the
  Smith normal form of \(T(y)\) has the same diagonal entries as \(T(y)\),
  which proves the claim on the invariant factors of \(\rmod\).
\end{proof}
\begin{proof}[Proof of \cref{prop:invariant-factors}]
  \Cref{cor:structure_module} implies that the determinantal degree~$\nu_m$ of
  $\rmod$ is the sum of the degrees of the elements of the first $\min(m,s)$
  elements of \cref{eq:inv_fact_lazard}. It follows that
  \[
    \dd\le\delta_k \deg(\polr_0) + \cdots + \delta_1 \deg(\polr_{k-1})=n,
  \]
  where the last identity comes from considering the $\field$-vector space
  dimension of $\xyRing/\ideal$. If $s\le m$, all the nontrivial invariant
  factors appear and the bound is reached, while otherwise \(m < s\) and
  \(\degdet{B}\), being the sum of the degrees of \(m\) nonconstant
  polynomials, is at least~$m$.

  If $B$ is a basis of $\rmod$, then there exists a unimodular
  matrix~$U\in\ymatRing{m}{m}$ such that~$UB=T$ with~$T$ as in the previous
  proof. If moreover~$B$ is upper triangular, then so is~$U$ and since
  \(\det(U) \in \field\setminus\{0\}\), the diagonal entries of \(U\) belong to
  \(\field\setminus\{0\}\). It follows that \(B\) has the same diagonal entries
  as \(T\) up to multiplication by nonzero elements of \(\field\).

  The columns of the characteristic matrix $y\idMat{n}-M_a$ represent the
  polynomials $x^k(y-a(x)) \rem f$ for $0 \le k < n$, making this matrix a
  matrix of relations of \(\rmod[n]\). It has determinantal degree
  \(\deg(\charpoly) = n\), which coincides with the determinantal degree of
  $\rmod[n]$, by the previous inequalities. Thus $y\idMat{n}-M_a$ is actually a
  basis of $\rmod[n]$ and its invariant factors are given by the previous
  paragraph.
\end{proof}

\subsubsection{Note}

For $m\in \{1,\ldots, n\}$, the module of relations $\rmod$ is
isomorphic to the module of vector generators for the matrix
sequence~$\{M_a^kX\}_{k\geq 0}$, where $X=\trsp{(\idMat{m}\;\;0)} \in
\matRing{n}{m}$ as above (this elementary fact is established
within the proof of \cref{lemma:denominators}, for instance); the
bases of relations are the {\em minimal generating polynomials} for
that sequence~\cite{Vil97:TR,KaVi05}. 

The relation between Coppersmith's block Wiedemann algorithm and
invariant factors of a characteristic matrix was described by Kaltofen
and Villard: they show that for
generic
projections~$V$ and~$W$ in $\field^{n \times \ell}$ and $\field^{n
  \times m}$, with $\ell \ge m$, the invariant factors of minimal
generating polynomial of the sequence $(\trsp{V}A^kW)_{k\ge0}$ are the
$m$ invariant factors of largest degree of the characteristic matrix
$y\idMat{n}-A$~\cite[Thm.\,2.12]{KaVi05}. In our more specific setting,
\cref{prop:invariant-factors} shows that this relation holds when the
right projection is the structured matrix~$X$ (see also
\cref{notes:reconstructing}).

\subsection{Composition using matrices of relations} 
\label{sec:bivcomposition}

Matrices of relations are used to reduce the univariate problem $\polp(a) \rem
f$ with $g\in \yRing$, to a bivariate one with better degree properties, thanks
to a matrix division.

\subsubsection{Division for polynomial matrices}
\label{sec:bivcomposition:division}

If $R$ is a nonsingular matrix in~$\ymatRing{m}{m}$ and $v_{\polp}$ is a vector
in $\yvecRing{m}$, then there exist quotient and remainder vectors~$w$
and~$v_{\tilde \polp}$ such that
\begin{equation} \label{eq:defvecdiv}
v_{\polp}=Rw + v_{\tilde \polp},
\end{equation}
and each entry of $v_{\tilde \polp}$ has degree less than that of the
corresponding row of $R$ \cite[Thm.\,6.3-15, p.\,389]{Kailath80}. The latter reference actually 
states a stronger condition on $v_{\tilde \polp}$, namely that the
matrix fraction $R^{-1} v_{\tilde \polp}$ is \emph{strictly proper} (see
\cref{sssec:fractions}); this implies the above degree condition
\cite[Lem.\,6.3-10, p.\,383]{Kailath80}, which is sufficient for our
needs.

For computing this division, it is customary to use \(\yRing\)-linear system
solving. For this, we rely on a kernel basis algorithm \cite{ZLS12}: this
returns $v$ in $\yvecRing{m}$ and $r$ in $\yRing$ such that $R^{-1}v_{\polp} =
v/r$, with~$r$ of minimal degree. From this the remainder is obtained as
$v_{\tilde \polp} =R\, (v \rem r)/r$, and here we do not need the quotient
vector \(w\).


\subsubsection{Composition Algorithm}
\label{sec:bivcomposition:algo}

In the case where $v_{\polp}=\trsp{(\polp\; 0\cdots 0)}$ and $R$ is a matrix of
relations of $\rmod$ of degree at most \(d\), the remainder in the above
division is a vector \(v_{\tilde \polp}\) of degree less than \(d\) whose
entries yield ${\tilde \polp}\in\xyRing_{<(m,d)}$ such that $\polp - \tilde
\polp \in \rmod$. Thus, analogously to a reduction modulo a Gr\"obner basis of
the ideal $\ideal = \genBy{y-a,f}$, this provides a bivariate polynomial
$\tilde\polp$ with smaller degree in $y$ and controlled degree in \(x\), and
such that $\tilde \polp - \polp \in \ideal$, that is, \(\tilde\polp(x,a) \equiv
\polp(a) \bmod f\).

\algoName{algo:BivariateModularCompositionWithRelationMatrix}{} is
given a matrix of relations \(R\) of $\rmod$ as a parameter and
performs this division; then it completes the composition by
evaluating $\tilde \polp(x,a) \rem f$ using
\algoName{algo:BivariateModularComposition}.
\algoName{algo:BivariateModularCompositionWithRelationMatrix} actually
accepts a slightly more general input: $\polp$ can be a bivariate
polynomial with $x$-degree less than $m$ (however, the rest of the
article focuses on the case of $\polp$ in $\yRing$ highlighted above).
The algorithm accepts $g$ of arbitrary degree in $y$, but the cost 
analysis is done under the assumption $\deg_y(g) \in O(n)$.

\begin{algorithm} 
  \algoCaptionLabel{BivariateModularCompositionWithRelationMatrix}{f,a,\polp,R}
  \begin{algorithmic}[1] 
    \Require
    \parbox[t]{0.8\textwidth}{
      $f$ of degree $n$ in $\xRing$, $a$ in $\xRing_{<n}$, $\polp$ in $\xyRing_{<(m,.)}$,

      $R\in\ymatRing{m}{m}_{\le d}$ a matrix of relations of~$\rmodfa$
    }
    \Ensure  $\polp(x,a)\rem f$
    \State Write $\polp(x,y)=\polp_0(y)+\polp_1(y)x+\dots+\polp_{m-1}(y)x^{m-1}$ and
    set $v_\polp \gets \trsp{(\polp_0 \cdots \polp_{m-1})} \in \yvecRing{m}$
    \State \CommentLine{Compute \(v \in \yvecRing{m}\) and \(r \in \yRing\) using \cite[Algo.\,1]{ZLS12}}
    \Statex \(\begin{pmatrix}  v \\ r \end{pmatrix} \in \yvecRing{m+1}
        \gets \hyperlink{cite.ZLS12}{\textproc{MinimalNullspaceBasis}}((R \,\;\; -v_\polp),(d,\ldots,d,\deg_y(\polp)))\)
        \label{algo:BivariateModularCompositionWithRelationMatrix:systemsolve}
    \State \(v_{\tilde \polp} \gets R\, (v \rem r)/r  \;\in \yvecRing{m}_{<d}\) \Comment{$v \rem r$ is the vector of entry-wise remainders}
      \label{algo:BivariateModularCompositionWithRelationMatrix:vecrem}
    \State $\tilde \polp(x,y) \gets$ the polynomial in $\xyRing_{<(m,d)}$ corresponding to $v_{\tilde \polp}$
      \label{algo:BivariateModularCompositionWithRelationMatrix:polyrem}

    \State \Return \Call{algo:BivariateModularComposition}{f,a,{\tilde \polp}} \Comment{\({\tilde \polp}(x,a) \rem f\)}, \cref{algo:BivariateModularComposition}
      \label{algo:BivariateModularCompositionWithRelationMatrix:modcomp}
  \end{algorithmic}
\end{algorithm}

\begin{proposition}
  \label{prop:comp-from-matrix}
  Given $f$ in $\field[x]$ of degree~$n$, $a$ in $\field[x]_{<n}$,
  $\polp$ in $\xyRing_{<(m,.)}$ with $\deg_y(\polp)=O(n)$ and a matrix
  of relations~$R$ in $\yRing_{\le d}^{m\times m}$ of $\rmodfa$,
  \algoName{algo:BivariateModularCompositionWithRelationMatrix}
  computes $\polp(x,a)\rem f$ using $\softO{m^{\omega} (d+n/m) +
    \bicost{d}}$ operations in~$\field$, \bicostrecall.
\end{proposition}
\begin{proof}
  First, \cref{algo:BivariateModularCompositionWithRelationMatrix:systemsolve}
  computes \(r\in\yRing\) and \(v = rR^{-1}v_\polp \in \yvecRing{m}\) with
  \(r\) of minimal degree. Indeed, since \(R\) is nonsingular, the right kernel
  of \((R \,\;\; -v_\polp) \in \ymatRing{m}{(m+1)}\) has rank \(1\). We use
  \cite[Algo.\,1]{ZLS12} to compute a basis \(\trsp{(\trsp{v} \;\; r)}\) of
  this kernel. Thus by construction $Rv=rv_\polp$ holds, and the fact that
  \(\trsp{(\trsp{v} \;\; r)}\) generates the kernel ensures that the greatest
  common divisor of \(r\) and all the entries of \(v\) is \(1\), hence the
  minimality of \(\deg(v)\) and \(\deg(r)\). 

  At \cref{algo:BivariateModularCompositionWithRelationMatrix:vecrem} one
  considers the vector \(\bar{v} = v \rem r \in \yvecRing{m}\) such that
  \(\deg(\bar v) < \deg(r)\) and \(v = rw + \bar v\) for some \(w \in
  \yvecRing{m}\). It follows that \(v_\polp = R v/r = R w + v_{\tilde \polp}\),
  where \(v_{\tilde \polp} = R \bar{v}/r\) is the vector computed at
  \cref{algo:BivariateModularCompositionWithRelationMatrix:vecrem};
  by construction the \(i\)th entry of \(v_{\tilde{\polp}}\) has degree less
  than that of the \(i\)th row of \(R\). In short,
  \cref{algo:BivariateModularCompositionWithRelationMatrix:systemsolve,algo:BivariateModularCompositionWithRelationMatrix:vecrem}
  compute a vector \(v_{\tilde \polp} \in \yvecRing{m}\) that has
  degree less
  than \(d\) and is a remainder of \(v_\polp\) modulo \(R\). Since \(R\) is a
  matrix of relations, the polynomial \(\tilde \polp(x,y)\) at
  \cref{algo:BivariateModularCompositionWithRelationMatrix:polyrem} is such
  that \(\tilde \polp(x,a) \equiv  \polp(x,a) \bmod f\). The correctness
  follows, since \(\tilde \polp(x,a) \rem f\) is the polynomial returned by
  \Call{algo:BivariateModularComposition}{f,a,\tilde \polp} (see
  \cref{prop:NuskenZiegler}).

  As required by Algorithm~1 of~\cite{ZLS12}, the tuple of
  integers
  \((d,\ldots,d,\deg_y(\polp))\in \ZZ^{m+1}\) bounds the column
  degrees of \((R \,\;\; -v_\polp)\). Then, since the sum of this
  tuple is \(md+\deg_y(\polp)\), with $\deg_y(\polp)=O(n)$,
  \cref{algo:BivariateModularCompositionWithRelationMatrix:systemsolve}
  costs \(\softO{m^{\omega} (d + n/m)}\) operations~\cite[Thm.\,4.1]{ZLS12}. The minimality of \(\deg(r)\) implies
  \(\deg(r) \le \degdet{R} \le md\), and then \(v\) has degree at most
  \(\degdet{R} R^{-1} v_\polp \le (m-1)d + n\) since \(\det(R)R^{-1}\)
  is the transpose of the cofactor matrix of $R$. Thus the computation
  of \(\bar{v} = v \rem r\) in
  \cref{algo:BivariateModularCompositionWithRelationMatrix:vecrem}
  uses \(\softO{m (md+n)}\) operations, which is smaller than the cost
  of
  \cref{algo:BivariateModularCompositionWithRelationMatrix:systemsolve}. Next,
  the matrix-vector product \(R\bar{v}\) can be performed in
  \(\softO{m^{\omega} d}\) operations: write the column \(\bar{v}\) of
  degree \(< md\) as \(m\) columns of degree \(<d\) via \(y^d\)-adic
  expansion; use a matrix-matrix product to left-multiply these
  columns by \(R\); finally recombine the resulting columns into a
  single column that gives \(R\bar{v}\). To obtain \(v_{\tilde
    \polp}\) it remains to divide each entry of \(R\bar{v}\) by~\(r\),
  which costs \(\softO{m^2 d}\) since \(\deg(R\bar{v}) < (m+1)d\).  By
  \cref{prop:NuskenZiegler}, the call at
  \cref{algo:BivariateModularCompositionWithRelationMatrix:modcomp}
  uses $\bicost{d}$ operations. The cost bound in the
  \namecref{prop:comp-from-matrix} follows.
\end{proof}

\subsubsection*{Note}

Comparing \cref{prop:comp-from-matrix} with \cref{prop:NuskenZiegler},
note that when $m\sim n^{\eta}$ and $d \sim n^{1-\eta}$ with~$\eta$
from \cref{eq:def-beta}, then the complexity bound of
\cref{prop:comp-from-matrix} is the same as the one given by the
N\"usken-Ziegler algorithm, however the $y$-degree of~$\polp$ can now
go up to the order of $n$.

\subsection{Annihilating polynomials using matrices of relations}
\label{sec:relmat:polynomials}

Our main algorithm requires an annihilating polynomial for~$a$, that is, a
polynomial $h$ in $\yRing$ such that $h(a) \equiv 0 \bmod f$. It can readily be
obtained from a matrix of relations.

\begin{proposition} \label{cor:annihilating}
  Let $R \in \ymatRing{m}{m}_{\le d}$ be a matrix of relations of $\rmodfa$. Its
  determinant is a nonzero annihilating polynomial for~$a$ modulo~$f$. It has
  degree at most $md$ in~$\yRing$ and can be computed from~$R$ using
  $\softO{m^{\omega}d}$ operations in \(\field\).
\end{proposition}
\begin{proof}
  As a polynomial combination of relations in $\rmod$, the entry $(1,1)$ of the
  (upper triangular) Hermite normal form of a matrix of relations is a relation
  in \(\genBy{f,y-a} \cap \yRing\), so it is a nonzero multiple of the minimal
  polynomial of $a$. This implies the same property for the determinant, since
  it is a multiple of that entry. The bound on the degree of the determinant is
  straightforward, and the cost bound is from~\cite[Thm.\,1.1]{LVZ17}.
\end{proof} 

\subsubsection*{Note} For the computations of the minimal polynomial and
of the characteristic polynomial of $a$ modulo $f$, see \cref{sec:minpoly}.


\section{Computing matrices of relations}
\label{sec:relmat_comp}

In this section, we give an algorithm computing a matrix of relations.
This study may be viewed as a specialization of the formalism developed by
Kaltofen and Villard for the block Wiedemann approach (see
\cref{subsubsec:blocks,subsubsec:effprojandsmallbivpols}) in terms of
manipulations of bivariate polynomials in the ideal generated by $y-a$ and $f$. 

As already done in \cref{sec:relmat_intro}, notation such as
\(\rmodfa\) and \(\ddfa\) is shortened into~\(\rmod\) and~\(\dd\)
in this section, except in the main statements.

In \cref{sec:relmat:relations_denominators}, we show that for $m \in \{1,
\ldots, n\}$, denominators
of irreducible right matrix fraction descriptions of $(y\idMat
{n}-M_a)^{-1} X$ with
$X=\trsp{(\idMat{m}\;\;0)} \in \field^{n\times m}$ yield bases of
$\rmodfa$. For efficiency reasons, a further truncation is
required: this leads us to introduce modules~$\rmod[\ell,m]^{(a,f)}$
whose bases are the denominators of irreducible right matrix fraction
descriptions of $\trsp{Y}(\charmat)^{-1}X$, where $\trsp{Y} =
(\idMat{\ell}\;\;0) \in \field^{\ell\times n}$, with $\ell \in \{1,
\ldots, n\}$; thus we use structured left and right block projections.
If $\ell=n$, $Y$ is the identity matrix of size $n$, and we recover
$\rmodfa$, but this value is too large for our cost
objectives. Instead, we focus on $\ell=m$, and thus \(Y=X\).

\cref{sec:relmat:approx_basis} describes how a basis of
$\rmod[\ell,m]^{(a,f)}$ can be reconstructed using 
\emph{minimal approximant bases}~\cite{BarBul92,BeLa94}, from sufficiently many terms of the
power series expansion of the matrix $H = \trsp{X}(\charmat)^{-1}X$.

This strategy is turned into an algorithm for computing matrices of
relations in \cref{subsec:candidate}: the expansion of \(H\) is
obtained via \algoName{algo:BlockTruncatedPowers}, while approximant
bases are computed using a matrix Pad\'e version of the
Berlekamp-Massey algorithm~\cite{BeLa94,GJV03}. The correctness and
efficiency of this approach depends on a fundamental condition on
$M_a$, i.e., on \(f\) and \(a\) (\cref{prop:compute_Mmm}, first
item). First, it expresses that the left projection does not prevent
us from getting the right denominators of $(\charmat)^{-1} X$ from
those of $H$. It also ensures the existence of matrices of relations
of ``small'' degree , and in this way appropriately limits the number
of terms of the expansion of $H$ that are required for the
reconstruction. We  prove in \cref{sec:genericity} that these
properties are satisfied for generic inputs; in
\cref{sec:composition_randomized}, we further study cases where
randomization can ensure such a condition.

Verifying the condition on \(\mulmat\), or verifying that a certain
matrix is a matrix of relations, are expensive tasks: except for some
restricted cases, the algorithm of \cref{subsec:candidate} does not
certify that its output is indeed a matrix of relations. As such, this
would lead to a Monte Carlo composition algorithm. To achieve Las
Vegas composition instead, in \cref{sec:relmat:certify} we propose an
algorithm which either detects that the output
mentioned above is not
a matrix of relations, or uses this output to build a certified matrix
of relations of slightly larger dimensions.

\subsection{Matrices of relations as denominators of matrix fractions}
\label{sec:relmat:relations_denominators}

This section relates denominators of some matrix fractions to bases of the
module of relations~$\rmod$ and of a truncated version \(\rmodlm\) of it.

\subsubsection{Definitions}
\label{sssec:fractions}

\paragraph{Matrix Fractions}

We first recall several notions on matrix fractions that can be found in
Kailath's book~\cite[Chap.\,6]{Kailath80}. {Let \(N\) be in
  \(\ymatRing{\ell}{m}\), let \(D \in \ymatRing{m}{m}\) be nonsingular, and
  consider the rational matrix $F=ND^{-1}\in \field(y)^{\ell\times m}$. Then
  \(N D^{-1}\) is called a \emph{right fraction description} of $F$. Similarly,
  if $F=\hat{D}^{-1}\hat{N}$, then \(\hat{D}^{-1}\hat{N}\) is called a
  \emph{left fraction description} of $F$.  The right fraction \(N D^{-1}\) is
  said to be \emph{irreducible} if \(N\) and \(D\) are right coprime, i.e.~any
  right divisor common to \(N\) and \(D\) is unimodular, or equivalently \(UN +
  VD = \idMat{m}\) for some \(U \in \ymatRing{m}{\ell}\) and \(V \in
  \ymatRing{m}{m}\) \cite[Lem.\,6.3.5 p.\,379]{Kailath80}. The fraction
  \(ND^{-1}\) is said to be \emph{strictly proper} if for each nonzero entry of
  the rational matrix $F=ND^{-1}$, the degree of
the numerator is less than the degree of the denominator.} A
matrix~$F\in\field(y)^{\ell\times m}$ is said to be \emph{describable in
degree~$d$} if it admits both a left and a right fraction description with
denominators of degree at most~$d$.

\paragraph{Truncated Module of Relations}

For efficiency reasons, we consider a $\yRing$-module similar to \(\rmod\), but
where only the first $\ell$ coefficients of the polynomials are required to be
$0$, for some positive integer $\ell$ {with \(\ell \le n=\deg(f)\)}.
Explicitly, for $\ell, m \in \mathbb{N}_{>0}$ we define the $\yRing$-modules
\[
  \rmod[\ell,m]^{(a,f)} = \left\{ \polr(x,y) \in \xyRing_{<(m,\cdot)}
  \mid \big[a(x)^k\polr(x,a(x))\rem f\big]_0^{\ell-1}=0 \text{ for all } k\ge 0 \right\},
\]
together with the usual simplified notation $\rmod[\ell,m]$. They
satisfy the inclusions $\rmod[1,m]\supseteq \rmod[2,m]\supseteq
\dots\supseteq \rmod[n,m]=\rmod$. 
{If $r(x,a(x)) \rem f =0$ then $a(x)^kr(x,a(x)) \rem f =0$ for all $k\geq 0$, 
but note that this is no longer true if truncated polynomials are considered. This explains the presence 
of $k$-th powers of~$a$ in the definition of $\rmod[\ell,m]^{(a,f)}$, while they are not necessary in 
the definition of $\rmodfa$ in \cref{sec:relmat:def_and_degdet}.}

 The determinantal degree
of~$\rmodlm$ is denoted $\ddlm$. Of particular interest is the
case when $\rmod[m,m]=\rmod$.

\subsubsection{Relation between bases of relations and
denominators of matrix fractions}

\begin{proposition}
  \label{prop:modprime_is_denom}
  For $\ell, m\in\{1,\dots,n\}$, the columns of a matrix~$D\in\ymatRing{m} {m}$
  form a basis of~$\rmod[\ell,m]^{(a,f)}$ if and only if $D$ is the denominator of an
  irreducible right fraction description~$ND^{-1}$ of
  \[
    (\idMat{\ell} \;\; 0)(\charmat)^{-1} X \;\;\; \in\ymatRing{\ell}{m};
 \]
 the denominator of any right fraction description of this matrix is a right
 multiple of any such basis~$D$.
\end{proposition}

\subsubsection{Proof of \cref{prop:modprime_is_denom}}

For a matrix of rational functions~$F\in \field(y)^{\ell\times m}$, we
let
\begin{equation}
  \label{eq:denomF}
  \denom{F} = \{v \in \yvecRing{m} \mid F v \in \yvecRing{\ell}\},
\end{equation}
which is a \(\yRing\)-submodule of \(\yvecRing{m}\) of rank \(m\). Then, we can
establish the relation between the module~$\rmodlm$ and the matrix in
\cref{prop:modprime_is_denom}.
\begin{lemma}\label{lemma:denominators}
  For $\ell, m$ in $\{1,\dots,n\}$, one has \(\rmodlm = \denombig{(\idMat{\ell} \;\; 0)(\charmat)^{-1} X}\).
\end{lemma}
\begin{proof}Taking $\trsp{Y} = (\idMat{\ell}\
0)$, define \(H(y)= \trsp{Y} (\charmat)^{-1} X\) and $H_k = \trsp{Y}
  M_a^k X \in \matRing{\ell}{m}$, so that, by power series
  expansion in~$y^{-1}$,
  \[
    H(y)
    = \sum_{k\ge0} H_k y^{-k-1}
    = \sum_{k\ge0} \trsp{Y} M_a^k X y^{-k-1}.
  \]

  Let \(\polr(x,y) = \sum_{0\le i\le d} \polr_i(x)y^i \in
  \xyRing_{<(m,\cdot)}\) be of \(y\)-degree \(d\), and let $v_i\in\vecRing{m}$
  be the coefficient vector of~$\polr_i$ for $i=0,\dots,d$. Then, for $k\geq
  0$,
  \begin{equation*}
    \left[a^k \polr(x,a) \rem f\right]_{0}^{\ell-1}
    = \left[\sum_{0\le i\le d} a^{k+i} \polr_i \rem f\right]_{0}^{\ell-1}
    = \sum_{0\le i\le d} \left[a^{k+i} \polr_i \rem f\right]_{0}^{\ell-1}
  \end{equation*}
  and \([a^{k+i}\polr_i \rem f]_{0}^{\ell-1}\) has coefficient vector \(\trsp{Y}
  \mulmat^{k+i} X v_i = H_{k+i} v_i \). Hence,  $\left[a^k \polr(x,a) \rem
  f\right]_{0}^{\ell-1}$ has coefficient vector $H_kv_0 + \dots +
  H_{k+d}v_{d}$. Therefore \(\polr(x,y)\) is in \(\rmodlm\) if and
  only if
  \begin{equation}
    \label{eq:Hankel-prod}
    H_kv_0 + \dots + H_{k+d}v_{d} = 0 \quad \text{for all } k\ge 0. 
  \end{equation}
  On the other hand, {defining \(v = \sum_{0 \le i \le d} v_i y^i\) and
  setting \(H_k=0\) for \(k<0\)}, the expansion of \(H v\) at infinity reads
  \begin{equation} \label{eq:smallnumer}
    H v
    = \sum_{k\ge 0} H_k y^{-k-1}  \sum_{0 \le i \le d} v_i y^i 
    = \sum_{k \ge -d} \left(H_kv_0 + \dots + H_{k+d}v_{d}\right) y^{-k-1},
  \end{equation}
  which implies that \cref{eq:Hankel-prod} holds if and only if $Hv$
  has polynomial entries.
\end{proof}

\cref{prop:modprime_is_denom} is then a direct consequence of the following
general result on matrix fractions, which is a reformulation of
\cite[Thm.\,6.5-4 and Lem.\,6.5-5, p.\,441]{Kailath80}.

\begin{lemma} \label{lem:basisdenom}
  Let \(F \in \field(y)^{\ell\times m}\) be a matrix of rational
  fractions.  The columns of \(D \in \ymatRing{m}{m}\)
  form a basis of $\denom{F}$ if and only if $D$ is the denominator of
  an irreducible right fraction description \(N D^{-1}\) of
  \(F\). Besides, the denominator of any right fraction description of
  $F$ is a right multiple of such a $D$.
\end{lemma}

\subsubsection{Notes}
\label{notes:reconstructing}

The role of the truncated modules~$\rmodlm$ is to reduce the cost of
computations: we decrease the dimension of the relevant matrices using
a structured left projection. The more usual approach~\cite{KaVi05}
uses generic projections matrices; our choice here is similar to the
one used for the efficient computation of generic
resultants~\cite{Vil18}.

Although not used in this work, genericity on the left is sufficient:
if $V \in \matRing{n}{\ell}$ is generic with~$\ell \in \{m, \ldots,
n\}$, then one has $\rmod = \denom{(\charmat)^{-1} X} =
\denom{\trsp{V} (\charmat)^{-1} X}$. The latter occurs if and only if
$ \rank{\trsp{V} P, \trsp{V} PA, \trsp{V} PA^2, \ldots} = \dd $ for a
well chosen full rank matrix $P \in \matRing{n}{\dd}$, and a
restriction $A \in \matRing{\dd}{\dd}$ of $M_a$ to the invariant
subspace generated by $X$~\cite[Lem.\,4.2]{Vil97:TR}. The rank condition is satisfied for a generic
projection~\cite[Cor.\,6.4 and its proof]{Vil97:TR}.

In terms of generators of matrix sequences, \cref{eq:Hankel-prod}
shows that the denominators of \cref{prop:modprime_is_denom} are bases
of modules of vector generators for the matrix sequence
$\{(\idMat{\ell} \;\; 0) M_a^kX\}_{k\geq 0}$ \cite[Lem.\,2.8]{KaVi05}.

\subsection{Reconstructing denominators of matrix fractions via approximant bases}
\label{sec:relmat:approx_basis}

\algoName{algo:BlockTruncatedPowers} from \cref{sec:truncated_pow}
allows one to compute a truncated power series expansion of \(H(y) =
\trsp{X}(\charmat)^{-1}X\). When the precision of this expansion is
sufficient, a basis of $\rmodmm$ can be reconstructed.

\subsubsection{Definitions} \label{subsubsec:definreconstruct}

\paragraph{Weak Popov matrices}

Let \(P \in \ymatRing{m}{m}\) be a matrix whose column \(j\) has degree \(d_j
\ge 0\). The \emph{(column) leading matrix} of \(P\) is the matrix in
\(\matRing{m}{m}\) whose entry \((i,j)\) is the coefficient of degree \(d_j\)
of the entry \((i,j)\) of \(P\). Then \(P\) is said to be \emph{(column)
reduced} if its leading matrix is invertible. This is the case if and only
if~\cite[Eq.\,(24), p.\,384]{Kailath80}
\begin{equation}
  \label{eq:colreddet}
  \degdet{P} = d_1+\cdots+d_m.
\end{equation}
A (column) reduced matrix is in \emph{(column) weak Popov form} if its leading
matrix is invertible and upper triangular. Any submodule of \(\yvecRing{m}\)
has at least one basis which is in weak Popov form~\cite{Kailath80,BeLaVi99}.

\paragraph{Approximant bases}

Let $F\in\ySeries^{m\times k}$ be a matrix of power series and $\sigma\in\NN$
be a nonnegative integer. A matrix $P\in\ymatRing{k}{k}$  is an
\emph{approximant basis} of $F$ at order $\sigma$ if its columns form a basis
of the \(\yRing\)-module \(\{v \in \yvecRing{k} \mid F v \equiv 0 \bmod
y^\sigma\}\), which is free of rank \(k\). This approximant basis is said to be
\emph{minimal} if it is reduced. Minimal approximant bases are also called
$\sigma$-bases, or order bases \cite{BarBul92,BeLa94}.


\subsubsection{Denominators from approximant bases}

We are going to use approximant bases for solving equations of the
type of~\cref{eq:blockPade}. As pointed out in
\cref{ssec:matrel_small_degree}, we use expansions at $y=0$ rather
than infinity (see \cref{rmk:shift}).

\begin{proposition}
  \label{lem:fraction-reconstruction}
  Let \(H \in \field(y)^{m\times m}\) be strictly proper, and \(\delta\) be 
  the determinantal degree of~\(\denom{H}\) (notation from \cref{eq:denomF}). Suppose that
  \(H\) has a power series expansion \( H = \sum_{k\ge 0} S_k y^k\) at \(y=0\),
  with \(S_k \in \matRing{m}{m}\).  Let 
  \[
    F=\left(\sum_{k=0}^{2d-1} S_ky^k \;\; -\idMat{m}\right)
    \in \ymatRing{m}{(2m)},
  \]
  and let 
  \[
    P = \begin{pmatrix} D & P_{1} \\ N & P_{2} \end{pmatrix}
    \in \ymatRing{(2m)}{(2m)}
  \]
  be an approximant basis at order \(2d\) of $F$ in weak Popov form, with each
  submatrix of size \(m \times m\). Then the following properties hold:
  \begin{enumerate}[label={\it (\roman*)},ref={\roman*}]
    \item\label{lem:fraction-reconstruction:item:general} \(D\) is weak Popov;
      \(\deg(N) < \deg(D)\); the sum of the degrees of the diagonal entries of
      \(D\) is \(\degdet{D}\) and satisfies $\degdet{D} \le \delta$.
    \item\label{lem:fraction-reconstruction:item:certif} If \(\degdet{D} =
      \delta\) and each of the \(m\) rightmost columns of \(P\) has degree at
      least \(\deg(D)\), then~\(ND^{-1}\) is an irreducible description of $H$.
    \item\label{lem:fraction-reconstruction:item:reconstruct} If 
      \(H\) is describable in degree~$d$, then~\(ND^{-1}\) is an
      irreducible description of $H$
      such that \(\deg(D) \le d\) and each
      of the \(m\) rightmost columns of \(P\) has degree at least \(\deg(D)\).
  \end{enumerate}
\end{proposition}
The first item gives general properties of the approximant basis in
weak Popov form, whereas
\cref{lem:fraction-reconstruction:item:certif,lem:fraction-reconstruction:item:reconstruct}
give sufficient conditions to guarantee it recovers an irreducible
fraction description of $H$.

\subsubsection{Proof of \cref{lem:fraction-reconstruction}}

\begin{lemma}
  \label{lem:wpopov_minimality}
  Let \(P \in \ymatRing{m}{m}\). If \(P\) is reduced and \(B \in
  \ymatRing{m}{m}\) is a right multiple \(B = PU\) with \(U\)
  nonsingular, then
  \(\deg(P) \le \deg(B)\). {If \(P \in \ymatRing{m}{m}\) is weak Popov, with
  diagonal degrees \(d_1,\ldots,d_m \in \NN\), and \(v
  \in \yvecRing{m}\) is a nonzero right multiple \(v = P u\) whose
  bottom-most
  entry of largest degree is in row \(i\) and has degree \(d\), then \(d_i \le
d = \max_{1 \le j \le m} (\deg(u_j)+d_j)\), where \(u = (u_j)_{1\le j\le m}\).}
\end{lemma}
\begin{proof}
  The first claim follows from the predictable degree property
  \cite[Thm.\,6.3-13, p.\,387]{Kailath80}, {and so does the identity
    \(\deg(v) = d = \max_{1 \le j \le m} (\deg(u_j)+d_j)\) since
    \((d_1,\ldots,d_m)\) are also the column degrees of \(P\) by definition of
    a weak Popov form}. The inequality \(d_i \le d\) is from
    \cite[Lem.\,1.17]{Neiger2016}.
\end{proof}

We now prove \cref{lem:fraction-reconstruction}.  Consider an
irreducible fraction description \(QR^{-1} = H\) for some \(Q \in
\ymatRing{m}{m}\) and some weak Popov \(R \in \ymatRing{m}{m}\).
Since \(H\) is strictly proper we have \(\deg(Q)<\deg(R)\) and, more precisely,
the \(i\)th column of \(Q\) has degree less than the \(i\)th column of \(R\). Thus
the \(i\)th column of \((\begin{smallmatrix} R
  \\ Q \end{smallmatrix})\) has its bottom-most entry of largest degree
in row \(i\); let \(d_i\) be this degree.

In \cref{lem:fraction-reconstruction:item:general}, the first two
claims follow from the definition of \(P\) being weak Popov. In
particular $D$ is column reduced, hence \cref{eq:colreddet} shows that
\(\degdet{D}\) is the sum of column degrees of \(D\), which is also
the sum of diagonal degrees of \(D\) since \(D\) is weak Popov. The
identity \((H \; -\idMat{m}) (\begin{smallmatrix} R \\ Q\end{smallmatrix}) = 0\)
  implies \(F (\begin{smallmatrix} R \\ Q\end{smallmatrix}) = 0 \bmod y^{2d}\), and therefore
  \((\begin{smallmatrix} R \\ Q \end{smallmatrix})\) is a
  right multiple of \(P\). Hence, by \cref{lem:wpopov_minimality},
  \(d_i\) is at least the degree of the \(i\)th column of \(P\), which
  is the degree of the \(i\)th column of \(D\); it follows that
  \(\degdet{D} \leq d_1+\cdots+d_m\). On the other hand, since \(R\)
  is reduced 
{and since by \cref{lem:basisdenom} its columns form a basis of 
~\(\denom{H}\)}
  we have \(d_1+\cdots+d_m = \degdet{R}=\delta\), 
  proving the last
  claim of \cref{lem:fraction-reconstruction:item:general}.

  Concerning \cref{lem:fraction-reconstruction:item:certif}, the assumption
  \(\degdet{D} = \delta = \degdet{R}\) implies that the sum of column degrees
  of \(D\) is \(d_1+\cdots+d_m\), while as showed above the \(i\)th column of
  \(D\) has degree at most \(d_i\). Thus \(D\) has the same column degrees
  \((d_1,\ldots,d_m)\) as \(R\). In particular \(\deg(D) = \deg(R) > \deg(Q)\).
  Then, since by assumption the \(m\) rightmost columns of \(P\) have
  bottom-most entries of largest degree in rows at least \(m+1\) and of degree
  at least \(\deg(D)\), one can deduce from \cref{lem:wpopov_minimality} that
  \((\begin{smallmatrix} R \\ Q \end{smallmatrix})\) is a right
  multiple of the
  leftmost \(m\) columns of \(P\). {Indeed, let \(u = (\begin{smallmatrix} u_1
    \\ u_2 \end{smallmatrix}) \in \yvecRing{2m}\), with \(u_1\) and \(u_2\)
    each of dimension \(m\), such that the \(i\)th column of
    \((\begin{smallmatrix} R \\ Q \end{smallmatrix})\) is \(P u =
    (\begin{smallmatrix} D u_1 + P_1 u_2 \\ N u_1 + P_2 u_2
    \end{smallmatrix})\); we want to prove \(u_2=0\). Using the last identity
    in \cref{lem:wpopov_minimality} on each of the weak Popov matrices \(P\),
    \(D\), and \(P_2\), we observe that \(\deg(Pu) = \max(\deg(D u_1), \deg(P_2
    u_2))\); note \(\deg(Pu) = d_i \le \deg(D)\) by construction. On the other
    hand, since all diagonal degrees of \(P_2\) are at least \(\deg(D)\),
    \cref{lem:wpopov_minimality} shows \(\deg(D) \le \deg(P_2 u_2)\), provided
    that \(u_2\neq 0\), which we now assume by contradiction. This implies
    \(\deg(Du_1) \le \deg(Pu) = \deg(P_2u_2) = \deg(D)\), hence \(\max_{1\le j
    \le m} d_j + \deg(u_{1,j}) \le \deg(D)\) using
    \cref{lem:wpopov_minimality}, where \(u_1 = (u_{1,j})_{1\le j\le m}\). Now
    by definition of weak Popov forms, the \(j\)th column of \(N\) has degree
    less than \(d_j\) for \(1\le j\le m\), hence \(\deg(Nu_1) < \max_{1\le j\le
    m} d_j + \deg(u_{1,j})\). This gives \(\deg(Nu_1 + P_2u_2) = \deg(D)\),
    which is a contradiction since \(Nu_1 + P_2u_2\) is the \(i\)th column of
    \(Q\) and has degree strictly less than \(d_i\), itself at most \(\deg(D)\). So,
  \(u_2=0\). Gathering this over all columns \(1 \le i \le m\), this means}
    \((\begin{smallmatrix} R \\ Q \end{smallmatrix}) = (\begin{smallmatrix} D
    \\ N \end{smallmatrix}) U = (\begin{smallmatrix} DU \\ NU
  \end{smallmatrix})\) for some \(U \in\ymatRing{m}{m}\), and \(U\) is
  unimodular since \(R\) and \(D\) are nonsingular with
  \(\degdet{R}=\degdet{D}\). Hence \(H = QR^{-1} = ND^{-1}\) and the fraction
  \(ND^{-1}\) is irreducible.

  The following proof of \cref{lem:fraction-reconstruction:item:reconstruct}
  reflects that of \cite[Lem.\,3.7]{GJV03}. The assumption implies first the
  existence of a left fraction \(H = \hat{R}^{-1} \hat{Q}\) with
  \(\deg(\hat{Q}) < \deg(\hat{R}) \le d\), and second the degree bound
  \(\deg(R)\le d\) thanks to the degree minimality of reduced bases (see
  \cref{lem:wpopov_minimality}). The above paragraph shows in particular
  \(\deg(D) \le \max_i(d_i) = \deg(R) \le d\).
  
  Now, since \(\hat{R} (\sum_{0\le k < 2d} S_ky^k) \equiv \hat{Q} \bmod y^
  {2d}\),
  multiplying on the left by \(\hat{R}\) both sides of \(F (
  \begin{smallmatrix} D \\ N
  \end{smallmatrix}) \equiv 0 \bmod y^{2d}\) shows
  that \(\hat{Q}D - \hat{R} N\) is a right multiple of \(y^{2d} 
  \hat{R}\).  On
  the other hand, \(\hat{Q}D - \hat{R} N\) has degree less than \(2d\).  Hence
  it is zero, and \(H = \hat{R}^{-1} \hat{Q} = N D^{-1}\). To prove that the
  latter fraction is irreducible, assume by contradiction that \(D\) and \(N\)
  have a nonsingular common right divisor \(B \in \ymatRing{m}{m}\), with
  \(\degdet{B} > 0\). Then \(H = (NB^{-1}) (DB^{-1})^{-1}\) yields \(F
  (\begin{smallmatrix} DB^{-1} \\ NB^{-1} \end{smallmatrix}) \equiv 0 \bmod
  y^{2d}\), and \(P\, \diag{B^{-1},\idMat{m}}\) is a right multiple of
  \(P\) (since $P$ is a basis):
  this is impossible since \(\degdet{P\, \diag{B^{-1},\idMat{m}}} <
  \degdet{P}\).

  It remains to prove the last degree assertion. By contradiction, assume that
  $P$ has a column $(\begin{smallmatrix}v_0 \\ v_1\end{smallmatrix})$ of index
  larger than $m$ with $v_0$ and $v_1$ in $\yvecRing{m}$ both of degree less
  than $d$. Then an argument similar to the one above shows that \(\hat{Q}v_0 -
  \hat{R} v_1=0\). Altogether we obtain a matrix \((\begin{smallmatrix} D &
  v_0 \\ N & v_1 \end{smallmatrix})\) of rank \(m+1\) which is in the right
  kernel of $(\hat{Q} \;\; \hat{R}) \in \ymatRing{m}{(2m)}$ whose rank is
  \(m\): this is not possible. 

\subsubsection{Notes} 

The existence of appropriate left and right descriptions of~$H$ was used before
for the reconstruction of matrix fractions within the approximant
framework~\cite[Sec.\,3.2]{GJV03}. Our proof is similar to that of
\cite[Lem.\,3.7]{GJV03}, with the additional use of the weak Popov form.

Reduced forms were introduced \cite{Wolovich74} as a way to get a better
control over the degrees when computing with polynomial matrices and matrix
fractions, see e.g.~\cite[Lem.\,6.3-11, p.\,385]{Kailath80} for proper
fractions, \cite[Thm.\,6.3-13, p.\,387]{Kailath80} for a predictable degree
property, and \cite[Thm.\,6.5-10, p.\,458]{Kailath80} concerning the
\emph{minimality} of the column degrees. Weak Popov forms were introduced later
\cite{BeLaVi99,MulSto03} (under the name quasi-Popov and up to column
permutation) and provide a refined degree control as illustrated by
\cref{lem:wpopov_minimality}.

\subsection{Candidate basis of relations}
\label{subsec:candidate}

\algoName{algo:CandidateBasis} takes as input a polynomial
$f\in\xRing$ of degree $n$ with $f(0)\neq 0$, a polynomial
$a\in\xRing_{<n}$ such that $\gcd(a,f)=1$, and two positive
integers~$m\leq n$ and~$d$. With this input, it computes an $m\times m$
matrix of degree at most $2d$. 

The algorithm starts by computing a truncated expansion at order~$2d$
of \(H = \trsp{X} (\charmat)^{-1}X\) at~$y=0$ using
\algoName{algo:BlockTruncatedPowers}. Then, it computes a~$2m\times
2m$ minimal approximant basis as in \cref{lem:fraction-reconstruction}
using the algorithm \hyperlink{cite.GJV03}{\textproc{PM-Basis}}
of~\cite{GJV03}, and extracts a potential basis of relations.  In some
cases we can certify that it is a indeed basis of $\rmod$, but it is
not always possible to do so; a flag is returned to indicate
this. This certification is actually an optimization, rather than
strictly necessary; \cref{sec:relmat:certify} discusses this question
in more detail.

\begin{algorithm}[t] 
  \algoCaptionLabel{CandidateBasis}{f,a,m,d}
  \begin{algorithmic}[1]
    \Require $f\in \xRing$ of degree \(n\), with $f(0)\neq 0$, $a\in
    \xRing_{<n}$ with \(\gcd(a,f)=1\), $m\leq n$ and $d$ in \(\NN_{>0}\)
    \Ensure a weak Popov matrix \(R \in \ymatRing{m}{m}_{\leq2d}\) and
    a flag in \(\{\Cert,\NoCert\}\); $R$ is a basis of $\rmod$ in
    either of the following cases:
    \begin{itemize}
      \item $\ddmm=\dd$ and \(H = \trsp{X}(\charmat)^{-1}X\) is
      describable in degree $d$, in which case $\deg(R)\leq d$ 
      \item the flag is \(\Cert\), which implies $\ddmm=\dd=n$
    \end{itemize}
    
    \State \CommentLine{Truncated expansion of \(H\): compute \(S_k = -\trsp{X}M_{a}^{-k-1}X\) for {\(0\leq k<2d\)} using \cref{algo:BlockTruncatedPowers}} \label{algo:CandidateBasis:series}
      \Statex $(A^*_{i,k})_{\substack{0 \le i < m\\ {0 \le k < 2d+1}}} \gets \Call{algo:BlockTruncatedPowers}{f,a^{-1} \bmod f,m,{2d+1}}$
      \Statex \(S_{i,k} \in \vecRing{m} \gets\) vector of coefficients of \(-A^*_{i,k+1} \in \xRing_{<m}\), for \(0 \le i < m\) and \(0 \le k < 2d\)
    \State \CommentLine{Fraction reconstruction: compute approximant basis using algorithm from \cite{GJV03,JeannerodNeigerVillard2020}}
    \Statex \(F \in \ymatRing{m}{2m}_{< 2d} \gets (\sum_{0 \le k < 2d} S_k y^k \;\; -\idMat{m})\) where
    \(S_k = (S_{0,k} \;\; \cdots \;\; S_{m-1,k}) \in \matRing{m}{m}\)
    \Statex \(P \in \ymatRing{2m}{2m}_{\leq 2d} \gets \trsp{\hyperlink{cite.GJV03}{\textproc{PM-Basis}}(\trsp{F},2d,0)}\), with $P$ in weak Popov form
      \label{algo:CandidateBasis:approx}
    \State \CommentLine{Return candidate matrix and result of basic certification}
    \label{algo:CandidateBasis:end}
    \Statex $R \gets P_{1..m,1..m}$ 
    \Statex\IF{} the sum of diagonal degrees of \(R\) is equal to $n$
          \Comment{\cref{lem:fraction-reconstruction:item:certif} of \cref{lem:fraction-reconstruction}}
    \Statex\AND{} each of the \(m\) rightmost columns of \(P\) has degree \(\ge \deg(R)\)
    \Statex\THEN{} \Return \((R,\Cert)\) \ELSE{} \Return \((R,\NoCert)\)
  \end{algorithmic}
\end{algorithm}

\begin{proposition}
  \label{prop:compute_Mmm}
  Given $f\in\field[x]$ of degree~$n$
 with $f(0)\neq 0$, $a\in\field[x]_{<n}$ such that \(\gcd(a,f)=1\), 
and  two positive integers $m\leq n$ and  $d$,   
  \algoName{algo:CandidateBasis} uses $\softO{m^{\omega} d +
    \bicost{d}}$ operations in~$\field$, \bicostrecall, and computes a weak
    Popov matrix \(R \in \ymatRing{m}{m}_{\leq2d}\). 
    The matrix $R$ is a basis of \(\rmodfa\) in
    either of the following cases:
  \begin{itemize}
    \item  The determinantal degree $\dd[m,m]^{(a,f)}$ is equal to $\ddfa$ and   the fraction \(H
    (y)= \trsp{X}(\charmat)^{-1}X\) is
    describable in degree~$d$; in that case we further have $\deg(R)
    \leq d$; if in addition \(\ddfa = n\) then the flag is \(\Cert\).
      \item The flag is \(\Cert\), which implies \(\dd[m,m]^{(a,f)} = \ddfa = n\). 
  \end{itemize}
\end{proposition}

\begin{proof}
  \Cref{prop:block_truncated_powers} shows that
  \cref{algo:CandidateBasis:series} uses $\softO{m^2 d + \bicost{d}}$
  operations to compute the vectors \(S_{i,k} \in \vecRing{m}\). These
  vectors are such that the matrices \(S_k\) built in
  \cref{algo:CandidateBasis:approx} are \(S_k =
  -\trsp{X}M_{a}^{-k-1}X\); as a result, the matrix \(S = \sum_{0 \le
    k < 2d} S_k y^k\) considered at \cref{algo:CandidateBasis:approx}
  is the power series expansion of~\(H\) truncated at order
  \(2d\). Then \cref{algo:CandidateBasis:approx} correctly computes a
  weak Popov approximant basis \(P\) for $F = (S \;\; -\idMat{m})$ at
  order \(2d\) with $\deg(P) \leq 2d$ using \(\softO{m^{\omega} d}\)
  operations \cite[Thm.\,2.4]{GJV03}
  \cite[Prop.\,3.2]{JeannerodNeigerVillard2020}. (Note that transposes
  are used at Step 2 because in
  \cite{GJV03,JeannerodNeigerVillard2020} approximant bases are
  considered row-wise, rather than column-wise here.) The claimed cost
  bound for \algoName{algo:CandidateBasis} is proved.

  For the first item, assume that \(H\) is describable in degree $d$.  Then
  \cref{lem:fraction-reconstruction:item:reconstruct} of
  \cref{lem:fraction-reconstruction} ensures that \(R\) is the denominator of
  an irreducible right fraction description of \(H\), that \(\deg(R) \le d\),
  and that each of the \(m\) rightmost columns of \(P\) has degree at least
  \(\deg(R)\). From \cref{prop:modprime_is_denom} we obtain that $R$ is a basis
  of $\rmodmm$, hence a basis of $\rmod$ when $\ddmm=\dd$. This also proves the
  last claim of the item: if $\ddmm=\dd=n$, then \(\degdet{R} = n\) and this is the sum of
  diagonal degrees of \(R\) since this matrix is in weak Popov form; hence the
  flag \Cert{} is returned.

  For the second item, assume that the output flag is \Cert. Then the sum of
  diagonal degrees of \(R\) is~\(n\); according to
  \cref{lem:fraction-reconstruction:item:general} of
  \cref{lem:fraction-reconstruction}, this sum is also \(\degdet{R}\) and is at
  most \(\delta\), the determinantal degree of bases of \(\denom{H}\). On the
  other hand \cref{prop:modprime_is_denom} implies that \(\delta\) is the
  determinantal degree \(\ddmm\) of \(\rmodmm\). Hence \(n = \degdet{R} \le
  \delta = \ddmm\), from which we deduce  $\degdet{R} = \delta = \ddmm = \dd =
  n$, since \(\ddmm\le\dd\le n\) always holds. Since the output flag is
  \Cert{} we know in addition that each of the \(m\) rightmost columns of
  \(P\) has degree at least \(\deg(R)\). Thus
  \cref{lem:fraction-reconstruction:item:certif} of
  \cref{lem:fraction-reconstruction} applies, and \(R\) is the denominator of
  an irreducible right fraction description of \(H\). We conclude as done for
  the first item that~\(R\) is a basis of \(\rmod\).
\end{proof}

\begin{remark}
  \label{rmk:shift}
  The assumption that $f$ and $a$ are coprime is used here to ensure
  that \(M_a\) is invertible, so that the expansion \(H = \sum_{k\ge
    0} S_k y^k = \sum_{k\ge 0} (- \trsp{X} \mulmat^{-k-1} X) y^k\) at
  $y=0$ can be used for fraction reconstruction. This is different
  from what happened in the proof of
  \cref{prop:modprime_is_denom}, where we used the
  expansion at infinity \(H = \sum_{k\ge 0} H_k y^{-k-1}\).
{The latter expansion involves powers of $M_a$ 
and thus our formalism remains close to that of \cite{KaVi05} with Krylov sequences.  
 From an algorithmic point of view, expansions at $y=0$ allow us to use directly   
the existing efficient algorithms for matrix fraction reconstruction~\cite{GJV03,JeannerodNeigerVillard2020}, and exploit their properties.}

  This assumption on $\gcd(f,a)$ is harmless in our context: in the
  computation of $\polp(a)\rem f$, one can instead evaluate
  $\polp(y-c)$ at $y=a+c$ for a randomly chosen~$c \in \field$,
  ensuring \(\gcd(a+c,f) = 1\) with good probability. See
  \cref{step:shiftp-mainalgo} in
  \algoName{algo:ModularCompositionBaseCase}.
\end{remark}

\subsubsection*{Notes}

For some families of approximation instances,
\hyperlink{cite.GJV03}{\textproc{PM-Basis}} has been used to design faster
minimal approximant basis algorithms
\cite{ZhoLab12,JeannerodNeigerVillard2020}. Yet, the instances considered here
are ones where \hyperlink{cite.GJV03}{\textproc{PM-Basis}} is the fastest known
algorithm.

A candidate matrix of relations in $\yRing_{\le 2d}^{m\times m}$ corresponds to
$m$ polynomials in~$\xyRing_{\le(m,2d)}$. Using
\algoName{algo:SimultaneousBivariateModularComposition} to verify that
the
evaluations of these polynomials at~$a\bmod f$ are zero
uses~$\softO{\bicost{d^2}}$ operations in~$\field$, by \cref{lemma:NZ}. For
the values of~$m$ and~$d$ used to obtain the exponent~$\kappa<1.43$ in our main
algorithm, this is~$\bigO{n^{2.55}}$, and thus too costly.

\subsection{Certified matrix of relations}
\label{sec:relmat:certify}

In general, when \algoName{algo:CandidateBasis} does not certify its
result, we do not know methods to verify that the matrix it returns is
a matrix of relations within our complexity bound.

Instead, from a matrix $R$ computed by
\algoName{algo:CandidateBasis}, \algoName{algo:MatrixOfRelations}
either
detects that it is not a matrix of relations of \(\rmod\), or
constructs from \(R\) a matrix of relations of \(\rmod[m']\) of
degree at most \(2d\), for some \(m' < 2m\). This is the key towards
making our modular composition algorithm Las Vegas, rather than Monte
Carlo.

To achieve this, instead of evaluating all columns of \(R\) at \(a
\bmod f\), \algoName{algo:MatrixOfRelations} evaluates only two
polynomials built randomly from these columns (and only one polynomial
in the special case $m=1$), which is within our target complexity
using the N\"usken-Ziegler algorithm. If these evaluations are not
both zero, then \(R\) was not a matrix of relations. Otherwise the
algorithm constructs a Sylvester matrix from these two vectors (see
e.g.,~\cite[Sec.\,6.3]{GaGe99} for the definition and properties of the
Sylvester matrix). When this matrix is nonsingular, it is a matrix of
relations of a  module $\rmod[m']$ for \(m' \leq
\max(1,2(m-1))\); since \(m'\) cannot be much larger than \(m\),
this matrix can be used for efficient composition.

\begin{algorithm} 
  \algoCaptionLabel{MatrixOfRelations}{f,a,m,d,(r_i)_{3 \le i \le m}}

  \begin{algorithmic}[1]
  \Require
  \parbox[t]{0.8\textwidth}{
    $f\in \xRing$ of degree \(n\), with $f(0)\neq 0$, $a\in \xRing_{<n}$ with \(\gcd(a,f)=1\),

    $m\leq n$ and $d$ in \(\NN_{>0}\), $(r_i)_{3 \le i \le m} \in \vecRing{m-2}$
    }
  \Ensure either \Fail{} or a matrix $R' \in\ymatRing{m'}{m'}_{\leq 2d}$ of
  relations of $\rmod[m']$ with $m'\leq\max(1,2(m-1))$ 
  \State \CommentLine{Use \cref{algo:CandidateBasis} to find a candidate basis of relations}
          \label{algo:MatrixOfRelations:candidate}
  \Statex \((R,\Flag) \in \ymatRing{m}{m}_{\leq 2d} \times \{\Cert,\NoCert\} \gets \Call{algo:CandidateBasis}{f,a,m,d}\) 
  \Statex\InlineIf{\(\Flag=\Cert\)} \Return $R$
\State\CommentLine{Case $m=1$, check that $R_{1,1} \in \field[y]_{<2d+1}$} annihilates $a\bmod f$ \label{algo:MatrixOfRelations:BKtest}
 \Statex \IF{} $m=1$ \THEN{}    
\Statex \hspace*{0.4cm}\IF{} $\Call{algo:ModularComposition-BrentKung}{f,a,R_{1,1}} \neq 0$ \THEN{} \Return \Fail
 \Statex \hspace*{0.4cm}\ELSE{} \Return $R$

  \State \CommentLine{Build candidate relations and verify them}
          \label{algo:MatrixOfRelations:NZ}
  \Statex $r(x,y) \gets R_{*,1}$;  $s(x,y) \gets R_{*,2} + r_3 R_{*,3} + \ldots + r_m R_{*,m}$ 
          \Comment{both in \(\xyRing_{<(m,2d+1)}\)}
  \Statex\IF{} $\Call{algo:BivariateModularComposition}{f,a,r} \neq 0$
  \Statex\OR{} $\Call{algo:BivariateModularComposition}{f,a,s} \neq 0$
  \Statex\THEN{} \Return \Fail \Comment{\cref{algo:BivariateModularComposition}}
  \Statex \InlineIf{$m=2$}{\Return \(R\)}

  \State \label{algo:MatrixOfRelations:gcd}
  \CommentLine{Construct and return the Sylvester matrix of $f$ and
  $s$, if it is nonsingular}
  \Statex \InlineIf{$\gcd_x(r,s) \neq 1$}{\Return \Fail} \Comment{\(r\) and \(s\) not coprime as elements of \(\field(y)[x]\)}
  \Statex \Return the Sylvester matrix of \((r,s)\) as in \cite[Sec.\,6.3,
  Eq.\,(5)]{GaGe99}, with rows in reversed order, viewing \(r\) and \(s\) as
  polynomials in \(x\) over \(\yRing\)
  \end{algorithmic}
\end{algorithm}

\begin{proposition}
  \label{prop:certificate}
  Given $f\in\field[x]$ of degree~$n$ with $f(0)\neq 0$,
  $a\in\field[x]_{<n}$ with $\gcd(a,f) = 1$, two positive integers $m(\leq n)$ and $d$, and $(r_i)_{3\le i\le
  m}\in\field^{m-2}$,
  \algoName{algo:MatrixOfRelations} uses
  $\softO{m^{\omega} d + \bicost{d}}$ operations in~$\field$, \bicostrecall, and returns
  either \(\Fail\) or a matrix of relations $R' \in\ymatRing{m'}{m'}_{\leq 2d}$
  of \(\rmod[m']^{(a,f)} \) where $m'\leq\max(1,2(m-1))$.

  If $\dd[m,m]^{(a,f)}=\ddfa$, the fraction \(H
  = \trsp{X}
  (\charmat)^{-1}X\) is describable in degree~$d$, and 
  $(r_3,\ldots,r_m)$ are chosen uniformly and independently at random from a
  finite subset $S$ of $\field$, then failure happens with probability at most
  $(m-1)/\card{S}$ and in case of success, $\deg(R')\leq d$.
\end{proposition}
\begin{proof}
  If \(\Flag=\Cert\) at \cref{algo:MatrixOfRelations:candidate}, then from the
  second item of \cref{prop:compute_Mmm} an appropriate matrix of relations is
  returned. Now assume that \(\Flag=\NoCert\) and
  \algoName{algo:MatrixOfRelations} does not return \Fail.  If
  $m=1$ then the relation has been checked at 
  \cref{algo:MatrixOfRelations:BKtest}, proving the result. Otherwise,
  let \(R'
  \in
  \ymatRing{m'}{m'}\) be the output matrix, which is constructed from the
  polynomials \(r,s\) of \(x\)-degree less than \(m\); in particular, $m'= 
  \deg_x(r) + \deg_x(s) \leq 2(m-1)$ \cite[Sec.\,6.3]{GaGe99}. The fact that the
  test at \cref{algo:MatrixOfRelations:gcd} has not failed ensures that $r$ and
  $s$ are coprime as univariate polynomials in $\field(y)[x]$, and therefore
  \(R'\) is nonsingular \cite[Cor.\,6.15]{GaGe99}. Furthermore, since the tests
  at \cref{algo:MatrixOfRelations:NZ} have not failed, \(r\) and \(s\) are
  relations of \(\rmod\). It follows that the columns of~\(R'\), which are by
  construction multiples of \(r\) and \(s\) in \(\xyRing\) represented as
  vectors in \(\yvecRing{m'}\), are relations of \(\rmod[m']\).  Besides, the
  construction of the Sylvester matrix does not increase the $y$-degree, hence
  $\deg(R') \leq \deg(R) \leq 2d$.  We have proved the fact that if the output
  is not \Fail, then it is a matrix of relations of \(\rmod[m']\). 

  For the complexity bound, the cost for finding $R$ is given in
  \cref{prop:compute_Mmm}, while the ones for checking that $R_{1,1}$, $r$ and $s$ are
  relations are given in \cref{lemma:BK,prop:NuskenZiegler}. As for the gcd test at
  \cref{algo:MatrixOfRelations:gcd}, it can be done via the resultant of $r$
  and $s$ with respect to $x$, computed using $\softO{m^2d}$ operations
  \cite{Rei97}.

  It remains to prove the third assertion and the probability bound. Since
  when \(\Flag=\Cert\) a basis is returned with no randomization, assume
  \(\Flag = \NoCert\). The assumptions here and the first item of
  \cref{prop:compute_Mmm} ensure that $R$ is a basis of~$\rmod$ with $\deg(R)
  \leq d$, hence $\deg(R')\leq d$. In that case failure never occurs at
\cref{algo:MatrixOfRelations:BKtest} for $m=1$. It never occurs either at 
  \cref{algo:MatrixOfRelations:NZ} for $m\geq 2$, and $r$ and $s$ are relations of \(\rmod\).
  The columns of \(R\) represent bivariate polynomials \(b_1,\ldots,b_m \in
  \xyRing_{<(m,d+1)}\) and we claim that \(\gcd_x(b_1,\ldots,b_m) = 1\), meaning
  that there is a \(\yRing\)-linear combination of \(b_1,\ldots,b_m\) 
  in \(\yRing\setminus\{0\}\). Since $R$ is nonsingular, the first column of a
  transformation for the (upper triangular) Hermite normal form of~$R$ provides
  such a combination.  It follows that \Fail{} is returned with probability at
  most $(m-1)/\card{S}$ at
  \cref{algo:MatrixOfRelations:gcd}~\cite[Thm.\,6.46]{GaGe99}. 
\end{proof}

\subsubsection*{Note}

The computation of \Cert{} by \algoName{algo:CandidateBasis} is  only
an optimization. \algoName{algo:MatrixOfRelations} works as it is,
even if \Cert{}
is never returned. When the candidate matrix \(R\) at
\cref{algo:MatrixOfRelations:candidate} is not a matrix of relations, this is
often detected at \cref{algo:MatrixOfRelations:NZ}, but not always. Even if
it is not detected, it suffices to find two coprime polynomials \(r(x,y)\) and
\(s(x,y)\) that are relations to ensure that 
\algoName{algo:MatrixOfRelations}
returns a matrix of relations. For example, it may happen that \(R\) is not a
matrix of relations but some columns of it still give low-degree relations of~$\rmod$.


\section{Change of basis}
\label{sec:changeofbasis}

In this section we present an algorithm for performing a change of
basis in~$\quotient =\xRing/\genBy{f}$. This algorithm is used in
a randomized manner in \cref{sec:composition_randomized}, in order to
handle arbitrary inputs with good complexity bounds. Our
approach is based on an extension of the approximant bases used in
\cref{sec:relmat_comp}; we start with necessary definitions.


\subsection{Definitions}

We  use an extension of the forms of polynomial matrices introduced in
\cref{subsubsec:definreconstruct}, called \emph{shifted} forms
\cite{BarBul92,BeLaVi99}. For a given tuple \(t = (t_1,\ldots,t_m) \in \ZZ^m\)
and a column vector \(v \in \yvecRing{m}\), the \(t\)-shifted degree of \(v\)
is \(\max_{1 \le i \le m} (\deg(v_i) + t_i)\). Then, for a matrix \(P \in
\ymatRing{m}{m}\) whose \(j\)th column has \(t\)-shifted degree \(d_j \in
\ZZ\), the \emph{(column) \(t\)-shifted leading matrix} of \(P\) is the matrix
in \(\matRing{m}{m}\) whose entry \((i,j)\) is the coefficient of degree
\(d_j-t_i\) of the entry \((i,j)\) of \(P\). Then \(P\) is said to be
\emph{\(t\)-shifted weak Popov} if this \(t\)-shifted leading matrix is
invertible and upper triangular.

We  also need the corresponding normal form: \(P\) is said to be
\emph{\(t\)-shifted Popov} if it is \(t\)-shifted weak Popov and its \emph{row}
leading matrix is the identity of \(\matRing{m}{m}\) \cite{Kailath80,BeLaVi99}.
For a given \(t\), any submodule of rank \(m\) of \(\yvecRing{m}\) admits a
unique basis in \(t\)-shifted Popov normal form \cite[Thm.\,3.7]{BeLaVi99}.
By definition, \(t\)-shifted Popov matrices are also (nonshifted) 
\emph{row}
reduced; in particular, Hermite normal forms are \(t\)-shifted Popov for an appropriate choice of $t$, hence 
are row reduced. 

Row reduced matrices allow for a {\em division with remainder} with stronger properties than the one for general
nonsingular matrices presented in \cref{sec:bivcomposition:division}; namely
they ensure uniqueness of the remainder. Precisely, if a matrix \(P \in
\ymatRing{m}{m}\) is row reduced, for any vector \(v \in \yvecRing{m}\) there
exists a unique vector \(\tilde{v} \in \yvecRing{m}\) such that \(v -
\tilde{v}\) is a right multiple of \(P\) and the \(i\)th entry of \(\tilde{v}\)
has degree less than the \(i\)th row of \(P\) \cite[Thm.\,6.3-15,
p.\,389]{Kailath80}.

We  also use the fact that, by definition, for any block
decomposition \(P = (\begin{smallmatrix} P_{11} & P_{12} \\ P_{21} & P_{22}
\end{smallmatrix})\) of a matrix \(P\), if \(P\) is in Hermite
(resp.~\(t\)-shifted Popov) normal form, then: 
\begin{itemize}
  \item \(P_{11}\) and \(P_{22}\) are in Hermite normal form (resp.~in shifted
    Popov normal form with respect to the corresponding subtuple of \(t\)); 
  \item each column of \(P_{12}\) (resp.~\(P_{21}\)) is its own remainder
    in the division by \(P_{11}\) (resp.~\(P_{22}\)).
\end{itemize}

Finally, \(t\)-shifted forms induce the notion  of {\em \(t\)-shifted
approximant bases}~\cite{BeLaVi99,ZhoLab12,JeannerodNeigerVillard2020}, which
are approximant bases (see \cref{subsubsec:definreconstruct}) in \(t\)-shifted
Popov normal form. 


\subsection{Inverse modular composition and change of basis via approximant bases}

Let $f$ be in $\field[x]$ of degree $n$. A core ingredient for the
randomization in our composition algorithm is an instance of
\emph{inverse modular composition}, which is used to change the basis
of~$\quotient =\xRing/\genBy{f}$ from \((1,x,\ldots,x^{n-1})\) to
\((1,\gamma,\ldots,\gamma^{n-1}) \bmod f\), for some $\gamma
\in\xRing$ whose minimal polynomial $\mu_{\gamma}$ modulo \(f\) has
degree $n$. This change of basis induces the $\field$-algebra
isomorphism
\begin{equation} \label{def:changeofbasis}
  \phi_\gamma:\quotient\to\yRing/\genBy{\mu_{\gamma}}, 
\end{equation}
which maps any $u \in \quotient$ to $v$ such that $v(\gamma)\equiv
u\bmod f$.  Given $a$ in $\field[x]_{<n}$, this section explains how
to compute the unique polynomial representative $\alpha \in
\yRing_{<n}$ of $\phi_\gamma(a\bmod f)$, i.e., the unique $\alpha \in
\yRing_{<n}$ such that $\alpha(\gamma) \rem f = a$.

Reversing the path followed in our modular composition approach, we
first find a bivariate $\tilde \alpha \in \xyRing$ such that $\tilde
\alpha - \alpha \in \rmodfg$, hence $\tilde \alpha(x,\gamma) \equiv
\alpha(\gamma) \bmod f$. Then the univariate solution~$\alpha$ is
recovered from~$\tilde \alpha$ and a basis of relations \(R\) of
\(\rmodfg\) by reversing the division from \cref{eq:defvecdiv}; this
corresponds to a division by the Hermite normal form of $R$.

Our algorithm for computing $\tilde{\alpha}$ can be seen as a
generalization to \(m\ge 1\) of Shoup's algorithm for computing
\(\alpha\), mentioned in \cref{subsubsec:pr}. The latter algorithm
deals with the case \(m=1\): from the power projections
$(\ell(1),\ell(\gamma),\dots,\ell(\gamma^{2n-1}))$ and
$(\ell(a),\ell(\gamma a),\dots,\ell(\gamma^{n-1}a))$, it obtains both
$\alpha$ and $\mu_{\gamma}$ by solving two Pad\' e approximation
problems. In the matrix case \(m\ge 1\), \algoName{algo:ChangeOfBasis}
computes solutions to equations similar to
\cref{eq:blockPade,eq:blockShoup} given in the introduction. These are
matrix generalizations of the Pad\'e approximation problems; their
solutions provide respectively a basis of relations $R$ of $\rmodfg$
and $\tilde{\alpha}$.

In more details,
\cref{algo:ChangeOfBasis:expansionU,algo:ChangeOfBasis:expansionS}
first compute the power series expansions involved in
\cref{eq:blockPade,eq:blockShoup}, which amounts to a type of
generalized power projections. Then both approximation problems are
solved at once using shifted approximant bases:
\begin{itemize}
\item The choice of the first \(2m\) columns of $F = (S \; -\idMat{m}
  \; s)$ at \cref{algo:ChangeOfBasis:reconstruct}, which are the same
  as in \cref{algo:CandidateBasis:approx} of
  \algoName{algo:CandidateBasis}, and the use of a corresponding
  ``zero shift'' (first $2m$ entries of the tuple $t$ at
  \cref{algo:ChangeOfBasis:reconstruct}), make this equivalent to the
  computation in \cref{subsec:candidate} (compare
  \cref{algo:ChangeOfBasis:expansionS,algo:ChangeOfBasis:reconstruct,algo:ChangeOfBasis:verif}
  of \algoName{algo:ChangeOfBasis} to
  \cref{algo:CandidateBasis:series,algo:CandidateBasis:approx,algo:CandidateBasis:end}
  of \algoName{algo:CandidateBasis}). This yields a basis $R$ of
  $\rmodfg$.
\item \Cref{eq:blockShoup} is solved thanks to an additional series
  expansion in \(F\) (its last column), and the use of a sufficiently
  large shift (the last entry $2d$ of the tuple $t$). This yields a
  bivariate polynomial $\tilde \alpha$ that is the remainder of
  the requested $\alpha$ in the division by $R$.
\end{itemize}
Finally, this sought $\alpha$ can be obtained by reversing this
division, using a Hermite normal form computation which also provides
the minimal polynomial \(\minpoly[\gamma]\)
(\cref{algo:ChangeOfBasis:hnf,algo:ChangeOfBasis:alpha}).

The assumptions in \algoName{algo:ChangeOfBasis} yield a slightly
stronger statement in \cref{prop:algo:ChangeOfBasis} than in
\cref{prop:compute_Mmm} for \algoName{algo:CandidateBasis}. Indeed, we
 suppose that $\gamma$ is such that $\deg(\mu_{\gamma})=n$,
whereas we make no such assumption in
\algoName{algo:CandidateBasis}. From the module properties in
\cref{prop:invariant-factors}, we deduce that $\deg(\mu_{\gamma})=n$
implies \(\ddfg=n\), which  allows us to certify the basis of
relations $R$ when \Fail{} is not returned.

\algoName{algo:ChangeOfBasis} may still return \Fail; 
\cref{sec:composition_randomized} shows that when it is
called with a random
$\gamma$, then with high probability, it does not fail, at least under
some assumptions on $f$.

\begin{algorithm}
  \algoCaptionLabel{ChangeOfBasis}{f,\gamma,a,m,d}
  \begin{algorithmic}[1]
    \Require $f$ of degree $n$ in $\xRing$, with $f(0)\neq 0$, $\gamma \in \xRing_{<n}$, $a\in \xRing_{<n}$,  \(m \leq n\) and \(d\) in \(\NNp\)
    \Ensure either \Fail{} or $(R,\minpoly[],\alpha)$
    where 
    \(R \in \ymatRing{m}{m}_{\le 2d}\) is the Popov basis of \(\rmodfg\),
    \(\minpoly[]\) is the minimal polynomial of \(\gamma\) in \(\xRing/\genBy{f}\) and has degree \(n\),
    and \(\alpha\in\yRing_{<n}\) with \(\alpha(\gamma) \equiv a \bmod f\)

    \State\InlineIf{\(\gcd(\gamma,f)\neq 1\)}{\Return \Fail}
    \label{algo:ChangeOfBasis:invertible}

    \State\CommentLine{Truncated expansion of \(-\trsp{X}(\charmat[\gamma])^{-1}v_a\)
        using \cref{algo:TruncatedPowers}, $v_a \in \field ^n$ is the coefficient vector of~$a$} 
        \label{algo:ChangeOfBasis:expansionU}
    \Statex \((r_{k})_{0\le k < 2d}\gets \Call{algo:TruncatedPowers}{f,\gamma^{-1}\bmod f, \gamma^{-1} a \bmod f, m, 2d}\) 
    \Statex \(s \in\yvecRing{m} \gets \sum_{0 \le k < 2d} s_{k} y^k\) where
            \(s_{k}\in\vecRing{m}\) is the coefficient vector of \(r_{k}\)

    \State \CommentLine{Truncated expansion of \(\trsp{X}(\charmat[\gamma])^{-1}X\)
        using \cref{algo:BlockTruncatedPowers} (analogous to \cref{algo:CandidateBasis:series} of \cref{algo:CandidateBasis})}
        \label{algo:ChangeOfBasis:expansionS}
    \Statex $(\Gamma_{i,k})_{\substack{0 \le i < m\\0 \le k <2d+2}} \gets \Call{algo:BlockTruncatedPowers}{f,\gamma^{-1} \bmod f,m,2(d+1)}$
    \Statex \(S_{i,k} \in \vecRing{m} \gets\) vector of coefficients of \(-\Gamma_{i,k+1} \in \xRing_{<m}\), for \(0 \le i < m\) and \(0 \le k < 2d\)
    \Statex \(S \in \ymatRing{m}{m}_{<2d} \gets \sum_{0 \le k < 2d} S_k y^k\) where \(S_k = (S_{0,k} \;\; \cdots \;\; S_{m-1,k}) \in \matRing{m}{m}\) 

    \State \CommentLine{Fraction reconstruction using \cite{GJV03,JeannerodNeigerVillard2020} (analogous to \cref{algo:CandidateBasis:approx} of \cref{algo:CandidateBasis})}
        \label{algo:ChangeOfBasis:reconstruct}

    \Statex \(F \in \ymatRing{m}{(2m+1)}_{< 2d} \gets (S(y) \;\;\; -\idMat{m} \;\;\; s(y))\)
    \Statex \(t \in \NN^{2m+1} \gets (0,\ldots,0,2d)\)
    \Statex \(\bar P \in \ymatRing{(2m+1)}{(2m+1)}_{\leq 2d} \gets \trsp{\hyperlink{cite.JeannerodNeigerVillard2020}{\textproc{Popov-PM-Basis}}(\trsp{F},2d,t)}\), with $\bar P$ in \(t\)-shifted Popov normal form
    \Statex $P \gets \bar P_{1..2m,1..2m}$
    \Statex $v_{\bar \alpha} \in \yvecRing{m}_{<\deg(R)} \gets \bar P_{1..m,2m+1}$ 
    \Comment{represents \(\bar \alpha(x,y)\), expected to satisfy \(\bar \alpha(x,\gamma) \equiv a \bmod f\)}

    \State \label{algo:ChangeOfBasis:verif}
          \CommentLine{Ensure \(R\) is a basis of \(\rmodfg\), from \cref{lem:fraction-reconstruction:item:certif} of \cref{lem:fraction-reconstruction} (analogous to \cref{algo:CandidateBasis:end} of \cref{algo:CandidateBasis})}
          \Statex  $R \gets P_{1..m,1..m}$ 
    \Statex\IF{} the sum of diagonal degrees of \(R\) is less than \(n\) {}
    \Statex\OR{} among the \(m\) rightmost columns of \({P}\), one has degree \(< \deg(R)\)
    \THEN{} \Return \Fail
    \State \CommentLine{Compute \(\minpoly[\gamma]\), and ensure it has degree \(n\)}
        \label{algo:ChangeOfBasis:hnf}
    \Statex \(T \in \ymatRing{m}{m} \gets\) Hermite normal form of \(R\)
        \Comment{using \cite[Algo.\,1 and 3]{LVZ17}}
    \Statex \(\minpoly[] \in \yRing \gets T_{1,1}\);
    \InlineIf{\(\deg(\minpoly[])<n\)}{\Return \Fail}

    \State \CommentLine{Deduce \(\alpha\) and return}
        \label{algo:ChangeOfBasis:alpha}
    \Statex \(\alpha \in \yRing_{<n} \gets \bar\alpha_1 - (T_{1,2}\bar\alpha_2 + \cdots + T_{1,m}\bar\alpha_{m}) \rem \minpoly[]\),
    where \(v_{\bar \alpha} = (\bar\alpha_1 \;\; \cdots \;\; \bar\alpha_{m})\)
    \Statex \Return \((R,\minpoly[],\alpha)\)
  \end{algorithmic}
\end{algorithm}

\begin{proposition}
  \label{prop:algo:ChangeOfBasis}
  Given $f\in\field[x]$ of degree~$n$ with $f(0)\neq 0$,
  $\gamma$ and $a$ in $\field[x]_{<n}$, $m \leq n$ and $d$ in~$\mathbb{N}_
  {>0}$, 
  \algoName
  {algo:ChangeOfBasis} uses
  $\softO{m^{\omega} d + \bicost{d}}$ operations in~$\field$, \bicostrecall, to return
  either \Fail{} or $(R,\minpoly[],\alpha)$ where \(R \in\ymatRing{m}{m}_{\le
  2d}\)~is the Popov basis of \(\rmodfg\), $\minpoly[]$ is the minimal
  polynomial \(\minpoly[\gamma]\) of $\gamma \bmod f$ and has degree~$n$, and
  \(\alpha\) is the unique polynomial in \(\yRing_{<n}\) such that
  \(\alpha(\gamma) \equiv a \bmod f\). 

  If \(\gcd(\gamma,f)=1\), \(\ddmmfg =
  \ddfg\), $\deg(\mu_{\gamma})=n$ and the fraction \(H = \trsp{X}(\charmat[\gamma])^{-1}X\)
  is describable in degree~$d$, then the output is not \Fail{}; in that case we
  further have $\deg(R) \leq d$.
\end{proposition}

\begin{proof}
  We start by showing that if the algorithm does not fail, then the truncated
  module \(\rmodmmfg\) and the module of relations \(\rmodfg\) are equal, and
  \(R\) is a basis of this module.

\paragraph*{\crefrange{algo:ChangeOfBasis:invertible}{algo:ChangeOfBasis:reconstruct}: the approximant basis \(\bar{P}\)}

If the test at \cref{algo:ChangeOfBasis:invertible} does not fail then the
specifications for
\cref{algo:ChangeOfBasis:expansionU,algo:ChangeOfBasis:expansionS} are met. 
At \cref{algo:ChangeOfBasis:expansionU},
\algoName{algo:TruncatedPowers} returns \(r_k = [a \, \gamma^
{-k-1}\rem
f]_0^{m-1}\) for \(0\le k<2d\) using~$\softO{\bicost{d}}$ operations,
according to \cref{prop:TruncatedPowers}. Thus  the coefficient vector $s_k \in
\vecRing{m}$ of \(r_k\) is \(\trsp{X}\mulmat[\gamma]^{-k-1}v_a\), where $v_a
\in \vecRing{n}$ is the coefficient vector of \(a\), so that the polynomial
vector \(s(y)\) computed at \cref{algo:ChangeOfBasis:expansionU} is the power
series expansion of \(-\trsp{X}(\charmat[\gamma])^{-1}v_a\) truncated at order
\(2d\).
From \cref{prop:block_truncated_powers}, the computation of \(S(y)\) at
\cref{algo:ChangeOfBasis:expansionS} uses $\softO{m^2 d + \bicost{d}}$
operations;   $S$ is the power series expansion of
\(\trsp{X}(\charmat[\gamma])^{-1}X\) truncated at order \(2d\).

\cref{algo:ChangeOfBasis:reconstruct} computes the \(t\)-shifted Popov
approximant basis \(\bar P\) for $F = (S \; -\idMat{m} \; s)$ at order \(2d\),
which uses \(\softO{m^{\omega} d}\) operations
\citetext{\citealp[Thm.\,2.4]{GJV03};
\citealp[Sec.\,3]{JeannerodNeigerVillard2020}}. Writing
\[
  \bar P = \begin{pmatrix} {P} & u \\ z & \lambda \end{pmatrix}
  \text{ for some } {P} \in \ymatRing{(2m)}{(2m)}_{\le 2d},
   \lambda \in \yRing_{\le 2d},
   u \in \yvecRing{2m}_{\le 2d},
  \text{ and } z \in \ymatRing{1}{(2m)}_{\le 2d},
\]
the fact that \(\bar P\) is \(t\)-shifted Popov and the choice \(t =
(0,\ldots,0,2d)\) ensure that \({P}\) is (nonshifted) Popov, that
\(\lambda\neq 0\), and that \(\deg(z) + 2d < \deg({P}) \le 2d\). The latter
degree bound yields \(z=0\), hence~\({P}\) is the Popov approximant basis of
\((S \;\; -\idMat{m})\) at order \(2d\).
The fact that \(\bar P\) is \(t\)-shifted Popov also ensures that the
(unique) remainder in the division of \(u\) by \(P\) is \(u\) itself, and that
the \(i\)th entry of \(u\) has degree less than the \(i\)th diagonal degree of
\(P\).

\paragraph*{After \cref{algo:ChangeOfBasis:verif}, \(R\) is a basis of \(\rmodmmfg=\rmodfg\)}

Let $R$ be the \(m\times m\) leading principal submatrix of \(P\) (and of
\(\bar{P}\)), and let \(v_{\bar \alpha} \in \yvecRing{m}\) be the length-\(m\)
top subvector of \(u\). Similarly to the above, \(R\) is (nonshifted) Popov,
and the \(i\)th entry of \(v_{\bar{\alpha}}\) has degree less than the \(i\)th
diagonal degree of \(R\); in particular, \(\deg(v_{\bar \alpha})<\deg(R)\).
Considering $H(y)=\trsp{X}(\charmat[\gamma])^{-1}X$, recall from
\cref{prop:modprime_is_denom} that \(\denom{H} = \rmodmmfg\), and recall that
\(\rmodmmfg \supseteq \rmodfg\) with equality if and only if \(\ddmmfg =
\ddfg\). In particular, bases of \(\denom{H}\) have determinantal degree
\(\ddmmfg \le \ddfg \le n\). Applying
\cref{lem:fraction-reconstruction:item:general} of
\cref{lem:fraction-reconstruction} to \(H\) and the approximant basis \({P}\)
shows that the sum of diagonal degrees of \(R\) is \(\degdet{R}\), and is at
most \(\ddmmfg\).

As a result, if \cref{algo:ChangeOfBasis:verif} does not return \Fail, then \(n
\le \degdet{R} \le \ddmmfg\), hence \(\ddmmfg=\ddfg=n\) and \(\denom{H} =
\rmodmmfg=\rmodfg\). Furthermore,
\cref{lem:fraction-reconstruction:item:certif} of
\cref{lem:fraction-reconstruction} shows that \(R\) is a basis of~\(\rmodfg\).

\paragraph*{After \cref{algo:ChangeOfBasis:hnf}, \(\minpoly[]\) is \(\minpoly[\gamma]\) and has degree \(n\)}

Using \(\softO{m^{\omega}d}\) operations~\cite{LVZ17},
\cref{algo:ChangeOfBasis:hnf} finds the Hermite normal form \(T\) of \(R\).
Since \({T}\) is a basis of \(\rmodfg\) in upper triangular form,
\cref{prop:invariant-factors} states that its first diagonal entry is the
minimal polynomial of \(\gamma\) in \(\xRing/\genBy{f}\). Hence \(\minpoly[]\)
computed at \cref{algo:ChangeOfBasis:hnf} is this minimal polynomial. It has
degree at most \(n\), and the algorithm returns \Fail{} at this step if and
only if \(\deg(\mu)<n\).

\paragraph*{After \cref{algo:ChangeOfBasis:hnf}, \(v_{\bar\alpha}\) represents
\(\bar \alpha(x,y)\) such that \(\bar \alpha(x,\gamma) \equiv a \bmod f\)}

Let \(\bar\alpha \in \xyRing_{<(m,\deg(R))}\) be the polynomial whose
coefficient vector is \(v_{\bar\alpha}\), that is, \(\bar\alpha =
\bar\alpha_1(y) + x\bar\alpha_2(y) + \cdots + x^{m-1}\bar\alpha_m(y)\) using
notation from \cref{algo:ChangeOfBasis:alpha}.

The fact that \(\minpoly[]=\minpoly[\gamma]\) has degree \(n\) also ensures
that there exists a unique \(\alpha\in\yRing_{<n}\) such that \(\alpha(\gamma)
\equiv a \bmod f\). Then let \(v_\alpha = \trsp{(\alpha \;\; 0 \; \cdots \; 0)}
\in \yvecRing{m}\), and let \(v_{\tilde \alpha} \in \yvecRing{m}\) be the
unique remainder in the division of \(v_\alpha\) by \(R\). The entries of
\(v_{\tilde \alpha}\) have degree strictly less than that of the corresponding
row of $R$: the degree of the \(i\)th entry of \(v_{\tilde\alpha}\) is less
than the \(i\)th diagonal degree of \(R\). We also define \(\tilde \alpha \in
\xyRing_{<(m,\deg(R))}\) as the polynomial whose coefficient vector is
\(v_{\tilde \alpha}\); in particular \(\tilde \alpha(x,\gamma) = \alpha(\gamma)
\equiv a \bmod f\). 

We now show that $\bar \alpha = \tilde \alpha$, which yields \(\bar
\alpha(x,\gamma) = a \bmod f\). By construction, \(\tilde{\alpha}(x,y) - a(x)\)
is in \(\relmod{n}{\gamma}{f}\), and since \(\charmat[\gamma]\) is a basis of
\(\relmod{n}{\gamma}{f}\) (see \cref{prop:invariant-factors}) there is a vector
\(v \in \yvecRing{n}\) such that \((\charmat[\gamma]) v = Xv_{\tilde {\alpha}}
- v_a\). Applying the predictable degree property \cite[Thm.\,6.3-13,
p.\,387]{Kailath80} to the column reduced matrix \(\charmat[\gamma]\), all of
whose columns have degree \(1\), we obtain that \(\deg(v)+1 =
\deg(Xv_{\tilde{\alpha}} - v_a) = \deg(v_{\tilde \alpha})\). Furthermore from
\((\charmat[\gamma]) v = Xv_{\tilde {\alpha}} - v_a\) we get
\[
  \trsp{X}(\charmat[\gamma])^{-1}X v_{\tilde \alpha} - \trsp{X}v - \trsp{X}(\charmat[\gamma])^{-1}v_a = 0,
\]
and considering truncated power series it follows that \(F q = (S \;\;
-\idMat{m} \;\; s)  q \equiv 0 \bmod y^{2d}\), where
\[
  q = \begin{pmatrix} \tilde {u} \\ 1 \end{pmatrix} \in \yvecRing{2m+1}
  \quad\text{and}\quad
  \tilde u  = \begin{pmatrix} v_{\tilde \alpha} \\ \trsp{X}v \end{pmatrix} \in \yvecRing{2m}.
\]
Therefore \(q\) is a right multiple of the approximant basis \(\bar P =
(\begin{smallmatrix} {P} & u \\ 0 & \lambda \end{smallmatrix})\), which shows
that \(\lambda\) is an element of \(\field\) {(we had shown \(\lambda \neq 0\) in $\field[y]$ at \cref{algo:ChangeOfBasis:reconstruct}), hence $\lambda=1$ as it is a monic polynomial in the Popov form.  It follows that \(\tilde u -u\)
is a right multiple of \({P}\),} and we check finally that the remainder of
\(\tilde u\) in the division by \({P}\)---which is \(u\) by construction---is
\(\tilde u\) itself. Indeed the degree of the \(i\)th entry of
\(v_{\tilde\alpha}\) is less than the \(i\)th diagonal degree of \(R\), which
is the \(i\)th diagonal degree of \({P}\); and as seen above all entries of
\(\trsp{X}v\) have degree at most \(\deg(v_{\tilde\alpha})-1 < \deg(R)-1\),
with \(\deg(R)\) being itself at most the \(i\)th diagonal degree of \({P}\)
for \(m+1\le i\le 2m\). In particular \(v_{\bar \alpha} = v_{\tilde \alpha}\),
hence \(\bar \alpha =\tilde \alpha\) and \(\bar \alpha(x,\gamma) = a \bmod f\).

\paragraph*{\cref{algo:ChangeOfBasis:alpha} computes \(\alpha\in\xRing_{<n}\) such that \(\alpha(\gamma) = a \bmod f\).}

Since \(\deg(\minpoly[]) = n\), the Hermite normal form of \(R\) has the shape
\(T = (\begin{smallmatrix} 
  \minpoly[] & T_{1,*} \\
  0 & \idMat{m-1}
\end{smallmatrix})\), with \(T_{1,*} = (T_{1,2} \;\; \cdots \;\; T_{1,m})\) 
and $\deg(T_{1,j}) < \deg(\mu) = n$ for $2 \leq j\leq m$.
Then, the polynomial \(\alpha = \bar\alpha_1 - (T_{1,2}\bar\alpha_2 + \cdots +
T_{1,m}\bar\alpha_{m}) \rem \minpoly[]\) constructed at
\cref{algo:ChangeOfBasis:alpha} has degree less than \(n\) and, by construction
as well, the vector \(v_\alpha = \trsp{(\alpha \;\; 0 \;\; \cdots \;\; 0)} \in
\yvecRing{m}\) is such that \(v_{\bar \alpha}-v_\alpha\) is a right multiple of
\(T\). (In fact, \(v_\alpha\) is the unique remainder in the division of
\(v_{\bar\alpha} \) by \(T\).) In particular, \(v_{\bar \alpha}-v_\alpha\) is a
right multiple of \(R\), meaning that \(v_\alpha\) is equal to \(v_{\bar
\alpha}\) modulo relations of~\(\rmodfg\), which implies \(\alpha(\gamma) =
\bar\alpha(x,\gamma) \equiv a \bmod f\). The computation of \(\alpha\) costs
\(\softO{nm}\) operations in \(\field\).

This concludes the proof of the properties of \((R,\minpoly[],\alpha)\) in the
case where the algorithm does not return \Fail. Furthermore, adding the above
costs yields the cost bound claimed in the lemma, which therefore holds in
general since the cost can only be smaller when the algorithm returns \Fail.

\paragraph*{Proof of the last claim}

The assumption \(\gcd(\gamma,f)=1\) ensures that
\cref{algo:ChangeOfBasis:invertible} does not return \Fail, in which case we
have seen that \({P}\) is a weak Popov approximant basis of \((S \;\;\;
-\idMat{m})\) at order \(2d\). 

From $\deg(\mu_{\gamma})=n$ and \cref{prop:invariant-factors} we know that
\(\ddfg = n\), hence with the assumption \(\ddmmfg=\ddfg \) we have
\(\ddmmfg=\ddfg=n \).  Using \(\denom{H} = \rmodmmfg=\rmodfg\), and
\cref{lem:fraction-reconstruction:item:reconstruct} of
\cref{lem:fraction-reconstruction} thanks to the assumption on
$H(y)=\trsp{X}(\charmat[\gamma])^{-1}X$, we deduce that the \(m\) rightmost
columns of \({P}\) have degree at least \(\deg(R)\) and that \(R\) is a basis
of \(\rmodfg\) with \(\deg(R) \le d\). In particular \(\degdet{R}=n\), and it
follows that \cref{algo:ChangeOfBasis:verif} does not return \Fail.

Then, the assumption on the degree of the minimal polynomial also ensures,
using \cref{prop:invariant-factors} as above, that the first diagonal entry of
the Hermite normal form \(T\) of \(R\) is \(\minpoly[\gamma]\), and is the
polynomial \(\minpoly[]\) computed at \cref{algo:ChangeOfBasis:hnf}. Therefore
\cref{algo:ChangeOfBasis:hnf} does not return \Fail{} either: we have proved
that, under the assumptions \(\gcd(\gamma,f)=1\), \(\ddmmfg = \ddfg\),
$\deg(\mu_{\gamma})=n$, and \(H = \trsp{X}(\charmat[\gamma])^{-1}X\) is
describable in degree~$d$, then the output is not \Fail{} and \(\deg(R)\le d\).
\end{proof}

\subsubsection*{Notes.}

Shoup's algorithm for computing \(\alpha\) in the case \(m=1\) uses only \(n\)
terms of the sequence $(\ell(\gamma^k a))_{k\ge0}$, or more generally
\(d\)
terms, where \(d\) is a known bound on \(\deg(\minpoly)\). Here as well, if one
knows that the sought basis of relations satisfies \(\deg(R) \le d\), for
example under the conditions of \cref{prop:algo:ChangeOfBasis} ensuring
success, then the algorithm may be modified so as to require only \(d\) terms
of the expansion of \(-\trsp{X}(\charmat[\gamma])^{-1}v_a\) instead of \(2d\).
The vector \(v_{\tilde{\alpha}}\) would appear in the approximant basis at
order \(d\), and from there one would consider a residual approximant problem
focusing on obtaining the missing part of \(R\). This is not detailed here, as
this would complicate the presentation without bringing an improvement to the
asymptotic complexity.

Modular composition and inverse composition are very similar. They both involve
the computation of a matrix of relations and use symmetric steps with similar
complexities. Indeed, the division with remainder of
\cref{sec:bivcomposition:division} is used in both algorithms to
change
between univariate and bivariate representations efficiently. Also, the application of
\algoName{algo:BivariateModularComposition}{} at the last step
of composition in 
\algoName{algo:BivariateModularCompositionWithRelationMatrix} is
reflected
by
\Call{algo:TruncatedPowers}{} as starting step of inverse composition in
\algoName{algo:ChangeOfBasis}. Both these steps have cost
$\softO{\bicost{d}}$
from \cref{prop:NuskenZiegler,prop:TruncatedPowers}, respectively (see also \cref{subsubsec:transpNZ}).


\section{The block Hankel matrix \texorpdfstring{$\hkfa$}{Hk(a,f)} and its generic properties}
\label{sec:genericity}

Matrices of relations are obtained either by \algoName{algo:MatrixOfRelations}
directly, or by \algoName{algo:ChangeOfBasis} after a change of basis. In both
cases, for the correctness of the computation to be granted via
\cref{prop:certificate,prop:algo:ChangeOfBasis}, we need $\dd[m,m]^{(a,f)}$ and
$\ddfa$ to be equal (and, equivalently, $\rmodmmfa = \rmodfa$) and the fraction
$H(y)= \trsp{X}(\charmat)^{-1}X$ to be describable in degree~$d$, or the same
statement with $\gamma$ in place of~$a$. It is thus important to understand
when these properties hold. 

Recall from \cref{subsec:linalginterpret} the matrices
$\Ra$ and $\La$, that are defined for $m\in \{1,\ldots, n\}$ by
{\[
  \La=\begin{pmatrix} \trsp{X}\\ \vdots\\ \trsp{X}M_a^{d-1}\\ \end{pmatrix}
  \in\matRing{(md)}{n}
  \quad\text{and}\quad
   \Ra=\begin{pmatrix}X&\dotsb&M_a^{d-1}X\end{pmatrix}
  \in\matRing{n}{(md)}
\]}

and that correspond to Algorithms \nameref{algo:BivariateModularComposition}
and \nameref{algo:TruncatedPowers} respectively, and also to the maps $\rhoa$
and $\lba$. Their product forms the \emph{block Hankel matrix}
\begin{equation}
  \label{eq:def_Hk}
  \hkfa  = \La\Ra = 
  \left(
    \begin{array}{cccc}
      H_0  & H_1 & \dots &  H_{d-1}\\
      H_1 & ~~~\iddots & ~~\iddots & H_{d} \\
      \vdots & \iddots & \iddots & \vdots \\
      H_{d-1} & H_{d}  & \dots & H_{2d -2}
    \end{array}
  \right) \in \matRing{(md)}{(md)},
\end{equation}
with $H_k=\trsp{X}M_a^kX$ for $k$ in $\NN$. This matrix, and in particular its
rank, is strongly related to the two properties mentioned above
\citetext{\citealp{Vil97:TR}; \citealp[p.\,97]{KaVi05}}.

The outcomes of this section are the following. For any positive
parameters $m\leq n$ and $d$, as soon as~$\rank\hkfa=\dd ^{(a,f)}$, then
$\dd[m,m]^{(a,f)}= \ddfa$ and $H$ is describable in
degree~$d$ (\cref{subsec:blockHankel}).  This happens in particular
when $f(0)\neq0$, $d\ge \lceil n/m\rceil$ and either $\deg(a)=m$
(\cref{subsec:rankn}) or for a generic choice of $a$
(\cref{subsubsec:genina}). Also, for generic choices of the roots of
$f$ and of the values of $a$ at these roots, $\rank\hkfa=\dd ^{(a,f)}$
as soon as $d\geq \lceil \dd ^{(a,f)}/m\rceil$
(\cref{ssec:genranksep}).
As in previous sections, notation such as
$\ddfa, \hkfa, \lba$, etc. is often shortened into $\dd, \hkfasimp,
\lbasimp$, etc. 

These results will be used in \cref{sec:composition_randomized} for
the analysis of the randomized composition algorithm when $f$ is
separable (\cref{sec:composition_randomized:proof_separable}), or when
$f$ is purely inseparable, which includes the case of power series
composition (\cref
{sec:composition_randomized:proof_inseparable,sec:composition_randomized:proof_inseparable_small}).

\subsection{Relation between block Hankel matrix rank and fraction
description degree} 
\label{subsec:blockHankel}

The key condition to control the degrees of fraction descriptions of $H(y)$ and
obtain matrices of relations is the equality
\[
  \rank{\hkfasimp} = \dd.
\]
The special case when $\rank{\hkfasimp} = n$ is common, and appears
naturally later on. The proof of the following result relies in an
essential manner on \cref{lem:hankel_rank}, which we give next
(the references we cite only give a sketch of proof).

\begin{proposition}
  \label{prop:hankel_rank_nu}
  Given $f\in\field[x]$ of degree~$n$, $a\in\field[x]_{<n}$, and positive
  integers $m \leq n$ and $d$, the rank of $\hkfa$ is at most $\ddfa$.
  In case of equality, we have $\dd[m,m]^{(a,f)}=\ddfa$  and $H(y)=
  \trsp{X}(\charmat)^{-1}X$ is describable in degree~$d$.

  In particular, if $\hkfa$ has rank $n$, then $\dd[m,m]^{(a,f)}=\ddfa=n$ and
  $H(y)$ is describable in degree~$d$.
\end{proposition}

\begin{proof}
  Using \cref{prop:invariant-factors} the inclusion \(\rmod \subseteq
  \rmodmm\) implies $\ddmm \leq \dd \leq n$, so that by
  \cref{lem:hankel_rank} below we have $\rank\hkfasimp \leq \ddmm \leq
  \dd\le n$.  If $\hkfasimp$ has rank $\dd$, then $\dd = \ddmm$, and
  the claim on $H$ follows again from \cref{lem:hankel_rank}. The case
  where the rank is~$n$ follows similarly.
\end{proof}

\begin{lemma}[{\cite[Sec.\,2.1]{KaVi05} and \cite[Lem.\,2.4]{Vil18}}]
  \label{lem:hankel_rank}
  For positive integers $m \leq n$ and $d$, the rank of $\hkfa$ is at most
  $\dd[m,m]^{(a,f)}$, with equality if and only if $H(y)=
  \trsp{X}(\charmat)^{-1}X$ is describable in degree~$d$.
\end{lemma}

\begin{proof}
We denote by $H_k = \trsp{X} M_a^k X \in \matRing{m}{m}$ the coefficient in the
expansion of \(H\) at infinity:
\begin{equation*}
  H(y) = \trsp{X}(\charmat)^{-1} X
  = \sum_{k\ge0} H_k y^{-k-1}
  = \sum_{k\ge0} \trsp{X} M_a^k X y^{-k-1}.
\end{equation*}

To show that the rank is at most $\ddmm$ we first note that $\hkfasimp$ is a
submatrix of $\hkfasimp[m,d+1]$ for  \(d\ge 0\), the sequence
\((\rank{\hkfasimp})_{d\ge 0}\) is thus nondecreasing.  Since \(\rmodmm\) is
the module of vector generators for the sequence \(\{H_k\}_{k \ge 0}\)
(\cref{sec:relmat:relations_denominators}), the minimal generating polynomial
$F \in \ymatRing{m}{m}$ in Popov form for that sequence is a basis of \(\rmodmm\)
(\cite[Def.\,2.5]{Vil97:TR} and \cite[Def.\,2.3]{KaVi05}). It follows that
$\degdet{F}=\ddmm$, and \cite[Eq.\,(2.6)]{KaVi05} shows that for
\(d\ge n\), the rank of $\hkfasimp$ is $\ddmm$. So the first claim is proved.

{From \cref{lemma:denominators}, $F$ is also a basis of $\denom{H}$; we now study the descriptions of $H$ 
by exploiting identities that we used to prove this lemma.}
If the rank of~$\hkfasimp$ is equal to $\ddmm$, then this rank is also
that of the infinite matrix corresponding to the system (see also
\cref{eq:Hankel-prod})
\begin{equation} \label{eq:appendixHinfty}
  H_kv_0 + \dots + H_{k+d}v_{d} = 0 \quad \text{for } k \geq 0,
\end{equation}
thus a solution to 
\begin{equation} \label{eq:appendixHtrunc}
  H_kv_0 + \dots + H_{k+d}v_{d} = 0 \quad \text{for } 0 \leq k \leq d-1
\end{equation}
is also a solution to \cref{eq:appendixHinfty}. Since the rank of
$\hkfasimp$ is maximal we also know that the last block column of
$\hkfasimp[m,d+1]$ is a linear combination of the previous ones.  This
provides with~$m$ linearly independent $R_1,\ldots,R_m \in
\yvecRing{m}$, of degree $d$, whose coefficient vectors in~$y$ are
solutions to~\cref{eq:appendixHtrunc}, hence to
\cref{eq:appendixHinfty}. Let~$R$ be the matrix in $\ymatRing{m}{m}$
whose $j$th column is $R_j$.  Using \cref{eq:appendixHinfty} we deduce
that $HR=Q$ with~$Q\in \ymatRing{m}{m}$ (see also
\cref{eq:smallnumer}). This gives a right fraction description
$H=QR^{-1}$ (which may not be irreducible) with denominator of
degree~$d$. The same reasoning on the left side gives a left matrix
description of degree~$d$, hence $H$ is describable in degree~$d$.

Conversely, a right matrix description $H=QR^{-1}$ with $R$ of degree at
most~$d$ gives $R_j$'s whose coefficient vectors are solutions to
\cref{eq:appendixHinfty}. Since $F$ is a basis of the module of
vector generators for \(\{H_k\}_{k \ge 0}\), $R$ must be a multiple of
$F$. By minimality $F$ has degree at
most~$d$~\cite[Thm.\,6.5-10, p.\,458]{Kailath80}, and
using~\cite[Eq.\,(2.6)]{KaVi05} the rank of the infinite block Hankel matrix
restricted to its  first~$d$ block columns is maximal. Starting from a left
description, in an analogous way we obtain that the rank restricted to the
first $d$ block rows is maximal, which yields that $\hkfasimp$ has
rank $\degdet{F}=\ddmm$. 
\end{proof}

\subsection{Families with \texorpdfstring{$\hkfa$}{block Hankel matrix} of rank \texorpdfstring{$n$}{n}} 
\label{subsec:rankn}

A simple condition implies the equality $\rank{\hkfasimp}=n$
of \cref{prop:hankel_rank_nu}.

\begin{proposition}\label{lem:hankelwithxm}
  Let $f\in\field[x]$ have degree~$n$, let $a\in\field[x]_{<n}$, and let \(m\)
  be a positive integer. If $f(0)\neq 0$ and \(\deg(a) = m\) (hence \(1 \le m < n\)),
  then the block Hankel matrix $\hkfa \in
  \matRing{(md)}{(md)}$ has rank $n$ for all $d \ge \lceil n/m\rceil$.
\end{proposition}

The rest of this subsection is devoted to the proof of this result. It is a
basis for the genericity result in the next subsection.

\begin{proof}
  {For a given $c$ in $\field$, by construction of this block-Hankel matrix, one
  has
  \[
    {\hankel{m,d}{c a}{f}} = C \hkfa C,
    \text{ where }
    C = \diag{\underbrace{1,\ldots,1}_{m \text{ times}},\underbrace{c,\ldots,c}_{m \text{ times}},\ldots,\underbrace{c^{d-1},\ldots,c^{d-1}}_{m \text{ times}}}.
  \]
  }%
  It follows that $\rank {\hankel{m,d}{c a}{f}} = \rank \hkfa$ for any $c\neq 0$, and
  therefore in the rest of the proof we can assume that $a$ is {\em monic} of
  degree $m$.

  By \cref{eq:def_Hk}, it is sufficient to show
  that the mappings $\rhoasimp$ and
  $\lbasimp$ associated to $\Rasimp$ and
  $\Lasimp$ are surjective and injective, respectively.

\paragraph{The mapping $\rhoasimp$ is surjective.}

By assumption, $n \le md$ so that surjectivity of $\rhoasimp$ is equivalent to
the matrix $\Rasimp \in \matRing{n}{(md)}$ from \cref{eq:factor_Hk} having full
row rank $n$. Indeed, the first $n$ columns of $\Rasimp$ are the coefficients
of the family of polynomials $x^ia^j \rem f$, for $0\le i<m$ and $0\le j<d$,
with $0 \le i + jm < n$.  Since $\deg(a)=m$, these columns form an upper
triangular matrix, with $1$'s on the diagonal; this proves the claim.

\paragraph{The mapping $\lbasimp$ is injective.}

Equivalently, we have to show that $\Lasimp$ has full column rank $n$. This
 follows from the structure of this matrix, seen at the level of
polynomials. 

\begin{lemma}
  With the notation and hypotheses of \cref{lem:hankelwithxm}, let
  \[
    p_{i}=[ax^{n-m+i}\rem f]_0^{m-1},\quad
    i=0,\dots,m-1.
  \]
  Then,
  \begin{enumerate}[label={\it (\roman*)},ref={\roman*}]
    \item\label{lem:hankelwithxm:proofitemi}
      if $m \le n/2$, the $m$ polynomials $p_0,\dots,p_{m-1}$ are linearly
      independent;
    \item\label{lem:hankelwithxm:proofitemii}
      if $n/2 < m$, the $n-m$ polynomials $p_{2m-n},\dots,p_{m-1}$ are
      linearly independent.
  \end{enumerate}
\end{lemma}
\begin{proof}
  The two cases require different proofs, sharing common ingredients. For $i
  \ge 0$, let $r_i = x^{n+i} \rem f$. For $b$ in $\xRing_{<n}$, we then have
  \begin{equation}\label{eq:mulb}
    x^i b \rem f = [x^i b]_{0}^{n-1} +  \delta_{b,i},  
  \end{equation}
  for some $\delta_{b,i}$ in $\Span(r_0,\dots,r_{i-1})$, in particular
  $\delta_{b,0}=0$. Applying this to $b=r_0 = x^n \rem f$ yields $r_i = [x^i
  r_0]_{0}^{n-1} + \delta_{r_0,i}$.  Taking this relation modulo $x^m$ gives
  $[r_i]_0^{m-1} = [x^i r_0]_{0}^{m-1} + \mu_{i}$, with $\mu_{i}$ in
  $\Span([r_0]_0^{m-1},\dots,[r_{i-1}]_0^{m-1})$ for $i>0$ and $\mu_0 = 0$. By induction on $i \ge 0$, one deduces that
  \[
    \Span([r_0]_0^{m-1},\dots,[r_{i}]_0^{m-1})
    =\Span([r_0]_0^{m-1},\dots,[x^i r_{0}]_0^{m-1}).
  \]
  Writing $f=f_0 + \cdots + f_{n-1}x^{n-1} +x^n$, we get $r_0 =
  -f_0-f_1x-\dots-f_{n-1}x^{n-1}$. Since $f_0 \ne 0$ by assumption, $[x^i
  r_{0}]_0^{m-1}$ has valuation $i$ for $0 \le i < m$; this implies that
  $\Span([r_0]_0^{m-1},\dots,[r_{i}]_0^{m-1})$ has dimension $i+1$ for
  $0 \le i < m$.

  \paragraph{Proof of \cref{lem:hankelwithxm:proofitemi}.} 

  Let $b=ax^{n-m} \rem f$ in~\cref{eq:mulb}. Upon reduction modulo $x^m$, for
  $0 \le i < m$, we obtain the relation $p_i = [x^i b]_{0}^{m-1} + \mu'_{i}$,
  with $\mu'_{i}$ in $\Span([r_0]_0^{m-1},\dots,[r_{i-1}]_0^{m-1})$.

  Since $ax^{n-m}$ is monic of degree $n$ ($a$ has degree $m$), with valuation
  at least $n-m \ge m$ (here, $m \le n/2$), we get $[b]_0^{m-1} =
  [r_0]_0^{m-1}$, and thus $[x^i b]_0^{m-1} = [x^i r_0]_0^{m-1}$ for $0\le
  i<m$.  This gives $p_i = [r_i]_{0}^{m-1} + \mu'_{i}-\mu_i$, with
  $\mu'_{i}-\mu_i$ in $\Span([r_0]_0^{m-1},\dots,[r_{i-1}]_0^{m-1})$.  In
  particular, taking all $i$ up to $m-1$, we get the equality
  $\Span(p_0,\dots,p_{m-1}) =\Span([r_0]_0^{m-1},\dots,[r_{m-1}]_0^{m-1})$, and
  we saw that the latter has dimension $m$.  \Cref{lem:hankelwithxm:proofitemi}
  is proved.

  \paragraph{Proof of \cref{lem:hankelwithxm:proofitemii}.}

  Assume that $q=c_mx^m+\dots+c_{n-1}x^{n-1}$ is such that $[ a q \rem
  f]_0^{m-1}=0$. We prove that all $c_i$'s vanish.

  We can rewrite $a q \rem f$ as $x^m b \rem f$, with $b = a (q/x^m)$; since
  $a$ has degree $m$, $b$ is in $\xRing_{<n}$. Applying \cref{eq:mulb} to $b$
  and $i=m$, our assumption that $[x^m b \rem f]_0^{m-1}=0$ implies
  $[\delta_{b,m} ]_{0}^{m-1} = 0$. {Writing \(\delta_{b,m} =
    \sum_{j=0}^{m-1} \bar\delta_{j} r_j\) for some
    \(\bar\delta_0,\ldots,\bar\delta_{m-1} \in \field\), we get
    \(\sum_{j=0}^{m-1} \bar\delta_{j} [r_j]_0^{m-1} = [\delta_{b,m}]_0^{m-1} =
    0\). The linear independence of $[r_0]_0^{m-1},\dots,[r_{m-1}]_0^{m-1}$
    ensures \(\bar\delta_j = 0\) for all \(j\), showing that $\delta_{b,m}$
    itself is zero. Hence $x^m b \rem f = [x^m b]_{0}^{n-1}$,
    from which we deduce \(x^n [x^m b]_{n}^{m-1} \rem f = 0\)
    using
    \[
      x^m b \rem f = ([x^m b]_0^{n-1} + x^n [x^m b]_{n}^{m-1}) \rem f
      = [x^m b]_0^{n-1} + (x^n [x^m b]_{n}^{m-1} \rem f).
    \]
    Since \(x^n [x^m b]_{n}^{m-1} = x^n [b]_{n-m}^{m-1}\), and since $f_0\ne 0$
  ensures that $x$ is invertible modulo $f$, it follows that $[b]_{n-m}^{m-1}$
vanishes modulo $f$, or equivalently that $[b]_{n-m}^{m-1}=0$.} Since $a$ is
monic of degree $m$, and since $n-m < m$, the definition of $b$ then implies
that all coefficients $c_i$'s vanish. Hence,
\cref{lem:hankelwithxm:proofitemii} is proved.
\end{proof}

  Let now $v \in \xRing_{<n}$ be such that
  \[
    [v]_0^{m-1}=[av\rem f]_0^{m-1}=\dots=[a^{d-1}v\rem f]_0^{m-1}=0.
  \]
  We prove that $\deg(v) < n-mi$ for $i=0,\dots,d-1$, by induction. For $d=\lceil n/m\rceil$,
  this gives $\deg(v) < m$; together with the assumption $[v]_0^{m-1}=0$, this
  proves that $v=0$.

  The base case of the induction is for $i=0$, and \(\deg(v) < n\) holds by assumption. If the claim
  holds for some index $i < d-1$, since $a$ has degree~$m$, for any $w$ in
  $\xRing_{<n}$, the polynomial $[aw\rem f]_0^{m-1}$ splits into two parts:
  \[
    [aw\rem f]_0^{m-1}=
    [a[w]_0^{n-m-1}]_0^{m-1}+[a x^{n-m} [w]_{n-m}^{m-1} \rem f]_0^{m-1}.
  \]
  Apply this identity with $w=a^iv \rem f$. Then, both the left-hand side and
  the first summand vanish: the former because $[a^{i+1} v\rem f]_0^{m-1}=0$,
  the latter because $[a^{i} v\rem f]_0^{m-1}=0$, i.e., $w=a^{i} v\rem f$ has
  valuation at least $m$. We deduce that $[a x^{n-m} [w]_{n-m}^{m-1}\rem f]_0^{m-1}=0$,
  with $w=a^iv \rem f$. 

  \begin{itemize}
    \item If $m \le n/2$, the linear independence of the polynomials
      $p_j=[ax^{n-m+j} \rem f]_0^{m-1}$, for $j=0,\dots,m-1$, then shows that
      $[w]_{n-m}^{m-1} = [a^iv \rem f]_{n-m}^{m-1}$ vanishes.
    \item If $m > n/2$, then the assumption that $w$ has valuation at
      least $m$, with thus $m > n-m$, shows that
      $[w]_{n-m}^{m-1}=x^{2m-n}[w]_{m}^{n-m-1}$. In this case, the linear
      independence of the polynomials $p_j$ for $j=2m-n,\dots,m-1$
      shows that $[w]_{n-m}^{m-1}=0$.
  \end{itemize}
  In other words, in both cases, we have proved that $w=a^iv \rem f$ has degree
  less than $n-m$.

  On the other hand, the induction assumption that $\deg(v) < n-mi$ implies
  that $a^iv \rem f= a^iv$, so the latter has degree less than $n-m$. Since
  $a^i$ has degree $mi$, this shows that $\deg(v) < n-m(i+1)$, as claimed.
\end{proof}

\subsection{Generic regularity in \texorpdfstring{$a$}{a} and \texorpdfstring{$f$}{f}} 
\label{subsubsec:genina}

In all this document, genericity is understood in the Zariski sense:
\begin{definition}
  A property $\mathcal{P}$ of certain parameters $(u_1,\dots,u_s)$ holds for a
  \emph{generic} choice of $(u_1,\dots,u_s)$ in $\vecRing{s}$ if there exists a
  nonzero polynomial $\Delta$ in $\field[\bar u_1,\dots,\bar u_s]$ (where the
  $\bar u_i$'s are new indeterminates) such that $\Delta(u_1,\dots,u_s) \ne 0$
  implies that $\mathcal{P}(u_1,\dots,u_s)$ holds.
\end{definition}
Note that if $\field$ is finite, there may be no choice of the $u_i$'s
in $\field$ for which $\Delta$ does not vanish, but such points exist
in a finite extension of $\field$ of sufficiently large degree (such
as $\bigO{\log(n)}$ when the degree of~$\Delta$ is polynomial in~$n$,
as is the case below).

\begin{proposition}\label{prop:generic_a}
  Let $f$ in $\xRing$ be of degree $n$ and such that $f(0)\neq
  0$. For any $m \in \{1, \ldots , n\}$ there
  exists a nonzero polynomial $\Delta_{f,m}$ in $\field [\bar
    a_0,\dots,\bar a _{n-1}]$ of degree at most ${2n^2/m}$ such that
  for $a=a_0+\dots+a_{n-1}x^{n-1}$ in $\xRing_{<n}$, if
  $\Delta_{f,m}(a_0,\dots,a_{n-1})\neq0$ then $\hkfa \in
  \matRing{(md)}{(md)}$ has rank~$n$ for any $d \ge \lceil n/m\rceil$.
\end{proposition}

\subsubsection{Proof of \cref{prop:generic_a}}
\begin{lemma} \label{cor:hankelwithxm}
  Let $m$, $n$ be positive integers, with $m \in \{1, \ldots , n\}$, and let
  $\varf = \varf_0 + \cdots + \varf_{n-1} x^{n-1} +x^n$ and
  $\vara = \vara_0 + \cdots + \vara_{n-1} x^{n-1}$ be polynomials in
  $\ZZ[\vara_0,\dots,\vara_{n-1},\varf_0,\dots,\varf_{n-1}][x]$.
  Then any $n$-minor of $\hankel{m,\lceil n/m\rceil}{\vara}{\varf}$
  has degree at most $2n^2/m$ in $\vara_0,\dots,\vara_{n-1}$ and
  $2n^2(n-1)/m$ in $\varf_0,\dots,\varf_{n-1}$.
\end{lemma}
\begin{proof}
{The multiplication matrix $M_{\vara}$
can be written as  $M_{\vara}=\sum_{k=0}^{n-1} \vara _k M_{x}^k$, where $M_{x}$ is the companion matrix of 
$\varf$. 
The entries of $M_{\vara}$, which are the
coefficients of $x^k\vara\rem \varf$ for $k=0,\dots,n-1$, are
therefore polynomials of degree~$1$ in the
coefficients~$\vara_0,\dots,\vara_{n-1}$ and at most $n-1$ in the
coefficients $\varf_0,\dots,\varf_{n-1}$.} In turn, the coefficients of
$M^j_{\vara}$ have degree at most~$j$ in~$\vara_0,\dots,\vara_{n-1}$
and $j(n-1)$ in $\varf_0,\dots,\varf_{n-1}$.  For $0 \le i,j < \lceil
n/m\rceil$, the $m\times m$ block of coordinates $(i,j)$ in
$\hankel{m,\lceil n/m\rceil}{\vara}{\varf}$ is a submatrix of
$M^{i+j}_{\vara}$; it has degree at most~$i+j$ in
$\vara_0,\dots,\vara_{n-1}$ and $(i+j)(n-1)$ in
$\varf_0,\dots,\varf_{n-1}$. As a result, any $n$-minor of this matrix
has degree at most $m \lceil n/m\rceil(\lceil n/m\rceil-1) \le 2n^2/m$
in $\vara_0,\dots,\vara_{n-1}$ and $m\lceil n/m\rceil (\lceil
n/m\rceil-1)(n-1) \le 2n^2(n-1)/m$ in $\varf_0,\dots,\varf_{n-1}$.
\end{proof}
Take $f$ of degree $n$ with $f(0) \ne 0$. \Cref{lem:hankelwithxm} with
$a=x^m$ shows that at least one $n$-minor of $\hankel{m,\lceil
  n/m\rceil}{x^m}{f}$ is nonzero, so the corresponding $n$-minor
of $\hankel{m,\lceil n/m\rceil}{\vara}{f}$ is not identically zero. We
take this minor for $\Delta_{f,m}$, and its degree is then bounded by
\cref{cor:hankelwithxm}.

\subsubsection{Note: basis of relations for a generic \texorpdfstring{$a$}{a}.}
\label{notes:genregularity}

For any $f$ in $\xRing$ with $f(0)\neq 0$, and for a generic $a$ in~$\xRing_{<
n}$, \cref{prop:generic_a} shows that the rank of $\hankel{m,d}{a}{f}$ is $n$,
with $d=\lceil n/m\rceil$. From \cref{prop:hankel_rank_nu} we then obtain
$\ddfa=\dd[m,m]^{(a,f)}=n$ and the describability of~$H$ in degree $d$.
Therefore, by \cref{prop:compute_Mmm}, \algoName{algo:CandidateBasis} returns a
basis of \(\rmodfa\) and the flag \(\Cert\).

\subsection{Generic rank for a separable~\texorpdfstring{$f$}{f}}
\label{ssec:genranksep}

We now study the rank of $\hkfasimp$, for a generic choice of the roots of
$f$, and for a generic choice of the values of $a$ at these roots,
subject to certain combinatorial conditions.

\subsubsection{Definitions.}
Consider pairwise distinct $\xi_1,\dots,\xi_n$ in an algebraic closure
$\Kbar$ of~$\field$. To such points, we associate the 
polynomial $f = (x- \xi_1)\dotsm(x-\xi_n)$.  We also consider
 $a \in \xRing_{<n}$, and we say that $a$ {\em takes values
  $\lambda_1,\dots,\lambda_r$ at $\xi_1,\dots,\xi_n$ with
  multiplicities} $\ell_1,\dots,\ell_r$ if the following holds:
\begin{itemize}
\item[-]$\lambda_1,\dots,\lambda_r$ are pairwise distinct elements in $\Kbar$;
\item[-] $\ell_1 + \dots + \ell_r=n$, with all $\ell_i$ positive integers;
\item[-] for $i=1,\dots,r$, $a(\xi_{\sigma_i+1}) = \cdots =
  a(\xi_{\sigma_i+\ell_i})= \lambda_i$, where we write $\sigma_i = \ell_1 +
  \cdots +\ell_{i-1}$ (the empty sum for $i=1$ is zero).
\end{itemize}
In view of our application, we also
assume that the $\xi_i$'s are such that $f$ is in $\xRing$.

\subsubsection{Generic rank.}
\begin{proposition}\label{prop:separable}
  Fix positive integers $m \in \{1,\dots,n\}$ and $\ell =
  (\ell_1,\dots,\ell_r)$ such that $\ell_1 + \cdots + \ell_r = n$.  Then, there
  exists a nonzero polynomial $\polbU_{\ell,m}\in\ZZ
  [\varxi_1,\dots,\varxi_n,\varlb_1,\dots,\varlb_r]$ such that the following
  holds.
  For pairwise distinct nonzero $\xi_1,\dots,\xi_n$ in $\Kbar$ such that
  $f = c(x- \xi_1)\dotsm (x-\xi_n)$ with $c\in\field\setminus\{0\}$ is
  in~$\xRing$ and for $a\in\xRing$ that takes values
  $\lambda_1,\dots,\lambda_r$ at $\xi_1,\dots,\xi_n$ with multiplicities
  $\ell_1,\dots,\ell_r$, if $\polbU_ {\ell,m}
  (\xi_1,\dots,\xi_n,\lambda_1,\dots,\lambda_r)$ is nonzero, then
  \[
    \rank{\hkfa}=\ddfa\quad\text{for any }\; d\ge\lceil\ddfa/m\rceil,
  \]
  with in addition the equality
  \[
    \ddfa = \sum_{i=1}^r{\min(\ell_i,m)}.
  \]
  Finally, for any pairwise distinct $\lambda_1,\dots,\lambda_r$,
  the polynomial
  $\polbU_{\ell,m}(\varxi_1,\dots,\varxi_n,\lambda_1,\dots,\lambda_r)$ is
  nonzero and has degree at most $2n^2$.
\end{proposition}

\subsubsection{Proof of \cref{prop:separable}.}

The rather long proof is decomposed as follows. First, the expression for the
determinantal degree~$\dd$ is established. For the proof of the rest of the
proposition we exploit the factorization $\hkfasimp = \Lasimp \Rasimp$, that is
analyzed through a series of lemmas. {All along, we use classical linear
algebra notions concerning invariant factors and Smith normal forms, and their
relation to eigenvalues in the case of diagonalizable matrices; see e.g.\
\cite[Ch.\,II and III]{New72} for more background on these aspects.}

The ranks of the matrices $\Rasimp$ and $\Lasimp$ are related to that of a
simple matrix $ \mathcal{P}_{\ell,m,d}$ (see \cref{eq:plmd}). This leads to the
proof that for $d=\lceil \dd/m\rceil$, the rank of $\Rasimp$ and $\Lasimp $ is
$\nu_m$ generically. Then we prove that generically, taking any $d_0 \ge \lceil
\dd/m\rceil$ is sufficient for studying the rank of $\hkfasimp$. The proof is
concluded by establishing that the rank is $\dd$ when $d_0$ is $r$, the number
of distinct values $a(\xi_k)$'s: for this value of $d_0$, we establish that the
intersection of the image of~$\Rasimp$ with the kernel of $\Lasimp$ is reduced
to~0. The polynomial~$\Gamma_{\ell,m}$ and the degree bounds are derived from
the proof.

\paragraph{\it Determinantal degree~$\ddfa$.}

As in the proposition, let $\xi_1,\dots,\xi_n$ be pairwise distinct
in~$\Kbar$ and let $f = c(x- \xi_1)\dotsm(x-\xi_n)$. The Lagrange
interpolation polynomials
\begin{equation}\label{eq:lagrange-basis}
  \lag_k(x)=\frac1{f'(\xi_k)}\frac{f(x)}{x-\xi_k}=\prod_{\ell\neq k}
  \frac{x-\xi_{\ell}}{\xi_k-\xi_\ell},\quad k=1,\dots,n.
\end{equation}
form a basis of $\overline\quotient:=\Kbar[x]/\genBy{f}$.  For
any~$a\in\overline\quotient$, the matrix of multiplication by $a$ is
diagonalizable, its eigenvalues are the values of $a$ at
the~$\xi_j$'s, and the Lagrange polynomials are eigenvectors. The
characteristic polynomial $\chi_a$ of $a$ modulo $f$ is therefore
given by
\[
  \chi_a = \prod_{k=1}^n (y-a(\xi_k)) \in \Kbar[y].
\]
For $1\leq i \leq r$, we define $S_i = \{ k \in \{ 1,\dots,n\} \mid
a(\xi_k)=\lambda_i\}$ and use that
\begin{equation} \label{eq:defSi}
  S_ i= \{\sigma_i+1,\dots,\sigma_i+\ell_i\}.
\end{equation}
With these conventions we have the factorization 
\[
  \chi_a = \prod_{i=1}^r (y-\lambda_i)^{\ell_i},
\]
where the factors $(y-\lambda_i)$ are pairwise coprime. The Smith normal form of $\charmat$ is then known and an explicit
expression for the determinantal degree $\dd$ can be given:
$\charmat$ has~$\max (\ell_i)$ nontrivial invariant factors; for $1
\leq k \leq \max (\ell_i)$, the $k$th one is $\prod_{1 \le i \le r}
(y-\lambda_i)^{\varepsilon_{i,k}}$, where $\varepsilon_{i,k}= 1$ if $k
\le \ell_i $ and $0$ otherwise. From there, recalling from~\cref{eq:defnu} that for $m$ in
$\{1,\dots,n\}$, $\dd$ is the sum of the degrees of the first $m$ such
invariant factors, we have:
\begin{equation}\label{eq:nu:separable}
  \dd=\sum_{k=1}^{\min(m,\max(\ell_i))}{\card{\{i\mid\ell_i\le k\}}}
  =\sum_{i=1}^r{\min(\ell_i,m)}.
\end{equation}
This proves the claim regarding $\dd$ in the proposition (this claim thus
holds without further assumption on the $\xi_i$'s and $\lambda_i$'s).

\paragraph{Maximal rank of $\hkfa$.}
\begin{lemma} \label{lem:dimkrylov}
  Let $A \in \matRing{n}{n}$ and $m \in \NN_{>0}$, and let $\nu$ be the sum of
  the degrees of the $\min(m,n)$ highest degree invariant factors of
  $y\idMat{n}-A$. Then for any collection of \(m\) vectors $v_1, \ldots, v_m \in
  \vecRing{n}$, one has $\dim(\Span(A^i v_j , 0\leq i, 1\leq j\leq m)) \leq
  \nu$.
\end{lemma}
\begin{proof}
  We let $\tilde{\nu} = \dim(\Span(A^i v_j, 0\leq i, 1\leq j\leq m))$.
For $1\leq j \leq m$, let $d_j\geq 0$ be the first index such that
$A^{d_j}v_j \in \Span(v_j, Av_j, \ldots, A^{d_j-1}v_j, \{A^i v_k
\mid 0\leq i,0\leq k < j\})$; if $l\geq d_j$ then $A^l v_j$
also belongs to the latter subspace of~$\field ^n$, which is therefore
stable under left multiplication by~$A$. This holds for any $1\leq
j\leq
m$, hence $d_1+\ldots + d_m=\tilde \nu$, and the matrix
\[
  P_1 =
  \begin{pmatrix}
    v_1 \;\;\; A v_1 \;\;\; \cdots \;\;\; A v_1^{d_1-1} \;\;\; \cdots \;\;\; v_m \;\;\; A v_m \;\;\; \cdots \;\;\; A v_m^{d_m-1}
  \end{pmatrix} \in \vecRing{n\times \tilde \nu}
\]
has rank~$\tilde \nu$ and can be completed into a nonsingular matrix
$P=(P_1\;\;P_2) \in \matRing{n}{n}$. By applying the change of basis $P^{-1}AP$ 
we obtain 
\begin{equation} \label{eq:Frobenius}
P^{-1}(y\idMat{n} - A)P= \begin{pmatrix} y\idMat{\tilde \nu} - C &  B_1 \\ 0 & 
y\idMat{ n - \tilde \nu}- B_2 \end{pmatrix} \in \ymatRing{n}{n},
\end{equation}
where $C\in \matRing{\tilde \nu}{\tilde \nu}$, $B_1\in \matRing{\tilde \nu}{(n-\tilde \nu)}$,
 $B_2\in \matRing{(n-\tilde \nu)}{(n-\tilde \nu)}$. Thanks to the form of $P_1$, 
 the matrix  $C\in \matRing{\tilde \nu}{\tilde \nu}$ is block upper triangular with at most $m$ companion blocks $C_j$ of dimensions $d_j$ on the diagonal (there is no block for $d_j=0$, and at most $n$ of the $d_j$'s are nonzero).
By a unimodular row transformation $U_j\in \ymatRing{d_j}{d_j}$, a matrix
$y\idMat{d_j} - C_j$ can be brought into an upper triangular form
$T_j(y)=U_j(y)(y\idMat{d_j} - C_j)$, which has diagonal entries \(1\) except
for the last entry which is the characteristic polynomial
$\chi^{(j)}=y^{d_j}-\chi^{(j)}_{d_j-1}y^{d_j-1}- \dots - \chi^{(j)}_0$ of $C_j$:
\[
\left(\begin{array}{cccc}
 &  - 1&  &  \\
 & & \ddots  & \\ 
 &  & & -1 \\ 
1 & y & \ldots & y^{d_j-1}
\end{array}\right) \left(y\idMat{d_j} - 
\left(\begin{array}{cccc}
 &  &  &  \chi ^{(j)}_0\\
1 &  &  & \chi ^{(j)}_1\\ 
 & \ddots & &  \vdots \\ 
 &  & 1 &  \chi^{(j)}_{d_j-1}
\end{array}\right)\right)=
\left(\begin{array}{cccc}
1 &  \cdot & \cdot & \cdot \\
& \ddots & \cdot  & \cdot \\ 
 &  & 1 & \cdot \\ 
 &  &  & \chi^{(j)}(y)
\end{array}\right) \in \yRing ^{d_j \times d_j},
\]
Therefore \cref{eq:Frobenius} can be
rewritten as
\begin{equation} \label{eq:FrobeniusSmith}
U(y)P^{-1}(y\idMat{n} - A)P= \begin{pmatrix} T(y) &  \bar B_1(y) \\ 0 & y\idMat{n - \tilde \nu}- B_2 \end{pmatrix}
= 
\begin{pmatrix} \idMat{\tilde \nu} &  \bar B_1(y) \\ 0 & y\idMat{\tilde n- \nu}- B_2 \end{pmatrix}
\begin{pmatrix} T(y) &  0\\  0 & \idMat{n- \tilde \nu} \end{pmatrix},
\end{equation}
where $U = \textrm{diag}(U_1,\ldots U_m, \idMat{ n- \tilde \nu})$ is
unimodular (with no $U_j$ if $d_j=0$), and $T \in \field [y]^{\tilde
  \nu \times \tilde \nu}$ is block upper triangular with diagonal
blocks the $T_j$'s. The matrix $T$ is triangular with $1$'s on the
diagonal except for at most~$m$ entries.  We deduce that the gcd of
the minors of dimension $k$ of $T$ is a unit for $1 \leq k \leq \tilde
\nu -m$, and that $T$ has at most $m$ nontrivial invariant
factors~\cite[Ch.\,II, Eq.\,(13)]{New72}.  The product of these
invariant factors is $\det(T) = \prod _j \chi^{(j)}$, whose degree is
$d_1+\cdots+d_m= \tilde \nu$.  From the matrix product on the right-hand side
of \cref{eq:FrobeniusSmith}, these latter invariant factors divide the
$m$ highest degree invariant factors of $y\idMat{n} -
A$~\cite[Thm.\,II.14]{New72}. From the definition of $\nu$ we
obtain~$\tilde \nu \leq \nu$.
\end{proof}

With $A=M_a$ or $\trsp{M_a}$, and
$\Rasimp, \Lasimp$ from ~\cref{eq:factor_Hk}, for any positive integer
$d$, \cref{lem:dimkrylov} gives
\begin{equation} \label{eq:rankKmd}
\rank\Rasimp\le\dd, ~\rank\Lasimp\le\dd ~\textrm{and}~ \rank \hkfasimp \leq \dd.
\end{equation}

Next, we show that the ranks of both $\Rasimp[\lceil
  \dd/m\rceil]$ and $\Lasimp[\lceil \dd/m\rceil]$ are $\dd$
generically.

\paragraph{The relation of $\Ra$ and
$\La$ to the matrix $\mathcal{P}_{\ell,m,d}$.}

{For $\ell = (\ell_1,\dots,\ell_r)$, $m$ in $\{1,\dots,n\}$, and a
positive integer $d$, we define the matrix
\begin{equation}\label{eq:plmd}
  \mathcal{P}_{\ell,m,d}=
  \begin{pmatrix}
    V_{\varxi} & D_{\varlb}V_{\varxi} & \cdots & D_{\varlb}^{d-1} V_{\varxi}
  \end{pmatrix}
  \in\ZZ[\varxi_1,\dots,\varxi_n,\varlb_1,\dots,\varlb_r]^{n\times md},
\end{equation}
where
\[
  V_{\varxi} = 
  \begin{pmatrix}
    1      & \varxi_1 & \cdots & \varxi_1^{m-1} \\
    \vdots & \vdots   &        & \vdots         \\ 
    1      & \varxi_n & \cdots & \varxi_n^{m-1}
  \end{pmatrix}
  \quad\text{and}\quad
  D_{\varlb} = \diag{\underbrace{\varlb_1,\ldots,\varlb_1}_{\ell_1 \text{ times}},\ldots,\underbrace{\varlb_r,\ldots,\varlb_r}_{\ell_r \text{ times}}}
  .
\]
}
The following lemma summarizes the key properties of
this matrix in relation with the rank of $\Rasimp$ and $\Lasimp$. 

\begin{lemma}\label{lemma:dimL}
  Let $\ell$, $\xi_1,\dots,\xi_n$, $\lambda_1,\dots,\lambda_r$, $f$,
  $a$ and $m$ be as in \cref{prop:separable}, and let $d$ be a positive
  integer. The following holds:
  \begin{itemize}
  \item the rank of $\Rasimp$ is equal to the rank of
    $\mathcal{P}_{\ell,m,d}(\xi_1,\dots,\xi_n,\lambda_1,\dots,\lambda_r)$;
  \item if all $\xi_i$'s are nonzero, the rank of $\Lasimp$ is equal to
    the rank of
    $\mathcal{P}_{\ell,m,d}(1/\xi_1,\dots,1/\xi_n,\lambda_1,\dots,\lambda_r)$.
  \end{itemize}
\end{lemma}
\begin{proof}
 We use the same notation 
 \[\rhoasimp:\Kbar[x,y]_{<(m,d)}
   \to \Kbar[x]/\genBy{f} \quad\text{and}\quad \lbasimp:\Kbar[x]_{<n}
 \to\Kbar[x]_{<m}^d\] for the mappings induced by scalar extension
 from $\rhoasimp$ and $\lbasimp$ from \cref{subsec:linalginterpret}.

  Taking $(x^iy^j)_{0\le i<m,0\le j<d}$ for basis of~$\Kbar[x,y]_{<(m,d)}$ and the Lagrange basis $\lag_1,\dots,\lag_n$
  for~$\Kbar[x]_{<n}$, the matrix of $\rhoasimp$ is
  $\mathcal{P}_{\ell,m,d}(\xi_1,\dots,\xi_n,\lambda_1,\dots,\lambda_r)$.
  This proves the first point. 

  To prove the second point, take $k$ in $\{1,\dots,n\}$, and let $i$
  in $\{1,\dots,r\}$ be such that $a(\xi_k)=\lambda_i$. The image of
  the Lagrange polynomial~$\lag_k$ by~$\lbasimp$ is the  polynomial vector 
  \begin{align*} 
    \lbasimp(\lag_k)&=\left([\lag_k]_0^{m-1},[a \lag_k\rem f]_0^{m-1},\dots,[a^
        {d-1}\lag_k\rem f]_0^{m-1}\right) \in \Kbar[x]_{<m}^d,\\
    \intertext{and since the Lagrange polynomials are eigenvectors of
      multiplication by~$a$, we get}
       \lbasimp(\lag_k) &=\left([\lag_k]_0^{m-1},[\lambda_i \lag_k]_0^{m-1},\dots,[\lambda_i^
        {d-1}\lag_k]_0^{m-1}\right) \\
      &=\left([\lag_k]_0^{m-1},\lambda_i [\lag_k]_0^{m-1},\dots,\lambda_i^{d-1}
        [\lag_k]_0^{m-1}\right).
  \end{align*}
  Let $L' \in  \bmatRing{(md)}{n}$ be the matrix whose $k$-th
  column (for $k=1,\dots,n$) contains the $md$ coefficients of the
  entries of $\lbasimp(\lag_k)$. This is the matrix of
  $\lbasimp$, if we take the Lagrange basis for the domain~$\Kbar[x]_{<n}$.

  Since all $\xi_i$'s are nonzero, we get $f(0)\neq0$, so that $f$ is
  invertible as a power series.  Because the $\field$-linear
  transformation $b \in  \xRing_{<m} \mapsto [b/f]_0^{m-1}$ is
  invertible, $L'$ has the same rank as the matrix whose columns are the
  coefficients of the vectors
  \[ \left( [[\lag_k]_0^{m-1}/f]_0^{m-1},\lambda_i [[\lag_k]_0^{m-1}/f]_0^{m-1},\dots,\lambda_i^{d-1} [[\lag_k]_0^{m-1}/f]_0^{m-1}\right), \]
  for $i$ and $k$ as above.
  On the other hand, we have
  $[[\lag_k]_0^{m-1}/f]_0^{m-1}=[\lag_k/f]_0^{m-1}$ and
  \[\frac{\lag_k}{f}=\frac{1}{f'(\xi_k)}\frac 1{x-\xi_k}.\] This shows
  that to determine the rank of $L'$, we may as well consider the vectors
  \[\left(\left[\frac1{x-\xi_k}\right]_0^{m-1},\lambda_i\left[\frac1{x-\xi_k}\right]_0^{m-1},\dots,{\lambda_i^{d-1}}\left[
    \frac1{x-\xi_k}\right]_0^{m-1}\right).\] 
  Now, note that 
  \[\left[\frac1{x-\xi_k}\right]_0^{m-1} = -\xi_k \left (1 + \frac{1}{\xi_k} x + \cdots + \frac{1}{\xi_k^{m-1}} x^{m-1}\right).\]
  Thus, up to the factors~$-\xi_k$, taking the $md$ coefficients of
  these vectors and putting them in columns gives us the transpose
  of
  $\mathcal{P}_{\ell,m,d}(1/\xi_1,\dots,1/\xi_n,\lambda_1,\dots,\lambda_r)$.
  This proves the rank equality claimed in the second item.
\end{proof}

\paragraph{The rank of $\Ra$ and $\La$ for $d=\lceil \ddfa/m\rceil$.}

Together with \cref{lemma:dimL}, the next lemma establishes that
the generic rank of $\Rasimp[\lceil \dd/m\rceil]$ and $\Lasimp[\lceil
  \dd/m\rceil]$ is~$\dd$.
Let $\mathcal{R}_{\ell,m}$ be the
$\dd \times \dd$ submatrix of $\mathcal{P}_{\ell,m,\lceil \dd/m\rceil}$ obtained
by extracting the
first~$\min(\ell_i,m)$ rows containing $\varlb_i$, for $i=1,\dots,r$
(see \cref{eq:nu:separable}), and the first~$\dd$ columns (note that
$\mathcal{P}_{\ell,m,\lceil \dd/m\rceil}$ has $m\lceil\dd/m\rceil \ge
\dd$ columns).

\begin{lemma}\label{lemma:Delta1} 
  For $\ell = (\ell_1,\dots,\ell_r)$, $n=\ell_1 + \cdots + \ell_r$ and
  $m$ in $\{1,\dots,n\}$, and for any pairwise distinct
  $\lambda_1,\dots,\lambda_r$ in $\bvecRing{r}$, the determinant
  $\polb_{\ell,m}(\varxi_1,\dots,\varxi_n,\lambda_1,\dots,\lambda_r)$
  of the $\dd \times \dd$ matrix $\mathcal{R}_{\ell,m}$ at
  $\lambda_1,\dots,\lambda_r$ is nonzero.
\end{lemma}
\begin{proof}
  We prove the nonvanishing property by exhibiting a vector
  $(\xi_1,\dots,\xi_n) \in \bvecRing{n}$ for which the evaluation
  $\polb_{\ell,m} (\xi_1,\dots,\xi_n,\lambda_1,\dots,\lambda_r)$ is
  not zero.  In what follows, for $i=1,\dots,r$, recall that we write
  $\sigma_i = \ell_1 + \cdots + \ell_{i-1}$, so that the rows involving
  $\varlb_i$ in $\mathcal{P}_{\ell,m,\lceil \dd/m\rceil}$ have indices
  $\sigma_i+1,\dots,\sigma_i+\ell_i$ (see \cref{eq:defSi}).

  Assume first that $m$ is invertible in~$\field$, and choose $\delta$
  in $\Kbar$ such that $\delta + \lambda_i \ne 0$ for
  $i=1,\dots,r$. Then, for all $i$, the
  polynomial~$x^m-(\delta+\lambda_i)$ is separable, since its
  discriminant is $m^m(\delta+\lambda_i)^{m-1}$, and we choose
  $\xi_{\sigma_i+1},\dots,\xi_{\sigma_i+\min(\ell_i,m)}$ to be pairwise
  distinct roots of this polynomial in $\Kbar$. If $m < \ell_i$,
  we further take $\xi_{\sigma_i+m+1},\dots,\xi_{\sigma_i+\ell_i}$ arbitrary
  in $\Kbar$ (note that $\polb_{\ell,m}$ does not depend on
  these quantities). Now, for any $\xi,\lambda$ such that $\xi^m =
  \delta+\lambda$, and for $j \ge 1$, we have $\lambda^j=\xi^{jm} +
  \sum_{k=1}^{j} \binom{j}{k} (-\delta)^k\xi^{(j-k)m}$. Up to
  invertible linear combinations of its columns, $\mathcal{R}_{\ell,m}
  (\xi_1,\dots,\xi_n,\lambda_1,\dots,\lambda_r)$ is thus the
  Vandermonde matrix at the roots
  $\xi_{\sigma_i+1},\dots,\xi_{\sigma_i+\min(\ell_i,m)}$, $i=1,\dots,r$.
  Since the $\lambda_i$'s are pairwise distinct, all these roots are
  pairwise distinct too, so the determinant
  $\polb_{\ell,m}(\xi_1,\dots,\xi_n,\lambda_1,\dots,\lambda_r)$ is nonzero.

  If $m$ is~0 in~$\field$, then for all $i$, $x^m+x-\lambda_i$ is 
  separable, since its discriminant is $(-1)^{m(m-1)/2} \neq 0$. Again,
  choosing distinct roots of these polynomials and performing linear
  combinations of the columns of~$\mathcal{R}_{\ell,m}$ leads to a
  nonzero Vandermonde determinant.
\end{proof}

{In several steps, we now study the rank of $\hkfa$ and show it is $\ddfa$ for $d$ large enough. 
Note that unlike in \cref{subsec:rankn}  where we were working with $\dd=n$,  additional ingredients are necessary 
in order to deduce this rank from those of  $\Ra$ and $\La$.}

\paragraph{If the rank of $\hkfa$ is $\ddfa$ for some \(d\ge0\), then it is $\ddfa$ for all $d\ge \lceil \ddfa/m\rceil$.}

\begin{lemma}\label{lemma:Delta1-nonvanish} 
  Let $\ell$, $\xi_1,\dots,\xi_n$, $\lambda_1,\dots,\lambda_r$, $a$,
  $f$ and $m$ be as in \cref{prop:separable}.  If $\hkfasimp[m,d_0]$
  has rank~$\dd$ for some~$d_0 \ge \lceil \dd/m\rceil$, and if
  $\polb_{\ell,m}(\xi_1,\dots,\xi_n,\lambda_1,\dots,\lambda_r)$ and
  $\polb_{\ell,m}(1/\xi_1,\dots,1/\xi_n,\lambda_1,\dots,\lambda_r)$
  from \cref{lemma:Delta1} are nonzero, then $\hkfasimp$ has rank
  $\dd$ for all $d\ge\lceil\dd/m\rceil$.
\end{lemma}

\begin{proof} 
  Since $\polb_{\ell,m}(\xi_1,\dots,\xi_n,\lambda_1,\dots,\lambda_r)$
  is nonzero, $\mathcal{P}_{\ell,m,\lceil
    \dd/m\rceil}(\xi_1,\dots,\xi_n,\lambda_1,\dots,\lambda_r)$ has
  rank at least $\dd$, and so does $\Rasimp[\lceil \dd/m\rceil]$
  (\cref{lemma:dimL}).

  As a result, for $d\ge\lceil\dd/m\rceil$, $\Rasimp[d]$ still has
  rank exactly $\dd$ (recall that this rank cannot exceed $\dd$, by
  \cref{eq:rankKmd}). Thus, for such $d$, there exists a nonsingular
  $P\in \matRing{(md)}{(md)}$ such that $\Rasimp[d]P=[\Rasimp[\lceil
      \dd/m\rceil]~0]$, where the zero matrix is $n\times (m(d-\lceil
  \dd/m\rceil))$.  In the same way, since
  $\polb_{\ell,m}(1/\xi_1,\dots,1/\xi_n,\lambda_1,\dots,\lambda_r)$ is
  nonzero, \cref{lemma:dimL} also implies that $\Lasimp[\lceil
    \dd/m\rceil]$ has rank $\dd$, therefore there exists a nonsingular
  $Q\in \matRing{(md)}{(md)}$ such that $Q\Lasimp[d]=
  \trsp{[\trsp{(\Lasimp[\lceil \dd/m\rceil])} ~0]}$. We obtain
  \[
    Q\hkfasimp P = Q \Lasimp[d]\Rasimp[d] P = \begin{pmatrix}
    \hkfasimp[m,\lceil \dd/m\rceil] & 0 \\ 0 & 0 \end{pmatrix}
    \in \matRing{(md)}{(md)},
  \]
  which shows that for $d\ge\lceil\dd/m\rceil$ we have
  $\rank{\hkfasimp}=\rank{\hkfasimp[m,\lceil \dd/m\rceil]}$.
\end{proof}

\paragraph{The rank of $\hkfa[m,r]$ is $\ddfa$ generically.}

To establish that the rank of~$\hkfasimp[m,r]$ is~$\dd$ for generic
choices of $\xi_1,\dots,\xi_n$, we introduce a decomposition into
vector spaces associated to the $\lambda_i$'s. We then study these
spaces separately; their dimensions are ${\min(\ell_i,m)}$,
respectively, leading as expected to a total dimension
$\sum_{i=1}^r{\min(\ell_i,m)}=\dd$.

This is achieved through a description of the images of the
mappings
$\rhoasimp[m,d]$ and $\lbasimp[m,d]$ in terms of polynomials.  Given positive
integers $\ell=(\ell_1,\dots,\ell_r)$ and $\xi_1,\dots,\xi_n$ in
$\bvecRing{n}$, define
\begin{equation}\label{eq:basis-imR}
  P_{i,j}=\sum_{k\in S_i}\xi_k^j\lag_k \in \Kbar[x],\qquad
  i=1,\dots,r,\quad j \ge 0,
\end{equation}
with the Lagrange polynomials $\lag_1,\dots,\lag_n$ and the sets
$S_1,\dots,S_r$ from~\cref{eq:defSi}.

\begin{lemma}\label{lemma:imRseparable}
  Let $\ell$, $\xi_1,\dots,\xi_n$, $\lambda_1,\dots,\lambda_r$, $a$,
  $f$ and $m$ be as in \cref{prop:separable} and let $d$ be a positive
  integer.  The image of~$\rhoasimp$ lies in the linear span of the~$\dd$
  linearly independent polynomials $P_{i,j}$ from \cref{eq:basis-imR},
  for $1 \le i \le r$ and $0 \le j < \min(\ell_i,m)$.
\end{lemma}
\begin{proof}
  Let $V(x,y)=\sum_{j=0}^{m-1}{c_j(y)x^j}$ belong to~$\Kbar[x,y]_{<
    (m,d)}$. Lagrange interpolation gives
  \[\rhoasimp(V)=V(x,a)\rem f=\sum_{k=1}^n{V(\xi_k,a(\xi_k))\lag_k}.\]
  Since $V(\xi_k,a(\xi_k))= \sum_{j=0}^{m-1}{c_j(a(\xi_k))\xi_k^j}$,
  we deduce
  \[\rhoasimp(V)=\sum_{i=1}^r{\sum_{j=0}^{m-1}{c_j(\lambda_i)P_{i,j}}}.\]
  For $i=1,\dots,r$, at most~$\ell_i$ of the polynomials $P_{i,j}$,
  $j=0,\dots,m-1$, can be linearly independent, since they are all
  linear combinations of $\ell_i$ linearly independent~$\lag_k$. On the
  other hand, the polynomials~$P_{i,j}$ for $j=0,\dots,\ell_i-1$ are
  linearly independent, due to the linear independence of the
  polynomials~$\lag_k$, and the invertibility of the Vandermonde matrix
  $[\xi_k^j]_{0 \le j < \ell_i} \in
  \bmatRing{\ell_i}{\ell_i}$. This proves that the image
  of~$\rhoasimp$ is included in the span of the polynomials $P_{i,j}$, for
  $i=1,\dots,r$ and $j=0,\dots,\min(\ell_i,m)-1$, as claimed. 
\end{proof}
This polynomial-based interpretation then allows us to use the following
decomposition.
\begin{lemma}\label{lemma:splitimL}
  Let $\ell$, $\xi_1,\dots,\xi_n$, $\lambda_1,\dots,\lambda_r$, $a$,
  $f$ and $m$ be as in \cref{prop:separable}. The rank of
  $\hkfasimp[m,r]$ is the sum of the dimensions of the vector spaces
\begin{equation}\label{eq:defVi}
  V_i=\Span\!\left([P_{i,j}]_0^{m-1},j=0,\dots,\min
  (\ell_i,m)-1\right)
  \end{equation}
  with the polynomials $P_{i,j}$ from \cref{eq:basis-imR}
  for $i=1,\dots,r$.
\end{lemma}
\begin{proof}
  We first claim that for $d=r$, $\Rasimp[r]$ has rank $\dd$, or
  equivalently (\cref{lemma:dimL}) that $\mathcal{P}_{\ell,m,r}$ has
  rank $\dd$ at $(\xi_1,\dots,\xi_n,\lambda_1,\dots,\lambda_r)$.
  Indeed,  we can extract from
  $\mathcal{P}_{\ell,m,r}$ a $\dd\times \dd$ submatrix by keeping the
  first~$\min(\ell_i,m)$ rows indexed by $\varlb_i$, for
  $i=1,\dots,r$, and the columns containing the monomials
  $\varlb^{i-1},\dots,\varlb^{i-1}\xi^{\min(\ell_i,m)-1}$, for
  $i=1,\dots,r$. {The columns of this matrix contain 
  the evaluations
  of the polynomials~$x^iy^j$ for $i=0,\dots,r-1$ and
  $j=0,\dots,\min
  (\ell_{i+1},m)-1$ at the points $(\lambda_i,\xi_{ij})$ for
  $i=1,\dots,r$
  and $\xi_{ij}=\xi_{\ell_1+\dots+\ell_{i-1}+j}$ for
  $j=1,\dots,\min(\ell_i,m)$. Now, up to linear combinations of its
  columns, the determinant of this matrix is the same as that of the
  matrix whose columns evaluate the polynomials 
  \[Q_{ij}(x,y):=\prod_{h=1}^{i-1}{(x-\lambda_h})\prod_{k=1}^{j-1}
  (y-\xi_{ik}),\qquad i=1,\dots,r,\quad j=1,\dots,\min
  (\ell_i,m). \]
  The latter matrix is triangular; its diagonal elements are $Q_{ij}
  (\lambda_i,\xi_{ij})\neq0$, showing that the matrix is nonsingular
  and therefore that $\Rasimp[r]$ has rank at least $\dd$. (More
  general determinant factorizations of this kind are considered by 
  \cite{BuckColeyRobbins1992}, \cite[\S2]{GascaMartinez1987}.)}
  Using \cref{eq:rankKmd} we deduce that $\Rasimp[r]$ has rank exactly $\dd$ as announced, and from \cref{lemma:imRseparable}, we
  know that the image of  
  $\rhoasimp[m,r]$ is the span of the polynomials $P_{i,j}$ defined in
  that lemma.

  It follows that the rank of $\hkfasimp[m,r]$ is the dimension of the
  span of the image $\lbasimp[m,r](P_{i,j})$. For $i =1,\dots,r$, and
  $j=0,\dots,\min(\ell_i,m)-1$,
  \begin{align*}
    \lbasimp[m,r](P_{i,j})&=\left([P_{i,j}\rem f]_0^{m-1},[a P_{i,j}\rem f]_0^{m-1},\dots,[a^{r-1}P_{i,j}\rem f]_0^{m-1}\right), \\
    \intertext{and since the Lagrange polynomials are eigenvectors of
      multiplication by~$a$, we get}
    \lbasimp[m,r](P_{i,j})    &=\left([P_{i,j}]_0^{m-1},[\lambda_i P_{i,j}]_0^{m-1},\dots,[\lambda_i^{r-1}P_{i,j}]_0^{m-1}\right) \\
    &=\left([P_{i,j}]_0^{m-1},\lambda_i [P_{i,j}]_0^{m-1},\dots,\lambda_i^{r-1}[P_{i,j}]_0^{m-1}\right).
  \end{align*}

{Let $v_{i,j}=\left(0,\ldots, 0, [P_{i,j}]_0^{m-1}, 0,\ldots , 0\right)$ be $[P_{i,j}]_0^{m-1}$ times the $i$-th canonical vector in $\Kbar[x]_{<m}^r$, seen as row vector.  The span of the $\lbasimp[m,r](P_{i,j})$'s multiplied on the right by the inverse of the Vandermonde matrix
  associated to the $\lambda_i$'s is the span of the 
$v_{i,j}$'s . By grouping the $v_{i,j}$'s for each $i$,} this yields a block-diagonal matrix with
  blocks that span the spaces~$V_i$ of the lemma. The result on the rank
  follows.
\end{proof}
The dimensions of the vector spaces from \cref{eq:defVi} can now be
analyzed separately.
\begin{lemma}\label{lemma:leftpol}
  Fix positive integers $\ell = (\ell_1,\dots,\ell_r)$ such that
  $\ell_1 + \cdots + \ell_r = n$, and $m$ in $\{1,\dots,n\}$.  There
  exists a nonzero polynomial
  $\polc_{\ell,m}\in\ZZ[\varxi_1,\dots,\varxi_n]$ of degree at
  most $(n-1)(n-\dd)$ such that if pairwise distinct nonzero
  $\xi_1,\dots,\xi_n$ do not form a zero of $\polc_{\ell,m}$, then
  $V_i$ from \cref{eq:defVi} has dimension $\min(\ell_i,m)$ for all
  $i$.
\end{lemma}
\begin{proof} 
  Take $i$ in $\{1,\dots,r\}$, consider the set of indices $S_i=
  \{\ell_1 + \cdots + \ell_{i-1} + 1,\dots,\ell_1 + \cdots +
  \ell_{i}\}$ from~\cref{eq:defSi}.  Let then $A_i=\prod_{k\in
    S_i}(x-\xi_k)$, $B_i=\prod_{k \notin S_i}(x-\xi_k)$ and $C_i=1/
  B_i\bmod A_i$. Note that $A_i$ and $B_i$ have respective degrees
  $\ell_i$ and $n-\ell_i$, and that $C_i$ is well defined, since $B_i$
  and $A_i$ have no common root.

  For $k$ in $S_i$, by construction, $B_i$ divides the Lagrange
  polynomial $\lag_k$, with a quotient of degree
  $n-1-\deg(B_i)=\ell_i-1$. In view of \cref{eq:basis-imR}, $B_i$
  divides $P_{i,j}=\sum_{k \in S_i} \xi_k^j \lag_k$, for all $j\ge 0$,
  and the quotient has degree less than $\ell_i$. We now prove that it
  is actually equal to $x^jC_i\rem A_i$. Since $f = A_i B_i$, for $k$
  in $S_i$, the Lagrange polynomial~$\lag_k=f/(f'(\xi_k)(x-\xi_k))$
  satisfies
  \[\frac{\lag_k}{ B_i}
  =\frac{1}{ f'(\xi_k)}\frac{f}{B_i(x-\xi_k)}
  =\frac{1}{ f'(\xi_k)}\frac{A_i}{x-\xi_k}
  =\frac{C_i(\xi_k)}{A_i'(\xi_k)}\frac{A_i}
  {x-\xi_k}.
  \]
  In particular, for $j \ge 0$, $\xi_k^j \lag_k / B_i$ takes the value
  $\xi_k^j C_i(\xi_k)$ at $\xi_k$, and $0$ at all other roots of
  $A_i$.  Taking the sum over all $k$ in $S_i$ then proves our claim
  that $P_{i,j}/B_i = x^j C_i \rem A_i$. Since $\xi_1,\dots,\xi_n$
  are nonzero, $B_i(0)$ as well is nonzero, so $B_i$ is invertible
  as a power series and $[P_
    {i,j}/B_i]_0^{m-1}=[[P_{i,j}]_0^{m-1}/B_i]_0^{m-1}$. Thus the
  truncated polynomials $[P_{i,j}]_0^{m-1}$, for $0 \le j <
  \min(\ell_i,m)$, are linearly independent if and only if the
  truncated polynomials $[x^jC_i\rem A_i]_0^{m-1}$ are.

  When $\ell_i\le m$, the polynomials $x^jC_i\rem A_i$ have degree
  less than~$m$ and their linear independence follows from that of the
  polynomials $x^j$, $j=0,\dots,\ell_i-1$, since $C_i$ is invertible
  modulo $A_i$. Thus in this case, we always have $\dim(V_i) = \ell_i
  = \min(\ell_i,m)$.

  When $\ell_i>m$, we are going to prove that the polynomials {$[x^jC_i\rem
  A_i]_0^{m-1}$}, $j=0,\dots,m-1$, are linearly independent for a generic choice
  of $\xi_1,\dots,\xi_n$. To achieve this, define the matrix $M_{C_i}$
  whose entry $(j,\ell)$ is the coefficient of $x^{\ell-1}$ in
  $x^{j-1} C_i\rem A_i$ for $j=1,\dots,\ell_i$ and
  $\ell=1,\dots,\ell_i$; this is the multiplication matrix by $C_i$
  modulo $A_i$.  We  also consider its inverse, the multiplication
  matrix $M_{B_i}$ by $B_i$ modulo $A_i$.

  For our claim to hold, it is enough to guarantee that the $m\times
  m$ leading principal minor $K_i$ of $M_{C_i}$ be nonzero. We 
  view this minor as a rational function in $\varxi_1,\dots,\varxi_n$:
  this is done by introducing the polynomials $\bar A_i=\prod_{k\in
    S_i}(x-\varxi_k)$, $\bar B_i=\prod_{k \notin S_i}(x-\varxi_k)$ and
  $\bar C_i=1/\bar B_i\bmod \bar A_i$, all of which are in
  $\mathbb{Q}(\varxi_1,\dots,\varxi_n)[x]$. We can then define the
  matrices $M_{\bar C_i}$ and $M_{\bar B_i}$ of multiplication by
  respectively $\bar C_i$ and $\bar B_i$ modulo $\bar A_i$, and the 
  $m\times m$ leading principal minor $\bar K_i$ of $M_{\bar
    C_i}$. This is a rational function of $\varxi_1,\dots,\varxi_n$,
  whose evaluation at {$\xi_1,\dots,\xi_n$} gives the scalar $K_i \in
  \field$.

  Note first that $\bar K_i$ is not identically zero: if we evaluate
  all $\varxi_g$ at $0$, for $g$ in $S_i$, $\bar A_i$ becomes
  $x^{\ell_i}$, and the matrix $M_{\bar C_i}$ becomes {lower 
  triangular}, with $1/B_i(0)\neq0$. It then remains to estimate the
  degree of a numerator of $\bar K_i$. The Schur complement formula
  gives $\bar K_i = \det(M_{\bar C_i}) \bar L_i$, where $\bar L_i$ is
  the $(\ell_i-m)\times(\ell_i-m)$ lower right minor of the inverse
  $M_{\bar B_i}$ of $M_{\bar C_i}$. The determinant of $M_{\bar C_i}$
  is the resultant of $\bar C_i$ and $\bar A_i$, that is, $1/\prod_{g
    \in S_i,h \notin S_i} (\varxi_g-\varxi_h)$. On the other hand,
  $\bar L_i$ is a polynomial in $\ZZ[\varxi_1,\dots,\varxi_n]$ (since
  $\bar B_i$ and $\bar A_i$ have coefficients in
  $\ZZ[\varxi_1,\dots,\varxi_n]$, and $\bar A_i$ is monic in $x$).

  For $s \ge 0$, write $x^s \rem \bar A_i= c_{s,0} + \cdots +
  c_{s,\ell_i-1}x^{\ell_i-1}$, for $c_{s,t} \in
  \ZZ[\varxi_1,\dots,\varxi_n]$. By induction on~$s$, we obtain the
  bound $\deg(c_{s,t}) \le s-t$. From this, it follows that all entries
  of $M_{\bar B_i}$ have degree at most $n-1 $, and that $\bar L_i$
  has degree at most $(n-1)(\ell_i-m) \le n \ell_i$.  To conclude the
  proof, we let $\polc_{\ell,m}$ be the product of the polynomials
  $\bar L_i$, for $i$ such that $\ell_i > m$. The degree bound follows
  from remarking that $\sum_{\ell_i > m} (\ell_i-m) = n-\dd$.
\end{proof}

\paragraph{Genericity polynomials and degree bounds.}

Until here, the conditions we have seen are the nonvanishing of
$\polb_{\ell,m}(\xi_1,\dots,\xi_n,\lambda_1,\dots,\lambda_r)$,
$\polb_{\ell,m}(1/\xi_1,\dots,1/\xi_n,\lambda_1,\dots,\lambda_r)$,
and $\polc_{\ell,m}(\xi_1,\dots,\xi_n)$.  When nonzero, the first
two quantities allow us to apply \cref{lemma:Delta1-nonvanish} and
obtain the rank of $\hkfasimp[m,\lceil n/m\rceil]$ from any
$\hkfasimp[m,d_0]$ with $d_0 \ge \lceil \dd/m\rceil$;
the third condition
$\polc_{\ell,m}(\xi_1,\dots,\xi_n)\neq 0$  allows us to take
$d_0=r$ thanks to \cref{lemma:splitimL,lemma:leftpol}

The bound on the degree of $\polb_{\ell,m}$ in
$\varxi_1,\dots,\varxi_n$ follows from summing the degrees of the
columns in $\mathcal{P}_{\ell,m,\lceil \dd/m\rceil}$.  Each block of
$m$ columns involves degrees $1+ \cdots + (m-1)=m(m-1)/2$, and we
consider~$\lceil \dd/m\rceil $ such blocks (the last one may not be
complete), for a total of at most $(\dd+m)(m-1)/2$.
Next, consider the term
$\polb_{\ell,m}(1/\varxi_1,\dots,1/\varxi_n,\varlb_1,\dots,\varlb_r)$, which is
not a polynomial in the $\varxi_i$'s. To estimate the degree of its
numerator, observe that it is a $\dd \times \dd$-minor of the matrix
\begin{equation*}\label{eq:bigmat-inv}
  \begin{pmatrix}
    1&\frac 1{\varxi_1}&\dots&\frac 1{\varxi_1^{m-1}}&\varlb_1&\frac 1{\varxi_1}\varlb_1&\dots&\frac 1{\varxi_1^{m-1}}\varlb_1^{\lceil \dd/m\rceil-1}\\
    &&\vdots&&&&\vdots&\\ 
    1&\frac 1{\varxi_n}&\dots&\frac 1{\varxi_n^{m-1}}&\varlb_r&\frac 1{\varxi_n}\varlb_r&\dots&\frac 1{\varxi_n^{m-1}}\varlb_r^{\lceil \dd/m\rceil-1}
  \end{pmatrix}
\end{equation*}
Factoring out (on the right) the diagonal matrix with diagonal
$(1/\varxi_i^{m-1})_{1 \le i \le n}$, we see that the nonvanishing of
$\polb_{\ell,m}(1/\xi_1,\dots,1/\xi_n,\lambda_1,\dots,\lambda_r)$ is
equivalent to the nonvanishing of the corresponding $\dd \times
\dd$-minor $\tilde{\polb}_{\ell,m}$ in
\begin{equation*}
  \begin{pmatrix}
    \varxi_1^{m-1}&\varxi_1^{m-2}&\dots&1& \varxi_1^{m-1}\varlb_1& \varxi_1^{m-2}\varlb_1&\dots&\varlb_1^{\lceil \dd/m\rceil-1}\\
    &&\vdots&&&&\vdots&\\ 
    \varxi_n^{m-1}&\varxi_n^{m-2}&\dots&1& \varxi_n^{m-1}\varlb_r& \varxi_n^{m-2}\varlb_r&\dots&\varlb_r^{\lceil \dd/m\rceil-1}
  \end{pmatrix}.
\end{equation*}
The degree upper bound for $\tilde{\polb}_{\ell,m}$ is
$(\dd+m)(m-1)/2$, as for $\polb_{\ell,m}$.

We then take $\polbU_{\ell,m}=\polb_{\ell,m}\tilde{\polb}_{\ell,m}
\polc_{\ell,m}$ to prove \cref{prop:separable}. For the degree
estimate, note that $(\dd+m)(m-1) + (n-1)(n-\dd) \le 2n^2$.  For
correctness, take pairwise distinct nonzero $\xi_1,\dots,\xi_n$ in
$\Kbar$ and let $a\in\xRing$ take distinct values
$\lambda_1,\dots,\lambda_r$ at $\xi_1,\dots,\xi_n$, with
multiplicities $\ell_1,\dots,\ell_r$. As before, we write
$f=(x-\xi_1)\cdots(x-\xi_n)$, and we assume that $f$ is in
$\xRing$. Finally, we suppose that
$\polbU_{\ell,m}(\xi_1,\dots,\xi_n,\lambda_1,\dots,\lambda_r)$ is
nonzero.  \cref{lemma:splitimL,lemma:leftpol} show that for $d_0=r$,
we have $\rank{\hkfa[m,r]}=\dd$. Since $r\ge\lceil\dd/m\rceil$, by
\cref{lemma:Delta1-nonvanish}, it is then also the case for $\hkfa$
for all $d\ge\lceil\dd/m\rceil$, as claimed.

The only remaining claim is that for any pairwise distinct
$\lambda_1,\dots,\lambda_r$, $\polbU_{\ell,m}(\bar \xi_1,\dots,\bar
\xi_n,\lambda_1,\dots,\lambda_r)$ is a nonzero polynomial in $\bar
\xi_1,\dots,\bar \xi_n$. That $ \polc_{\ell,m}$ is nonzero is in
\cref{lemma:leftpol} (this polynomial does not depend on
$\lambda_1,\dots,\lambda_r$); \cref{lemma:Delta1} proves that
$\polb_{\ell,m}(\bar \xi_1,\dots,\bar
\xi_n,\lambda_1,\dots,\lambda_r)$ is nonzero. That lemma also implies that
$\polb_{\ell,m}(1/\bar \xi_1,\dots,1/\bar
\xi_n,\lambda_1,\dots,\lambda_r)$ is nonzero (as a rational
function), and as a consequence, this is also the case for $\tilde
\polb_{\ell,m}(\bar \xi_1,\dots,\bar
\xi_n,\lambda_1,\dots,\lambda_r)$. The claim for $\polbU_{\ell,m}$
is thus proved.

\section{A randomized composition algorithm through change of basis}
\label{sec:composition_randomized}

In this section we give the base case of our modular composition
algorithm that is used when~$f$ is either separable or purely
inseparable (which includes the case of power series). The core
\algoName{algo:ModularCompositionBaseCase} is studied in
\cref{sec:composition_randomized:algo}, and a variation for computing
annihilating polynomials is given in \cref{sec:minpoly_randomized}.

The algorithm of \cref{sec:bivcomposition} performs bivariate modular
composition within our target complexity bound, assuming the knowledge
of a matrix of relations with appropriate dimension and degree.  Since
such a matrix of relations of $\rmodfa$ may not exist for general $a$
and $f$, \algoName{algo:ModularCompositionBaseCase} transports the
computation of $\polp(a)$ in $\quotient=\xRing/\genBy{f}$ to an
isomorphic algebra that is expected to be more favorable to the
computation.

More precisely, we pick a random $\gamma \in \xRing _{<n}$;
generically, its minimal polynomial $\mu_\gamma \in \yRing$ has degree
\(n\) and is also its characteristic polynomial $\chi_{\gamma}$, so
that the powers of $\gamma$ generate $\quotient$. This induces the
$\field$-algebra isomorphism $\phi_\gamma$ of
\cref{def:changeofbasis}; \cref{step:changebasis-mainalgo} of
\algoName{algo:ModularCompositionBaseCase} then computes a polynomial
representative $\alpha$ of $\phi_\gamma (a \bmod f)$ using the change
of basis algorithm of \cref{sec:changeofbasis}.  Note that a matrix of
relations $\rmatfg$ is also obtained at
\cref{step:changebasis-mainalgo} in preparation for the final stage.
Then, with good probability, the conditions for the efficient
computation of a certified matrix of relations $\rmatma$ of $\rmodma$
via the approach of \cref{sec:relmat:certify} are fulfilled.
\cref{basecase:Ralpha} of \algoName{algo:ModularCompositionBaseCase}
 computes this matrix of relations, which then
allows us to obtain the polynomial \(\beta = \polp(\alpha) \rem
\mu_{\gamma}\) at \cref{step:compnewbasis-mainalgo} as seen in \cref{sec:bivcomposition}.
The solution $b= \polp(a) \rem f$ to the initial problem is finally
recovered by applying $\phi ^{-1}_\gamma$ to $\beta \bmod \mu
_{\gamma}$, which amounts to computing $b=\beta (\gamma) \rem f$.
Since we already have $\rmatfg$ at our disposal, $b$ is obtained
with the algorithm of \cref{sec:bivcomposition} as well.

\cref{prop:algo-mod_comp} in \cref{sec:composition_randomized:algo}
shows the correctness of this strategy and bounds its complexity.  We
then study the probability of success for $f$ separable and $f$ purely
inseparable.  The main point is to ensure that appropriate matrices of
relations $\rmatfg$ and $\rmatma$ are actually available.  For
\cref{step:changebasis-mainalgo,step:inversechange-mainalgo} where a
random $\gamma$ is involved, we directly rely on the generic
properties of the associated block Hankel matrix $\hankel{m,\lceil
  n/m\rceil}{\gamma}{f}$ (\cref{prop:generic_a}).  For the
computation of $\rmatma$ we use the fact that $\alpha$ and $\mu
_{\gamma}$ are sufficiently generic, hence also give access to good
properties for the associated block Hankel matrix after the change of
basis.

The probability of failure for a general separable~$f$ is bounded in
\cref{sec:composition_randomized:proof_separable}.  The power series
case and, more generally, the case of purely inseparable $f$ are
treated in
\cref{sec:composition_randomized:proof_inseparable,sec:composition_randomized:proof_inseparable_small}.
For such $f$, the success of
\algoName{algo:ModularCompositionBaseCase} is proven in
\cref{sec:composition_randomized:proof_inseparable} under some
assumptions on the valuation of the input polynomial~$a$ and the
characteristic of $\field$. Still in the case of $f$ purely
inseparable, a complete algorithm is then given in
\cref{sec:composition_randomized:proof_inseparable_small}: when the
valuation is large (with respect to the target value $m \sim n^{\eta}$ with
$\eta$ from \cref{eq:def-beta}), then the minimal polynomial of $a$
modulo $f$ has small degree and we use the extension of Shoup's
algorithm seen in \cref{subsubsec:minpoly}.  For fields $\field$ of
small characteristic, we adapt Bernstein's composition algorithm for
power series~\cite{Ber98} to our general context.

\subsection{Randomized composition}
\label{sec:composition_randomized:algo} 

The procedure is detailed in
\algoName{algo:ModularCompositionBaseCase}.  It uses $n+m$ parameters
from $\field$ that are available as a sequence $r$ of length $n+m$.
The coefficients of the random polynomial~$\gamma$ are given as part
of the input as $r_i$, for $3\leq i \leq n+2$; we require further
parameters in order to reduce to the case where $f(0) \ne 0$ and
$\gcd(a,f)=1$ (\cref{rmk:fat0,rmk:shift}), and for the random column
combination performed by \algoName{algo:MatrixOfRelations}{}.

The parameter $m$ could be taken arbitrarily in $\{1,\ldots , n\}$,
but we choose the specific value $m=\lceil n^{\eta} \rceil$, with
$\eta$ from \cref{eq:def-beta}, as this choice minimizes the overall
cost. The following proposition describes the output of the procedure;
the probability of failure is bounded in
\cref{sec:composition_randomized:proof_inseparable,sec:composition_randomized:proof_separable}.

\begin{algorithm}
  \algoCaptionLabel{ModularCompositionBaseCase}{f,a,\polp,r}  
  \begin{algorithmic}[1]
    \Require $f$ of degree $n$ in $\xRing$, $a\in \xRing_{<n}$, $\polp\in \yRing$, $r \in \vecRing{n+\lceil n^{\eta} \rceil}$ 
    \Ensure $b=\polp(a)\rem f$ or \Fail
    
    \State\InlineIf{$n=1$}{\Return$\polp(a)$} \Comment{$a\in \field$}

    \State $\polp \gets \polp(y-r_1)$, $a \gets a(x) + r_1$; \InlineIf{$\gcd(a,f)\ne 1$} \Return \Fail \label{step:shiftp-mainalgo}

    \State $f \gets f(x + r_2)$; $a \gets a(x+r_2)$; \InlineIf{$f(0)=0$} \Return \Fail \label{step:shifta-mainalgo}
    
    \State $m \gets \lceil n^{\eta} \rceil$  \Comment{With $\eta$ from \cref{eq:def-beta}}
    \label{basecase:m}
    \State\CommentLine{Change of basis: compute a polynomial $\alpha$
    such that $\alpha \equiv \phi_\gamma (a \bmod f)\bmod \mu _{\gamma}$} 
    \label{step:changebasis-mainalgo}
    \Statex \CommentLine{Getting a basis of relations $\rmatfg[{}]$ and the minimal polynomial $\mu_{\gamma}$ of \(\gamma\bmod f\)}
    \Statex $\gamma \gets r_3 + r_4 x + \cdots + r_{n+2}x^{n-1}$
    \Statex $(\rmatfg[{}],\mu_{\gamma},\alpha) \gets \Call{algo:ChangeOfBasis}{f,\gamma,a,m,\lceil n/m\rceil}$
            \Comment{\cref{algo:ChangeOfBasis}}
    \Statex\InlineIf{this call returned \Fail}{\Return\Fail}
    \State \InlineIf{$\mu_{\gamma}(0)=0$} \Return \Fail \label{step:mugammazero}
    \State substitute ``$y$'' by ``$x$'' in $\mu_{\gamma}$ and $\alpha$, which are then in $\xRing$
        \label{basecase:subs1}

    \State \CommentLine{Compute a matrix of relations for $(\alpha, \mu_{\gamma})$} 
           \label{basecase:Ralpha}
    \Statex $\rmatma[{}]\gets\Call{algo:MatrixOfRelations}{\mu_{\gamma},\alpha,m,\lceil n/m\rceil,(r_{n+i})_{3 \le i \le m}}$
    \Comment{\cref{algo:MatrixOfRelations}}
    \Statex\InlineIf{this call returned \Fail}{\Return \Fail}
    \State \CommentLine{Bivariate modular composition in the new basis: \(\beta
    \equiv \polp(\alpha) \bmod \mu_{\gamma}\)}\Comment{\cref{algo:BivariateModularCompositionWithRelationMatrix}}
      \label{step:compnewbasis-mainalgo}
     \Statex $\beta\!\gets\Call{algo:BivariateModularCompositionWithRelationMatrix}{\mu_{\gamma},\alpha,\polp,\rmatma[{}]}$
    \State substitute ``$x$'' by ``$y$'' in $\beta$, which  is then in $\yRing$
      \label{basecase:subs2}
    \State\CommentLine{Inverse change of basis: $b \equiv \phi^{-1}(\beta \bmod \mu_{\gamma}) \bmod f$} 
    \label{step:inversechange-mainalgo}
    \Statex $b\gets\Call{algo:BivariateModularCompositionWithRelationMatrix}{f,\gamma,\beta,\rmatfg[{}]}$
    \Comment{\cref{algo:BivariateModularCompositionWithRelationMatrix}}
    \State \Return $b(x-r_2)$ \label{step:return-mainalgo}
  \end{algorithmic}
\end{algorithm}

\begin{proposition}
  \label{prop:algo-mod_comp} 
  Given $f\in\field[x]$ of degree~$n$, $a\in\field[x]_{<n}$,
  $g\in\field[y]$ with $\deg(g)=O(n)$ and $r\in\field^{n+m}$ with $m=\lceil n^{\eta}\rceil$ and
  $\eta$ from \cref{eq:def-beta},
  \algoName{algo:ModularCompositionBaseCase} returns either $\polp(a)
  \rem
  f$ or \Fail{}; it uses $\softO{n^{\kappa}}$ operations in $\field$,
  with $\kappa < 1.43$ as in \cref{eq:def-gamma}.
\end{proposition}

\begin{proof}
  If $n=1$ then as $a$ has degree~0, the result is~$\polp(a) \in \field$ and the
  algorithm is correct. The rest of the proof assumes~$n>1$.

{\Cref{step:shiftp-mainalgo,step:shifta-mainalgo} ensure} that $\gcd (a,f)=1$ and $f(0)\neq 0$.
   This does not impact the complexity, as shifting a polynomial of
   degree~$O(n)$ can be achieved in~$\softO{n}$ arithmetic operations
   \cite[Chap.\,1, Pb.\,3.5]{BiPa94}. The same observation applies to
   the last step.
   
   At \cref{step:changebasis-mainalgo}, if
   \algoName{algo:ChangeOfBasis}{} does not return $\Fail$ then by
   \cref{prop:algo:ChangeOfBasis} the matrix $\rmatfg[{}]$ is a basis
   of relations of \(\rmodfg\), $\mu_{\gamma}=\chi_{\gamma}$, and
   $\alpha(\gamma)\equiv a \bmod f$. It follows that $\mu
   _{\alpha}=\mu _a$ since the quotient algebras are isomorphic, and
   $\mu_{\alpha}(0)=\mu_a(0)$ implies $\gcd (\alpha,\mu_{\gamma})=
   \gcd (a,f)=1$.  If the test at \cref{step:mugammazero} does not
   fail then the specifications for the call to
   \algoName{algo:MatrixOfRelations}{} are met; from
   \cref{prop:certificate}, if \cref{basecase:Ralpha} does not return
   $\Fail$ then the matrix~$\rmatma[{}]$ is a matrix of relations
   in~$\langle\mu_\gamma,y-\alpha\rangle$. Both these matrices of
   relations have dimension at most~$2(m-1)$, and degree at
   most~$2\lceil n/m\rceil$; they are obtained in $\softO{m^{\omega} d
     + \bicost{d}} = \softO{m^{\omega}d+md^{\omega_2/2}}$ operations,
   with $d=\lceil n/m\rceil$. This is $\softO{n^{\kappa}}$ arithmetic
   operations, according to \cref{eq:def-beta,eq:def-gamma} and the
   choice of~$m$ at \cref{basecase:m}.

  The variable substitutions at \cref{basecase:subs1,basecase:subs2} are
  harmless; they make notation match with that in Algorithms
  \nameref{algo:MatrixOfRelations} and
  \nameref{algo:BivariateModularCompositionWithRelationMatrix}.

  At \cref{step:compnewbasis-mainalgo}, within the same complexity
  bound as above by \cref{prop:comp-from-matrix}, $\beta$ is
  computed such that $\beta \equiv \polp(\alpha) \bmod \mu_{\gamma}$
  (these polynomials are temporarily in $x$).  After the
  substitution of \cref{basecase:subs2} the latter relation implies
  the existence of a polynomial $h\in\yRing$ such that
  \[\beta(y)=\polp(\alpha(y))+h(y)\mu_\gamma(y).\]
  Since~$\mu_\gamma(\gamma)\equiv 0 \bmod f$, evaluating this
  identity at~$y=\gamma$ results in $b=\beta(\gamma) =  \polp(a) \rem
  f$ at \cref{step:inversechange-mainalgo}.
\end{proof}

\subsection{Randomized annihilating polynomial}
\label{sec:minpoly_randomized} 

If the choice of $\gamma$ ensures that the isomorphism $\phi_{\gamma}$
is well defined (the powers of $\gamma$ generate $\quotient$), then a
univariate polynomial~$\mu$ over $\field$ is such that $\mu(a)\equiv 0
\bmod f$ if and only $\mu(\alpha)\equiv 0 \bmod \mu_{\gamma}$.  Since
\algoName{algo:ModularCompositionBaseCase} computes a matrix of
relations
in~$\langle\mu_\gamma,y-\alpha\rangle$ at \cref{basecase:Ralpha}, an
algorithm for computing such a $\mu$ follows from the results of
\cref{sec:relmat:polynomials}.

\begin{algorithm}
  \algoCaptionLabel{AnnihilatingPolynomial}{f,a,r} 
  \begin{algorithmic}[1]
   \Require $f$ of degree $n$ in $\xRing$, $a\in \xRing_{<n}$,
   $r \in \vecRing{n+\lceil n^{\eta} \rceil}$ 
    \Ensure $\mu$ nonzero in $\yRing_{\leq 4n}$ such that $\mu(a) \equiv 0 \bmod f$ or \Fail
    \State\InlineIf{$n=1$}{\Return $y-a$} \Comment{$a\in \field$}
    \State \CommentLine{Compute a matrix of relations $\rmatma[{}]$ for ($\alpha,\mu_{\gamma}$)}
    with $\alpha=\phi_\gamma(a)$
    \Statex execute \crefrange{step:shifta-mainalgo}{basecase:Ralpha}  of \cref{algo:ModularCompositionBaseCase} 
    \Statex\InlineIf{\Fail{} has been returned by one of these steps}{\Return \Fail}
    \State $\mu \gets \det(\rmatma[{}])$ \Comment{~\cite[Algo.\,2]{LVZ17}} 
  
    \State \Return $\mu$
  \end{algorithmic}
\end{algorithm}

\begin{corollary}
  \label{cor:annihilatingbasecase}
  Given $f\in\field[x]$ of degree~$n$, $a\in\field[x]_{<n}$ and
  $r\in\field^{n+m}$ with $m=\lceil n^{\eta} \rceil$ and~$\eta$ from
  \cref{eq:def-beta}, \algoName{algo:AnnihilatingPolynomial} returns either
  \Fail{} or a nonzero $\mu \in \yRing_{\leq 4n}$ such that $\mu(a)\equiv 0 \bmod
  f$; it uses~$\softO{n^{\kappa}}$ operations in $\field$, with $\kappa < 1.43$
  as in \cref{eq:def-gamma}.
\end{corollary}
\begin{proof}If $n=1$ then as $a \in \field$, $\mu =y-a$ is such that $\mu(a)=0$ and the
  algorithm is correct. Now assume that $n>1$. 
  The annihilating polynomials are left unchanged by the substitution
  $x\gets x+r_2$.  As in the proof of \cref{prop:algo-mod_comp}, if
  failure does not occur then \cref{basecase:Ralpha} computes a matrix
  of relations of~$\rmodma[m']$, for some $m' \le 2(m-1)$, within the
  claimed complexity bound; this matrix has degree at most $2\lceil
  n/m\rceil$. Then \cref{cor:annihilating} shows that $\mu$
  annihilates $\alpha \bmod \minpoly[\gamma]$ and thus $a \bmod f$, and that it
  has degree $\deg(\mu) \leq 4(m-1)\lceil n/m\rceil$. This is at most $4n$ when
  $m\le\sqrt{n}$, which is the case when $m=\lceil n^\eta\rceil$ with $\eta$ as in
  \cref{eq:def-beta}. The complexity then follows from the proof of \cref{prop:algo-mod_comp} and 
  \cref{cor:annihilating} again.
\end{proof}

\subsection{Success of randomization for separable \texorpdfstring{$f$}{f}} 
\label{sec:composition_randomized:proof_separable}

The probabilistic properties of the previous algorithms in the
separable case are summarized in the following.
\begin{proposition} \label{prop:composition-separable} 
  Let $a,f$ be polynomials in $\xRing$ and $\polp$ be in $\yRing$, with $f$
  separable of degree $n$ and $\deg(a) < n$. If $r_1, \ldots , r_ {n+\lceil n^{\eta} \rceil} \in \field$
  are chosen uniformly and independently from a finite subset $S$ of
  $\field$, then Algorithms \nameref{algo:ModularCompositionBaseCase} and
  \nameref{algo:AnnihilatingPolynomial}  return \Fail{} with probability at
  most $6n^2/\card{S}$.
\end{proposition} 
\begin{proof}
The success of modular composition in
\algoName{algo:ModularCompositionBaseCase} and of the computation of
an
annihilating polynomial in \algoName{algo:AnnihilatingPolynomial}
relies
on: finding good shifts~$r_1$ and~$r_2$ in the first two steps; a
choice of $\gamma$ such that $\mu_\gamma (0)\neq0$ and $\mu_\gamma$
has degree $n$; the availability of matrices of relations
$\rmatfg[{}]$ and~$\rmatma[{}]$. The probability estimate is obtained
by showing the existence of polynomials whose zero sets contain the
values of the parameters~$r_i$ where these properties do not hold. The
probability of avoiding these zero sets is then handled by the
Schwartz-Zippel lemma. In what follows, as in the algorithm, we write
$m=\lceil n^\eta \rceil$.

\paragraph{(1) A value of $r_1$ such that $\gcd(a+r_1,f)\neq1$}
  The resultant of $a(x)+r_1$ and $f(x)$ is nonzero of degree $n$ in
  $r_1$. Bad choices thus occur with probability at most~$n/\card{S}$.

\paragraph{(2) A value of $r_2$ such that $f(0)\neq0$ after the shift ``\(x\gets x+r_2\)''}
The same reasoning as above applies to the coefficient of degree zero of $f(x+r_2)$.

~ 

The next properties all concern the same parameters
$(r_3,\dots,r_{n+m})$, so their failures are not independent
events, and their joint probability is bounded using a product of
polynomials encoding each of them.  Below, we write $\bar \gamma =
\vargamma_0 + \cdots + \vargamma_{n-1}x^{n-1}$, with the
$\vargamma_i$'s new indeterminates, and consider polynomials in $\Kbar[\vargamma_0,\dots,\vargamma_{n-1}]$ to quantify probabilities
of failure.

  \paragraph{(3) The constant coefficient $\mu_\gamma(0)$ is not~0}
  Write $f = c(x- \varphi_1) \cdots (x-\varphi_n)$, for pairwise
  distinct $\varphi_i$ in $\Kbar$ and $c\in\field\setminus\{0\}$.
  The roots of $\mu_\gamma$ are
  the values $\bar\gamma(\varphi_i)$, so $\mu_\gamma(0)$ being nonzero
  is equivalent to $\gcd(\gamma, f)$ being trivial.
  Thus, we let $\Delta_0 \in
  \field[\vargamma_0,\dots,\vargamma_{n-1}]$ be the resultant of
  $\bar\gamma$ and $f$. This polynomial has degree $n$, and choosing
  $\gamma=1$ shows that it is not identically zero.

  \paragraph{(4) The minimal polynomial~$\mu_\gamma$ has degree~$n$}
 For any $\gamma=\gamma_0 +
  \cdots + \gamma_{n-1}x^{n-1}$ in $\xRing_{<n}$, the characteristic
  polynomial $\chi_\gamma \in \yRing$ of $\gamma\bmod f$ factors over
  $\Kbar[y]$ as $\chi_\gamma = \prod_{i=1}^n (y-\xi_i)$, where
  $\xi_i = \gamma(\varphi_i)$ for all $i$. We can thus let $\Delta_1
  \in \Kbar[\vargamma_0,\dots,\vargamma_{n-1}]$ be the product
  $\prod_{1 \le i < j \le n} (\bar\gamma(\varphi_i)
  -\bar\gamma(\varphi_j))$.  This is a polynomial of degree 
  $n(n-1)/2$, and the previous discussion shows that
  $\Delta_1(\gamma_0,\dots,\gamma_{n-1}) \ne 0$ implies that
  $\chi_\gamma$ is separable. In that case, 
{since the $n$ distinct roots of $\chi_\gamma$ must be roots of $\mu_\gamma$, we have 
$\chi_\gamma=\mu_\gamma$.
}
 Finally, the polynomial~$\Delta_1$ itself is
  nonzero since its value at~$(0,1,0,\dots,0)$, i.e. at~$\gamma=x$, is
  not zero.

  \paragraph{(5) The computation of $\rmatfg[{}]$ does not fail}
  Since $f(0)\neq0$, \cref{prop:generic_a} shows that the associated
  block Hankel matrix $\hankel{m,\lceil n/m\rceil}{\gamma}{f}$ has
  rank~$n$ as soon as the coefficients of $\gamma$ avoid the zero set
  of a polynomial~$\Delta_{f,m}$ of degree at most~$2n^2/m$.

  When this condition holds, \cref{prop:hankel_rank_nu} shows that the
  matrix fraction \(\trsp{X}(\charmat[\gamma])^{-1}X\) is describable
  in degree~$\lceil n/m\rceil$, and that $\ddfg = \ddmmfg$. Since
  we also have $\gcd(\gamma, f)=1$ by the item above, and since the
  minimal polynomial~$\mu_\gamma$ of $\gamma\bmod f$ has degree~$n$,
  \cref{prop:algo:ChangeOfBasis} concludes that the computation of
  $\rmatfg[{}]$ is successful.

  \paragraph{(6) The rank of $\hkma$ is equal to $\ddma$ for $d\ge\lceil\ddma/m\rceil$} When the previous properties are
  all satisfied, there exists a $\field$-algebra isomorphism
  $\phi_{\gamma}:\xRing/\genBy{f}\rightarrow \yRing/
  \genBy{\mu_\gamma}$ that maps $a$ to $\alpha$ such that $\alpha
  (\gamma)\equiv a \bmod f$.
  Up to changing the indices of the roots $\varphi_i$, we can assume
  that $a$ takes values $\lambda_1,\dots,\lambda_r$ at
  $\varphi_1,\dots,\varphi_n$ with multiplicities
  $\ell_1,\dots,\ell_r$, for some positive integers
  $\ell_1,\dots,\ell_r$, and pairwise distinct
  $\lambda_1,\dots,\lambda_r$ in $\Kbar$ (as in
  \cref{ssec:genranksep}, the $\varphi_i$'s are assumed to be ordered
  such that $a(\varphi_1) = \cdots =a(\varphi_{\ell_1})=\lambda_1$,
  etc).  Then, since $\xi_i = \gamma(\varphi_i)$ for all $i$, the
  relation $\alpha (\gamma)\equiv a \bmod f$ implies that
  $\alpha(\xi_i)=a(\varphi_i)$ for all $i$, so that $\alpha$ takes the
  values $\lambda_1,\dots,\lambda_r$ at $\xi_1,\dots,\xi_n$ with
  multiplicities $\ell_1,\dots,\ell_r$.

  The assumptions of \cref{prop:separable} are satisfied. If
  $\polbU_{\ell,m} \in
  \ZZ[\varxi_1,\dots,\varxi_n,\varlb_1,\dots,\varlb_r]$ is the
  polynomial defined in that proposition, then when
  $\polbU_{\ell,m}(\xi_1,\dots,\xi_n,\lambda_1,\dots,\lambda_r)$ is
  nonzero, the rank of $\hkma$ is $\ddma$ for
  $d\ge\lceil\ddma/m\rceil$.  

  The relevant polynomial is thus $\Delta_3=\polbU_{\ell,m}
  (\bar\gamma(\varphi_1),\dots,\bar\gamma (\varphi_
  n),\lambda_1,\dots,\lambda_r)\in\Kbar
  [\bar\gamma_0,\dots,\bar\gamma_{n-1}]$. \cref{prop:separable} states
  that $\polbU_{\ell,m}
  (\bar\gamma_0,\dots,\bar\gamma_{n-1},\lambda_1,\dots,\lambda_r)$ is
  nonzero of degree at most~$2n^2$; this is thus also the case for
  $\Delta_3$, since the transformation
  $(\bar\gamma_0,\dots,\bar\gamma_{n-1}) \mapsto
  (\bar\gamma(\varphi_1),\dots,\bar\gamma (\varphi_ n))$ is linear and
  invertible (its matrix is the Vandermonde matrix at
  $\varphi_1,\dots,\varphi_n$).
  
\paragraph{(7) The computation of $\rmatma[{}]$ does not fail} 
  When the previous properties are all satisfied,
  \cref{prop:hankel_rank_nu} applies with $a=\alpha$ and
  $f=\mu_{\gamma}$ and shows that
  $\nu_{m,m}^{(\alpha,\mu_\gamma)}=\nu_m^{(\alpha,\mu_\gamma)}$ and
  \(\trsp{X}(\charmat[\alpha])^{-1}X\) is describable in
  degree~$\lceil\ddma/m\rceil$, where $\mulmat[\alpha]$ is the
  multiplication matrix of $\alpha$ modulo $\mu_\gamma$. Since
  $\mu_\gamma(0)\neq 0$ and $\gcd(\alpha,\mu_\gamma)=\gcd(a,f)=1$, the
  assumptions of \cref{prop:certificate} are satisfied for the
  successful computation of $\rmatma[{}]$ (\algoName{algo:MatrixOfRelations})
  with a probability of failure depending on the choices of
  $(r_{n+3},\dots,r_{n+m})$ and bounded by~$ (m-1)/\card S$.

  \paragraph{Case $n=1$} In that situation steps, (5)--(7) above simplify.
  Since its top left corner is the identity matrix, the rank of the
  block-Hankel matrix is at least~1, which is equal to~$n$, and thus (5)--(7)
  succeed with probability~1 in that case.
  
  \paragraph{Probability bounds}
  The polynomial~$\Delta_0\Delta_1\Delta_{f,m}\Delta_3\in
  \Kbar[\bar\gamma_0,\dots,\bar\gamma_{n-1}]$ is nonzero and has
  degree at most
  \[d_{n,m}={n+\frac{n(n-1)}2}+\frac{2n^2}m+2n^2.\]
  A choice of $(r_3,\dots,r_{n+2})$ that avoids its zero set ensures
  that the properties (3)--(6) hold. The other probabilities have been
  discussed in steps (1), (2) and (7) above. In summary, the
  probability of success is at least
\[
  \begin{cases}
  \left(1-\frac{1}{\card{S}}\right)^{\!3}\ge 1 - \frac 3{\card {S}}&\text{if $n=1$,}\\
  \left(1-\frac{n}{\card{S}}\right)^{\!2}\left(1-
  \frac{d_{n,m}}{\card{S}}\right)\left(1-\frac{m-1}{\card{S}}\right)
  \ge 1-\frac{2n+d_{n,m}+m-1}{\card{S}}&\text{otherwise.}
  \end{cases}\]
 In the second expression, dividing the numerator of the fraction for
 $n\ge 2$ by~$n^2$ gives
\[\frac52+\frac5{2n}+\frac2m+\frac{m-1}{n^2},\]
which decreases as a function of $n$ for $n\ge0$ and, for fixed $n$, decreases
as a function of~$m$ for $m\le n$. Thus it reaches its maximum at $m=1,n=2$,
where its value is $23/4<6$, proving the probability bound for $n\ge2$, while
$3<6$ deals with the case~$n=1$.

The assertion for \algoName{algo:AnnihilatingPolynomial} follows: 
\cref{step:shiftp-mainalgo} apart, it fails in the same cases as 
\algoName{algo:ModularCompositionBaseCase}.
\end{proof}

\subsection{Success of randomization for \texorpdfstring{$f$}{f}  purely inseparable: small valuation} 
\label{sec:composition_randomized:proof_inseparable}

\begin{definition} A degree $n$ polynomial $f$ in $\xRing$ is 
  \emph{purely inseparable} if it has
  only one root in an algebraic closure $\Kbar$, so that it
  factors as $f=(x-\xi)^n$ in $\Kbar[x]$; if $n$ is a unit in
  $\field$, $\xi$ itself is in $\field$.
\end{definition}

{In this section, we study the probabilistic aspects of
\algoName{algo:ModularCompositionBaseCase}{} for such
polynomials.
If $a=a_0+a_v(x-\xi)^v +a_{v+1}(x-\xi)^{v+1}+
  \cdots + a_{n-1}(x-\xi)^{n-1}$, with $a_0=a(\xi)$ and $a_v \ne 0$, then the valuation $v=\val_\xi(a-a(\xi))$ is 
    the order of vanishing of $a-a(\xi)$ at $x=\xi$.  For the moment we work under
two additional assumptions on this valuation:
it is not~0 in~$\field$, and it is at most the value chosen for $m$
(which is $\lceil n^{\eta}\rceil$ in the algorithm, for the target
complexity bound).  The other cases are discussed in the next section.}

\begin{proposition} \label{prop:series} 
  Let $a,f$ be polynomials in~$\xRing$ and $\polp$ be in $\yRing$, with $f=(x-\xi)^n \in \xRing$
  where $\xi\in \Kbar$, and $\deg(a) < n$.  Let $p$ be the
  characteristic of~$\field$.  Suppose that $v=\val_\xi (a-a(\xi))$
  satisfies the following inequalities, with $\eta$ as in
  \cref{eq:def-beta}:
  \[v \le \lceil n^{\eta} \rceil,\qquad p=0\text{ or }{v < p}.\]
  Take $r_1=0$ if $\gcd(a,f)=1$ and $r_1=1$ otherwise, $r_2=0$ if
  $f\neq x^n$ and $r_2=1$ otherwise.  If $r_3, \ldots, r_{n+\lceil
    n^\eta \rceil}$ are chosen uniformly and independently from a
  finite subset $S$ of $\field$, then Algorithms \nameref{algo:ModularCompositionBaseCase} and
  \nameref{algo:AnnihilatingPolynomial}
  return \Fail{} with probability at most~$\probainsep/\card{S}$
\end{proposition} 
\begin{proof}
  The proof follows the same steps as in
  \cref{sec:composition_randomized:proof_separable}. As before, we
  write $m=\lceil n^\eta \rceil$.
  
  \paragraph{(1), (2) Values of $r_1$ and $r_2$}
  The choice of $r_1$ gives $\gcd(a+r_1,f)=1$, and $r_2$ modifies the
  constant coefficient of~$f$ if necessary.  The first two steps of
  \algoName{algo:ModularCompositionBaseCase} therefore provide
  polynomials that satisfy $\gcd(a,f)=1$, $f(0)\neq0$, and
  $v=\val_\xi(a-a(\xi))\le m$.

  \paragraph{(3) The constant coefficient~$\mu_\gamma(0)$ is not 0}
  For any $\gamma=\gamma_0+\dots+\gamma_{n-1}x^{n-1}$, the roots of
  the characteristic polynomial $\chi_\gamma$ of $\gamma$ modulo $f$
  are the values taken by $\gamma$ at the roots of $f$, counted with
  multiplicities. Since $f=(x-\xi)^n$ over $\Kbar[x]$, this
  implies that {$\chi_\gamma=(y-\gamma(\xi))^n$}. The minimal
  polynomial $\mu_\gamma$ then admits a similar factorization as
  {$(y-\gamma(\xi))^c$}, for some positive $c$.

  Set
  $\Delta_0(\vargamma_0,\dots,\vargamma_n)=\sum_{i=0}^{n-1}{\vargamma_i\xi^i}$;
  this is a (nonzero) polynomial of degree~1 which is such that
  $\Delta_0(\gamma_0,\dots,\gamma_{n-1})=\gamma(\xi)$, so the
  nonvanishing of this quantity gives the same property
  for~$\mu_\gamma(0)$. 

  \paragraph{(4) The minimal polynomial~$\mu_\gamma$ has degree~$n$}
  Consider now
  $\Delta_1(\vargamma_0,\dots,\vargamma_{n-1})=\sum_{i=1}^{n-1}{ i
    \vargamma_i\xi^{i-1}}$, which is also a nonzero polynomial of degree~1.
  It is such that $\Delta_1(\gamma_0,\dots,\gamma_{n-1})=\gamma'(\xi)$, so
  the nonvanishing of this quantity implies that $\val_\xi
  (\gamma-\gamma(\xi))=1$. This implies that the powers
  $1,\gamma-\gamma(\xi),(\gamma-\gamma(\xi))^2,\dots,(\gamma-\gamma(\xi))^{n-1}
  \rem (x-\xi)^n$ have respective valuations $0,1,\dots,n-1$ at $\xi$,
  and thus are linearly independent. It follows that the minimal
  polynomial of~$\gamma-\gamma(\xi)$ has degree~$n$, and the same then
  holds for $\gamma$ itself.

  \paragraph{(5) The computation of $\rmatfg[{}]$ does not fail}
  Here the argument of the previous section applies verbatim and
  relies on a polynomial~$\Delta_{f,m}$ of degree at most~$2n^2/m$.

  \paragraph{(6) The rank of \/$\hkma$ is equal to $n$
    for $d\ge\lceil n/m\rceil$} This step is the difficult one in the
  proof; note that the statement slightly deviates from the one in the
  separable case in the definition of the threshold degree $\lceil n/m
  \rceil$.

  The result is obtained by bounding the degree of the numerator of a
  nonzero $n\times n$ minor of~$\hkma[m,\lceil n/m\rceil]$, seen as a
  polynomial in~$\gamma_0,\dots,\gamma_{n-1}$.  We first show the
  existence of $\gamma \in \xRing_{<n}$ and $\alpha \in \yRing_{<n}$
  such that the block Hankel matrix $\hkma[m,\lceil n/m\rceil]$ has
  rank~$n$. This implies the existence of a nonzero $n\times n$ minor
  of this matrix; the degree of this minor as a polynomial in the
  coefficients of~$\alpha$ and~$\mu_\gamma$ is controlled by
  \cref{cor:hankelwithxm}. These in turn are related to the
  coefficients of~$\gamma$, using its explicit form for~$\mu_\gamma$
  and a linear system for the coefficients of~$\alpha$.

  \paragraph{(6a) Generic behavior}
  We start by proving the existence of $\alpha$ of degree $m$ in
  $\Kbar[y]$ and $\gamma$ in $\Kbar[x]_{<n}$ such that we
  have $\alpha(\gamma) \equiv a \bmod f$, $\gamma(\xi) \ne 0$ and
  $\gamma'(\xi)\ne 0$.

  Write $a=a_0+a_v(x-\xi)^v+
  \cdots + a_{n-1}(x-\xi)^{n-1}$, with $a_0=a(\xi)$ and, by definition
  of $v=\val_\xi(a-a(\xi))$, $a_v \ne 0$ and $v>0$.  Since we also
  assume that the characteristic $p$ of $\field$ is either zero, or
  greater than~$v$, this means in particular that $v$ is a unit
  in~$\field$. Let
  \[\tilde a(x) = \frac{a-a_0}{a_v (x-\xi)^v} = 1 + \sum_{1 \le  i < n-v}{\tilde a_i (x-\xi)^i},\]
  with coefficients $\tilde a_i = a_{i+v}/a_v$.  
  \begin{itemize}
  \item[--] If $v=m$, we define $\tilde \alpha(y) = y^v = y^m$. Since
    $v \ne 0$ in $\field$, $\tilde \alpha'(1) \ne 0$ and Newton
    iteration guarantees the existence of a unique $\tilde{\gamma} = 1
    + \sum_{1 \le i < n}{\tilde{\gamma}_i (x-\xi)^i}$ such that
    $\tilde{\gamma}^v \equiv \tilde a \bmod f$.

  \item[--] If $v < m$, we define $\tilde \alpha(y) = y^v + y^m$. This
    time, we let $\tilde{\gamma}$ be the unique polynomial of the form
    $\tilde{\gamma} = 1 + \sum_{1 \le i < n}{\tilde{\gamma}_i
      (x-\xi)^i}$ such that $\tilde{\gamma}^v +
    (x-\xi)^{m-v}\tilde{\gamma}^m \equiv \tilde a \bmod f$.  As
    previously, existence follows from Newton iteration, using the
    assumption $v \ne 0$ in $\field$.
  \end{itemize}
  In both cases, we set $ \gamma= 1+ (x-\xi) \tilde{\gamma}\rem f \in
  \Kbar[x]$ and $\alpha = a_0 + a_v \tilde \alpha(y-1) \in
  \Kbar[y]$. We can then verify that all requirements
  $\alpha(\gamma) \equiv a \bmod f$, $\gamma(\xi) \ne 0$ and
  $\gamma'(\xi)\ne 0$ are satisfied.

  Since $\mu_\gamma$ then has degree $n$, and since
  $\mu_\gamma(0)=(-\gamma(\xi))^n$ is nonzero, \cref{lem:hankelwithxm}
  shows that $\hkma[m,d]$ has rank $n$ for $d\geq \lceil n/m\rceil$.

\paragraph{(6b) A polynomial in~$\field[\bar a_0,\dots,\bar a_{n-1},\varf_0,\dots,\varf_{n-1}]$}
  The existence of $\gamma$ and $\alpha$ implies that
  of a nonzero $n\times n$ minor $\delta$ of $\hkca[m,\lceil
  n/m\rceil]$. 
  Let then $\polDelta \in
  \field[\bar a_0,\dots,\bar a_{n-1},\varf_0,\dots,\varf_{n-1}]$ be
  the corresponding minor of $\hankel{m,\lceil
    n/m\rceil}{\bar a}{\varf}$, where $\bar a=\bar a_0 +\cdots +
  \bar a_{n-1} x^{n-1}$ and $\varf=\varf_0 +\cdots + \varf_{n-1}
  x^{n-1} + x^n$ are polynomials whose coefficients are indeterminates. 
  \cref{cor:hankelwithxm} shows that this is a 
  polynomial of degree at most $2n^2/m$ in $\bar a_0,\dots,\bar
  a_{n-1}$ and $2n^2(n-1)/m$ in $\varf_0,\dots,\varf_{n-1}$.

  \paragraph{(6c) The rational functions $\bar\alpha_0,\dots,\bar\alpha_{n-1}$}
  Next, with $\vargamma = \vargamma_0 + \cdots + \vargamma_{n-1}
  x^{n-1}$ a polynomial whose coefficients are indeterminates, we
  consider $\bar \alpha$ such that~$\bar \alpha(\vargamma)\equiv
  a\bmod f$.  The coefficients of $\bar\alpha$ are given as solutions
  of the linear system $\bar \alpha(\vargamma)\equiv a\bmod f$, thus
  they are rational functions $\bar\alpha_0,\dots,\bar\alpha_{n-1}$
  in~$\field (\vargamma_0,\dots,\vargamma_{n-1})$. In this paragraph,
  we bound the degrees of their numerators and denominators in
  $\field[\vargamma_0,\dots,\vargamma_{n-1}]$, using power series
  inversion and composition.

  We first consider the solution $u$ to $u(\vargamma)\equiv
  x-\xi\bmod(x-\xi)^n$, or equivalently $u(\bar \varphi)\equiv x\bmod
  x^n$, with $\bar \varphi=\vargamma(x+\xi)$. We write $\bar \varphi
  =\bar \varphi_0+\bar \varphi_1 x+\dots+\bar \varphi_{n-1}x^{n-1}$,
  where the coefficients $\bar \varphi_0,\dots,\bar \varphi_{n-1}$ are
  linear in~$\vargamma_0,\dots,\vargamma_{n-1}$, with in particular
  $\bar \varphi_0=\Delta_0$ and $\bar \varphi_1=\Delta_1$. We can then
  write $u =\sum_{j=1}^{n-1} u_j (y-\Delta_0)^j$, where for $j \ge 1$,
  the coefficient $u_j$ is a rational function
  in~$\vargamma_1,\dots,\vargamma_{n-1}$, with numerator of
  degree~$j-1$ in $\vargamma_1,\dots,\vargamma_{n-1}$ and
  denominator~$\Delta_1^{2j-1}$. More generally, for $i \ge 1$, the
  power $u^i$ has valuation $i$, and for $j \ge i$, the coefficient of
  $ (y-\Delta_0)^j$ in it is a rational function with numerator of
  degree~$j-i$ in $\vargamma_1,\dots,\vargamma_{n-1}$ and
  denominator~$\Delta_1^{2j-i}$.

  It follows that if we write $a =a_0 + a_v (x-\xi)^v + \cdots +
  a_{n-1} (x-\xi)^{n-1}$, then the solution $\bar \alpha$ to the
  equation $\bar \alpha(\vargamma)\equiv a\bmod f$ is given by $\bar
  \alpha=a_0 + a_v u^v+\dots+a_{n-1} u^{n-1}\rem (y-\Delta_0)^n$.
  Once we rewrite $\bar \alpha$ as $\bar\alpha_0 + \cdots +
  \bar\alpha_{n-1}y^{n-1}$, we see that the coefficients
  $\bar\alpha_0,\dots,\bar\alpha_{n-1}$ are rational functions with
  numerator of degree at most $2n-3$ in
  $\vargamma_0,\dots,\vargamma_{n-1}$, and denominator
  $\Delta_1^{2n-3}$.

\paragraph{(6d) The polynomial $\Delta_2$}
We now evaluate the indeterminates $\bar a_i$ and $\bar f_i$ in the minor
$\Delta$ of (6b) at the coefficients of $\bar \alpha$ and $\chi
_{\bar \gamma}=(y- \Delta_0)^n$, respectively.  Write $(y-
\Delta_0)^n$ as $\bar q_0+ \cdots + \bar q_{n-1} y^{n-1}+y^n$, so that
$\bar q_i = \binom{n}{i} (-\Delta_0)^{n-i}$ for all $i$. It follows
that $\polDelta(\bar\alpha_0,\dots,\bar\alpha_{n-1},\bar
q_0,\dots,\bar q_{n-1})$ is a rational function in the indeterminates
$\vargamma_0,\dots,\vargamma_{n-1}$, which can be written as
\begin{align}\label{eqdef:Delta}
  \Delta(\bar\alpha_0,\dots,\bar\alpha_{n-1},\bar q_0,\dots,\bar q_{n-1})=
  \frac{\Delta_2(\vargamma_0,\dots,\vargamma_{n-1})}{\Delta_1(\vargamma_0,\dots,\vargamma_{n-1})^{\epsilon}},
\end{align}
for
some polynomial~$\Delta_2$ of degree at most 
\[\frac{2n^2}m(2n-3)+ \frac{2n^2(n-1)}m(n-1)=\frac{2n^2(n^2-2)}m,\]
and for some integer exponent $\epsilon \le 2n^2(2n-3)/m$. 

Consider again the polynomials $\gamma$ and $\alpha$ in (6a), and
their coefficients $\gamma_0,\dots,\gamma_{n-1}$ and
$\alpha_0,\dots,\alpha_{n-1}$ (with actually $\alpha_{m+1}
=\cdots=\alpha_{n-1}=0$). We saw that $\gamma$ satisfies
$\Delta_1(\gamma_0,\dots,\gamma_{n-1})=\gamma'(\xi)\neq0$, which
implies that the rational
functions~$\bar\alpha_0,\dots,\bar\alpha_{n-1}$ are well defined at
$\gamma_0,\dots,\gamma_{n-1}$ and take $\alpha_0,\dots,\alpha_{n-1}$
for values there. This implies that the nonzero minor~$\delta$ is
$\delta=\Delta_2(\gamma_0,\dots,\gamma_{n-1})/\Delta_1(\gamma_0,\dots,\gamma_
{n-1})^\epsilon$, and in particular that $\Delta_2$ is a nonzero
polynomial.

\paragraph{Probability bounds}
The end of the proof is as in the previous section. The polynomial
$\Delta_0\Delta_1\Delta_{f,m}\Delta_2$ in
$\Kbar[\vargamma_0,\dots,\vargamma_{n-1}]$ has degree at most
\[
  1+1+\frac{2n^2}m+\frac{2n^2(n^2-2)}m=\frac{2(n^4-n^2+m)}m;
\]
we can now readily verify that a choice of $(r_3,\dots,r_{n+2})$ that
avoids its zeros ensures that properties (3)-(6) hold. For (3)-(5),
this follows immediately from the definitions. 

To see that (6) holds, that is, that $\hkma$ has rank $n$ for
$d\ge\lceil n/m\rceil$, recall that the algorithm constructs $\gamma =
r_3 + r_4 x + \cdots + r_{n+2}x^{n+1}$. Properties (3)-(4) show that
$\mu_\gamma$ has degree $n$, and that its constant coefficient is
nonzero. Since in particular $\Delta_1(r_3,\dots,r_{n+2}) \ne 0$, we
deduce that the rational functions
$\bar\alpha_0,\dots,\bar\alpha_{n-1}$ of (6c) are well defined at
$(r_3,\dots,r_{n+2})$, and that they give the coefficients of the
unique polynomial $\alpha$ such that $\alpha(\gamma)\equiv a \bmod f$.
Since $\Delta_2(r_3,\dots,r_{n+2}) \ne 0$, it follows from
\cref{eqdef:Delta} that $\hkma[m,\lceil n/m\rceil]$ has rank $n$ (and
thus similarly for $\hkma$, for $d\ge\lceil n/m\rceil$).

The other probabilities have been discussed in step~(1)-(2) above and
in step~(7) of the previous section.
Altogether, this gives a probability of success at least
\[\left(1-\frac{2(n^4-n^2+m)/m}{\card S}\right)\left(1-\frac{m-1}{\card
S}\right)\ge 1-\frac{2(n^4-n^2+m)/m+m-1}{\card S}.\]
Dividing the numerator of the last fraction by~$n^4/m$ gives
\[2-\frac{2n^2-m^2-m}{n^4}\le 2,\]
where the last inequality comes from $m\le n$.
\end{proof}

\subsubsection{Note}
\label{rem:counterexample} In \cref{prop:series}, the role of the condition   on the valuation being nonzero in~$\field$ is shown by
the following example.
Take a field~$\field$ of characteristic~2, $n=6$, 
$m=3$, $f=(x-1)^6$ and $a=(x-1)^2$. Then for any~$\gamma\in\xRing _{<6}$, 
the four polynomials~$(1,a,\gamma^2,a \gamma^2)\rem f$
belong to the vector space generated by~$(1,x^2,x^4)$ and are
therefore linearly dependent. Using the expression of $M_{\alpha}$ from \cref{eq:defMalpha}, we see that this implies that the block Krylov 
matrix~$\Rama[n/m]$ of \cref{eq:factor_Hk} is singular, and
thus so is~$\hkma[m,n/m]$ regardless of the choice of $\gamma$.

A more general version of this counterexample when~$\field$ has 
characteristic~$p>0$ is obtained with $m=p+1$, $d=p$, $n=md$, and 
$\val_\xi(a)=p$.

\subsection{Complete algorithm  for \texorpdfstring{$f$}{f}  purely inseparable}
\label{sec:composition_randomized:proof_inseparable_small}

We now extend \cref{prop:series} in order to cover all cases of
composition modulo a purely inseparable polynomial~$f$. 

If $p$ is the characteristic of~$\field$, any purely
inseparable~$f$ can be written as $f(x)=(x^{p^e}-c)^\mult$ with $c$ in
$\field$ and $e,\mult$ in $\NN$ such that~$p$ does not divide $\mult$,
and $e = 0$ if $p=0$~\cite{GiTr96}; in particular the degree~$n$ of~$f$ is equal to $p^e \mult$. 
We assume that the parameters $e$, $\mult$ and $c$ are known, since this is the
case when our algorithms have to handle this situation; indeed in the
next section we introduce separable factorization techniques that
allow us to
compute them.

\subsubsection{Large valuation} \label{subsec:comp_special1}
If $v=\val_\xi (a (x)-a(\xi))$ satisfies $v > \lceil n^\eta \rceil$,
the minimal polynomial of $a$ in~$\xRing/\langle f \rangle$ factors
over $\Kbar$ as $\mu_a (y) = (y-a(\xi))^ \delta$, with $\delta =
\lceil n/v \rceil \leq \lceil n^{1-\eta} \rceil$. Since the latter degree is
small compared to $n$, this case is handled efficiently by
\algoName{algo:ModularComposition-SmallMinimalPolynomial}{} from
\cref{ssec:BK}.

\subsubsection{Small characteristic} \label{subsec:comp_special2}
In the case $0 < p \le \lceil n^\eta\rceil$, our algorithm is based on
Bernstein's composition algorithm for power series~\cite{Ber98}, which
we adapt to work modulo $f(x)=(x^{p^e}-c)^\mult$. See
also~\cite[Algorithm 3.1]{HoeLec17} for another extension of
Bernstein's result, which is however not sufficient to reach our target cost
for the specific kind of modulus we work with.
 
If $e=0$, $p^e=1$ and we are working modulo $f=(x-c)^\mult$, with
$\mult=n$. In this case, to compute $b=\polp(a) \rem f$, we write
$\tilde a(x) = a(x+c)$, we compute $\tilde b = \polp(\tilde a) \rem
x^\mult$, then we obtain $b$ as $\tilde b(x-c)$.  The bottleneck is
the computation of $\polp(\tilde a) \rem x^\mult$, which can be done
in $\softO {p\mult}$ operations in~$\field$ using Bernstein's
algorithm (in \algoName{algo:CompositionModuloInseparable-SmallCharacteristic},
that algorithm is called
\textproc{PowerSeriesComposition-SmallCharacteristic}).  

Suppose now that $e \ge 1$. Write $\polp=\sum_{i=0}^{p-1} \polp_i(y^p)
y^i$, with $\polp_i \in \yRing$ of degree less than $p^{e-1}\mult$.
Write also $a(x)=\sum_{i=0}^{n-1} a_i x^i$, and let $\bar
a(x)=\sum_{i=0}^{n-1} a_i^p x^i$, so that $a^p(x) = \bar a(x^p)$. It
follows that
\begin{align*}
  \polp(a) \rem (x^{p^e}-c)^\mult &= \sum_{i=0}^{p-1} \bar\polp_i a^i\rem (x^{p^e}-c)^\mult
\end{align*}
where, for all $0\leq i \leq p-1$,
\[
 \bar\polp_i(x)= \polp_i(a^p(x)) \rem (x^{p^e}-c)^\mult
= \polp_i(\bar a(x^p))\rem (x^{p^e}-c)^\mult.
\]
If we define  $h_i= \polp_i(\bar a)\rem (x^{p^{e-1}}-c)^\mult$,
it follows that $\bar g_i = h_i(x^p)$, so that
 \begin{align*}
  \polp(a) \rem (x^{p^e}-c)^\mult&=   \sum_{i=0}^{p-1} h_i(x^p) a^i \rem (x^{p^e}-c)^\mult.
 \end{align*}
The following lemma summarizes the cost of this procedure.

\begin{algorithm}
  \algoCaptionLabel{CompositionModuloInseparable-SmallCharacteristic}{c,e,\mult,a,\polp}
  \begin{algorithmic}[1]
  \Require
  \parbox[t]{0.8\textwidth}{
    $\field$ has characteristic $p>0$,

    $c$ in $\field$, $e$ in $\NN$ and $\mult$ in $\NN_{>0}$ such that $f= (x^{p^e}-c)^\mult$ has degree $n=\mult p^e$,

    $a$ in $\xRing_{< n}$, $\polp$ in $\yRing_{<n}$ 
  }
  \Ensure $\polp(a) \rem f$ 
  
  \If {$e=0$}
  \State $a \gets a(x+c)$
  \State $b \gets \textproc{PowerSeriesComposition-SmallCharacteristic}(x^n,a,\polp)$ 
  \Comment{\cite[Sec.\,2]{Ber98}}
  \State \Return $b(x-c)$
  \Else
  \State Write $\polp=\polp_0(y^p) + \cdots + \polp_{p-1}(y^p) y^{p-1}$
  \State Write $a = a_0 + \cdots + a_{n-1} x^{n-1}$
  \State $\bar a \gets  a_0^p + \cdots + a_{n-1}^p x^{n-1}$\label{smallcar-8}
  \For{$i=0,\dots,p-1$}
  \State $h_i \gets  \Call{algo:CompositionModuloInseparable-SmallCharacteristic}{c,e-1,\mult,\bar a,\polp_i}$
  \EndFor
  \State {$f \gets (x^{p^e}-c)^\mult$} \label{smallcar-11a}
  \State \Return $h_0(x^p) + \cdots + h_{p-1}(x^p) a^{p-1} \rem f$\label{smallcar-11}
  \EndIf
  \end{algorithmic}
\end{algorithm}

\begin{lemma}\label{lemma:modifiedDJB}
  For a field $\field$ of characteristic $p>0$, given a purely inseparable
  polynomial $f=(x^{p^e}-c)^\mult$ of degree $n=\mult p^e$, $a\in\field
  [x]_{<n}$ and $g\in\field[y]_{<n}$,
  \algoName{algo:CompositionModuloInseparable-SmallCharacteristic}{} returns
  $\polp(a)
  \rem f$ and uses $\softO {pn}$ operations in $\field$.
\end{lemma}
\begin{proof}
  Correctness follows from the previous description. For the runtime
  analysis when $e=0$, the result is Bernstein's. For $e>0$, apart from the
  $p$ recursive calls, \cref{smallcar-8} takes $\softO{n}$
  operations (we raise all coefficients of $a$ to the power $p \le n$), 
{\cref{smallcar-11a} takes $\softO n$ operations by repeated squaring}, 
  and \cref{smallcar-11} takes $\softO {pn}$ operations, using
  Horner's rule. Remembering that $n=\mult p^e$, we deduce that the
  runtime $T(e,p,\mult)$ satisfies $T(e,p,\mult) = p T(e-1,p,\mult) +
  \softO{p^{e+1}\mult}$ and $T(0,p,\mult) \in \softO{p\mult}$. This
  resolves to $T(e,p,\mult) \in \softO{ p^{e+1} \mult}$, which is~$\softO{pn}$.
\end{proof}

\subsubsection{Main algorithm}

Combining the previous results gives
\algoName{algo:CompositionModuloInseparable}. It first tests whether
the characteristic of $\field$ is small enough for
\algoName{algo:CompositionModuloInseparable-SmallCharacteristic} to run within
our prescribed runtime. Otherwise, rather than computing the valuation
$v$, it simply calls
\algoName{algo:ModularComposition-SmallMinimalPolynomial}; in case
of failure, it falls back on
\algoName{algo:ModularCompositionBaseCase}.  As previously, the
algorithm takes as input a vector $r$ that plays the role of random
parameters.

\begin{algorithm}
  \algoCaptionLabel{CompositionModuloInseparable}{c,e,\mult,a,\polp,r}
  \begin{algorithmic}[1]
    \Require
    \parbox[t]{0.9\textwidth}{
      $c$ in $\field$, $e$ in $\NN$ and $\mult$ in $\NN_{>0}$ such that $f= (x^{p^e}-c)^\mult$ has degree $n=\mult p^e$, 
      where $p$ is the characteristic of~$\field$,

      $a$ in $\xRing_{< n}$, $\polp$ in $\yRing_{<n}$, $r\in \field ^{n+\lceil n^{\eta} \rceil}$ with $\eta$ from \cref{eq:def-beta}
    }
    \Ensure $b=\polp(a) \rem f$  or \Fail

    \State $n \gets \mult p^e$
    \If {$0 < p \leq \lceil n^{\eta}\rceil$}
    \Statex 
     \Return \Call{algo:CompositionModuloInseparable-SmallCharacteristic}{c,e,\mult,a,\polp}
    \EndIf
    \State $f \gets (x^{p^e}-c)^\mult$
    \State $b \gets \Call{algo:ModularComposition-SmallMinimalPolynomial}{f,a,\polp,\lceil n^{1-\eta} \rceil, (r_i)_{0\le i < n}}$\label{algoinsep:h}
    \Statex \InlineIf {$b \ne \Fail$}{\Return $b$}
    
    \State     \InlineIfElse{$\gcd(a,f)=1$}{$r_1=0$}{$r_1=1$}; \InlineIfElse{$c\neq 0$}{$r_2=0$}{$r_2=1$}
    \State \Return \Call{algo:ModularCompositionBaseCase}{f,a,\polp,r}\label{algoinsep:BC}
    \Comment{\cref{prop:series} }
  \end{algorithmic}
\end{algorithm}

\begin{proposition}\label{thm:CompositionModuloInseparable}
  For a field~$\field$ of characteristic~$p$, given $c,e,\mult$  such that
  $f=(x^{p^e}-c)^\mult$ is purely inseparable of degree $n=\mult p^e$ ($e=0$ if
  $p=0$), $a\in\field [x]_{<n}$, $g\in\field[y]_{<n}$ and
  $r\in\field^{n+m}$ with $m=\lceil n^{\eta}\rceil$ and $\eta$ from \cref{eq:def-beta}, 
  \algoName{algo:CompositionModuloInseparable}{}
  uses~$\softO{n^{\kappa}}$ operations in $\field$, with $\kappa <
  1.43$ as in \cref{eq:def-gamma}, and returns either $\polp(a) \rem
  f$, or \Fail.
  
  If the entries of $r$ are chosen uniformly and independently from a
  finite subset $S$ of $\field$, then the algorithm returns $\polp(a)
  \rem f$ with probability at least $1-\probainsep/\card{S}$.
\end{proposition}
\begin{proof}
  We first analyze the runtime. For a small characteristic $0 < p \leq \lceil n^{\eta}\rceil$, then
\algoName{algo:CompositionModuloInseparable-SmallCharacteristic}{} has cost 
  $\softO {pn}$ by \cref{lemma:modifiedDJB}, which is thus $\softO {n^
  {1+\eta}}\in \softO {n^{1+(\omega-1)\eta}}=\softO{n^\kappa}$ from~\cref{eq:def-beta}. Computing $f$
  takes time $\softO n$ by repeated squaring. By
  \cref{lemma:smallminpoly}, the call to
  \algoName{algo:ModularComposition-SmallMinimalPolynomial}{} uses
  \[\softO { n^{1 + (1-\eta)(\omega_2/2-1)}}=
\softO { n^{\eta + (1-\eta)(\omega_2/2)}}=
  \softO{n^\kappa}
  \]
  operations in $\field$, and by \cref{prop:algo-mod_comp}, it is also the
  case for \algoName{algo:ModularCompositionBaseCase}{}.  The
  specifications of the subroutines imply that the output can be
  either $\polp(a) \rem f$ or \Fail, so only the probability analysis remains.

  If $0 < p \leq \lceil n^{\eta}\rceil$, \cref{lemma:modifiedDJB} shows that
  the output is $\polp(a) \rem f$; hence, we may now assume that $p
  > \lceil n^{\eta}\rceil$, or $p=0$. Let $\xi=c^{1/p^e} \in
  \Kbar$, so that $f=(x-\xi)^n$ in $\Kbar[x]$; let further $v$
  be the valuation of $a-a(\xi)$ at $\xi$. The minimal
  polynomial of $a$ modulo $f$ has degree $\delta=\lceil n/v\rceil$.

  Suppose first that $v \le \lceil n^\eta \rceil$, so that $\delta
  \ge\lceil n^1/\lceil n^\eta\rceil \rceil$. The value $b$ computed at
  \cref{algoinsep:h} is either $\polp(a) \rem f$, or \Fail; let $\pi$
  be the probability of the former (for instance, by
  \cref{lemma:smallminpoly}, $\pi=0$ if $\delta > \lceil n^{1-\eta}
  \rceil$).  If $\Fail$ is returned at \cref{algoinsep:h}, then we
  enter \cref{algoinsep:BC}. At this stage, we have inequalities $v
  \le \lceil n^\eta \rceil < p$, or $v \le \lceil n^\eta\rceil$ and
  $p=0$, so by \cref{prop:series} the call to
  \algoName{algo:ModularCompositionBaseCase}{} returns $\polp(a) \rem
  f$ with probability at least $1-\probainsep/\card{S}$. Overall, the
  probability of returning $\polp(a) \rem f$ in this case is at least
  $\pi + (1-\pi) (1-\probainsep/\card{S})$, which is at least
  $1-\probainsep/\card{S}$.
  
  Suppose on the other hand that $v > \lceil n^\eta \rceil$, so that
  we have in particular $v \ge n^\eta$, and thus $\delta = \lceil
  n/v\rceil \le \lceil n^{1-\eta}\rceil$. By
  \cref{lemma:smallminpoly}, $b$ computed at \cref{algoinsep:h} is
  $\polp(a) \rem f$ with probability at least $1-n/\card{S}$. If it is
  not the case, the algorithm enters \algoName
  {algo:ModularCompositionBaseCase}{}, which computes $\polp(a) \rem
  f$ with a certain probability $\pi' \ge 0$. Overall, we return
  $\polp(a) \rem f$ with probability at least $1-n/\card{S} + \pi' \ge
  1-n/\card{S}$.
\end{proof}


\section{Algorithm for general \texorpdfstring{$f$}{f}} \label{sec:generalalgo}

We now present our Las Vegas \algoName{algo:ModularComposition} that 
computes $\polp(a) \rem f$ for arbitrary input $\polp,a,f$. The analysis of this
algorithm in \cref{ssec:mainalgo} proves~\cref{thm:intro}.

The starting point is the \emph{separable decomposition} of $f$
(\cref{ssec:separabledecomp}), a generalization of square-free
decomposition from fields of characteristic zero to arbitrary base
fields. This yields a partial factorization $f=f_1\dotsm f_s$ into
pairwise coprime factors. The algorithm then proceeds by computing
$\polp(a)$ modulo each of these factors and the final result is
obtained by Chinese remaindering in quasi-linear
complexity~\cite[\S10.3]{GaGe99}.
If $p$ is the characteristic of $\field$ then the factors~$f_i$ of the
separable decomposition of $f$ are the form
$h_i(x^{p^{e_i}})^{\mult_i}$ (or more simply~$h_i(x)^{\mult_i}$ when
$p=0$), with integers $e_i,\mult_i$ and separable $h_i \in
\xRing$. Composition modulo such an~$f_i$ is achieved via a
$\field$-algebra isomorphism
\[\Psi_i:\quotient_i=\xRing/\genBy{f_i(x)}\to\quotientB_i=\field
[\theta,z]/\genBy{h_i (\theta),(z^{p^ {e_i}}-\theta)^{\mult_i}}
\] that
maps~$x$ to~$z$ (\cref{prop:Phi}). If $\quotientL_i$ denotes $\field
[\theta]/\langle h_i(\theta) \rangle$,
then, as a $\field$-vector space, $\quotientB_i\simeq\quotientL_i
[z]/\genBy{(z^{p^{e_i}}-\bar \theta _i)^{\ell_i}}$  with $\bar\theta_i$ the class of $\theta$
in~$\quotientL_i$. The computation of
$\polp
(a)\rem
f_i$ over $\field$ is thus mapped to the composition 
$$
\polp (A_i) \bmod  (z^{p^
{e_i}}-\bar\theta_i)^
{\ell_i}
$$ over $\quotientL_i$, with $A_i=\Psi_i (a\mod f_i)$ and modulo the
purely inseparable $(z^{p^ {e_i}}-\bar\theta_i)^ {\ell_i}$.  In order
to perform this last composition efficiently, it is also necessary to
decrease the degree of~$\polp$ by first reducing~$g$ modulo the
characteristic polynomial of~$A_i$ in $\quotientL_i[z]/\genBy{(z^{p^
    {e_i}}-\bar\theta_i)^{\ell_i}}$. We call \emph{reduction} of~$g$
that step of the process (\cref{lemma:algo-gen-bivar-red}). It
produces a representative of~$G_i\in\quotientL_i[y]$ such
that~$B_i=\polp (A_i)\in\quotientB_i$ is obtained through the
univariate modular composition
$$
G_i(A_i) \bmod  (z^{p^
{e_i}}-\bar\theta_i)^
{\ell_i},
$$ which is computed with coefficients in $\quotientL_i$.  Finally,
the class $\polp (a) \bmod f_i \in \quotient$ is recovered as $\Psi
_i^{-1}(B_i)$.  In practice, the algorithms working with elements
of~$\quotientL_i$ use polynomial representatives in $\field
[\theta]_{<\deg(f_i)}$, that are the canonical lifts of their class.

The idea of using these homomorphisms was introduced by van der Hoeven
and Lecerf in the case~$e_i=0$~\cite{HoeLec17}; it is extended to
the general case in \cref{ssec:untangling,ssec:modulo_powers}. We keep
their terminology, calling \emph{untangling} an algorithm that
computes the map~$\Psi_i$ and \emph{tangling}, one which computes the
reverse map. Both these operations can be performed efficiently
(\cref{ssec:untangling}).

The univariate modular composition in $\quotientL_i[z]$ modulo the
purely inseparable polynomial~$ (z^{p^{e_i}}-\bar\theta )^{\mult_i}$
can be
achieved by \algoName{algo:CompositionModuloInseparable} of
\cref{sec:composition_randomized:proof_inseparable_small}
when~$\bigfield_i$ is a field. In general however, $\bigfield_i$ is
\emph{a product of fields}. In \cref{sec:overseparable}, the extension
of the scope of our algorithms to this setting is obtained using
a paradigm also due to van der Hoeven and Lecerf
called \emph{directed evaluation}~\cite{HoeLec20}.

\paragraph{Conventions} For $h$ of degree $d$ in $\field[\theta]$ and $f$ in $\field[\theta,z]$,
monic of degree $n$ in $z$, and for any $P$ in $\field[\theta,z]$, we
denote by $P \rem \genBy{h, f} \in \field[\theta,z]_{<(d,n)}$ the
polynomial obtained by reducing $P$ first by~$f$, then by $h$ (this is
the normal form of $P$ modulo $(h,f)$, if we see the latter as a
Gr\"obner basis for the lexicographic order induced by $\theta \prec
z$). Thus $P\rem\genBy{h,f}$ is a canonical lift of the
class of $P$ in $\field[\theta,z]/\genBy{h,f}$.
If $P\in\field[\theta,z]$, we use the
notation~$\bar P(z)$ to denote the class (projection) of~$P$
in~$\quotientL[z]$, where $\quotientL$ will be clear from the
context.

\subsection{Separable decomposition} \label{ssec:separabledecomp}

Let $p$ be the characteristic of the field $\field$ and let $f$ in
$\xRing$ be of degree $n$. The {\em separable decomposition} of
$f$ is
the set
\[
  \mathcal{S} = \{(h_1,e_1,\mult_1),\dots,(h_s,e_s,\mult_s)\},
  \quad \text{with~} 
  h_i \in \xRing \text{~and~} e_i,\mult_i \in \NN \text{~for
  all $i$},
\]
that satisfies the following properties, where we write
$f_i = h_i\!\left(x^{p^{e_i}}\right)^{\mult_i}$:
\begin{enumerate}
\item $f =  cf_1 \cdots f_s$ with $c\in\field\setminus\{0\}$
\label{step:sepdec1};
\item for all $i\ne j$ in $\{1,\dots,s\}$, $f_i$ and $f_j$ are coprime;\label{item:coprime}
\item for all $i$ in $\{1,\dots,s\}$, $h_i \in \xRing$ is separable,
monic and of positive degree $d_i$;
\item for all $i$ in $\{1,\dots,s\}$, $e_i=0$ (if $p=0$) or $e_i$ is
  in $\NN$ (if $p > 0$);
\item for all $i$ in $\{1,\dots,s\}$, $\mult_i$ is not divisible by
  $p$;
\item for all $i\ne j$ in $\{1,\dots,s\}$, $(e_i,\mult_i) \ne (e_j,\mult_j)$.
\end{enumerate}
The separable decomposition of $f$ can be computed in~$\softO{n}$ operations
in~$\field$ using an algorithm due to
Lecerf~\cite{Lec2008}. The special case when~$p=0$ recovers the more classical
\emph{square-free} factorization.

\subsection{Composition  over products of fields, modulo purely inseparable \texorpdfstring{$f$}{f}}\label{sec:overseparable}

Let $h$ be separable of degree $d$ in $\field[\theta]$, and consider
$f$ of the form $f=(z^{p^e}-c(\theta))^\mult \in\field[\theta,z]$, for
integers $e \in \NN$ and $\mult \in \NN_{>0}$, where $p$ is the
characteristic of $\field$. Given $A$ in $\field[\theta,z]_{<(d,n)}$
and $G$ in $\field[\theta,y]_{<(d,n)}$, with $n=\deg_z(f)=\mult p^e$,
we consider here the computation of $B=G(\theta, A) \rem \genBy{h, f}$.

This question is mapped to a univariate
composition problem with
coefficients in $\bigfield=\field [\theta]/\genBy{h}$: if we let $\bar
A$, $\bar G$, $\bar B$ and $\bar c$ be the projections of respectively
$A$, $G$, $B$  and $c$ in $\bigfield[z]$, $\bigfield[y]$,
$\bigfield[z]$ and $\bigfield$ (the
degree constraints show that $\bar A,\bar G,\bar B$ can be obtained
without any calculation from $A,G,B$, and conversely), then $\bar B =
\bar G(\bar A) \rem (z^{p^e}-\bar c)^\mult$ as an equality in
$\bigfield[z]$.

 When $h$ is irreducible, so that $\bigfield$ is a field, the algorithm of
 \cref{sec:composition_randomized:proof_inseparable_small} applies
 over
 $\bigfield$; as reported in \cref{thm:CompositionModuloInseparable},
 if $n=\deg(f) = \mult p^e$, the runtime is
 $\softO{dn^{\kappa}}$ operations in $\field$, coming from $\softO{
 (\mult
   p^e)^{\kappa}}=\softO{n^{\kappa}}$ times a factor in $\softO{d}$
 for the cost of arithmetic operations in~$\bigfield$.  However, we
 only assume $h$ separable, so that $\bigfield$ is a \emph{product of
   fields}.  The key difference is the presence of zero-divisors in
 $\bigfield$: a nonzero element of~$\bigfield$ is not necessarily
 invertible.  Since the procedures in
 \cref{sec:composition_randomized:proof_inseparable_small} use
 zero-tests and divisions, their direct application is not possible.

\subsubsection{Directed evaluation} 

The technique of \emph{directed evaluation},
due to van~der~Hoeven and Lecerf~\cite{HoeLec20}, is an
efficient version of the classical \emph{dynamic evaluation}
process~\cite{D5}. 

In dynamic evaluation, prior to each zero-test or
inversion, say by a quantity $q \in \bigfield$, the computation of
$h_1 = \gcd(q,h)$ gives the factorization $h=h_1 h_2$. Since $h$ is
separable, $h_1$ and $h_2$ are coprime, and $\bigfield$ can be
decomposed as the product $\bigfield_1 \times \bigfield_2$, with $q=0$
in $\bigfield_1=\field[\theta]/\langle h_1 \rangle$ and $q$ invertible
in~$\bigfield_2=\field[\theta]/\langle h_2 \rangle$. Under the dynamic
evaluation paradigm, the calculation can then be continued in two
branches, working modulo~$h_1$ and~$h_2$ separately.

In directed evaluation, the idea is rather to run the entire program
in a unique branch, then to apply the process recursively in residual
branches after reduction of input data modulo the corresponding
polynomial.  We do not detail the underlying techniques, for which we
refer to Sections~3 and~4 of~\cite{HoeLec20}, and simply apply their
\emph{panoramic evaluation} procedure~\cite[Algo.\,2]{HoeLec20}. It
takes as input a computation tree ${\mathcal T}$ over $\field$ 
(see \cref{sec:preliminaries}), a defining separable polynomial $h$ of degree
$d$ for $\bigfield$, and $\lambda=(\lambda_1,\ldots,\lambda_s)$ in
$\field[\theta]_{<d}^s$ (representing an input to ${\mathcal T}$ in $\bigfield^s$); it then returns a \emph{panoramic value},
defined as follows.

\begin{definition}[{\cite[Def.\,1 and Lem.\,2]{HoeLec20}}] \label{def:panoramic}
  Given an input $(h, \lambda, {\mathcal T})$ as above, a \emph{panoramic value} of ${\mathcal T}$
  at $\lambda$ is a set of pairs
  $\{(h_1,\varepsilon_1),\ldots,(h_t,\varepsilon_t)\}$, where
  \begin{itemize}
  \item $h_1,\dots,h_t$ are polynomials in $\field[\theta]$ that
    satisfy $h=h_1\cdots h_t$ (thus $ \bigfield \simeq
    \bigfield _1 \times \dots \times \bigfield _t$, with $\bigfield _i =
    \field[\theta]/\langle h_i\rangle$);
  \item for all $i$, $\varepsilon_i$ is in $\field[\theta]_{< d_i}^{\ell_i}$
   (representing an output in 
  $\bigfield _i ^{\ell _i}$), with $d_i = \deg(h_i)$ and $\ell_i$ in~$\NN$;
  \item   for all $1 \le i \le
  t$, let $h_{i,1},\dots,h_{i,k_i}$ be the factorization of $h_i$ into
  irreducibles. For $1 \le j \le k_i$, let $\bigfield_{i,j}$ be the
  the field $\field[\theta]/\langle h_{i,j}\rangle$, and denote by
  $\pi _{i,j}: \field[\theta] \to \bigfield_{i,j}$ the canonical
  projection $a \mapsto a \bmod h_{i,j}$ (the notation carries over to
  vectors over $\field[\theta]$).  Then~$\mathcal T$ is supposed to be
evaluable at  $\pi_{i,j}(\lambda)\in
  \bigfield _{i,j}^s$ for all $i,j$,
  and
  $\pi _{i,j}(\varepsilon_i) \in
  \bigfield_{i,j}^{\ell_i}$ is the result of evaluating ${\mathcal T}$
  (seen as a computation tree over~$\bigfield_{i,j}$) at $\pi_{i,j}
  (\lambda)$, 
  using the same branch of~$\mathcal T$ for all~$j$.
  \end{itemize}
\end{definition}
The application of this method requires that one uses computation
trees as the underlying computational model, which is the case
here~(\cref{sec:preliminaries}). Crucially, the 
cost overhead is then~$\softO{d}$~\cite[Thm.\,1]{HoeLec20}, i.e. similar (up to logarithmic
factors) to the one incurred if $h$ were irreducible.

\subsubsection{Algorithm} \label{subsec:algooverseparable}

With \algoName{algo:CompositionModuloInseparable-ProductOfFields}
we apply panoramic evaluation  (called \textproc{Panoramic} in our pseudocode)  to
\algoName{algo:CompositionModuloInseparable} for modular composition
over $\field$. Note that in addition to field elements, the latter
algorithm also takes two integers $e,\mult$ as input. Panoramic
evaluation can still be used in this context, since each choice of the
parameters~$e,\mult$ corresponds to a computation tree, to which the
techniques described above apply. This yields a factorization of $h$,
and performs the compositions modulo the corresponding factors;
the final result is then reconstructed using Chinese remaindering.

\begin{algorithm}
  \algoCaptionLabel{CompositionModuloInseparable-ProductOfFields}{h,c,e,\mult,A,G,r}
  \begin{algorithmic}[1]
    \Require
    \parbox[t]{0.8\textwidth}{
      $h$ separable of degree $d$ in $\field[\theta]$,

      $c$ in $\field[\theta]_{<d}$, 
      $e$ in $\NN$ and $\mult$ in $\NN_{>0}$ such that $f=(z^{p^e}-c)^{\mult}$ has degree $n=\mult p^e$, 

      $A \in \field[\theta,z]_{<(d,n)}$, 
      $G \in \field[\theta,y]_{<(d,n)}$, 
      $r \in \vecRing{n+\lceil n^{\eta} \rceil}$  
    }

    \Ensure $B=G(\theta,A) \rem \langle h,  f \rangle$, or \Fail 
    \State  \label{step:panoramic} \CommentLine{Splitting 
    $\bigfield \simeq  
\field[\theta]/\genBy{h_1} \times \dots \times \field[\theta]/\genBy{h_t}$ and reductions of $B$, accordingly, using 
\cite[Algo.\,2]{HoeLec20} }  

\Statex $\{(h_1 , B_1),\dots,\!(h_t, B_t)\} \gets$ \textproc{Panoramic}(\Call{algo:CompositionModuloInseparable}{}, $h,c, e, \mult, A, G,r$)
    \State\InlineIf{any of the $B_i$'s equals \Fail}{\Return \Fail}
    \State \Return $\textproc{ChineseRemaindering}((B_1,\dots,B_t),(h_1,\dots,h_t))$
  \end{algorithmic}
\end{algorithm}

\begin{proposition} \label{prop:insep-POF} 
  For a field~$\field$ of characteristic~$p$, given
  $h\in\field[\theta]$ separable of degree~$d$, $c$ in
  $\field[\theta]_{<d}$, integers $e$ in $\NN$ and $\mult$ in
  $\NN_{>0}$, $A$ in $\field[\theta,z]_{<(d,n)}$, $G$ in
  $\field[\theta,y]_{<(d,n)}$, $r$ in $\field^{n+\lceil n^\eta\rceil}$
  with $n=\mult p^e$ and~$\eta$ from \cref{eq:def-beta},
  \algoName{algo:CompositionModuloInseparable-ProductOfFields}
  uses $\softO{d (\mult p^e)^{\kappa}} = \softO{d n^{\kappa}}$
  operations in $\field$, with $\kappa < 1.43$ as in
  \cref{eq:def-gamma}.

  It returns either $G(\theta,A) \rem \langle h, f \rangle \in
  \field[\theta,z]_{<(d,n)}$ or \Fail, with $f=(z^{p^e}-c)^\mult$. If
  the entries of~$r$ are chosen uniformly and independently from a
  finite subset $S$ of $\field$, then the algorithm returns
  $G(\theta,A) \rem \langle h, f \rangle$ with probability at least
  $1- 2d n^4/\card{S}$.
\end{proposition}

\begin{proof}
Combined with our \cref{thm:CompositionModuloInseparable}, Theorem~1
in~\cite{HoeLec20} gives the runtime estimate. 
In the pseudocode, the
output of the panoramic evaluation is written as $\{(h_1,B_1),\dots,(h_t,B_t)\}$, where
$h_1\dotsm h_t$ is a factorization of $h$ (not necessarily into
irreducibles), and for all $i$, either $B_i \in \field[\theta,z]_{<
(d_i,n)}$  with $d_i=\deg(h_i)$,
or~$B_i=\Fail$. At the
level of computation trees, a flag such as $\Fail$ is obtained
by setting a dedicated output value to 1 (and 0 otherwise); call $\flag _i$
this value, for $1\le i \le t$. If $\flag_i=1$ (failure), we set $B_i=0$
by convention, so in the rest of this proof, $B_i$ is an element of
$\field[\theta,z]$ for all $i$.

We use the following notation: for $1\leq i\leq t$, the
irreducible factors of $h_i$ are written
$h_{i,1},\dots,h_{i,k_i}$. For $1 \le j \le k_i$, we then define
$\bar c_{i,j},\bar A_{i,j},\bar G_{i,j}$ by taking~$c, A, G$ modulo
$h_{i,j}$ and
seeing them over the field $\bigfield_{i,j} = \field[\theta]/\langle
h_{i,j}\rangle$, so $\bar c_{i,j}$ is in $\bigfield_{i,j}$, $\bar A_
{i,j}$ in
$\bigfield_{i,j}[z]$ and $\bar G_{i,j}$ in $\bigfield_{i,j}[y]$. The
elements in the vector $r$ are already in $\field$, and thus in
$\bigfield_{i,j}$. 
Finally, we let $\bar B_{i,j}$ be the polynomial obtained by taking
$B_i
\in \field[\theta,z]$ and projecting it to $\bigfield_{i,j}[z]$
through reduction modulo $h_{i,j}$, and we set $\flag_{i,j}=\flag_i$ (recall
that $\flag_i \in \field$ is either $0$ or $1$).

Then, from~\cref{def:panoramic}, the key property of the output of the first step is that for all
indices $i,j$, $\flag_{i,j}$ and $\bar B_{i,j}$ are the result of
calling
\algoName{algo:CompositionModuloInseparable}{} on input
$\bar c_{i,j},e,\mult,\bar A_{i,j},\bar G_{i,j},r$ over the field
$\bigfield_{i,j}$. This implies in particular that our algorithm
returns $\Fail$ if and only if the computation fails over one of the
fields $\bigfield_{i,j}$.
  
To quantify the probability of this event, we apply
\cref{thm:CompositionModuloInseparable} over all fields
$\bigfield_{i,j}$. For any given~$i,j$,
\cref{thm:CompositionModuloInseparable} shows that $\flag _{i,j}=1$ occurs
with probability at most $2n^4/\card{S}$. Since there are at most $d$
such indices $i,j$, the probability that this happens for at least one
pair of indices is at most $2d n^4/\card{S}$.
Assume none of the $\flag_{i,j}$'s is~1, so that the algorithm does
not return $\Fail$. Then, for all $i,j$, $\bar B_{i,j} \in
\bigfield_{i,j}[z]_{<n}$ is equal to $\bar G_{i,j}(\bar A_{i,j}) \rem
(z^{p^e}-\bar c_{i,j})^\mult$. In terms of bivariate polynomials, the
Chinese Remainder Theorem then implies that for all $i$, $B_i$ itself
is equal to $G(\theta, A) \rem \genBy{h_i,(z^{p^e}-c)^\mult} \in
\field[\theta,z]_{<(d_i,n)}$. In the last step of the algorithm, we
further apply the Chinese Remainder Theorem coefficient-wise to the
$B_i$'s with respect to $z$; this gives us $G(\theta,A) \rem \langle
h, (z^{p^e}-c)^\mult \rangle$ as a polynomial in
$\field[\theta,z]_{<(d,n)}$. The cost of this last step is in
$\softO{d\mult p^e}$, so the proof is complete.
\end{proof}

The complexity bound $\softO{d n^{\kappa}}$ in \cref{prop:insep-POF} indicates that the
overhead coming from operations modulo $h(\theta)$ is $\softO{d}$,
as pointed out previously.

\subsection{Untangling and tangling}  \label{ssec:untangling}

In this subsection, we give the main tools (tangling, untangling and bivariate reduction) that are needed for
reducing composition modulo powers of separable polynomials to the
situation of the previous subsection. The central results are due to
van~der~Hoeven and Lecerf~\cite{HoeLec17} with $f=h(x)^\ell$ and $h$ separable 
(\cref{subsubsec:tandu,subsubsec:bivred}). We slightly generalize them
to the case $f=h(x^{p^e})^\ell$ with~$e > 0$ (\cref{subsubsuc:untang,subsubsec:redg}).

\subsubsection{Tangling and untangling} \label{subsubsec:tandu}
The starting point is the following observation.
\begin{lemma}[{\cite[\S4.2]{HoeLec17}}] \label{lemma:phi}
  For $h$ of degree $d$ in $\field[x]$ and for a
  positive integer $\mult$, there exists a $\field$-algebra homomorphism
  \begin{align*}
    \psi_{h,\ell}:~~ \xRing/\genBy{ h(x)^\mult} & \to
    \field[\theta,z]/\genBy{h(\theta), (z-\theta)^\mult } \\ x & \mapsto z.
  \end{align*}
  If moreover $h$ is separable then $\psi_{h,\ell}$ is an isomorphism.
\end{lemma}
This homomorphism is a variant of the homomorphism~$\pi_{h,\ell}$
considered by
van der Hoeven and Lecerf, that maps $u\in\xRing/\genBy{ h(x)^\mult}$
to $u(z+\theta)\in\field[\theta,z]/\genBy{h(\theta), z^\mult }$. The
morphism $\psi_{h,\ell}$ is obtained by composing $\pi_{h,\ell}$ with
a translation~$z\mapsto z-\theta$. It turns out that $\psi_{h,\ell}$
is more convenient than~$\pi_{h,\ell}$ for our generalization in
\cref{subsubsuc:untang}. van der Hoeven and Lecerf
call
$\textproc{Untangling}(h,\mult,u)$ the algorithm which implements 
$\pi_{h,\ell}$; we use this terminology for the algorithm
that implements~$\psi_{h,\ell}$: given $u$ in $\field[x]_{< d\mult}$,
it computes
$U\in
\field[\theta,z]_{<(d,\ell)}$ such that $U=u(z)
\rem \genBy{h(\theta), (z-\theta)^\mult }$. When $h$ is separable, the
inverse
operation is called
$\textproc{Tangling}(h,\mult,U)$. Again, we use their terminology 
for the inverse of $\psi_{h,\ell}$.

\begin{lemma}\label{lemma:untang}
  $\textproc{Untangling}$ and $\textproc{Tangling}$ (when defined)
  take $\softO{d
    \mult}$ operations in~$\field$.
\end{lemma}
\begin{proof}This is mostly in~\cite{HoeLec17}. First, it is
easy to 
check that the algorithms~4.3 and~4.5 and the proofs of Prop.~4.6
and~4.10 of
that reference do not make
use of the separability of $h$.
Next, translation 
can be performed
in quasi-linear complexity over an arbitrary 
ring~\cite[Thm.~4.5]{Gerhard2004}, so that the complexity estimate is
unchanged for our variant of these algorithms.
\end{proof}

\subsubsection{Bivariate reduction} \label{subsubsec:bivred}
The computation of the composition $\polp(a)\rem h(x)^\mult$ for a
separable~$h$ reduces to computing~$\psi_{h,\ell}^{-1}(\polp(\psi_{h,\ell}
(a\bmod h(x)^\ell)))$, where
the inner composition is performed as a univariate composition in~$
\bigfield[z]$ modulo~$(z-\bar\theta)^\ell$, with $\mathbb L=\field
[\theta]/\genBy{h}$. 

In order to make use of the algorithms of the previous
sections to perform this composition, it is necessary to first reduce the degree
of~$\polp$.  
Denote by $A$ the canonical lift of 
$\psi_{h,\ell}
(a\bmod h(x)^\ell)$, and by $\bar A$ its projection in $\quotientL[z]$. 
The idea is to reduce $\polp$ modulo the characteristic
polynomial $(y-\bar A(\bar\theta))^\ell\in\quotientL[y]$ of $\bar A(z)$
modulo $(z-\bar\theta)^\ell$.

This is achieved in two steps. For  $h$ of degree $d$, 
we let $\alpha  \in \field[\theta]_{<d}$ be the canonical lift of
$\bar A(\bar\theta) \in \quotientL$.
First, one
computes the canonical lift of~$\psi_{\mu,\ell}
(g\bmod\mu^\ell)$,
where $\mu$ is an annihilating polynomial of~$\alpha \bmod h$. This produces~$\tilde G
(z,y)\in\field[z,y]_{<(\deg\mu,\ell)}$ such that
\[\tilde G(z,y)=\sum_{i=0}^{\ell-1}\tilde{G}_i(z)y^i=g(y)+\tilde U
(z,y)\mu
(z)+\tilde V(z,y)(y-z)^\ell\]
for some polynomials~$\tilde U,\tilde V$ in~$\field[z,y]$. 

Next, in view of~$\mu(\alpha)\equiv0\bmod h$, a modular composition of each of
the $\ell$ coefficients of this
polynomial~$\tilde G$ in~$y$
with~$\alpha(\theta)$ modulo~$h(\theta)$ gives $G(\theta,y)\in\field
[z,y]_
 {<(\deg\mu,\ell)}$ such that
\begin{equation}\label{eq:bivariatereduction}
G(\theta,y)=g(y)+U(\theta,y)h(\theta)+V(\theta,y)(y-\alpha
(\theta))^\ell,
\end{equation}
for some polynomials~$U,V$
in~$\field[\theta,y]$. \cref{eq:bivariatereduction} may also be read
as $\bar G(\bar A) = g(\bar A) \rem (z-\bar\theta)^\ell$ over~$\quotientL$.

These two steps are detailed in
\algoName{algo:BivariateReduction} below
  and correspond
to Steps~(2)-(4) of \cite[Algo.\,4.2]{HoeLec17}.
The runtime and probability analyses are new; they are based on the
results of the previous sections.

\begin{algorithm}
  \algoCaptionLabel{BivariateReduction}{h, \mult, \alpha, \polp, r}
  \begin{algorithmic}[1]
    \Require $h$ separable, monic, of degree $d$ in $\field[\theta]$,
    $\mult$ in $\NN_{>0}$, $\alpha$ in $\field[\theta]_{<d}$, $\polp$
    in $\yRing$, $r$ in $\vecRing{d+\lceil d^{\eta} \rceil}$ 
    \Ensure $G(\theta,y)= \polp(y) \rem \genBy{ h(\theta), 
    (y-\alpha(\theta))^\mult}\in
    \field[\theta,y]_{<(d,\mult)}$, or \Fail 
    \State
    \CommentLine{Either $\mu=\Fail$, or $\mu$ is nonzero in
      $\field[\gamma]_{\le 4d}$ and $\mu(\alpha) \equiv 0 \bmod h$}
    \Statex $\mu\gets \Call{algo:AnnihilatingPolynomial}{h, \alpha,
      r}$ \Comment{\cref{algo:AnnihilatingPolynomial}}
  
    \Statex    \InlineIf{$\mu=\Fail$}{\Return \Fail}
    \State $\tilde G \gets \textproc{Untangling}(\mu,\mult,\polp \rem \mu^\mult)$
    \Comment{$\tilde G(\gamma,z) \in \field[\gamma,y]_{<(\deg(\mu),\mult)}$, \cref{lemma:untang}}
    \State Write $\tilde G=\sum_{0 \le i < \mult} \tilde G_i(\gamma) y^i$ \Comment{$\tilde G_i \in \field[\gamma]_{<\deg(\mu)}$}
    \For{$i=0,\dots,\mult-1$}
    \Statex \hspace*{0.4cm}$G_i \gets \Call{algo:ModularCompositionBaseCase}{h,
    \alpha, \tilde G_i, r}$ \Comment{$G_i= \tilde G_i(\alpha) \rem h$ or 
      $\Fail$, \cref{algo:ModularCompositionBaseCase}}\label{loop:bred-gren}
    \Statex    \hspace*{0.4cm}\InlineIf{$G_i=\Fail$}{\Return \Fail}
    \EndFor
    \State $G \gets \sum_{0 \le i < \mult} G_i y^i$ \Comment{$G$ is in $\field[\theta,y]_{<(d,\mult)}$}
    \State $\Return$ $G$ \label{laststep:bred-gen}
  \end{algorithmic}
\end{algorithm}

\begin{lemma}\label{lemma:BivariateReduction} 
  Given $h$ in $\field[\theta]$ monic, separable and of degree $d$,
  $\alpha$ in $\field[\theta]_{<d}$, $g$ in $\yRing$, $r$ in
  $\field^{d+\lceil d^{\eta} \rceil}$ with $\eta$ from
  \cref{eq:def-beta}, and $\mult$ in $\NN_{>0}$,
  \algoName{algo:BivariateReduction}{} uses $\softO{\deg(g) +
    d^{\kappa} \mult}$ operations in $\field$ with $\kappa < 1.43$ as
  in \cref{eq:def-gamma}, and returns either $\polp \rem \genBy{h
  (\theta),
    (y-\alpha(\theta))^\mult}$ or \Fail. If the entries of~$r$ are
    chosen
  uniformly and independently from a finite subset $S$ of $\field$,
  then the algorithm returns $\polp \rem \genBy{ h, (y-\alpha)^\mult
  }$ with probability at least $1-6(\mult+1)d^2/\card{S}$.
\end{lemma}
\begin{proof}
The reduction of~$g\bmod\mu^\ell$ is justified by the fact that
$\mu(a)^\ell=0\bmod h^\ell$. The correction of the rest of the algorithm
when \cref{laststep:bred-gen} is reached follows from the discussion
above.

  Since $h$ is separable, \cref{prop:composition-separable} applies;
  it shows that the first step computes an annihilating polynomial for
  $\alpha$ modulo $h$ with probability at least $1-6d^2/\card{S}$. It
  also shows that each call to
  \algoName{algo:ModularCompositionBaseCase}{} succeeds with at least
  the same probability. Altogether, the probability of success of the
  whole algorithm is thus at least $1-6(\mult+1)d^2/\card{S}$.
  
  By \cref{cor:annihilatingbasecase}, the first step
  uses~$\softO{d^{\kappa}}$ operations in $\field$. Since $\deg(\mu)$
  is in~$\bigO{d}$, computing $g \rem \mu^\mult$ takes $\softO{\deg(g)
    + d\mult}$ operations in $\field$, and \cref{lemma:untang} shows
  that deducing $\tilde G$ takes a further $\softO{ d\mult}$
  cost. Finally, by \cref{prop:algo-mod_comp}, each pass in the loop
  at \cref{loop:bred-gren} takes $\softO{d^{\kappa}}$ operations, so
  that the overall runtime is $\softO{\deg(g) + d^\kappa\mult}$.
\end{proof}

\subsubsection{General Tangling and Untangling}\label{subsubsuc:untang}
In fields of positive characteristic, the isomorphism of
\cref{lemma:phi} and the complexity of its realization generalize as
follows.
\begin{proposition}\label{prop:Phi} Let $f=h(x^{p^e})^\ell$ be of degree~$n$,
with $h$ of degree~$d$ in $\field[x]$, and $\field$ of
characteristic~$p$ ($e=0$ if $p=0$). There exists a $\field$-algebra homomorphism

  \begin{align*}
    \Psi_{h,\ell}:~~  \xRing/\genBy{f}& \to \field[\theta,z]/\genBy{h
    (\theta),
    (z^{p^e}-\theta)^\mult}\\    
    x & \mapsto  z.
  \end{align*}
  If moreover $h$ is separable then $\Psi_{h,\ell}$ is an isomorphism.
  Applying $\Psi_{h,\ell}$ or its inverse when the latter is defined
  takes
  quasi-linear time $\softO{n}=\softO{d\mult p^e}$ over $\field$.
\end{proposition}
\begin{proof}
  Write $\quotient= \xRing/\genBy{f}$ and
  $\quotientB=\field[\theta,z]/\genBy{h(\theta), (z^{p^e}-\theta)^\mult}$.
  When $h$ is separable, we prove that
  the minimal polynomial of $z$ in
  the $\field$-algebra $\quotientB$ is $f$.  This implies that
  $\quotient$ is $\field$-isomorphic (as a $\field$-algebra) to the
  subalgebra of $\quotientB$ generated by $z$. Since $\quotientB$ has
  $\field$-dimension $n=\deg(f)$, this subalgebra is $\quotientB$
  itself, and the first claim will follow.

  To determine the minimal polynomial of~$z$, we can work in
  $\overline\quotientB =
  \Kbar[\theta,z]/\genBy{h(\theta),(z^{p^e}-\theta)^\mult}$,
  where~$\Kbar$ is an algebraic closure of $\field$. If we let
  $\xi_1,\dots,\xi_d$ be the roots of $h$ in $\Kbar$ (which are
  pairwise distinct), then $\overline\quotientB$ is isomorphic, as a
  $\Kbar$-algebra, to the product
  \[\Kbar[\theta,z]/\genBy{\theta-\xi_1, (z^
  {p^e}-\xi_1)^\mult}\times
  \dots \times \Kbar[\theta,z]/ \genBy{\theta-\xi_d,
    (z^{p^e}-\xi_d)^\mult}.\] The minimal polynomial of $z$ in the
  $i$th factor above is $\mu_i=(x^{p^e}-\xi_i)^\mult$ for $1\leq i
  \leq d$.  These polynomials are pairwise coprime: since $t \mapsto
  t^{p^e}$ is a bijection in $\Kbar$, $\mu_i$ has a unique root
  in $\Kbar$, which is the $p^e$-th root of $\xi_i$, and these
  roots are pairwise distinct, since the $\xi_i$'s are. As a result, the minimal
  polynomial of $z$ in $\overline\quotientB$, or equivalently in
  $\quotientB$, is the product $\mu_1 \cdots \mu_d=f$.

  For the second claim, we take $a$ in $\xRing$ of degree less than~$n$, and  
  write it as $a=\sum_{0 \le i < p^e} a_i(x^{p^e}) x^i$,
  with all $a_i$'s of degree less than $n/p^e=d\mult$. Then,
  \begin{alignat}{3}
    \Psi_{h,\ell}(a\bmod f) &\equiv \sum_{0 \le i < p^e} a_i(z^{p^e})
    z^i&\quad
    \bmod \langle h
    (\theta), (z^{p^e}-\theta)^\mult\rangle, \nonumber\\
    &\equiv \sum_{0 \le i < p^e} \tilde A_i(\theta,z^{p^e})
    z^i&\quad\bmod
    \langle h
    (\theta), (z^{p^e}-\theta)^\mult\rangle,     \label{eq:phia_red}
  \end{alignat}
  where $\tilde A_i(\theta,z) = a_i(z) \rem \langle h(\theta),
  (z-\theta)^\mult\rangle$ is in $\field[\theta,z]_{<(d,\mult)}$;
  these degree bounds show that the expression in \cref{eq:phia_red}
  is indeed reduced modulo $\langle f(\theta),
  (z^{p^e}-\theta)^\mult\rangle$.  Each $\tilde A_i = \psi_{h,\ell}
  (a_i)$
  can
  be computed in time $\softO{d\mult}$ by \cref{lemma:untang}, so that
  one application of $\Psi_{h,\ell}$ takes $\softO{d\mult p^e}=\softO
  {n}$
  operations in $\field$, as claimed.

  Conversely, any element $B$ in $\field[\theta,z]_{<(d,\mult p^e)}$
  can be written as in \cref{eq:phia_red}, for some $\tilde B_i$'s in
  $\field[\theta,z]_{<(d,\mult)}$. Applying $\psi_{h,\ell}^{-1}$ to
  each of
  them allows us to recover $b =\Psi_{h,\ell}^{-1}(B)$, by
  reversing the steps
  above.  The cost analysis is similar to the one for $\Psi_{h,\ell}$.
\end{proof}
We call $\textproc{Untangling-General}(h,e,\mult,a)$ the
algorithm outlined in this proof that applies $\Psi_{h,\ell}$ to (the
class
modulo $f$ of) $a \in \xRing_{<n}$, and returns the canonical lift
of
$\Psi_{h,\ell}(a \bmod f)$ to $\field[\theta,z]_{<(d, \mult p^e)}$;
equivalently,
$A(\theta,z)=a(z) \rem \genBy{h(\theta), (z^{p^e}-\theta)^\mult}$.
For $B$ in $\field[\theta,z]_{<(d,\mult p^e)}$, the inverse operation
is written $\textproc{Tangling-General}(h,e,\mult,B)$.

\subsubsection{Main reduction}\label{subsubsec:redg}

A more general form of bivariate reduction is needed in
\cref{ssec:modulo_powers}.
With $h$ of degree~$d$ as
before, given $g$ in $\field[y]$ and now a bivariate $A$ in
$\field[\theta,z]_{<(d,\mult p^e)}$, the aim is to reduce the degree
of~$g$ before performing the composition 
in~$\quotientL[z]$ modulo
$(z^{p^e}-\bar\theta)^\ell$ with $\quotientL=\field[\theta]/\genBy{h}$. Denoting by $\bar A$
the projection of $A$ in $\quotientL[z]$, the idea is to compute $\bar
G=\polp \rem \chi_{\bar A}$ in $\quotientL[z]$, where $\chi_{\bar A}
\in \quotientL[y]$ is the characteristic polynomial of $\bar A \in
\quotientL[z]$ in the extension $\quotientL \to
\quotientL[z]/\langle (z^{p^e}-\bar\theta)^\mult \rangle$.
Thus, $\bar G\in \quotientL[y]$ has degree less than $\mult p^e$; its
canonical lift $G \in \field[\theta,y]_{<(d,\mult p^e)}$ is the
output.

 The computation of $\polp\rem\chi_{\bar A}$ is made easy by
an explicit formula for the characteristic polynomial~$\chi_{\bar A}$. In the following lemma,
we let $\sigma: \quotientL
\to \quotientL$ be the $p^e$th-power operator; we write the image
of $\Lambda \in \quotientL$ as $\Lambda ^\sigma$.  This
notation is extended to the coefficient-wise action on
polynomial rings over
$\quotientL$.

\begin{lemma} \label{lem:charpolyextension}
  The characteristic polynomial of $\bar A$ relative to the extension
  $\quotientL \to \quotientL[z]/\langle (z^
  {p^e}-\bar\theta)^\mult \rangle$ is $\chi_{\bar A}=(y^
  {p^e}-\bar\alpha)^\mult
  \in \quotientL[y]$, where $\bar\alpha = {\bar A}^\sigma(\bar\theta)
  \in
  \quotientL$.
\end{lemma}
\begin{proof}
  The characteristic polynomial $\chi_{\bar A}$ can be computed
  relative to the extension $\quotientL^*\to \quotientL^*[z]/\genBy{
    (z^{p^e}-\bar\theta)^\mult}$, where we set $\quotientL^* =
  \quotientL[w]/\langle w^{p^e}-\bar\theta\rangle$.
  In $\quotientL^*[z]$, we have the factorization
  \[(z^{p^e}-\bar\theta)^\mult = (z^{p^e}-w^{p^e})^\mult = (z-w)^
  {\mult
  p^e},\]
   so the characteristic
  polynomial of $\bar A$ in
  $\quotientL^*[z]/\genBy{(z^{p^e}-\bar\theta)^\mult}$ is 
  \[(y-\bar
A(w))^{\mult p^e}=(y^{p^e}-\bar A(w)^{p^e})^\mult=(y^{p^e}-\bar
A^\sigma(\bar\theta))^\mult.\qedhere\]
\end{proof}

The reduction of $\polp$ by this characteristic polynomial is
described in
\algoName{algo:MainReduction}. First, the canonical lift
$\alpha\in\field
[\theta]_{<d}$ of $\bar\alpha\in\quotientL$ from
\cref{lem:charpolyextension} is computed. Next, in 
\cref{algo:genbivred:gi}, the polynomial~$\polp$ is rewritten as a
polynomial in~$y$ of degree less than~$p^e$, with coefficients~$g_i
(y^{p^e})$. Each of these polynomials~$g_i(y)$ can then be reduced
modulo~$\genBy{h,(y-\alpha)^\ell}$ by \algoName{algo:BivariateReduction},
producing a polynomial~$G_i(\theta,y)$ (\cref{algo:modpower:loop}).
Thus, $G_i(\theta,y)\equiv g_i(y)\bmod \genBy{h,(y-\alpha)^\ell}$, whence $G_i(\bar \theta, y^
{p^e})\equiv g_i(y^{p^e})\bmod\chi_{\bar A}$. Recombining these
coefficients yields~$G(\theta,y)$ such that $G(\bar \theta, y)\equiv g
(y)\bmod\chi_{\bar A}$. Finally, since $\chi_{\bar A}(\bar A)\equiv 0$ in
$\quotientL[z]/\genBy{(z^{p^e}-\bar\theta)^\ell}$, it
follows
that $G(\theta,A)\equiv
  g(A)\bmod \genBy{h(\theta),(z^{p^e}-\theta)^\mult}$. 

\begin{algorithm}
  \algoCaptionLabel{MainReduction}{h,e,\mult,A,\polp,r}
  \begin{algorithmic}[1]
  \Require $h$ separable, monic, of degree $d$ in $\xRing$, $e$ in 
  $\NN$, $\mult$ in $\NN_{>0}$, $A$ in  $\field[\theta,z]_{<(d,\mult p^e)}$, $\polp$ in $\yRing$, $r$ in
  $\vecRing{d+\lceil d^{\eta} \rceil}$ 

  \Ensure $G\in \field[\theta,y]_{<(d,\mult p^e)}$ such that 
   $G(\theta, A)\equiv\polp(A) \bmod \genBy{h(\theta),(z^{p^e}-\theta)^\ell}$, or \Fail

  \State Write $A = \sum_{0 \le i < \mult p^e} A_i z^i$ \Comment{$A_i \in \field[\theta]_{<d}$}

  \State \CommentLine{Compute $\alpha$ s.t.~the characteristic
  polynomial of~$\bar A$ is
    $(y^{p^e}-\alpha)^\mult$ (see \cref{lem:charpolyextension})} \Statex
    $\alpha \gets \sum_{0
    \le i <
    \mult p^e} {A_i}^{p^e} \theta^i$; $\alpha \gets \alpha \rem h$
  \Comment{$\alpha \in \field[\theta]_{<d}$} \label{algo:modpower:k} \State Write $\polp=\sum_{0
    \le i < p^e} \polp_i(y^{p^e}) y^i$
  \label{algo:genbivred:gi}
  \Comment{$\deg(g_i) \le  \deg(g)/p^e$}
  \For{$i=0,\dots,p^e-1$}
  \Statex \hspace*{0.4cm} {$G _i \gets \Call{algo:BivariateReduction}{h,\mult,\alpha,\polp_i,r}$} 
  \label{algo:modpower:loop}  
  \Comment{$G_i \in \field[\theta,y]_{<(d,\mult)}$}
  \EndFor
  
  \State $G \gets \sum_{0 \le i < p^e} G_i(\theta,y^{p^e}) y^i$ \label{untangredgend}
  \Comment{$G \in \field[\theta,y]_{<(d,\mult p^e)}$}

  \State \Return $G$
  \end{algorithmic}
\end{algorithm}

\begin{proposition}\label{lemma:algo-gen-bivar-red} Given $h$
  separable, monic, of degree $d$ in $\field[x]$, $e$ in $\NN$,
  $\mult$ in $\NN_{>0}$, $A$ in $\field[\theta,z]_{<(d,\mult p^e)}$,
  $\polp$ in $\yRing$, and $r$ in $\vecRing{d+\lceil d^{\eta} \rceil}$,
  \algoName{algo:MainReduction}{} uses $\softO{\deg(g) +
    n^\kappa}$ operations in $\field$, with $n=d\mult p^e$ and $\kappa
  < 1.43$ as in \cref{eq:def-gamma}. It returns
  $G\in\field[\theta,y]_{<(d,\mult p^e)}$ such that~$G(\theta,A)\equiv
  g (A)\bmod \genBy{h(\theta),(z^{p^e}-\theta)^\mult}$, or~\Fail.

  If the entries of $r$ are chosen
  uniformly and independently from a finite subset $S$ of $\field$,
  then the algorithm returns $G$ with probability at least
  $1-6(\mult+1)d^2p^e/\card{S}$.
\end{proposition}
\begin{proof}
The correction of the algorithm when it does not return \Fail\ follows
from the
discussion above.  

Working coefficient-wise, since $e=\bigO{\log(p^e)}$ the computation of
$\alpha$ at \cref{algo:modpower:k} takes $\softO{\mult p^e}$
operations on polynomials modulo $h$ of degree~$d$, so $\softO{n}$
operations in $\field$; reducing it modulo $h$ has the same complexity
bound.  The cost is thus governed by the loop, which uses
$\softO{\deg(g) + d^\kappa \mult p^e}=\softO{\deg(g) +
  (n/d)d^{\kappa}}$ operations by \cref{lemma:BivariateReduction}.
The latter also allows us to quantify the probability of success: each
of the $p^e$ calls to \algoName{algo:BivariateReduction}{} succeeds
with probability at least $1-6(\mult+1)d^2/\card{S}$.
\end{proof}

\subsection{Composition modulo powers}  \label{ssec:modulo_powers}

We now consider $f =h(x^{p^e})^\mult$, with $h$ separable of degree $d$
and integers $e,\mult$, with $\mult$ positive and not divisible by $p$
(and $e=0$ if $p=0$); the degree of $f$ is $n=d\mult p^e$.
\algoName{algo:ModularCompositionModuloPower}{} computes $\polp(a)
\rem f$, extending to $e\neq0$ the approach of van der Hoeven and
Lecerf~\cite{HoeLec17} outlined in 
\cref{subsubsec:bivred}.

We first compute $A(\theta,z)=a(z) \rem \genBy{h(\theta),
  (z^{p^e}-\theta)^\mult}$; this is done using the general untangling
operation of \cref{subsubsuc:untang}. The reduction of the degree of
$g$ is done by \algoName{algo:MainReduction}{}, giving $G$
in $\field[\theta,y]_{<(d,\mult p^e)}$, such that $G(\theta,A) \equiv
g(A) \bmod \genBy{h(\theta), (z^{p^e}-\theta)^\mult}$; the construction
of $A$ then implies $G(\theta,A) \equiv g(a(z)) \bmod \genBy{h(\theta),
  (z^{p^e}-\theta)^\mult}$. The quantity $B=G(\theta,A)\rem
\genBy{h(\theta), (z^{p^e}-\theta)^\mult}$ is obtained by
\algoName{algo:CompositionModuloInseparable-ProductOfFields}{} of \cref{subsec:algooverseparable}.
We finally apply the general tangling procedure of
\cref{subsubsuc:untang} to $B$; since tangling is a $\field$-algebra
isomorphism, the outcome is $b=g(a) \rem h(x^{p^e})^\mult$.  

\begin{algorithm}
  \algoCaptionLabel{ModularCompositionModuloPower}{h,e,\mult,a,\polp,r}
  \begin{algorithmic}[1]
  \Require $h$ separable, monic, of degree $d$ in $\xRing$, $e$ in
  $\NN$, $\mult$ in $\NN_{>0}$, such that $f =h(x^{p^e})^\mult$ has degree $n = d\mult p^e$, $a$ in
  $\xRing_{<n}$, $\polp$ in $\yRing$, $r$ in $\vecRing{\rho+\lceil
    \rho^{\eta} \rceil}$ where $\rho=\max(d,n/d)$ \Ensure
  $b=\polp(a)\rem f$ or \Fail
  \label{algo:vdHL}
  \State \CommentLine{Conversion of $a\in\xRing$ to a bivariate polynomial (\cref{prop:Phi})}
  \Statex $A \gets \textproc{Untangling-General}(h,e,\mult,a)$ \label{genuntanga} 
 \Comment{$A \in \field[\theta,z]_{<(d,\mult p^e)}$}
  \State \CommentLine{Reduction of $g$ modulo the characteristic polynomial of $\bar A$ 
  (\cref{lemma:algo-gen-bivar-red})}
  \Statex $G \gets\Call{algo:MainReduction}{h,e,\ell,A,g,(r_k)_{0\le k
  < d+\lceil
    d^{\eta} \rceil}}$\label{callgenbivred}
  \Comment{$G \in \field[\theta,y]_{<(d,\mult p^e)}$}
  \Statex \InlineIf{$G=\Fail$}{\Return \Fail}
  \State \CommentLine{Modular composition, $B= G(\theta,A) \rem \genBy{ h(\theta),
  (z^ {p^e}-\theta)^{\mult}}  \in \field[\theta,z]_{<(d,\mult p^e)}$ or $\Fail$}\label{algo:modpower:c}
  \Statex $B \gets 
  \Call{algo:CompositionModuloInseparable-ProductOfFields}{h, \theta, e, \mult,A, G,(r_k)_{0\le k < \frac{n}{d}+\lceil
    (\frac{n}{d})^{\eta} \rceil}}$
  \Statex \InlineIf{$B=\Fail$}{\Return \Fail}
  \State \CommentLine{Recovery of $b$ over $\field$ (\cref{prop:Phi})}
  \Statex $b \gets \textproc{Tangling-General}(h,e,\mult,B)$ \label{gentangb}
  \State \Return $b$
  \end{algorithmic}
\end{algorithm}

\begin{proposition} \label{prop:ModularCompositionModuloPower} 
  For a field~$\field$ of characteristic~$p$, given $h$ separable,
  monic and of degree $d$ in $\field[x]$, integers $e$ in $\NN$ and
  $\mult$ in $\NN_{>0}$, $a$ in $\field[x]_{<n}$, $g$ in $\field[y]$,
  $r$ in $\field^{\rho+\lceil \rho^\eta\rceil}$, with $n=d \mult p^e$,
  $\rho=\max(d,\mult p^e)$ and~$\eta$ from \cref{eq:def-beta},
  \algoName{algo:ModularCompositionModuloPower}{} uses $\softO{\deg(g)
    + n^\kappa}$ operations in $\field$, with $\kappa< 1.43$ as in
  \cref{eq:def-gamma}, and returns $\polp(a)\rem h(x^{p^e})^\mult$
  or~\Fail.

  If the entries of $r$ are chosen uniformly and independently from a
  finite subset $S$ of $\field$, then the algorithm returns $\polp(a)
  \rem h(x^{p^e})^\mult$ with probability at least $1-
  (2n^4+12n^2)/\card{S}$.
\end{proposition}
\begin{proof}
  That the output of the algorithm is $\polp(a)\rem h(x^{p^e})^\mult$
  or $\Fail$ follows from the previous discussion.
  By \cref{prop:Phi}, with $n=d\mult p^e$, the first and last step
  both take $\softO{n}$ operations in $\field$. 
  \Cref{lemma:algo-gen-bivar-red} shows that \cref{callgenbivred}
  takes $\softO{\deg(g) + n^\kappa}$ operations in $\field$.
  Finally, \cref{prop:insep-POF} shows that
  \cref{algo:modpower:c} takes $\softO{d (\mult p^e)^{\kappa}}=\softO{d(n/d)^{\kappa}}$
  operations in $\field$, so the runtime estimate is proved.

  The steps that may output $\Fail$ are the computation of~$G$ at
  \cref{callgenbivred} and that of $B$ at
  \cref{algo:modpower:c}. By \cref{lemma:algo-gen-bivar-red}, the
  former happens with probability at most 
   $6(\mult+1)d^2p^e/\card{S} \leq 12 n^2/\card{S}$; by
  \cref{prop:insep-POF}, the latter happens with probability
  at most $2d(\mult p^e)^4/\card{S} \leq 2n^4/\card{S}$.
\end{proof}

\subsection{Main algorithm and its analysis}\label{ssec:mainalgo}

We can now give \algoName{algo:ModularComposition} performing modular
composition with general polynomials, and prove~\cref{thm:intro}.

The separable decomposition $f_1\dotsm f_s$ of~$f$ allows us to reduce
the problem to compositions modulo the $f_i$'s, which are powers of
polynomials as in \cref{ssec:modulo_powers}. The polynomials $a$ and
$\polp$ are first reduced so that the compositions modulo the $f_i$'s
are called with inputs of appropriate degrees, then the result
$b=\polp(a) \rem f$ is recovered using Chinese remaindering.  The
number of random elements in $\field$ we use is an \emph{a priori}
bound that
can be refined if the separable decomposition of $f$ is known.

\begin{algorithm}
  \algoCaptionLabel{ModularComposition}{f,a,\polp,r}
  \begin{algorithmic}[1]
    \Require $f$ of degree $n$ in $\xRing$, $a$ in $\xRing_
    {<n}$, $\polp$ in $\yRing_{<n}$,
$r \in \vecRing{n+\lceil n^{\eta} \rceil}$
    \Ensure $b=\polp(a)\rem f$ or $\Fail$
    \State \CommentLine{Decomposition of $f$ \cite[Algo.\,3]{Lec2008}}
    \Statex $(h_1,e_1,\mult_1),\dots,(h_s,e_s,\mult_s) \gets \textproc{\sc SeparableDecomposition}(f)$
    \Comment {$h_i$ monic of degree $d_i$ in $\xRing$}
    \State $(f_1,\dots,f_s) \gets (h_1(x^{p^{e_1}})^{\mult_1},\dots,h_s(x^{p^{e_s}})^{\mult_s})$
    \Comment {$f_i$ of degree $n_i=d_i \mult_i p^{e_i}$ in $\xRing$}
    
    \State     \CommentLine {Degree reduction, $\deg(a_i) < n_i$}
    \Statex $(a_1,\dots,a_s) \gets (a \rem f_1,\dots,a \rem f_s)$
    \State \CommentLine{Annihilating polynomials of the $a_i$ modulo
    $f_i$} \label{step:modcomp:loop}
    \Statex \InlineFor{$i=1,\dots,s$} 
    \Statex \hspace*{0.4cm}Write $a_i = \sum_{0 \le k < n_i} a_{i,k} x^k$ 
         
  \Statex \hspace*{0.4cm}$\alpha _i \gets \sum_{0 \le k < n_i} {a_
  {i,k}}^{p^{e_i}} x^k$; 
      $\alpha _i \gets \alpha _i \rem h_i$
\Statex \hspace*{0.4cm}$\mu _i \gets \Call{algo:AnnihilatingPolynomial}{h_i, \alpha _i,(r_k)_{0\le k < d_i+\lceil
    d_i^{\eta} \rceil}
}$ 
    \Comment{$\mu_i(\alpha_i)\equiv 0\bmod h_i$, $\deg(\mu_i) \le 4d_i$}
    \Statex \hspace*{0.4cm}\InlineIf{$\mu_i=\Fail$}{\Return \Fail}
    \Statex \hspace*{0.4cm}$\chi_i \gets \mu_i(y^{p^{e_i}})^{\mult_i}$
    \Comment{$\chi_i(a_i) \equiv 0 \bmod f_i$, $\deg(\chi _i) \le 4 n_i$}
    
    \State\CommentLine{Degree reduction, $\deg(\polp_i) < 4 n_i$}
    \Statex $(\polp_1,\dots,\polp_s) \gets (\polp \rem \chi
    _1,\dots,\polp \rem \chi _s)$\label{step:modcomp:polpi}
    \State \CommentLine{Modular compositions, either  $b_i \equiv\polp(a) \bmod
    f_i$ or $\Fail$} \label{step:modcomp:loop2}
    \Statex \InlineFor{$i=1,\dots,s$}
    \Statex  
    \hspace*{0.4cm}$\rho_i\gets \max(d_i,{n_i}/{d_i})$
    \Statex  
    \hspace*{0.4cm}$b_i \gets  \Call{algo:ModularCompositionModuloPower} {h_i,e_i,\mult_i,a_i,\polp_i,(r_k)_{0\le k < \rho_i+\lceil
    \rho_i^{\eta} \rceil}}$
    \Statex \hspace*{0.4cm}\InlineIf{$b_i=\Fail$}{\Return \Fail}
     \State \Return $\textproc{ChineseRemaindering}((  b_1,\dots, b_s), (f_1,\dots,f_s))$  \label{finalCRT}
  \end{algorithmic}
\end{algorithm}

\begin{proof}[Proof of \cref{thm:intro}]
  First we prove correctness. Suppose that none of the subroutines
  returns~$\Fail$; we show that the output is $\polp(a)\bmod f$. 

  Using the same notation for $p^e$-th powering as in
  \cref{lem:charpolyextension}, at the $i$-th pass in the loop at
  \cref{step:modcomp:loop}, the polynomial $\mu_i$ satisfies
  $\mu_i(\alpha _i)\equiv 0 \bmod h_i$, with $\alpha _i = a_i^\sigma$
  (that is, the coefficients of $\alpha _i$ are the $p^{e_i}$-th
  powers of those of $a_i$). Raising this equality to the power
  $\mult_i$ gives $\mu_i^{\mult_i}(\alpha _i) \equiv 0\bmod
  h_i^{\mult_i}$. Evaluation at $y^{p^{e_i}}$ using the facts that
  $\alpha _i(y^{p^{e_i}})={a_i}^{p^{e_i}}$ and $\chi _i= \mu
  _i(y^{p^{e_i}})^{\mult_i}$ finally gives $\chi_i(a_i)\equiv0 \bmod
  f_i$.  The degree bound $\deg(\mu_i) \le 4 d_i$ follows from the
  specifications of \algoName{algo:AnnihilatingPolynomial}{}, and the
  degree bound for $\chi _i$ follows.

  In the second for loop at \cref{step:modcomp:loop2}, $b_i$ satisfies
  $b_i \equiv \polp_i(a_i) \bmod h_i(x^{p_{e_i}})^{\mult_i}\equiv
  \polp_i(a_i) \bmod f_i$. Since~$g_i = g \rem \chi_i$, and $\chi _i$
  cancels $a_i$ modulo $f_i$, $b_i$ is also equal to $\polp(a_i) \rem
  f_i$, and thus to $\polp(a) \rem f_i$. It follows that the return
  value, obtained by Chinese remaindering, is indeed $\polp(a) \rem
  f$.

  Next, we bound the overall cost. The call to $ \textproc{\sc
  SeparableDecomposition}(f)$
  takes $\softO{n}$ operations in $\field$~\cite[Prop.\,5]{Lec2008}.
  Using repeated squaring,
  the polynomials $f_1,\dots,f_s$ can be computed in quasi-linear time
  as well, and the same holds for the remainders $a_1,\dots,a_s$.

  Consider a fixed index $i$ in the loop at \cref{step:modcomp:loop},
  and denote $d_i \mult_i p^{e_i}$ by $n_i$. Working coefficient-wise,
  computing $\alpha_i=a_i^\sigma$ takes $\softO{ n_i}$ operations
  since $e_i=\bigO{\log(n_i)}$, and reducing it modulo $h_i$ has the
  same complexity bound. By \cref{prop:algo-mod_comp},
  \algoName{algo:AnnihilatingPolynomial}{} uses $\softO{d_i^\kappa}$
  operations in $\field$.  If it does not fail, $\chi_i$ is then
  deduced in $\softO{n_i}$ operations again, hence the cost of the
  loop is $\softO{n^\kappa}$.

  When~\cref{step:modcomp:polpi} is reached, since all $\chi_i$'s
  have respective degrees at most $4 n_i$, fast
  multiple remaindering gives the polynomials $\polp_i$ in $\softO{n}$
  operations, with $\deg(\polp_i) < 4n_i$.  Then, by
  \cref{prop:ModularCompositionModuloPower}, each call to
  \algoName{algo:ModularCompositionModuloPower}{} uses $\softO
  {n_i^\kappa}$ operations in $\field$, so their total
  cost is $\softO{n^\kappa}$ again. Finally, the cost of the last step
  (if reached) is $\softO{n}$. Altogether, the cost is
  $\softO{n^\kappa}$, as claimed.

  It remains to discuss the probability of failure. By
  \cref{prop:composition-separable}, the $i$th call to
  \algoName{algo:AnnihilatingPolynomial}{} fails with probability at
  most $6d_i^2/\card{S}$; hence, the probability that we successfully
  exit the first for loop is at least $1- 6n^2/\card{S}$.  Then, by
  \cref{prop:ModularCompositionModuloPower}, the $i$th call to
  \algoName{algo:ModularCompositionModuloPower}{} fails with
  probability at most $(2n_i^4+12n_i^2)/\card{S}$, so the probability
  that we successfully exit the second for loop is at least
  $1-(2n^4+12n^2)/\card{S}$. Altogether this gives a failure
  probability of at most $(2n^4+ 18n^2)/\card{S}$.
\end{proof}


\section{Applications}
\label{sec:applications}

We now list several variants of the modular composition problem and related
ones and sketch how the algorithms presented above can improve the best known
complexity.

\subsection{Annihilating polynomials}
\label{sec:minpoly}

\subsubsection{Annihilating polynomial}

A by-product of \algoName{algo:ModularComposition} is a Las Vegas
algorithm that takes~$\softO{n^\kappa}$ ($\kappa$ from \cref{thm:intro}) arithmetic operations for computing an
annihilating polynomial for $a$ of degree at most $4n$.  

Indeed, with the notation of the algorithm, for all $1 \leq i \leq s$,
since $\chi_i(a_i)\equiv0 \bmod f_i$ we have $\chi_i(a)= r_i f_i$ for
some $r_i\in \field[x]$.  Hence $\prod _{i=1}^s \chi _i$ is an
annihilating polynomial for $a$ modulo $f=\prod _{i=1}^s f_i$, whose
degree is at most $4\sum _{i=1}^s n_i = 4n$.

\subsubsection{Minimal and characteristic polynomial}
In general, our knowledge of the minimal and characteristic polynomial depends on whether we have a certified 
basis of relations. 

\begin{proposition}
  \label{cor:smithform}
  Let $R \in \ymatRing{m}{m}_{\le 2d}$ be the matrix produced by
  \algoName{algo:CandidateBasis}. If $R$ is a basis of $\rmodfa$, then
  the first $m$
  invariant factors of $\charmat$, hence in particular the minimal polynomial
  $\minpoly \in \yRing$ of \(a\) modulo \(f\), can be computed in
  $\softO{m^{\omega}d}$ operations in \(\field\). If furthermore \Cert{} is
  returned (implying that $R$ is a basis of \(\rmodfa\)), then the product of
  these invariant factors gives the characteristic polynomial $\charpoly \in
  \yRing$ of~\(a\) modulo \(f\). 
\end{proposition}
\begin{proof}
  If $R$ is a basis of $\rmod$, \Cref{prop:invariant-factors} shows that the
  Hermite normal form of~$R$ is a triangular basis of~$\rmod$ whose diagonal
  entries are the first invariant factors $\sigma _1, \ldots, \sigma _m$
  of~$\charmat$; in particular $\mu_a =\sigma _1$. If \Cert{} is returned, then
  $R$ is a basis of \(\rmod\) and $\dd =n$ (\cref{prop:compute_Mmm}). Hence
  $\degdet R=n$ and all the invariant factors are known; the characteristic
  polynomial is their product. The Hermite normal form of $R$ can be computed
  in $\softO{m^{\omega}d}$ operations~\cite[Thm.\,1.2]{LVZ17}.   
\end{proof}

One case of certification of the minimal polynomial is when \Cert{} is
returned by \algoName{algo:CandidateBasis}, which occurs in particular
for any $f$ in $\xRing$ with $f(0)\neq 0$ and a generic $a$ in
$\xRing_{< n}$ (see \cref{notes:genregularity}). Using
\cref{prop:compute_Mmm} and a shift as in \cref{rmk:fat0}, this
establishes the complexity bound $\softO{n^\kappa}$ for computing a
basis of relations and the minimal polynomial in the case of a
generic~$a \in \xRing_{<n}$.

Under the assumptions of \cref{prop:composition-separable} with the
additional hypothesis $\ddfa=n$ for $m=\lceil n^{\eta} \rceil$, a call
to \algoName{algo:CandidateBasis} instead of a call to
\algoName{algo:MatrixOfRelations} in
\algoName{algo:ModularCompositionBaseCase}, leads to a {\em certified} basis
of relations of $\rmodma[m]$ with good probability (use
\cref{prop:compute_Mmm} instead of \cref{prop:certificate} in the
proof of \cref{prop:composition-separable}).  From 
\cref{cor:smithform}, this also allows one to compute and certify the
minimal and characteristic polynomials in time $\softO{n^\kappa}$ when
$f$ is separable and $\ddfa[\lceil n^{\eta} \rceil]=n$.

The latter can be extended to the case $f$ irreducible and separable
since then the minimal polynomial~$\mu_a$ must be irreducible as well,
and therefore $\charmat$ has~$r$ nontrivial invariant factors all
equal to $\mu_a$.  If for $m=\lceil n^{\eta} \rceil$ the minimal
polynomial satisfies $\delta = \deg(\mu_a) \geq n/m$, then $r\leq m$
and $\ddfa=n$, hence the above certification when~$f$ is separable
leads to the minimal polynomial.  The low degree case $\delta < n/m$
can be treated directly using \cref{lemma:smallminpoly}, allowing to
compute $\mu_a$ in time $ \softO{n\delta^{(\omega_2/2) -1}}$, which is
$\softO{n^\kappa}$ since~$\delta < \lceil n^{1- \eta} \rceil$.

However, a matrix $R$ returned by \algoName{algo:CandidateBasis} might
not be a basis of \(\rmodma[m]\): without an efficient certification of
this property, \cref{cor:smithform} only gives a minimal polynomial
algorithm of the Monte Carlo kind.  Proceeding as done above, with a
call to \algoName{algo:CandidateBasis} instead of a call to
\algoName{algo:MatrixOfRelations} in
\algoName{algo:ModularCompositionBaseCase}, a Monte Carlo minimal
polynomial algorithm in $\softO{n^{\kappa}}$ can be derived under the
assumptions of \cref{prop:composition-separable}.

\subsection{Power series reversion and power series equations}
In this subsection, the characteristic of~$\field$ is~0.

For a given $a \in \xRing$ with $a(0)=0$ and $a'(0)\neq0$, power series
reversion (or functional inversion) asks for a power series~$g\in \field[[x]]$
such that 
\[
  a(\polp)=\polp(a)\equiv x\bmod x^n.
\]
By Newton's iteration, a composition algorithm in~$\softO{n^c}$
operations for
some $c>1$ induces a reversion algorithm in~$\softO{n^c}$ operations as
well~\cite{BK78}. Thus, we get a Las Vegas algorithm for power series reversion
in~$\softO{n^\kappa}$ operations in \(\field\).
Note that the
converse
reduction, from reversion to composition, also holds in this
situation~\cite{BK78}.

The approach for reversion extends partially to the resolution of a class of
power series equations. The aim is to solve an equation
\begin{equation}
  \label{eq:power-series-eq}
  \polp(x,y)=b\bmod x^n
\end{equation}
for $y\in\field[[x]]_{<n}$, when~$\polp\in\field[[x]][y]$ satisfies
$\polp(0,0)=b(0)$ and its partial derivative with respect to~$y$
is not~0 at~$(0,0)$.

By \cref{thm:CompositionModuloInseparable}, 
\algoName{algo:CompositionModuloInseparable}
computes
a
composition $\polp(x,a)$ in~$\softO{n^\kappa}$ operations for~$\polp$ in
$\xyRing_{<(m,n)}$ with $m=O(n^\eta)$ and $\eta$ from
\cref{eq:def-beta}. Together with Newton's iteration, this gives a Las
Vegas algorithm solving
\cref{eq:power-series-eq} in~$\softO{n^\kappa}$ operations
for $\polp\in\xyRing_{<(n^\eta,n)}$. Reversion is the special case with~$b=x$
and $\deg_x(\polp)=0$.

\subsubsection*{Note.}

It is known that the complexity of composition of power series (in
terms of nonscalar operations) is essentially that of
computing the coefficient of $x^{n-1}$ of $\polp(a)$~\cite{Ritzmann1986}.
By contrast, computing the coefficient of $x^{n-1}$ in the
reverse of~$a$ costs
only~$\softO{n}$ arithmetic operations~\cite{BK78}.

\subsection{Bivariate composition}
In this subsection, the characteristic of~$\field$ is~0.

Brent and Kung gave an algorithm that computes 
\[
  \polp(a,b)\rem x^n
\]
for $\polp\in\xyRing_{<(n,n)}$ and truncated power series $a,b \in \field[[x]]$
in only~$\softO{n^2}$ operations~\cite{BrentKung1977}. This is quasi-optimal,
since the number of coefficients of~$\polp$ is $\Theta(n^2)$ in general.  In
the simple situation where $a(0)=0$ and~$a'(0)=1$, the algorithm is as follows:
\begin{enumerate}
  \item by power series reversion, compute~$s(x)$ such that $a(s)=s(a)\equiv x\bmod x^n$;
  \item by univariate composition, compute $c=b(s) \rem x^n$;
  \item\label{biv:step3} by uni-bivariate composition, compute $d=\polp(x,c)\rem x^n$;
  \item by univariate composition, compute $d(a)\rem x^n$.
\end{enumerate}
The complexity is dominated by the uni-bivariate composition in Step~
(\ref{biv:step3}), which can be performed by Horner evaluation in~$\softO{n^2}$
operations. 

We obtain a Las Vegas algorithm with a complexity reduced
to~$\softO{n^\kappa}$ when $\polp\in\xyRing_{<(n^\eta,n)}$, where the uni-bivariate composition is done
in~$\softO {n^\kappa}$ as discussed in the case of power series equations, and
all the other steps are univariate compositions that are also performed
in~$\softO{n^\kappa}$ by our algorithm.

This method extends to the computation of 
\[
  \polp(a,b)\rem f
\]
with $f$ of degree \(n\) in \(\xRing\), and $a,b$ in $\xRing_{<n}$. The
algorithm becomes
\begin{enumerate}
  \item compute an annihilating polynomial~$\chi$ of $a$ modulo $f$;
  \item by inverse modular composition, compute $c$ such that $c(a)\equiv b\bmod f$;
  \item by uni-bivariate composition, compute $d=\polp(x,c)\rem\chi$;
  \item by univariate composition, compute $d(a)\rem f$.
\end{enumerate}
At least for generic~$a$, this is again a Las Vegas algorithm
in~$\softO{n^\kappa}$ operations when
$\polp\in\xyRing_{<(n^\eta,n)}$.

\section*{Acknowledgements}
We thank  an anonymous reviewer for an extremely detailed list of comments and very useful suggestions that helped us improve the quality of the manuscript.


\newcommand{\Hoeven}{\relax}\newcommand{\Gathen}{\relax}

\end{document}